\documentclass[12pt]{article}
\usepackage{amssymb}
\usepackage{amsfonts}
\usepackage{amsmath}
\usepackage{eurosym}
\usepackage{makeidx}
\usepackage{amssymb}
\usepackage{times}
\usepackage{anysize}
\usepackage{cite}
\usepackage{lineno}

\marginsize{2cm}{2cm}{1cm}{1cm}
\newtheorem{theorem}{Theorem}[section]

\newtheorem{corollary}[theorem]{Corollary}

\newtheorem{definition}[theorem]{Definition}

\newtheorem{lemma}[theorem]{Lemma}
\newtheorem{notation}[theorem]{Notation}

\newtheorem{proposition}[theorem]{Proposition}
\newtheorem{remark}[theorem]{Remark}

\newenvironment{proof}[1][Proof]{\noindent\textbf{#1.} }{\ \rule{0.5em}{0.5em}}
\input{tcilatex}
\begin{document}

\title{Quantum Dynamics Generated by Long-Range Interactions for Lattice
Fermions and Quantum Spins}
\author{J.-B. Bru \and W. de Siqueira Pedra}
\date{\today }
\maketitle

\begin{abstract}
\bigskip We study the macroscopic dynamics of fermion and quantum-spin
systems with long-range, or mean-field, interactions, which turns out to be
equivalent to an intricate combination of classical and short-range quantum
dynamics. In this paper we focus on the \emph{quantum} part of the
long-range macroscopic dynamics. The classical part is studied in a
companion paper. Altogether, the results obtained are far beyond previous
ones and required the development of a suitable mathematical framework. The
entanglement of classical and quantum worlds is noteworthy, opening new
theoretical perspectives, and is shown here to be a consequence of the
highly non-local character of long-range, or mean-field,
interactions.\bigskip

\noindent \textit{Dedicated to V.A. Zagrebnov for his important
contributions to the mathematics of quantum many-body theory.}\bigskip
\bigskip

 \bigskip

\end{abstract}

\setcounter{tocdepth}{2}
\tableofcontents%

\section{Introduction}

Following \cite{BruPedra-MFII}, we pursue our study on macroscopic dynamics
of fermion and quantum-spin systems with long-range, or mean-field,
interactions. In \cite{BruPedra-MFII} we only focus on the (effective)
classical part of this dynamics. In the current paper we study its
(effective short-range) quantum part in detail. The results obtained are far
beyond previous ones because the permutation-invariance of lattice-fermion
or quantum-spin systems is \emph{not} required:

\begin{itemize}
\item The short-range part of the corresponding Hamiltonian is very general
since only a sufficiently strong polynomial decay of its interactions and
translation invariance are necessary.

\item The long-range part is also very general, being an infinite sum (over $%
n$) of mean-field terms of order $n\in \mathbb{N}$. In fact, even for
permutation-invariant systems, the class of long-range, or mean-field,
interactions we are able to handle is much larger than what was previously
studied.

\item The initial state is only required to be periodic. By \cite[%
Proposition 2.2]{BruPedra-MFII}, observe that the set of all such initial
states is weak$^{\ast }$-dense within the set of all even states, the
physically relevant ones.
\end{itemize}

\noindent For an exhaustive historical discussion on fermion or quantum-spin
systems with long-range, or mean-field, interactions we refer to \cite[%
Section 1]{BruPedra-MFII}. Here, we add several observations concerning the
physical relevance of long-range interactions in Physics.

The most general form of a translation-invariant model for fermions (with
spin set $\mathrm{S}$)\ in a cubic box $\Lambda _{L}\doteq \{\mathbb{Z}\cap
\lbrack -L,L]\}^{3}$ of volume $|\Lambda _{L}|$, $L\in \mathbb{N}$, with a
quartic (in the fermionic fields) gauge- and translation-invariant
interaction is given in momentum space by 
\begin{eqnarray}
H &=&\underset{k\in \Lambda _{L}^{\ast },\ \mathrm{s}\in \mathrm{S}}{\sum }%
\left( \varepsilon _{k}-\mu \right) \tilde{a}_{k}^{\ast }\tilde{a}_{k} 
\notag \\
&&+\frac{1}{\left\vert \Lambda _{L}\right\vert }\underset{\mathrm{s}_{1},%
\mathrm{s}_{2},\mathrm{s}_{3},\mathrm{s}_{4}\in \mathrm{S}}{\underset{%
k,k^{\prime },q\in \Lambda _{L}^{\ast }}{\sum }}g_{\mathrm{s}_{1},\mathrm{s}%
_{2},\mathrm{s}_{3},\mathrm{s}_{4}}\left( k,k^{\prime },q\right) \tilde{a}%
_{k+q,\mathrm{s}_{1}}^{\ast }\tilde{a}_{k^{\prime }-q,\mathrm{s}_{2}}^{\ast }%
\tilde{a}_{k^{\prime },\mathrm{s}_{3}}\tilde{a}_{k,\mathrm{s}_{4}}\ .
\label{hamil general}
\end{eqnarray}%
See, for instance, \cite[Eq. (2.1)]{Metzner}. Here, $\Lambda _{L}^{\ast }$
is the reciprocal lattice of quasi-momenta (periodic boundary conditions)
associated with $\Lambda _{L}$ and the operator $\tilde{a}_{k,\mathrm{s}%
}^{\ast }$ (respectively $\tilde{a}_{k,\mathrm{s}}$) creates (respectively
annihilates) a fermion with spin $\mathrm{s}\in \mathrm{S}$ and (quasi-)
momentum $k\in \Lambda _{L}^{\ast }$. The function $\varepsilon _{k}$
represents the kinetic energy of a fermion with (quasi-) momentum $k$ and
the real number $\mu $ is the chemical potential. The last term of (\ref%
{hamil general}) corresponds to a general translation-invariant two-body
interaction written in the (quasi-) momentum space.

One important example of a lattice-fermion system with a long-range
interaction is given in the scope of the celebrated BCS theory -- proposed
in the late 1950s (1957) to explain conventional type I superconductors. The
lattice version of this theory is obtained from (\ref{hamil general}) by
taking $\mathrm{S}\doteq \{\uparrow ,\downarrow \}$ and imposing 
\begin{equation*}
g_{\mathrm{s}_{1},\mathrm{s}_{2},\mathrm{s}_{3},\mathrm{s}_{4}}\left(
k,k^{\prime },q\right) =\delta _{k,-k^{\prime }}\delta _{\mathrm{s}%
_{1},\uparrow }\delta _{\mathrm{s}_{2},\downarrow }\delta _{\mathrm{s}%
_{3},\downarrow }\delta _{\mathrm{s}_{4},\uparrow }f\left( k,-k,q\right)
\end{equation*}%
for some function $f$: It corresponds to the so-called (reduced) BCS\
Hamiltonian%
\begin{equation}
\mathrm{H}_{\Lambda }^{BCS}\doteq \sum\limits_{k\in \Lambda _{L}^{\ast
}}\left( \varepsilon _{k}-\mu \right) \left( \tilde{a}_{k,\uparrow }^{\ast }%
\tilde{a}_{k,\uparrow }+\tilde{a}_{k,\downarrow }^{\ast }\tilde{a}%
_{k,\downarrow }\right) -\frac{1}{\left\vert \Lambda _{L}\right\vert }%
\sum_{k,q\in \Lambda _{L}^{\ast }}\gamma _{k,q}\tilde{a}_{k,\uparrow }^{\ast
}\tilde{a}_{-k,\downarrow }^{\ast }\tilde{a}_{-q,\downarrow }\tilde{a}%
_{q,\uparrow }\ ,  \label{Hamiltonian BCS}
\end{equation}%
where $\gamma _{k,q}$ is a positive\footnote{%
The positivity of $\gamma _{k,q}$ imposes constraints on the choice of the
function $f$.} function. Because of the term $\delta _{k,-k^{\prime }}$, the
interaction of this model has a long-range character, in position space. The
simple choice $\gamma _{k,q}\doteq \gamma >0$ in (\ref{Hamiltonian BCS}) is
still very interesting since, even when $\varepsilon _{k}=0$, the BCS\
Hamiltonian qualitatively displays most of basic properties of real
conventional type I superconductors. See, e.g., \cite[Chapter VII, Section 4]%
{Thou}. Written in the $x$-space, the BCS interaction in (\ref{Hamiltonian
BCS}) is, in this case, equal to 
\begin{equation}
-\frac{\gamma }{\left\vert \Lambda _{L}\right\vert }\sum_{k,q\in \Lambda
_{L}^{\ast }}\tilde{a}_{k,\uparrow }^{\ast }\tilde{a}_{-k,\downarrow }^{\ast
}\tilde{a}_{-q,\downarrow }\tilde{a}_{q,\uparrow }=-\frac{\gamma }{%
\left\vert \Lambda _{L}\right\vert }\sum_{x,y\in \Lambda _{L}}a_{x,\uparrow
}^{\ast }a_{x,\downarrow }^{\ast }a_{y,\downarrow }a_{y,\uparrow }\ ,
\label{long range interactions}
\end{equation}%
the operators $a_{x,\mathrm{s}}^{\ast }$,$a_{x,\mathrm{s}}$ being
respectively the creation and annihilation operators of a fermion with spin $%
\mathrm{s}\in \{\uparrow ,\downarrow \}$ at lattice site $x\in \Lambda _{L}$%
. The right-hand side of the equality explicitly shows the long-range
character of the interaction. It is a mean-field interaction since%
\begin{equation*}
\frac{1}{\left\vert \Lambda _{L}\right\vert }\sum_{x,y\in \Lambda
_{L}}a_{x,\uparrow }^{\ast }a_{x,\downarrow }^{\ast }a_{y,\downarrow
}a_{y,\uparrow }=\sum_{y\in \Lambda _{L}}\left( \frac{1}{\left\vert \Lambda
_{L}\right\vert }\sum_{x\in \Lambda _{L}}a_{x,\uparrow }^{\ast
}a_{x,\downarrow }^{\ast }\right) a_{y,\downarrow }a_{y,\uparrow }\ .
\end{equation*}%
The (reduced) BCS\ Hamiltonian with $\gamma _{k,q}\doteq \gamma >0$ is an
important, albeit very elementary, example of the far more general case
treated in this paper.

Long-range, or mean-field, effective models are essential in condensed
matter physics to study, from microscopic considerations, macroscopic
phenomena like superconductivity. They come from different approximations or
Ans\"{a}tze, like the choice $\gamma _{k,q}\doteq \gamma >0$. The general
form of the (effective) BCS interaction in (\ref{Hamiltonian BCS}) comes
from the celebrated Fr\"{o}hlich electron-phonon interactions. What's more,
they are possibly not merely effective interactions.

Such models capture surprisingly well many phenomena in condensed matter
physics. For instance, recall that the BCS interaction (\ref{long range
interactions}) allows us to qualitatively display most of basic properties
of conventional superconductors \cite[Chapter VII, Section 4]{Thou}. Ergo,
one could wonder whether such interactions may have a more fundamental
physical relevance. Such a question is usually not addressed, because these
interactions seem to break the spacial locality of Einstein's relativity.
For instance, the BCS interaction (\ref{long range interactions}) can be
seen as a kinetic term for fermion pairs that can hop from $y\in \Lambda
_{L} $ to \emph{any} other lattice site $x\in \Lambda _{L}$, for each $L\in 
\mathbb{N}$.

This non-locality property is reminiscent of the inherent non-locality of
quantum mechanics, highlighted by Einstein, Podolsky and Rosen\ with the
celebrated EPR\ paradox. Philosophically, this general issue challenges
causality, in its local sense, as well as the notion of a material object%
\footnote{%
According to the spatio-temporal identity of classical mechanics, the same
physical object cannot be at the same time on two distinct points of the
phase space. This refers to Leibniz's Principle of Identity of
Indiscernibles \cite[p. 1]{FK}. The spatio-temporal identity of classical
mechanics is questionable in quantum mechanics. See, e.g., \cite{FK}.}. In 
\cite{Einstein1}, Einstein says the following: \bigskip

\textquotedblleft \textit{If one asks what, irrespective of quantum
mechanics, is characteristic of the world of ideas of physics, one is first
of all struck by the following: the concepts of physics relate to a real
outside world... it is further characteristic of these physical objects that
they are thought of as a range in a space-time continuum. An essential
aspect of this arrangement of things in physics is that they lay claim, at a
certain time, to an existence independent of one another, provided these
objects \textquotedblleft are situated in different parts of
space\textquotedblright . }

\textit{The following idea characterizes the relative independence of
objects far apart in space (A and B): external influence on A has no direct
influence on B... }

\textit{There seems to me no doubt that those physicists who regard the
descriptive methods of quantum mechanics as definitive in principle would
react to this line of thought in the following way: they would drop the
requirement... for the independent existence of the physical reality present
in different parts of space; they would be justified in pointing out that
the quantum theory nowhere makes explicit use of this requirement. }

\textit{I admit this, but would point out: when I consider the physical
phenomena known to me, and especially those who are being so successfully
encompassed by quantum mechanics, I still cannot find any fact anywhere
which would make it appear likely that (that) requirement will have to be
abandoned. }

\textit{I am therefore inclined to believe that the description of quantum
mechanics... has to be regarded as an incomplete and indirect description of
reality, to be replaced at some later date by a more complete and direct
one.\textquotedblright }\bigskip

The debate on non-locality in Physics, experimentally shown, refers to the
existence of quantum entanglement, used in quantum information theory. For a
discussion on locality and realism in quantum mechanics, see, e.g., \cite%
{Aspect} by A. Aspect, who is one of the main initiators of experimental
studies on quantum entanglement, in the beginning of the 1980s.

The non-locality of long-range, or mean-field, interactions like the BCS
interaction\footnote{%
The strength of the BCS interaction (\ref{long range interactions}) between
two points of the space does not decay at large distances.} (\ref{long range
interactions}) can be seen as an instance of the (controversial) intrinsic
non-locality of quantum physics. Mean field interactions are thus usually
not considered by the physics community as being fundamental interactions,
in order to avoid polemics. We partially agree with this position and see
long-range interactions as possibly resulting from (more fundamental)
interactions with (bosonic) mediators, like phonons in conventional
superconductivity.

Nonetheless, a long-range interaction like (\ref{long range interactions}),
being quantum mechanical, does not refer to an actuality (in Aristotle's
sense), \emph{but only to a potentiality}. Physical properties of any
(energy-conserving) physical system do not just depend on its Hamiltonian
but also on its state which accounts for the \textquotedblleft
environmental\textquotedblright\ part of the system: This situation is
analogous to the epigenetics\footnote{%
Quoting \cite{Epigenetics}: \textquotedblleft \textit{Epigenetics is
typically defined as the study of heritable changes in gene expression that
are not due to changes in DNA sequence. Diverse biological properties can be
affected by epigenetic mechanisms: for example, the morphology of flowers
and eye colour in fruitflies.}\textquotedblright} showing that the DNA
sequence\ is only a set of constraints and potentialities, the physical
realizations of which depend on the history and environment of the
corresponding organism. For instance, in a lattice-fermion system described
by the so-called (reduced) BCS\ Hamiltonian with $\gamma _{k,q}=\gamma $,
pairs of particles may (almost) never hop in arbitrarily large distances if
the state\footnote{%
I.e., a positive and normalized continuous functional on the CAR algebra.} $%
\rho $ of the corresponding system is such that 
\begin{equation*}
\lim_{L\rightarrow \infty }\frac{1}{\left\vert \Lambda _{L}\right\vert }%
\sum_{x,y\in \Lambda _{L}}\rho \left( a_{x,\uparrow }^{\ast }a_{x,\downarrow
}^{\ast }a_{y,\downarrow }a_{y,\uparrow }\right) =0\ .
\end{equation*}%
This is the case for equilibrium states of this model at sufficiently high
temperatures. It is thus too reductive to a priori eliminate such
interactions from \textquotedblleft fundamental\textquotedblright\
Hamiltonians of physical systems.

On the top of that, as is well-known, the thermodynamic limit of mean-field
dynamics is repres%
\-%
entation-%
\-%
dependent. This is basically Haag's original argument proposed in 1962 \cite%
{haag62} for the BCS model. In fact, the description of the full dynamics
requires an extended quantum framework \cite{Bru-pedra-MF-I}, which is an
intricate combination of classical and quantum dynamics, as observed by B%
\'{o}na already thirty years ago \cite{Bona91}. The paper \cite%
{Bru-pedra-MF-I} shows the emergence of classical mechanics defined from
Poisson brackets on state spaces without necessarily a disappearance of the
quantum world, offering a general formal mathematical framework to
understand physical phenomena with macroscopic quantum coherence. In the
context of lattice-fermion systems, it is explained in detail in \cite%
{BruPedra-MFII}. Such an entanglement of classical and quantum worlds is
noteworthy, opening new theoretical perspectives, and is here a direct
consequence of the highly non-local character of long-range, or mean-field,
interactions.

In the present paper, as already stressed, we derive the (effective
short-range) quantum part of the dynamics of fermion systems driven by
mean-field interactions. It is done in the following way:

\begin{itemize}
\item As explained in \cite{BruPedra-MFII}, a classical dynamics can be
defined by using the whole state space $E$ of the CAR algebra $\mathcal{U}$,
but some space homogeneity (periodicity) of the fermion system is crucial in
order to make sense of the (effective short-range) quantum part of the* full
long-range dynamics in the thermodynamic limit. We thus consider as initial
state some $\vec{\ell}$-periodic state $\rho \in E_{\vec{\ell}}$, where $E_{%
\vec{\ell}}\subseteq E$ denotes the space of $\vec{\ell}$-periodic states
for some fixed ($d$-dimensional, $d\in \mathbb{N}$) vector $\vec{\ell}\in 
\mathbb{N}^{d}$.

\item If the initial state $\rho \in E_{\vec{\ell}}$ is an extreme point of $%
E_{\vec{\ell}}$, then we are able to derive the thermodynamic limit of the
quantum part of the dynamics within the GNS representation\footnote{$%
\mathcal{H}_{\rho }$ is an Hilbert space, $\pi _{\rho }:\mathcal{U}%
\rightarrow \mathcal{B}\left( \mathcal{H}_{\rho }\right) $ a so-called
representation of $\mathcal{U}$, while $\Omega _{\rho }$ is a cyclic vector
with respect to $\pi _{\rho }(\mathcal{U})$, i.e., $\mathcal{H}_{\rho }$ is
the closure of (the linear span of) the set $\left\{ \pi (A)\Omega _{\rho
}:A\in \mathcal{U}\right\} $. See, e.g., \cite[Section 2.3.3]%
{BrattelliRobinsonI}.} $(\mathcal{H}_{\rho },\pi _{\rho },\Omega _{\rho })$
of the $C^{\ast }$-algebra $\mathcal{U}$, associated with the $\vec{\ell}$%
-periodic state $\rho $. This refers to Theorem \ref{theorem structure of
omega copy(2)}, which is proven by using two pivotal ingredients:\ the
well-known ergodicity of any extreme point of $E_{\vec{\ell}}$ \cite[Theorem
1.16]{BruPedra2} and Lieb-Robinson bounds \cite[Section 4.3]{brupedraLR}.

\item The next issue is an extension of this result to all\ $\vec{\ell}$%
-periodic initial states. To this end, it is natural to consider the Choquet
theorem \cite[Theorem 10.18]{BruPedra2}, applied to the metrizable and weak$%
^{\ast }$-compact convex set $E_{\vec{\ell}}$: Each state $\rho \in E_{\vec{%
\ell}}$ is the barycenter of a unique probability measure $\mu _{\rho }$
which is supported on the set $\mathcal{E}(E_{\vec{\ell}})\subseteq E_{\vec{%
\ell}}$ of extreme $\vec{\ell}$-periodic states:%
\begin{equation*}
\rho \left( A\right) \doteq \int_{\mathcal{E}(E_{\vec{\ell}})}\hat{\rho}%
\left( A\right) \mu _{\rho }\left( \mathrm{d}\hat{\rho}\right) \text{ }%
,\qquad A\in \mathcal{U}\ .
\end{equation*}

\item By Theorem \ref{theorem structure of omega copy(2)}, the full
long-range dynamics restricted to any extreme state of $E_{\vec{\ell}}$ is
effectively represented by a short-range quantum dynamics, explicitly
depending on a state that evolves according to the (effective) classical
part of the dynamics. In particular, the effective short-range quantum
dynamics is non-autonomous. In this way, the full long-range dynamics
emerges as an intricate combination of classical and short-range quantum
dynamics. This fact leads us to consider the unital $C^{\ast }$-algebra $%
C(E_{\vec{\ell}};\mathcal{U})$ of continuous functions from $E_{\vec{\ell}}$
to $\mathcal{U}$ as well as the extension of elements of $E_{\vec{\ell}}$ to
states on $C(E_{\vec{\ell}};\mathcal{U})$, via the definition%
\begin{equation*}
\rho \left( f\right) \doteq \int_{\mathcal{E}(E_{\vec{\ell}})}\hat{\rho}%
\left( f\left( \hat{\rho}\right) \right) \mu _{\rho }\left( \mathrm{d}\hat{%
\rho}\right) \text{ },\qquad f\in C\left( E_{\vec{\ell}};\mathcal{U}\right)
\ ,
\end{equation*}%
for any state $\rho \in E_{\vec{\ell}}$.\ This prescription is based on a
(non-commutative) conditional expectation naturally appearing in our
setting. Observe that the function space $C(E_{\vec{\ell}};\mathbb{C})$ is
the one on which the classical dynamics considered in \cite{BruPedra-MFII}
runs.

\item It is then tempting to use the direct integral 
\begin{equation}
\left( \mathcal{H}_{\rho }^{\oplus }\equiv \int_{\mathcal{E}(E_{\vec{\ell}})}%
\mathcal{H}_{\hat{\rho}}\mu (\mathrm{d}\hat{\rho}),\ \pi _{\rho }^{\oplus
}\equiv \int_{\mathcal{E}(E_{\vec{\ell}})}\pi _{\hat{\rho}}\mu (\mathrm{d}%
\hat{\rho}),\ \Omega _{\rho }^{\oplus }\equiv \int_{\mathcal{E}(E_{\vec{\ell}%
})}\Omega _{\hat{\rho}}\mu (\mathrm{d}\hat{\rho})\right)
\label{direct integral representation}
\end{equation}%
of the GNS representations $(\mathcal{H}_{\hat{\rho}},\pi _{\hat{\rho}%
},\Omega _{\hat{\rho}})$ of $\mathcal{U}$ associated with the extreme states 
$\hat{\rho}\in \mathcal{E}(E_{\vec{\ell}})$, along with the (effective
short-range) infinite-volume quantum dynamics in each fiber of the direct
integral, as given by Theorem \ref{theorem structure of omega copy(2)}.

\item Constant functions of $C(E_{\vec{\ell}};\mathcal{U})$ are canonically
seen here as elements of $\mathcal{U}$, i.e., $\mathcal{U}\subseteq C(E_{%
\vec{\ell}};\mathcal{U})$. Note also that 
\begin{equation*}
\pi _{\rho }^{\oplus }:C\left( E_{\vec{\ell}};\mathcal{U}\right) \rightarrow 
\mathcal{B}\left( \mathcal{H}_{\rho }^{\oplus }\right) .
\end{equation*}%
The triplet $\left( \mathcal{H}_{\rho }^{\oplus },\pi _{\rho }^{\oplus }|_{%
\mathcal{U}},\Omega _{\rho }^{\oplus }\right) $ should give a cyclic
representation of the $C^{\ast }$-algebra $\mathcal{U}$, associated with the
generally non-extreme state $\rho \in E_{\vec{\ell}}$. This property is a
consequence of Theorem \ref{coro Extension of states} whenever the
probability measure $\mu _{\rho }$ is orthogonal. This fact is reminiscent
of the Tomita theorem \cite[Proposition 4.1.22]{BrattelliRobinsonI} and
refers to a sufficient condition for the cyclicity of the direct-integral
representation.

\item Orthogonality of the ergodic decomposition $\mu _{\rho }$ of any $\vec{%
\ell}$-periodic state $\rho \in E_{\vec{\ell}}$ is therefore pivotal and
thus proven in Theorem \ref{theorem choquet}, using that $\vec{\ell}$%
-periodic states of the full CAR algebra $\mathcal{U}$ are in one-to-one
correspondence with $\vec{\ell}$-periodic states of the subalgebra $\mathcal{%
U}^{+}$ of even elements of $\mathcal{U}$. This identification is pivotal
because, in contrast with $\mathcal{U}$, $\mathcal{U}^{+}$ is asymptotically
abelian and this case is well-known. See, for instance, \cite[Propositions
4.3.3 and 4.3.7]{BrattelliRobinsonI}.

\item Theorems \ref{theorem choquet} and \ref{coro Extension of states}
yield Proposition \ref{coro Extension of states copy(1)} and the direct
integral of GNS representations given by (\ref{direct integral
representation}) can in fact be used also as a good cyclic representation of
the $C^{\ast }$-algebra $C(E_{\vec{\ell}};\mathcal{U})$, associated with any 
$\vec{\ell}$-periodic (initial) state (seen as a state on $C(E_{\vec{\ell}};%
\mathcal{U})$, by the above prescription). Applying Theorem \ref{theorem
structure of omega copy(2)} to each fiber of the direct integral, which
corresponds to an extreme $\vec{\ell}$-periodic state, we get -- by
integration -- the (effective short-range) quantum part of the full
long-range dynamics, yielding the final results, i.e., Theorem \ref{theorem
structure of omega} and Corollary \ref{theorem structure of omega copy(1)}.
\end{itemize}

\noindent To conclude, our main results are Proposition \ref{coro Extension
of states copy(1)}, Theorems \ref{theorem structure of omega}, \ref{theorem
choquet}, \ref{theorem structure of omega copy(2)}, \ref{coro Extension of
states} and Corollary \ref{theorem structure of omega copy(1)}.

The paper is organized as follows: Section \ref{Section FERMI0} explains the
algebraic formulation of lattice-fermion systems with short- and long-range
interactions and we make explicit the problem of the thermodynamic limit\ of
their associated dynamics. Like in \cite{BruPedra2,BruPedra-MFII}, note that
we prefer to use, from now on, the term \textquotedblleft
long-range\textquotedblright\ instead of \textquotedblleft
mean-field\textquotedblright , since the latter can refer to different
scalings. Our description of long-range dynamics requires the mathematical
framework of \cite{Bru-pedra-MF-I}, which is thus presented in Section \ref%
{State-Dependant Interactions}. Section \ref{Long-Range Dynamics} describes
long-range dynamics, which are based on self-consistency equations \cite[%
Theorem 6.5]{BruPedra-MFII}. Observe that we shortly explain the (effective)
classical part of the long-range dynamics in Section \ref{Classical Part},
while its (effective short-range) quantum part is given in detail in Section %
\ref{Quantum Part}, which gathers the main results of the paper. All proofs
are postponed to Section \ref{dddd}. Finally, Section \ref{app direct
integrals} is a detailed appendix on the theory of direct integrals of
measurable families of Hilbert spaces, operators, von Neumann algebras and $%
C^{\ast }$-algebra representations. In this appendix we aim at making the
paper self-contained and a section of pedagogical interest for the
non-expert. It will additionally be useful in future applications of the
general results presented here to the study of KMS states of lattice-fermion
or quantum-spin systems with long-range interactions.

In this paper, we only focus on lattice-fermion systems which are, from a
technical point of view, slightly more difficult than quantum-spin systems,
because of a non-commutativity issue at different lattice sites. However,
all the results presented here hold true for quantum-spin systems, via
obvious modifications.

\begin{notation}
\label{remark constant}\mbox{
}\newline
\emph{(i)} A norm on a generic vector space $\mathcal{X}$ is denoted by $%
\Vert \cdot \Vert _{\mathcal{X}}$ and the identity mapping of $\mathcal{X}$
by $\mathbf{1}_{\mathcal{X}}$. The space of all bounded linear operators on $%
(\mathcal{X},\Vert \cdot \Vert _{\mathcal{X}}\mathcal{)}$ is denoted by $%
\mathcal{B}(\mathcal{X})$. The unit element of any algebra $\mathcal{X}$ is
denoted by $\mathfrak{1}$, provided it exists. The scalar product of any
Hilbert space $\mathcal{H}$ is denoted by $\langle \cdot ,\cdot \rangle _{%
\mathcal{H}}$. \newline
\emph{(ii)} For any topological space $\mathcal{X}$ and normed space $(%
\mathcal{Y},\Vert \cdot \Vert _{\mathcal{Y}}\mathcal{)}$, $C\left( \mathcal{X%
};\mathcal{Y}\right) $ denotes the space of continuous maps from $\mathcal{X}
$ to $\mathcal{Y}$. If $\mathcal{X}$ is a locally compact topological space,
then $C_{b}\left( \mathcal{X};\mathcal{Y}\right) $ denotes the Banach space
of bounded continuous mappings from $\mathcal{X}$ to $\mathcal{Y}$ along
with the topology of uniform convergence.\newline
\emph{(iii)} The notion of an automorphism depends on the structure of the
corresponding domain. In algebra, a ($\ast $-) automorphism acting on a $%
\ast $-algebra is a bijective $\ast $-homomorphism from this algebra to
itself. In topology, an automorphism acting on a topological space is a
self-homeomorphism, that is, a homeomorphism of the space to itself.
\end{notation}

\section{Algebraic Formulation of Lattice-Fermion Systems\label{Section
FERMI0}}

The mathematical framework used here is the same than in \cite{BruPedra-MFII}%
, including the notation. We thus give it in a concise way and refer to \cite%
{BruPedra-MFII} for more details.

\subsection{CAR Algebra of Lattices\label{Algebra of Lattices}}

\subsubsection{Background Lattice}

Let $\mathfrak{L}\doteq \mathbb{Z}^{d}$ for some fixed $d\in \mathbb{N}$ and 
$\mathcal{P}_{f}\subseteq 2^{\mathfrak{L}}$ be the set of all non-empty
finite subsets of $\mathfrak{L}$. In order to define the thermodynamic
limit, for simplicity, we use the cubic boxes 
\begin{equation}
\Lambda _{L}\doteq \{(x_{1},\ldots ,x_{d})\in \mathfrak{L}:|x_{1}|,\ldots
,|x_{d}|\leq L\}\subseteq \mathfrak{L}\ ,\qquad L\in \mathbb{N}\ ,
\label{eq:def lambda n}
\end{equation}%
as a so-called van Hove net. We also use a positive-valued symmetric
function $\mathbf{F}:\mathfrak{L}^{2}\rightarrow (0,1]$ with maximum value $%
\mathbf{F}\left( x,x\right) =1$ for all $x\in \mathfrak{L}$ and satisfying%
\begin{equation}
\left\Vert \mathbf{F}\right\Vert _{1,\mathfrak{L}}\doteq \underset{y\in 
\mathfrak{L}}{\sup }\sum_{x\in \mathfrak{L}}\mathbf{F}\left( x,y\right)
<\infty  \label{(3.1) NS}
\end{equation}%
and 
\begin{equation}
\mathbf{D}\doteq \underset{x,y\in \mathfrak{L}}{\sup }\sum_{z\in \mathfrak{L}%
}\frac{\mathbf{F}\left( x,z\right) \mathbf{F}\left( z,y\right) }{\mathbf{F}%
\left( x,y\right) }<\infty \ .  \label{(3.2) NS}
\end{equation}%
Explicit examples of such a function are given in \cite[Section 3.1]%
{BruPedra-MFII}.

\subsubsection{The CAR $C^{\ast }$-Algebra}

For any subset $\Lambda \subseteq \mathfrak{L}$, $\mathcal{U}_{\Lambda }$ is
the separable universal unital $C^{\ast }$-algebra generated by elements $%
\{a_{x,\mathrm{s}}\}_{x\in \Lambda ,\mathrm{s}\in \mathrm{S}}$ satisfying
the canonical anti-commutation relations (CAR), $\mathrm{S}$ being some
finite set of spins. Note that we use the notation $\mathcal{U\equiv U}_{%
\mathfrak{L}}$ and the subspace 
\begin{equation}
\mathcal{U}_{0}\doteq \bigcup_{\Lambda \in \mathcal{P}_{f}}\mathcal{U}%
_{\Lambda }  \label{simple}
\end{equation}%
is a dense $\ast $-algebra of $\mathcal{U}$. Elements of $\mathcal{U}_{0}$
are called local elements. The (real) Banach subspace of all self-adjoint
elements of $\mathcal{U}$ is denoted by $\mathcal{U}^{\mathbb{R}%
}\varsubsetneq \mathcal{U}$.

\subsubsection{Even Elements}

The condition 
\begin{equation}
\sigma (a_{x,\mathrm{s}})=-a_{x,\mathrm{s}},\qquad x\in \Lambda ,\ \mathrm{s}%
\in \mathrm{S}\ ,  \label{automorphism gauge invariance}
\end{equation}%
defines a unique $\ast $-automorphism $\sigma $ of the $C^{\ast }$-algebra $%
\mathcal{U}$. The subspace 
\begin{equation}
\mathcal{U}^{+}\doteq \{A\in \mathcal{U}:A=\sigma (A)\}\subseteq \mathcal{U}
\label{definition of even operators}
\end{equation}%
is the $C^{\ast }$-subalgebra of so-called even elements. $\mathcal{U}^{+}$
should be seen as more fundamental than $\mathcal{U}$ in Physics, because of
the local causality in quantum field theory.

Note that the fact that the local causality in quantum field theory can be
invoked to see $\mathcal{U}^{+}$ as being more fundamental than $\mathcal{U}$
in Physics does not prevent us from considering long-range interactions as
possibly fundamental interactions, as explained in\ the introduction. The
choice of $\mathcal{U}^{+}$ only compel us to consider (local) observables
satisfying the local causality as measurable physical quantities, the full
energy of lattice Fermi systems with short-range or long-range interactions
being generally inaccessible in infinite volume. In fact, the long-range
part yields possibly non-vanishing background fields, in the spirit of the
Higgs mechanism of quantum field theory, in a given representation of the
observable algebra, which is fixed by the initial state.

\subsubsection{Translation Automorphisms}

Translations refer to the group homomorphism $x\mapsto \alpha _{x}$ from $(%
\mathbb{Z}^{d},+)$ to the group of $\ast $-automorphisms of $\mathcal{U}$,
which is uniquely defined by the condition%
\begin{equation}
\alpha _{x}(a_{y,\mathrm{s}})=a_{y+x,\mathrm{s}}\ ,\quad y\in \mathfrak{L},\;%
\mathrm{s}\in \mathrm{S}\ .  \label{transl}
\end{equation}%
This group homomorphism is used below to define the notion of (space)
periodicity of states as well as the translation invariance of interactions
of lattice-fermion systems.

\subsection{State Space}

\subsubsection{Full State Space}

The state space associated with $\mathcal{U}$ is defined by%
\begin{equation}
E\doteq \{\rho \in \mathcal{U}^{\ast }:\rho \geq 0,\ \rho (\mathfrak{1}%
)=1\}\ .  \label{states CAR}
\end{equation}%
As explained in Section \ref{Positive Functionals}, $E$ is a metrizable and
weak$^{\ast }$-compact convex subset of the dual space $\mathcal{U}^{\ast }$%
. It is also the state space of the classical dynamics studied in \cite%
{Bru-pedra-MF-I,BruPedra-MFII}. By the Krein-Milman theorem \cite[Theorem
3.23]{Rudin}, $E$ is the weak$^{\ast }$-closure of the convex hull of the
(non-empty) set $\mathcal{E}(E)$ of its extreme points, which turns out to
be weak$^{\ast }$-dense \cite{Bru-pedra-MF-I,BruPedra-MFII}: 
\begin{equation}
E=\overline{\mathrm{co}\mathcal{E}\left( E\right) }=\overline{\mathcal{E}(E)}%
\ .  \label{density extreme states}
\end{equation}%
All state spaces we define below have this peculiar geometrical structure.

We define $C\left( E;E\right) $ to be the set of weak$^{\ast }$-continuous
functions from the state space $E$ to itself, endowed with the topology of
uniform convergence. In other words, any net $(f_{j})_{j\in J}\subseteq
C\left( E;E\right) $ converges to $f\in C\left( E;E\right) $ whenever%
\begin{equation}
\lim_{j\in J}\max_{\rho \in E}\left\vert f_{j}(\rho )(A)-f(\rho
)(A)\right\vert =0\ ,\qquad \text{for all }A\in \mathcal{U}\ .
\label{uniform convergence weak*}
\end{equation}%
We denote by $\mathrm{Aut}\left( E\right) \varsubsetneq C\left( E;E\right) $
the subspace of all automorphisms of $E$, i.e., element of $C\left(
E;E\right) $ with weak$^{\ast }$-continuous inverse. As is usual, we
identify constant functions on $\mathbb{R}$ with elements of the codomain of
such a function:%
\begin{equation}
E\subseteq C\left( \mathbb{R};E\right) \text{\qquad and\qquad }\mathrm{Aut}%
\left( E\right) \subseteq C\left( \mathbb{R};\mathrm{Aut}\left( E\right)
\right) \ .  \label{identify}
\end{equation}

\subsubsection{Even States}

The physically relevant set of states is the metrizable and weak$^{\ast }$%
-compact convex set of \emph{even} states defined by 
\begin{equation}
E^{+}\doteq \left\{ \rho \in E:\rho \circ \sigma =\rho \right\} \ ,
\label{gauge invariant states}
\end{equation}%
$\sigma $ being the unique $\ast $-automorphism of $\mathcal{U}$ satisfying (%
\ref{automorphism gauge invariance}). The set $\mathcal{E}(E^{+})$ of
extreme points of $E^{+}$ is also a weak$^{\ast }$-dense\emph{\ }subset of $%
E^{+}$: 
\begin{equation*}
E^{+}=\overline{\mathrm{co}\mathcal{E}\left( E^{+}\right) }=\overline{%
\mathcal{E}(E^{+})}\ .
\end{equation*}%
See \cite[Proposition 2.1]{BruPedra-MFII}.

\subsubsection{Periodic States\label{Sect Periodic-State Space}}

For $\vec{\ell}\in \mathbb{N}^{d}$, consider the subgroup $(\mathbb{Z}_{\vec{%
\ell}}^{d},+)\subseteq (\mathbb{Z}^{d},+)$, where%
\begin{equation}
\mathbb{Z}_{\vec{\ell}}^{d}\doteq \ell _{1}\mathbb{Z}\times \cdots \times
\ell _{d}\mathbb{Z}\ .  \label{group}
\end{equation}%
Any state $\rho \in E$ satisfying $\rho \circ \alpha _{x}=\rho $ for all $%
x\in \mathbb{Z}_{\vec{\ell}}^{d}$ is called $\mathbb{Z}_{\vec{\ell}}^{d}$%
\emph{-invariant} on $\mathcal{U}$ or $\vec{\ell}$\emph{-periodic}, $\alpha
_{x}$ being the unique $\ast $-automorphism of $\mathcal{U}$ satisfying (\ref%
{transl}). Translation-invariant states refer to $(1,\cdots ,1)$-periodic
states. For any $\vec{\ell}\in \mathbb{N}^{d}$, the metrizable and weak$%
^{\ast }$-compact convex set 
\begin{equation}
E_{\vec{\ell}}\doteq \left\{ \rho \in E:\rho \circ \alpha _{x}=\rho \quad 
\text{for}\ \text{all}\ x\in \mathbb{Z}_{\vec{\ell}}^{d}\right\}
\label{periodic invariant states}
\end{equation}%
is called the $\vec{\ell}$-periodic-state space. By \cite[Lemma 1.8]%
{BruPedra2}, periodic states are even and, by \cite[Proposition 2.3]%
{BruPedra-MFII}, the set 
\begin{equation}
E_{\mathrm{p}}\doteq \bigcup_{\vec{\ell}\in \mathbb{N}^{d}}E_{\vec{\ell}}
\label{set of periodic states}
\end{equation}%
of all periodic states is a weak$^{\ast }$-dense subset of even states. For
any $\vec{\ell}\in \mathbb{N}^{d}$, the set $\mathcal{E}(E_{\vec{\ell}})$ of
extreme points of $E_{\vec{\ell}}$ is a weak$^{\ast }$-dense $G_{\delta }$
subset of $E_{\vec{\ell}}$:%
\begin{equation}
E_{\vec{\ell}}=\overline{\mathrm{co}\mathcal{E}(E_{\vec{\ell}})}=\overline{%
\mathcal{E}(E_{\vec{\ell}})}\ ,\qquad \vec{\ell}\in \mathbb{N}^{d}\ .
\label{cov heull l perio}
\end{equation}%
In fact, up to an affine homeomorphism, for any $\vec{\ell}\in \mathbb{N}%
^{d} $, $E_{\vec{\ell}}$ is the so-called Poulsen simplex \cite[Theorem 1.12]%
{BruPedra2}. This property is well-known and also holds true for lattice
quantum spin systems \cite[Example 4.3.26 and discussions p. 464]%
{BrattelliRobinsonI}.

\subsection{Banach Spaces of Short-Range Interactions\label{Section Banach
space interaction}}

\subsubsection{Complex Interactions}

A (complex) \emph{interaction} is a mapping $\Phi :\mathcal{P}%
_{f}\rightarrow \mathcal{U}^{+}$ such that $\Phi _{\Lambda }\in \mathcal{U}%
_{\Lambda }$ for all $\Lambda \in \mathcal{P}_{f}$. The set $\mathcal{V}$ of
all interactions can be naturally endowed with the structure of a complex
vector space and the involution 
\begin{equation}
\Phi \mapsto \Phi ^{\ast }\doteq (\Phi _{\Lambda }^{\ast })_{\Lambda \in 
\mathcal{P}_{f}}\ .  \label{involution}
\end{equation}%
An interaction $\Phi $ is, by definition, self-adjoint if $\Phi =\Phi ^{\ast
}$. The set of all self-adjoint interactions forms a real subspace of the
space of all interactions.

\subsubsection{Short-Range Interactions}

The separable Banach space of short-range interactions is defined by 
\begin{equation}
\mathcal{W}\doteq \left\{ \Phi \in \mathcal{V}:\left\Vert \Phi \right\Vert _{%
\mathcal{W}}<\infty \right\}  \label{banach space short range}
\end{equation}%
with the norm of $\mathcal{W}$ being defined, from the positive-valued
symmetric function $\mathbf{F}$ of Section \ref{Algebra of Lattices}, by%
\begin{equation}
\left\Vert \Phi \right\Vert _{\mathcal{W}}\doteq \underset{x,y\in \mathfrak{L%
}}{\sup }\sum\limits_{\Lambda \in \mathcal{P}_{f},\;\Lambda \supseteq
\{x,y\}}\frac{\Vert \Phi _{\Lambda }\Vert _{\mathcal{U}}}{\mathbf{F}\left(
x,y\right) }\ .  \label{iteration0}
\end{equation}%
The (real) Banach subspace of all self-adjoint interactions is denoted by $%
\mathcal{W}^{\mathbb{R}}\varsubsetneq \mathcal{W}$, similar to $\mathcal{U}^{%
\mathbb{R}}\varsubsetneq \mathcal{U}$.

\subsubsection{Translation-Invariant Interactions}

By definition, the interaction $\Phi $ is translation-invariant\ if%
\begin{equation}
\Phi _{\Lambda +x}=\alpha _{x}\left( \Phi _{\Lambda }\right) \ ,\qquad x\in 
\mathbb{Z}^{d},\ \Lambda \in \mathcal{P}_{f}\ ,  \label{ti interaction}
\end{equation}%
where%
\begin{equation}
\Lambda +x\doteq \left\{ y+x\in \mathfrak{L}:y\in \Lambda \right\} \ .
\label{translation box}
\end{equation}%
Here, $\{\alpha _{x}\}_{x\in \mathbb{Z}^{d}}$ is the family of (translation) 
$\ast $-automorphisms of $\mathcal{U}$ defined by (\ref{transl}). We then
denote by $\mathcal{W}_{1}\varsubsetneq \mathcal{W}$ the (separable) Banach
subspace of translation-invariant, short-range interactions on $\mathfrak{L}$%
.

\subsubsection{Finite-Range Interactions}

For any $\Lambda \in \mathcal{P}_{f}$, we define the closed subspace%
\footnote{%
This follows from the continuity and linearity of the mappings $\Phi \mapsto
\Phi _{\mathcal{Z}}$ for all $\mathcal{Z}\in \mathcal{P}_{f}$.}%
\begin{equation}
\mathcal{W}_{\Lambda }\doteq \left\{ \Phi \in \mathcal{W}_{1}:\Phi _{%
\mathcal{Z}}=0\text{ whenever }\mathcal{Z}\nsubseteq \Lambda \text{, }%
\mathcal{Z}\ni 0\right\}  \label{eq:enpersitebis}
\end{equation}%
of finite-range translation-invariant interactions. Note that, for any $%
\Lambda \in \mathcal{P}_{f}$ and $\vec{\ell}\in \mathbb{N}^{d}$, 
\begin{equation}
\mathcal{W}_{\Lambda }\subseteq \left\{ \Phi \in \mathcal{W}_{1}:\mathfrak{e}%
_{\Phi ,\vec{\ell}}\in \mathcal{U}_{\Lambda ^{(\vec{\ell})}}\right\}
\subseteq \mathcal{W}_{1}\ ,  \label{eq:enpersitebisbis}
\end{equation}%
where, for any $\Phi \in \mathcal{W}$ and $\vec{\ell}\in \mathbb{N}^{d}$,%
\begin{equation}
\mathfrak{e}_{\Phi ,\vec{\ell}}\doteq \frac{1}{\ell _{1}\cdots \ell _{d}}%
\sum\limits_{x=(x_{1},\ldots ,x_{d}),\;x_{i}\in \{0,\ldots ,\ell
_{i}-1\}}\sum\limits_{\mathcal{Z}\in \mathcal{P}_{f},\;\mathcal{Z}\ni x}%
\frac{\Phi _{\mathcal{Z}}}{\left\vert \mathcal{Z}\right\vert }
\label{eq:enpersite}
\end{equation}%
and 
\begin{equation}
\Lambda ^{(\vec{\ell})}\doteq \cup \left\{ \Lambda +x:x=(x_{1},\ldots
,x_{d}),\;x_{i}\in \{0,\ldots ,\ell _{i}-1\}\right\} \in \mathcal{P}_{f}\ .
\label{eq:enpersitebisbisbis}
\end{equation}%
From Equations (\ref{(3.1) NS}) and (\ref{iteration0}), observe that 
\begin{equation}
\Vert \mathfrak{e}_{\Phi ,\vec{\ell}}\Vert _{\mathcal{U}}\leq \left\Vert 
\mathbf{F}\right\Vert _{1,\mathfrak{L}}\left\Vert \Phi \right\Vert _{%
\mathcal{W}}\ ,\qquad \Phi \in \mathcal{W},\ \vec{\ell}\in \mathbb{N}^{d}.
\label{e phi}
\end{equation}

\subsection{Banach Space of Long-Range Models}

\subsubsection{Self-Adjoint Measures on Interactions}

Let $\mathbb{S}$ be the unit sphere of $\mathcal{W}_{1}$. For any $n\in 
\mathbb{N}$ and any finite signed Borel measure $\mathfrak{a}$ on the
Cartesian product $\mathbb{S}^{n}$ (endowed with its product topology), we
define the finite signed Borel measure $\mathfrak{a}^{\ast }$ to be the
pushforward of $\mathfrak{a}$ through the automorphism 
\begin{equation}
\left( \Psi ^{(1)},\ldots ,\Psi ^{(n)}\right) \mapsto ((\Psi ^{(n)})^{\ast
},\ldots ,(\Psi ^{(1)})^{\ast })\in \mathbb{S}^{n}
\label{push forward self-adjoint}
\end{equation}%
of $\mathbb{S}^{n}$ as a topological space. A finite signed Borel measure $%
\mathfrak{a}$ on $\mathbb{S}^{n}$ is, by definition, self-adjoint whenever $%
\mathfrak{a^{\ast }}=\mathfrak{a}$. For any $n\in \mathbb{N}$, the real
Banach space of self-adjoint, finite, signed Borel measures on $\mathbb{S}%
^{n}$ endowed with the norm%
\begin{equation}
\Vert \mathfrak{a}\Vert _{\mathcal{S}(\mathbb{S}^{n}\mathbb{)}}\doteq |%
\mathfrak{a}|(\mathbb{S}^{n})\ ,\qquad n\in \mathbb{N}\ ,
\label{definition 0}
\end{equation}%
is denoted by $\mathcal{S}(\mathbb{S}^{n}\mathbb{)}$.

\subsubsection{Sequences of Self-Adjoint Measures on Interactions}

Endowed with point-wise operations, $\mathcal{S}$ is the real Banach space
of all sequences $\mathfrak{a}\equiv (\mathfrak{a}_{n})_{n\in \mathbb{N}}$
of self-adjoint, finite signed Borel measures $\mathfrak{a}_{n}\in \mathcal{S%
}(\mathbb{S}^{n}\mathbb{)}$, along with the (finite) norm 
\begin{equation}
\left\Vert \mathfrak{a}\right\Vert _{\mathcal{S}}\doteq \sum_{n\in \mathbb{N}%
}n^{2}\left\Vert \mathbf{F}\right\Vert _{1,\mathfrak{L}}^{n-1}\left\Vert 
\mathfrak{a}_{n}\right\Vert _{\mathcal{S}(\mathbb{S}^{n}\mathbb{)}}\,,\qquad 
\mathfrak{a}\equiv (\mathfrak{a}_{n})_{n\in \mathbb{N}}\in \mathcal{S}\ .
\label{definition 0bis}
\end{equation}%
Recall that $\mathbf{F}:\mathfrak{L}^{2}\rightarrow (0,1]$ is any
positive-valued symmetric function with maximum value $\mathbf{F}\left(
x,x\right) =1$ for all $x\in \mathfrak{L}$, satisfying Equations (\ref{(3.1)
NS})-(\ref{(3.2) NS}).

\subsubsection{Long-Range Models}

The separable Banach space of long-range models is defined by%
\begin{equation}
\mathcal{M}\doteq \left\{ \mathfrak{m}\in \mathcal{W}^{\mathbb{R}}\times 
\mathcal{S}:\left\Vert \mathfrak{m}\right\Vert _{\mathcal{M}}<\infty
\right\} \ ,  \label{def long range1}
\end{equation}%
where the norm of $\mathcal{M}$ is defined from (\ref{iteration0}) and (\ref%
{definition 0bis}) by%
\begin{equation}
\left\Vert \mathfrak{m}\right\Vert _{\mathcal{M}}\doteq \left\Vert \Phi
\right\Vert _{\mathcal{W}}+\left\Vert \mathfrak{a}\right\Vert _{\mathcal{S}%
}\,,\qquad \mathfrak{m}\doteq \left( \Phi ,\mathfrak{a}\right) \in \mathcal{M%
}\ .  \label{def long range2}
\end{equation}%
The spaces $\mathcal{W}^{\mathbb{R}}$ and $\mathcal{S}$ are canonically seen
as subspaces of $\mathcal{M}$, i.e., 
\begin{equation}
\mathcal{W}^{\mathbb{R}}\subseteq \mathcal{M}\qquad \text{and}\qquad 
\mathcal{S\subseteq M}\ .  \label{identification model}
\end{equation}%
In particular, $\Phi \equiv \left( \Phi ,0\right) \in \mathcal{M}$ for $\Phi
\in \mathcal{W}^{\mathbb{R}}$. Similar to (\ref{eq:enpersitebis}), we define
the subspaces 
\begin{equation}
\mathcal{M}_{\Lambda }\doteq \mathcal{W}^{\mathbb{R}}\times \mathcal{S}%
_{\Lambda }\subseteq \mathcal{M}\ ,\qquad \Lambda \in \mathcal{P}_{f},\ 
\label{S00bis}
\end{equation}%
where, for any $\Lambda \in \mathcal{P}_{f}$,%
\begin{equation}
\mathcal{S}_{\Lambda }\doteq \left\{ (\mathfrak{a}_{n})_{n\in \mathbb{N}}\in 
\mathcal{S}:\forall n\in \mathbb{N},\ |\mathfrak{a}_{n}|(\mathbb{S}^{n})=|%
\mathfrak{a}_{n}|((\mathbb{S}\cap \mathcal{W}_{\Lambda })^{n})\right\} \ .
\label{S0bis}
\end{equation}%
Note that 
\begin{equation*}
\mathcal{M}_{0}\doteq \bigcup_{L\in \mathbb{N}}\mathcal{M}_{\Lambda _{L}}
\end{equation*}%
is a dense subspace of $\mathcal{M}$.

Long-range models $\mathfrak{m}\doteq \left( \Phi ,\mathfrak{a}\right) $ are
not necessarily translation-invariant, because of its short-range component $%
\Phi $ which may be non-translation-invariant, and we define 
\begin{equation}
\mathcal{M}_{1}\doteq \left( \mathcal{W}_{1}\cap \mathcal{W}^{\mathbb{R}%
}\right) \times \mathcal{S\varsubsetneq M}
\label{translatino invariatn long range models}
\end{equation}%
to be the Banach space of all translation-invariant long-range models.

\subsection{Local Energies}

\subsubsection{Local Energy Elements of Short-Range Interactions}

The local energy elements of any complex interaction $\Phi \in \mathcal{W}$
are defined by%
\begin{equation}
U_{L}^{\Phi }\doteq \sum\limits_{\Lambda \subseteq \Lambda _{L}}\Phi
_{\Lambda }\in \mathcal{U}_{\Lambda _{L}}\cap \mathcal{U}^{+}\ ,\qquad L\in 
\mathbb{N}\ .  \label{equation fininte vol dynam0}
\end{equation}%
Note that%
\begin{equation}
\left\Vert U_{L}^{\Phi }\right\Vert _{\mathcal{U}}\leq \left\vert \Lambda
_{L}\right\vert \left\Vert \mathbf{F}\right\Vert _{1,\mathfrak{L}}\left\Vert
\Phi \right\Vert _{\mathcal{W}}\ ,\qquad L\in \mathbb{N},\ \Phi \in \mathcal{%
W}\ .  \label{norm Uphi}
\end{equation}%
By \cite[Proposition 3.2]{BruPedra-MFII}, for any $\vec{\ell}\in \mathbb{N}%
^{d}$, $\vec{\ell}$-periodic state $\rho \in E_{\vec{\ell}}$ (\ref{periodic
invariant states}) and translation-invariant complex interaction$\ \Phi \in 
\mathcal{W}_{1}$, 
\begin{equation*}
\lim\limits_{L\rightarrow \infty }\frac{\rho \left( U_{L}^{\Phi }\right) }{%
\left\vert \Lambda _{L}\right\vert }=\rho (\mathfrak{e}_{\Phi ,\vec{\ell}})
\end{equation*}%
with $\mathfrak{e}_{\Phi ,\vec{\ell}}$ being the even observable defined by (%
\ref{eq:enpersite}).

\subsubsection{Local Energy Elements of Long-Range Models}

The local Hamiltonians of any model $\mathfrak{m}\doteq \left( \Phi ,%
\mathfrak{a}\right) \in \mathcal{M}$ are the (well-defined) self-adjoint
elements 
\begin{equation}
U_{L}^{\mathfrak{m}}\doteq U_{L}^{\Phi }+\sum_{n\in \mathbb{N}}\frac{1}{%
\left\vert \Lambda _{L}\right\vert ^{n-1}}\int_{\mathbb{S}^{n}}U_{L}^{\Psi
^{(1)}}\cdots U_{L}^{\Psi ^{(n)}}\mathfrak{a}_{n}\left( \mathrm{d}\Psi
^{(1)},\ldots ,\mathrm{d}\Psi ^{(n)}\right) \,,\qquad L\in \mathbb{N}\ .
\label{equation long range energy}
\end{equation}%
Note that $U_{L}^{\left( \Phi ,0\right) }\doteq U_{L}^{\Phi }$ for $\Phi \in 
\mathcal{W}^{\mathbb{R}}$ (cf. (\ref{identification model})) and
straightforward estimates yield the bound 
\begin{equation}
\left\Vert U_{L}^{\mathfrak{m}}\right\Vert _{\mathcal{U}}\leq \left\vert
\Lambda _{L}\right\vert \left\Vert \mathbf{F}\right\Vert _{1,\mathfrak{L}%
}\left\Vert \mathfrak{m}\right\Vert _{\mathcal{M}}\ ,\qquad L\in \mathbb{N}\
,  \label{energy bound long range}
\end{equation}%
by Equations (\ref{definition 0})-(\ref{def long range2}) and (\ref{norm
Uphi}).

\subsection{Dynamical Problem Associated with Long-Range Interactions\label%
{sect Lieb--Robinson copy(1)}}

\subsubsection{Local Dynamics on the CAR Algebra}

\noindent \underline{Local derivations:} The local (symmetric) derivations $%
\{\delta _{L}^{\mathfrak{m}}\}_{L\in \mathbb{N}}\varsubsetneq \mathcal{B}(%
\mathcal{U})$ associated with any model $\mathfrak{m}\in \mathcal{M}$ are
defined by%
\begin{equation}
\delta _{L}^{\mathfrak{m}}(A)\doteq i\left[ U_{L}^{\mathfrak{m}},A\right]
\doteq i\left( U_{L}^{\mathfrak{m}}A-AU_{L}^{\mathfrak{m}}\right) \ ,\qquad
A\in \mathcal{U},\ L\in \mathbb{N}\ .  \label{equation derivation long range}
\end{equation}%
Note that $\delta _{L}^{\left( \Phi ,0\right) }\equiv \delta _{L}^{\Phi }$
for $L\in \mathbb{N}$ and $\Phi \in \mathcal{W}^{\mathbb{R}}$ (cf. (\ref%
{identification model})). \bigskip

\noindent \underline{Local non-autonomous dynamics:} Let $\mathfrak{m}\in C(%
\mathbb{R};\mathcal{M})$ be a continuous function from $\mathbb{R}$ to the
Banach space $\mathcal{M}$. Then, for any $L\in \mathbb{N}$, there is a
unique (fundamental) solution $(\tau _{t,s}^{(L,\mathfrak{m})})_{_{s,t\in 
\mathbb{R}}}$ in $\mathcal{B}(\mathcal{U})$ to the (finite-volume) non-auto%
\-%
nomous evolution equations%
\begin{equation}
\forall s,t\in {\mathbb{R}}:\qquad \partial _{s}\tau _{t,s}^{(L,\mathfrak{m}%
)}=-\delta _{L}^{\mathfrak{m}\left( s\right) }\circ \tau _{t,s}^{(L,%
\mathfrak{m})}\ ,\qquad \tau _{t,t}^{(L,\mathfrak{m})}=\mathbf{1}_{\mathcal{U%
}}\ ,  \label{cauchy1}
\end{equation}%
and%
\begin{equation}
\forall s,t\in {\mathbb{R}}:\qquad \partial _{t}\tau _{t,s}^{(L,\mathfrak{m}%
)}=\tau _{t,s}^{(L,\mathfrak{m})}\circ \delta _{L}^{\mathfrak{m}\left(
t\right) }\ ,\qquad \tau _{s,s}^{(L,\mathfrak{m})}=\mathbf{1}_{\mathcal{U}}\
.  \label{cauchy2}
\end{equation}%
In these two equations, $\mathbf{1}_{\mathcal{U}}$ refers to the identity
mapping of $\mathcal{U}$. Note also that, for any $L\in \mathbb{N}$ and $%
\mathfrak{m}\in C(\mathbb{R};\mathcal{M})$, $(\tau _{t,s}^{(L,\mathfrak{m}%
)})_{_{s,t\in \mathbb{R}}}$ is a continuous two-para%
\-%
meter family of $\ast $-auto%
\-%
morphisms of $\mathcal{U}$ that satisfies the (reverse) cocycle property%
\begin{equation*}
\forall s,r,t\in \mathbb{R}:\qquad \tau _{t,s}^{(L,\mathfrak{m})}=\tau
_{r,s}^{(L,\mathfrak{m})}\tau _{t,r}^{(L,\mathfrak{m})}\ .
\end{equation*}%
Again, $\tau _{t,s}^{(L,(\Psi ,0))}\equiv \tau _{t,s}^{(L,\Psi )}$ for $%
s,t\in \mathbb{R}$, $L\in \mathbb{N}$ and $\Psi \in C(\mathbb{R};\mathcal{W}%
^{\mathbb{R}})$ (cf. (\ref{identification model})).

\subsubsection{Dynamical Problem at Infinite Volume}

For any $\Psi \in C(\mathbb{R};\mathcal{W}^{\mathbb{R}})$, $(\tau
_{t,s}^{(L,\Psi )})_{s,t\in {\mathbb{R}}}$, $L\in \mathbb{N}$, converges
strongly, uniformly for $s,t$ on compacta, to a strongly continuous two-para%
\-%
meter family $(\tau _{t,s}^{\Psi })_{s,t\in {\mathbb{R}}}$ of $\ast $-auto%
\-%
morphisms of $\mathcal{U}$, as stated in \cite[Proposition 3.7]%
{BruPedra-MFII}. The main aim of this paper is to make sense of the
thermodynamic limit%
\begin{equation*}
\lim_{L\rightarrow \infty }\tau _{t,s}^{(L,\mathfrak{m})}\left( A\right) \
,\qquad s,t\in \mathbb{R},\ A\in \mathcal{U}\ ,
\end{equation*}%
for any $\mathfrak{m}\in C(\mathbb{R};\mathcal{M}_{\Lambda })$, where $%
\mathcal{M}_{\Lambda }$ is the space of long-range models defined by (\ref%
{S00bis})-(\ref{S0bis}) for a fixed $\Lambda \in \mathcal{P}_{f}$. This
cannot be done within the $C^{\ast }$-algebra $\mathcal{U}$, as explained in 
\cite[Section 4.3]{BruPedra-MFII}, but by using appropriate representations
of $\mathcal{U}$.

\section{State-Dependent Interactions and Dynamics\label{State-Dependant
Interactions}}

The long-range dynamics takes place in the space $C(E;\mathcal{U})$ of $%
\mathcal{U}$-valued weak$^{\ast }$-continuous functions on the metrizable
compact space $E$, a quantum $C^{\ast }$-algebra of continuous functions on
states. As explained in \cite[Section 6.6]{BruPedra-MFII}, it is used to
construct the infinite-volume limit of non-autonomous dynamics of
time-dependent long-range models within a cyclic representation associated
with some periodic state. We show that, generically at infinite-volume, the
long-range dynamics is equivalent to an intricate combination of classical
and short-range quantum dynamics. Similar to \cite{Bru-pedra-MF-I}, the
existence of both dynamics will be a (non-trivial) consequence of the
well-posedness of a self-consistency problem, see \cite[Theorem 6.5]%
{BruPedra-MFII}.

\subsection{The Quantum $C^{\ast }$-Algebra of Continuous Functions on
States \label{State-Dependent Short-Range Interactions copy(1)}}

\subsubsection{Quantum Algebra}

Endowed with the point-wise $\ast $-algebra operations inherited from $%
\mathcal{U}$, $C(E;\mathcal{U})$ is the unital $C^{\ast }$-algebra denoted by%
\begin{equation}
\mathfrak{U}\equiv \mathcal{U}_{\mathfrak{L}}\doteq \left( C\left( E;%
\mathcal{U}\right) ,+,\cdot _{{\mathbb{C}}},\times ,^{\ast },\left\Vert
\cdot \right\Vert _{\mathfrak{U}}\right) \ .  \label{metaciagre set}
\end{equation}%
The unique $C^{\ast }$-norm $\left\Vert \cdot \right\Vert _{\mathfrak{U}}$
is the supremum norm: 
\begin{equation}
\left\Vert f\right\Vert _{\mathfrak{U}}\doteq \max_{\rho \in E}\left\Vert
f\left( \rho \right) \right\Vert _{\mathcal{U}}\ ,\qquad f\in \mathfrak{U}\ .
\label{metaciagre set bis}
\end{equation}%
The (real) Banach subspace of all $\mathcal{U}^{\mathbb{R}}$-valued
functions is denoted by $\mathfrak{U}^{\mathbb{R}}\varsubsetneq \mathfrak{U}$%
. We identify the primordial $C^{\ast }$-algebra $\mathcal{U}$ with the
subalgebra of constant functions of $\mathfrak{U}$, i.e., $\mathcal{U}%
\subseteq \mathfrak{U}$.

\subsubsection{Local Elements}

Similar to (\ref{simple}), we define the $\ast $-subalgebras%
\begin{equation}
\mathfrak{U}_{\Lambda }\doteq \left\{ f\in \mathfrak{U}:f\left( E\right)
\subseteq \mathcal{U}_{\Lambda }\right\} \ ,\qquad \Lambda \in \mathcal{P}%
_{f}\ ,  \label{local elements}
\end{equation}%
and 
\begin{equation}
\mathfrak{U}_{0}\doteq \bigcup_{\Lambda \in \mathcal{P}_{f}}\mathfrak{U}%
_{\Lambda }\subseteq \mathfrak{U}\ ,  \label{U frac 0}
\end{equation}%
which is a dense $\ast $-subalgebra of $\mathfrak{U}$. See \cite[Section 5.3]%
{BruPedra-MFII}.

\subsubsection{Even Elements}

The $\ast $-automorphism $\sigma $ of $\mathcal{U}$ uniquely defined by (\ref%
{automorphism gauge invariance}) naturally induces a $\ast $-automorphism $%
\Xi $ of $\mathfrak{U}$ defined by 
\begin{equation}
\left[ \Xi \left( f\right) \right] \left( \rho \right) \doteq \sigma \left(
f\left( \rho \right) \right) \ ,\qquad \rho \in E,\ f\in \mathfrak{U}\ .
\label{gauge invariant}
\end{equation}%
The set 
\begin{equation}
\mathfrak{U}^{+}\doteq \{f\in \mathfrak{U}:f=\Xi (f)\}=\left\{ f\in 
\mathfrak{U}:f\left( E\right) \subseteq \mathcal{U}^{+}\right\} \subseteq 
\mathfrak{U}  \label{even state dependent}
\end{equation}%
of all even $\mathcal{U}$-valued continuous functions is a $C^{\ast }$%
-subalgebra of $\mathfrak{U}$. Compare with (\ref{definition of even
operators}).

\subsubsection{Translation Automorphisms}

The $\ast $-automorphisms $\alpha _{x}$, $x\in \mathbb{Z}^{d}$, of $\mathcal{%
U}$ uniquely defined by (\ref{transl}) naturally induce a group homomorphism 
$x\mapsto \mathrm{A}_{x}$ from $(\mathbb{Z}^{d},+)$ to the group of $\ast $%
-automorphisms of $\mathfrak{U}$, defined by 
\begin{equation}
\left[ \mathrm{A}_{x}\left( f\right) \right] \left( \rho \right) \doteq
\alpha _{x}\left( f\left( \rho \right) \right) \ ,\qquad \rho \in E,\ f\in 
\mathfrak{U},\ x\in \mathbb{Z}^{d}\ .  \label{translatbis}
\end{equation}%
These $\ast $-automorphisms represent the translation group in $\mathfrak{U}$%
.

\subsection{The Classical $C^{\ast }$-Algebra of Continuous Functions on
States}

\subsubsection{Classical Algebra}

Endowed with point-wise vector space operations and complex conjugation, $%
C(E;\mathbb{C})$ is the unital commutative $C^{\ast }$-algebra denoted by 
\begin{equation}
\mathfrak{C}\doteq \left( C\left( E;\mathbb{C}\right) ,+,\cdot _{{\mathbb{C}}%
},\times ,\overline{\left( \cdot \right) },\left\Vert \cdot \right\Vert _{%
\mathfrak{C}}\right) \ ,  \label{metaciagre set 2}
\end{equation}%
where the corresponding (unique) $C^{\ast }$-norm is 
\begin{equation}
\left\Vert f\right\Vert _{\mathfrak{C}}\doteq \max_{\rho \in E}\left\vert
f\left( \rho \right) \right\vert \ ,\qquad f\in \mathfrak{C}\ .
\label{metaciagre set 2bis}
\end{equation}%
The (real) Banach subspace of all real-valued functions of $\mathfrak{C}$ is
denoted by $\mathfrak{C}^{\mathbb{R}}\varsubsetneq \mathfrak{C}$. The $%
C^{\ast }$-algebra $\mathfrak{C}$ is separable, $E$ being metrizable and
compact.

The classical dynamics takes place in $\mathfrak{C}$. This unital
commutative $C^{\ast }$-algebra is identified with the subalgebra of
functions of $\mathfrak{U}$ whose values are multiples of the unit $%
\mathfrak{1}\in \mathcal{U}$. In other words, we have the canonical
inclusion $\mathfrak{C\subseteq U}$.

\subsubsection{Poisson Bracket for Polynomial Functions\label{Poisson
Algebra}}

Elements of the (separable and unital)$\ C^{\ast }$-algebra $\mathcal{U}$
naturally define continuous affine functions $\hat{A}\in \mathfrak{C}$ by 
\begin{equation}
\hat{A}\left( \rho \right) \doteq \rho \left( A\right) \ ,\qquad \rho \in
E,\ A\in \mathcal{U}\ .  \label{fA}
\end{equation}%
This mapping $A\mapsto \hat{A}$ is a linear isometry from $\mathcal{U}^{%
\mathbb{R}}$ to $\mathfrak{C}^{\mathbb{R}}$. We denote by%
\begin{equation}
\mathfrak{C}_{\mathcal{U}_{0}}\doteq \mathbb{C}[\{\hat{A}:A\in \mathcal{U}%
_{0}\}]\subseteq \mathfrak{C}  \label{CU0}
\end{equation}%
the subalgebras of polynomials in the elements of $\{\hat{A}:A\in \mathcal{U}%
_{0}\}$, with complex coefficients. Note that $\mathfrak{C}_{\mathcal{U}%
_{0}} $ is dense in $\mathfrak{C}$, i.e., $\mathfrak{C}=\overline{\mathfrak{C%
}_{\mathcal{U}_{0}}}$ (the Stone-Weierstrass theorem).

In \cite[Section 5.2]{BruPedra-MFII} we define a Poisson bracket%
\begin{equation*}
\{\cdot ,\cdot \}:\mathfrak{C}_{\mathcal{U}_{0}}\times \mathfrak{C}_{%
\mathcal{U}_{0}}\rightarrow \mathfrak{C}\ ,
\end{equation*}%
i.e., a skew-symmetric biderivation satisfying the Jacobi identity. This
Poisson bracket can be extended to any continuously differentiable
real-valued functions on the state space $E$. For more details, see \cite[%
Section 5.2]{BruPedra-MFII} as well as \cite[Section 3]{Bru-pedra-MF-I}. In
fact, it is not really used in this paper and is only shortly mentioned in
order to establish a connection between the results of this paper and those
of \cite{BruPedra-MFII}.

\subsection{Banach Spaces of State-Dependent Short-Range Interactions\label%
{State-Dependent Short-Range Interactions}}

\subsubsection{Complex State-Dependent Interactions}

As is done in Section \ref{Section Banach space interaction}, a
state-dependent (complex) interaction is defined to be a mapping $\mathbf{%
\Phi }:\mathcal{P}_{f}\rightarrow \mathfrak{U}^{+}$ such that $\mathbf{\Phi }%
_{\Lambda }\in \mathfrak{U}_{\Lambda }$ for any $\Lambda \in \mathcal{P}_{f}$%
. See Equations (\ref{local elements}) and (\ref{even state dependent}). The
set $\mathfrak{V}$ of all state-dependent interactions is naturally endowed
with the structure of a complex vector space and with the involution 
\begin{equation}
\mathbf{\Phi }\mapsto \mathbf{\Phi }^{\ast }\doteq (\mathbf{\Phi }_{\Lambda
}^{\ast })_{\Lambda \in \mathcal{P}_{f}}\ .  \label{involutioninvolution}
\end{equation}%
Self-adjoint state-dependent interactions $\mathbf{\Phi }$ are, by
definition, those satisfying $\mathbf{\Phi }=\mathbf{\Phi }^{\ast }$.

\subsubsection{Short-Range State-Dependent Interactions}

Similar to the Banach space $\mathcal{W}$ of short-range interactions, we
define a Banach space 
\begin{equation*}
\mathfrak{W}\doteq \left\{ \mathbf{\Phi }\in \mathfrak{V}:\left\Vert \mathbf{%
\Phi }\right\Vert _{\mathfrak{W}}<\infty \right\}
\end{equation*}%
of state-dependent short-range interactions by using the norm%
\begin{equation*}
\left\Vert \mathbf{\Phi }\right\Vert _{\mathfrak{W}}\doteq \underset{x,y\in 
\mathfrak{L}}{\sup }\sum\limits_{\Lambda \in \mathcal{P}_{f},\;\Lambda
\supseteq \{x,y\}}\frac{\Vert \mathbf{\Phi }_{\Lambda }\Vert _{\mathfrak{U}}%
}{\mathbf{F}\left( x,y\right) }\ .
\end{equation*}%
Compare with (\ref{iteration0}). The (real) Banach subspace of all
self-adjoint state-dependent interactions is denoted by $\mathfrak{W}^{%
\mathbb{R}}\varsubsetneq \mathfrak{W}$.

\subsubsection{Approximating Interactions of Long-Range Models\label%
{Approximating Interactions2}}

In order to simplify the notation, for any $\mathbf{\Psi }\in C(\mathbb{R};%
\mathfrak{W})$ and $\rho \in E$, $\mathbf{\Psi }\left( \rho \right) \in C(%
\mathbb{R};\mathcal{W})$ stands for the time-dependent interaction defined
by 
\begin{equation}
\mathbf{\Psi }\left( \rho \right) \left( t\right) \doteq \mathbf{\Psi }%
\left( t;\rho \right) \ ,\qquad \rho \in E,\ t\in \mathbb{R}\ .
\label{notation state interactionbis}
\end{equation}%
For any $\Psi \in \mathcal{W}$ and $\rho \in E$, we define%
\begin{equation}
\left\lfloor \rho ;\Psi \right\rfloor _{\vec{\ell}}\doteq \Psi
\label{defined interand1}
\end{equation}%
and, for any $n\geq 2$ and all interactions $\Psi ^{(1)},\ldots ,\Psi
^{(n)}\in \mathcal{W}_{1}$, 
\begin{equation}
\left\lfloor \rho ;\Psi ^{(1)},\ldots ,\Psi ^{(n)}\right\rfloor _{\vec{\ell}%
}\doteq \sum_{m=1}^{n}\Psi ^{(m)}\prod\limits_{j\in \left\{ 1,\ldots
,n\right\} ,j\neq m}\rho (\mathfrak{e}_{\Psi ^{(j)},\vec{\ell}})\in \mathcal{%
W}_{1}\ ,\qquad \rho \in E\ .  \label{defined interand2}
\end{equation}%
The pertinence of these objects is explained in \cite[Section 6.4]%
{BruPedra-MFII}. They yield an approximating (state-dependent, short-range)
interaction, which is ubiquitous in the study of the infinite-volume
dynamics of lattice-fermion (or quantum-spin) systems with long-range
interactions.

\begin{definition}[Non-autonomous approximating interactions]
\label{definition BCS-type model approximated}\mbox{ }\newline
For $\vec{\ell}\in \mathbb{N}^{d}$ and any continuous functions $\mathfrak{m}%
=(\Phi (t),\mathfrak{\mathfrak{a}}(t))_{t\in \mathbb{R}}\in C\left( \mathbb{R%
};\mathcal{M}\right) $, $\xi \in C\left( \mathbb{R};E\right) $, we define
the mapping $\Phi ^{(\mathfrak{m},\xi )}$ from $\mathbb{R}$ to $\mathcal{W}^{%
\mathbb{R}}$ by%
\begin{equation*}
\Phi ^{(\mathfrak{m},\xi )}\left( t\right) \doteq \Phi \left( t\right)
+\sum_{n\in \mathbb{N}}\int_{\mathbb{S}^{n}}\ \left\lfloor \xi \left(
t\right) ;\Psi ^{(1)},\ldots ,\Psi ^{(n)}\right\rfloor _{\vec{\ell}}\ 
\mathfrak{a}\left( t\right) _{n}\left( \mathrm{d}\Psi ^{(1)},\ldots ,\mathrm{%
d}\Psi ^{(n)}\right) \,,\qquad t\in \mathbb{R}\ .
\end{equation*}%
If $\mathbf{\xi }\in C\left( \mathbb{R};\mathrm{Aut}\left( E\right) \right) $
then a mapping $\mathbf{\Phi }^{(\mathfrak{m},\mathbf{\xi })}$ from $\mathbb{%
R}$ to $\mathfrak{W}^{\mathbb{R}}$ is defined, for any $\rho \in E$ and $%
t\in \mathbb{R}$, by%
\begin{equation*}
\mathbf{\Phi }^{(\mathfrak{m},\mathbf{\xi })}\left( t;\rho \right) \doteq
\Phi \left( t\right) +\sum_{n\in \mathbb{N}}\int_{\mathbb{S}^{n}}\
\left\lfloor \mathbf{\xi }\left( t;\rho \right) ;\Psi ^{(1)},\ldots ,\Psi
^{(n)}\right\rfloor _{\vec{\ell}}\ \mathfrak{a}\left( t\right) _{n}\left( 
\mathrm{d}\Psi ^{(1)},\ldots ,\mathrm{d}\Psi ^{(n)}\right) \ .
\end{equation*}
\end{definition}

\noindent Recall that $\mathrm{Aut}\left( E\right) \varsubsetneq C\left(
E;E\right) $ is the subspace of all automorphisms of $E$, i.e., weak$^{\ast
} $-continuous state-valued functions over $E$ with weak$^{\ast }$%
-continuous inverse. The topology in $\mathrm{Aut}\left( E\right) $ is the
one of $C\left( E;E\right) $, that is, the topology of uniform convergence
(cf. (\ref{uniform convergence weak*})). Recall also the identifications (%
\ref{identify}). By \cite[Lemma 6.4]{BruPedra-MFII}, note finally that, for
any $\mathfrak{m}\in C\left( \mathbb{R};\mathcal{M}\right) $ and $\xi \in
C\left( \mathbb{R};E\right) $, 
\begin{equation}
\left\Vert \Phi ^{(\mathfrak{m},\xi )}\left( t\right) \right\Vert _{\mathcal{%
W}}\leq \left\Vert \mathfrak{m}\left( t\right) \right\Vert _{\mathcal{M}%
}\,,\qquad t\in \mathbb{R}\ .  \label{inequality trivial}
\end{equation}

\subsection{State-Dependent Quantum Dynamics}

\noindent \underline{Limit derivations:} Using that $\mathfrak{U}_{0}=%
\mathrm{span}\left\{ \mathfrak{C}\mathcal{U}_{0}\right\} $, we define the
symmetric derivations $\mathbf{\delta }^{\mathbf{\Phi }}$ associated with
any $\mathbf{\Phi }\in \mathfrak{W}$ on the dense subset $\mathfrak{U}_{0}$ (%
\ref{U frac 0}) by%
\begin{equation*}
\left[ \mathbf{\delta }^{\mathbf{\Phi }}(fA)\right] \left( \rho \right)
\doteq f\left( \rho \right) \delta ^{\mathbf{\Phi }\left( \rho \right) }(A)\
,\qquad \rho \in E,\ f\in \mathfrak{C},\ A\in \mathcal{U}_{0}\ .
\end{equation*}%
The right-hand side of the above equation defines an element of $\mathfrak{U}
$, by \cite[Corollary 3.5]{BruPedra-MFII}.\ If $\Phi \in \mathfrak{W}^{%
\mathbb{R}}$ then the symmetric derivation $\mathbf{\delta }^{\mathbf{\Phi }%
} $ is (norm-) closable. See \cite[discussions after Definition 6.1]%
{BruPedra-MFII}. \bigskip

\noindent \underline{State-dependent dynamics:} Any $\mathbf{\Psi }\in C(%
\mathbb{R};\mathfrak{W}^{\mathbb{R}})$ determines a two-parameter family $%
\mathfrak{T}^{\mathbf{\Psi }}\equiv (\mathfrak{T}_{t,s}^{\mathbf{\Psi }%
})_{_{s,t\in \mathbb{R}}}$ of $\ast $-auto%
\-%
morphisms of $\mathfrak{U}$ defined by%
\begin{equation}
\left[ \mathfrak{T}_{t,s}^{\mathbf{\Psi }}\left( f\right) \right] \left(
\rho \right) \doteq \tau _{t,s}^{\mathbf{\Psi }\left( \rho \right) }\left(
f\left( \rho \right) \right) \ ,\qquad \rho \in E,\ f\in \mathfrak{U},\
s,t\in \mathbb{R}\ ,  \label{definiotion tho frac}
\end{equation}%
where the two-para%
\-%
meter family $(\tau _{t,s}^{\Psi })_{s,t\in {\mathbb{R}}}$ of $\ast $-auto%
\-%
morphisms of $\mathcal{U}$ is the strong limit of $(\tau _{t,s}^{(L,\Psi
)})_{s,t\in {\mathbb{R}}}$, $L\in \mathbb{N}$, which is defined by (\ref%
{cauchy1})-(\ref{cauchy2}). See \cite[Proposition 3.7]{BruPedra-MFII}. The
right-hand side of (\ref{definiotion tho frac}) defines an element of $%
\mathfrak{U}$. In fact, by \cite[Proposition 6.4]{BruPedra-MFII}, for any $%
\mathbf{\Psi }\in C(\mathbb{R};\mathfrak{W}^{\mathbb{R}})$, $\mathfrak{T}^{%
\mathbf{\Psi }}\equiv (\mathfrak{T}_{t,s}^{\mathbf{\Psi }})_{_{s,t\in 
\mathbb{R}}}$ is a strongly continuous two-para%
\-%
meter family of $\ast $-auto%
\-%
morphisms of $\mathfrak{U}$, which is the unique solution in $\mathcal{B}(%
\mathfrak{U})$ to the non-auto%
\-%
nomous evolution equation%
\begin{equation}
\forall s,t\in {\mathbb{R}}:\qquad \partial _{t}\mathfrak{T}_{t,s}^{\mathbf{%
\Psi }}=\mathfrak{T}_{t,s}^{\mathbf{\Psi }}\circ \mathbf{\delta }^{\mathbf{%
\Psi }\left( t\right) }\ ,\qquad \mathfrak{T}_{s,s}^{\mathbf{\Psi }}=\mathbf{%
1}_{\mathfrak{U}}\ ,  \label{klsdf}
\end{equation}%
in the strong sense on the dense subspace $\mathfrak{U}_{0}\subseteq 
\mathfrak{U}$, $\mathbf{1}_{\mathfrak{U}}$ being the identity mapping of $%
\mathfrak{U}$. It satisfies, in particular, the reverse cocycle property:%
\begin{equation}
\forall s,r,t\in \mathbb{R}:\qquad \mathfrak{T}_{t,s}^{\mathbf{\Psi }}=%
\mathfrak{T}_{r,s}^{\mathbf{\Psi }}\mathfrak{T}_{t,r}^{\mathbf{\Psi }}\ .
\label{reverse cocycle0}
\end{equation}

\section{Long-Range Dynamics\label{Long-Range Dynamics}}

\subsection{Classical Part of Long-Range Dynamics\label{Classical Part}}

We show that, generically, long-range (or mean-field) dynamics are
equivalent to intricate combinations of classical \emph{and} short-range
quantum dynamics. The existence of both dynamics is a (non-trivial)
consequence of the well-posedness of the self-consistency problem of \cite[%
Theorem 6.5]{BruPedra-MFII}:

\begin{theorem}[Self-consistency equations]
\label{theorem sdfkjsdklfjsdklfj}\mbox{ }\newline
Fix $\Lambda \in \mathcal{P}_{f}$ and $\mathfrak{m}\in C_{b}(\mathbb{R};%
\mathcal{M}_{\Lambda })$. There is a unique $\mathbf{\varpi }^{\mathfrak{m}%
}\in C\left( \mathbb{R}^{2};\mathrm{Aut}\left( E\right) \right) $ such that%
\begin{equation*}
\mathbf{\varpi }^{\mathfrak{m}}\left( s,t;\rho \right) =\rho \circ \tau
_{t,s}^{\mathbf{\Phi }^{(\mathfrak{m},\mathbf{\varpi }^{\mathfrak{m}}\left(
\alpha ,\cdot \right) )}(\rho )}|_{\alpha =s}\ ,\qquad s,t\in {\mathbb{R}}\ ,
\end{equation*}%
with the strongly continuous two-para%
\-%
meter family $(\tau _{t,s}^{\Psi })_{s,t\in {\mathbb{R}}}$ being, at fixed $%
s,t\in {\mathbb{R}}$, the strong limit of the local dynamics $(\tau
_{t,s}^{(L,\Psi )})_{s,t\in {\mathbb{R}}}$ defined by (\ref{cauchy1})-(\ref%
{cauchy2}) for any $\Psi \in C(\mathbb{R};\mathcal{W}^{\mathbb{R}})$. See 
\cite[Proposition 3.7]{BruPedra-MFII}. Note that, above, we use the notation
(\ref{notation state interactionbis}).
\end{theorem}

The continuous family $\mathbf{\varpi }^{\mathfrak{m}}$ of Theorem \ref%
{theorem sdfkjsdklfjsdklfj} yields a family $(V_{t,s}^{\mathfrak{m}%
})_{s,t\in \mathbb{R}}$ of $\ast $-automorphisms of $\mathfrak{C}\doteq
C\left( E;\mathbb{C}\right) $ defined by 
\begin{equation}
V_{t,s}^{\mathfrak{m}}\left( f\right) \doteq f\circ \mathbf{\varpi }^{%
\mathfrak{m}}\left( s,t\right) \ ,\qquad f\in \mathfrak{C},\ s,t\in \mathbb{R%
}\ .  \label{classical evolution family}
\end{equation}%
It is a strongly continuous two-parameter family defining a classical
dynamics on the classical $C^{\ast }$-algebra $\mathfrak{C}$ of continuous
complex-valued functions on states, defined by (\ref{metaciagre set 2})-(\ref%
{metaciagre set 2bis}).

This classical dynamics is a Feller evolution, as explained in \cite[Section
6.5]{BruPedra-MFII}. Additionally, the classical flow $\mathbf{\varpi }^{%
\mathfrak{m}}\left( s,t\right) $, $s,t\in \mathbb{R}$, conserves the
even-state space $E^{+}$. Thus, $V_{t,s}^{\mathfrak{m}}$ can be seen as
either a mapping from $C(E^{+};\mathbb{C})$ to itself or from $C(E\backslash
E^{+};\mathbb{C})$ to itself:%
\begin{equation*}
V_{t,s}^{\mathfrak{m}}\left( f|_{E^{+}}\right) \doteq \left( V_{t,s}^{%
\mathfrak{m}}f\right) |_{E^{+}}\ ,\qquad V_{t,s}^{\mathfrak{m}}\left(
f|_{E\backslash E^{+}}\right) \doteq \left( V_{t,s}^{\mathfrak{m}}f\right)
|_{E\backslash E^{+}}\ ,\qquad f\in \mathfrak{C},\ s,t\in \mathbb{R}\ .
\end{equation*}%
Recall that $E^{+}$ is the (physically relevant) weak$^{\ast }$-compact
convex set of even states defined by (\ref{gauge invariant states}).

If $\mathbf{F}$ decays sufficiently fast, as $\left\vert x-y\right\vert
\rightarrow \infty $, then, using the local classical energy functions \cite[%
Definition 6.8]{BruPedra-MFII} associated with $\mathfrak{m}\in \mathcal{M}$%
, that is, the functions%
\begin{equation*}
\mathrm{h}_{L}^{\mathfrak{m}}\doteq \widehat{U_{L}^{\Phi }}+\sum_{n\in 
\mathbb{N}}\frac{1}{\left\vert \Lambda _{L}\right\vert ^{n-1}}\int_{\mathbb{S%
}^{n}}\ \widehat{U_{L}^{\Psi ^{(1)}}}\cdots \widehat{U_{L}^{\Psi ^{(n)}}}\ 
\mathfrak{a}\left( t\right) _{n}\left( \mathrm{d}\Psi ^{(1)},\ldots ,\mathrm{%
d}\Psi ^{(n)}\right) \,,\qquad L\in \mathbb{N}\ ,
\end{equation*}%
(see (\ref{fA})), we prove in \cite[Theorem 6.10]{BruPedra-MFII} that, for
any $s,t\in \mathbb{R}$, $f\in \mathfrak{C}_{\mathcal{U}_{0}}$ and $%
\mathfrak{m}\in C_{b}(\mathbb{R};\mathcal{M}_{\Lambda }\cap \mathcal{M}_{1})$%
, 
\begin{equation*}
\partial _{t}V_{t,s}^{\mathfrak{m}}\left( f\right) |_{E_{\mathrm{p}%
}}=\lim_{L\rightarrow \infty }V_{t,s}^{\mathfrak{m}}\left( \{\mathrm{h}_{L}^{%
\mathfrak{m}(t)},f\}\right) |_{E_{\mathrm{p}}}
\end{equation*}%
while, for any $\mathfrak{m}\in C_{b}(\mathbb{R};\mathcal{M}_{\Lambda })$,%
\begin{equation*}
\partial _{s}V_{t,s}^{\mathfrak{m}}\left( f\right) |_{E_{\mathrm{p}%
}}=-\lim_{L\rightarrow \infty }\{\mathrm{h}_{L}^{\mathfrak{m}(s)},V_{t,s}^{%
\mathfrak{m}}(f)\}|_{E_{\mathrm{p}}}\ ,
\end{equation*}%
where $\{\cdot ,\cdot \}$ is the Poisson bracket referred to in Section \ref%
{Poisson Algebra}. All limits have to be understood in the point-wise sense
on the metrizable weak$^{\ast }$-dense (in $E^{+}$) convex set $E_{\mathrm{p}%
}$ of periodic states defined by (\ref{set of periodic states}). These
non-trivial statements are proven from Lieb-Robinson bounds for
multi-commutators derived in \cite{brupedraLR}. In the autonomous situation,
we obtain the usual (autonomous) dynamics of classical mechanics written in
terms of Poisson brackets (see, e.g., \cite[Proposition 10.2.3]%
{classical-dynamics}), i.e., \emph{Liouville's equation}. See \cite[%
Corollary 6.11]{BruPedra-MFII}.

\subsection{Quantum Part of Long-Range Dynamics\label{Quantum Part}}

The classical part of the dynamics of lattice-fermion systems with
long-range interactions, which is defined within the classical $C^{\ast }$%
-algebra $\mathfrak{C}$ defined by (\ref{metaciagre set 2})-(\ref{metaciagre
set 2bis}), is shown to result from the solution to the self-consistency
equation of Theorem \ref{theorem sdfkjsdklfjsdklfj}.

As soon as only the classical part of the long-range dynamics is concerned,
there is no need to assume any additional property on initial states.
However, for the quantum part, we need periodic states as initial states.
Note that the set of all periodic states is still a weak$^{\ast }$-dense
subset of the physically relevant space $E^{+}$ of all even states, by \cite[%
Proposition 2.3]{BruPedra-MFII}.

Fix from now on $\vec{\ell}\in \mathbb{N}^{d}$ and consider the metrizable
and weak$^{\ast }$-compact convex set $E_{\vec{\ell}}$ of $\vec{\ell}$%
-periodic states defined by (\ref{periodic invariant states}), with $%
\mathcal{E}(E_{\vec{\ell}})$ being its (non-empty) set of extreme points. By
the Choquet theorem \cite[Theorem 10.18]{BruPedra2}, for any state $\rho \in
E_{\vec{\ell}}$, 
\begin{equation}
\rho \left( A\right) =\int_{E_{\vec{\ell}}}\hat{\rho}\left( A\right) \mu
_{\rho }\left( \mathrm{d}\hat{\rho}\right) =\int_{\mathcal{E}(E_{\vec{\ell}%
})}\hat{\rho}\left( A\right) \mu _{\rho }\left( \mathrm{d}\hat{\rho}\right) 
\text{ },\qquad A\in \mathcal{U}\ ,  \label{decomposition choquet}
\end{equation}%
where $\mu _{\rho }\equiv \mu _{\rho }^{(\vec{\ell})}$ is the (unique)
orthogonal probability measure of Theorem \ref{theorem choquet}. Then, any
state $\rho \in E_{\vec{\ell}}\subseteq \mathcal{U}^{\ast }$ naturally
extends to a state on the quantum $C^{\ast }$-algebra $\mathfrak{U}$ of
continuous $\mathcal{U}$-valued functions on states, also denoted by $\rho
\in \mathfrak{U}^{\ast }$, via the definition%
\begin{equation}
\rho \left( f\right) \doteq \int_{\mathcal{E}(E_{\vec{\ell}})}\hat{\rho}%
\left( f\left( \hat{\rho}\right) \right) \mu _{\rho }\left( \mathrm{d}\hat{%
\rho}\right) \text{ },\qquad f\in \mathfrak{U}\ .  \label{ddddd}
\end{equation}%
Cf. Definition \ref{def Extension of states} with $\mathcal{X}=\mathcal{U}$
and $F=E_{\vec{\ell}}$. Note that this extension equals $\mu _{\rho }\circ
\Xi $, where $\mu _{\rho }$ is a seen as a state on $\mathfrak{C}$ and $\Xi $
is the conditional expectation from $\mathfrak{U}$ to $\mathfrak{C}\subseteq 
\mathfrak{U}$ defined by 
\begin{equation*}
\Xi \left( f\right) \left( \rho \right) =\rho \left( f\left( \rho \right)
\right) \ ,\qquad \rho \in E\ .
\end{equation*}

Observe that this extension is $\vec{\ell}$-dependent: If $\vec{\ell}_{1},%
\vec{\ell}_{2}\in \mathbb{N}^{d}$ is such that $\mathbb{Z}_{\vec{\ell}%
_{1}}^{d}\subseteq \mathbb{Z}_{\vec{\ell}_{2}}^{d}$ (see (\ref{group})),
then $E_{\vec{\ell}_{2}}\subseteq E_{\vec{\ell}_{1}}$, but an extreme state
of $E_{\vec{\ell}_{2}}$ is not necessarily an extreme state of $E_{\vec{\ell}%
_{1}}$. Nevertheless, if $\rho \in E_{\vec{\ell}_{2}}$ then the unique
probability measure $\mu _{\rho }^{(\vec{\ell}_{1})}$ representing $\rho $
in $E_{\vec{\ell}_{1}}$ is directly related to the unique one $\mu _{\rho
}^{(\vec{\ell}_{2})}$\ representing $\rho $ in $E_{\vec{\ell}_{2}}$. In
fact, $\mu _{\rho }^{(\vec{\ell}_{2})}$ is the pushforward of $\mu _{\rho
}^{(\vec{\ell}_{1})}$ through a natural space-averaging mapping $\mathfrak{x}%
_{\vec{\ell}_{1},\vec{\ell}_{2}}$ (see (\ref{space-average fcuntion})). This
is shown in Corollary \ref{corllary-pushforward}. For any $\rho \in E_{\vec{%
\ell}_{2}}$ and any function $f\in \mathfrak{U}$ satisfying 
\begin{equation}
\Xi \left( f\right) \circ \mathfrak{x}_{\vec{\ell}_{1},\vec{\ell}_{2}}\left( 
\hat{\rho}\right) =\Xi \left( f\right) \left( \hat{\rho}\right) \ ,\qquad 
\hat{\rho}\in \mathcal{E}(E_{\vec{\ell}_{1}})\ ,  \label{errrrrrrrtttgg}
\end{equation}%
it follows that%
\begin{equation}
\rho ^{(\vec{\ell}_{1})}\left( f\right) \doteq \int_{\mathcal{E}(E_{\vec{\ell%
}_{1}})}\hat{\rho}\left( f\left( \hat{\rho}\right) \right) \mu _{\rho }^{(%
\vec{\ell}_{1})}\left( \mathrm{d}\hat{\rho}\right) =\int_{\mathcal{E}(E_{%
\vec{\ell}_{2}})}\hat{\rho}\left( f\left( \hat{\rho}\right) \right) \mu
_{\rho }^{(\vec{\ell}_{2})}\left( \mathrm{d}\hat{\rho}\right) \doteq \rho ^{(%
\vec{\ell}_{2})}\left( f\right) \ .  \label{eeeeeeeeeeeeee}
\end{equation}%
In general, $\rho ^{(\vec{\ell}_{1})}$ and $\rho ^{(\vec{\ell}_{2})}$ are
two different extensions to $\mathfrak{U}$ of the state $\rho \in E_{\vec{%
\ell}_{2}}\subseteq E_{\vec{\ell}_{1}}$. In other words, the state of $%
\mathfrak{U}^{\ast }$ defined by (\ref{ddddd}) is $\vec{\ell}$-dependent.
For instance, the set $E_{(1,\ldots ,1)}$ of all translation-invariant
states satisfies 
\begin{equation*}
E_{(1,\ldots ,1)}\subseteq \bigcap_{\vec{\ell}\in \mathbb{N}^{d}}E_{\vec{\ell%
}}
\end{equation*}%
and an arbitrary $\rho \in E_{(1,\ldots ,1)}$ generally leads to a different
extended state of $\mathfrak{U}^{\ast }$ for each $\vec{\ell}\in \mathbb{N}%
^{d}$.

A very nice characterization of cyclic representations of such an extension
of a periodic state $\rho \in E_{\vec{\ell}}\subseteq \mathcal{U}^{\ast }$
is given in Theorem \ref{coro Extension of states}, the main assertions of
which can be phrased as follows, for lattice fermion systems:

\begin{proposition}[Cyclic representations of periodic states]
\label{coro Extension of states copy(1)}\mbox{ }\newline
Fix $\vec{\ell}\in \mathbb{N}^{d}$ and $\rho \in E_{\vec{\ell}}$, seen as a
state of either $\mathcal{U}^{\ast }$ or $\mathfrak{U}^{\ast }$\emph{.}%
\newline
\emph{(i)} Let $(\mathcal{H}_{\rho },\pi _{\rho },\Omega _{\rho })$ be any
cyclic representation of $\rho \in \mathcal{U}^{\ast }$. Then, there exists
a unique representation $\Pi _{\rho }$ of $\mathfrak{U}$ on $\mathcal{H}%
_{\rho }$ such that $\Pi _{\rho }|_{\mathcal{U}}=\pi _{\rho }$ and $(%
\mathcal{H}_{\rho },\Pi _{\rho },\Omega _{\rho })$ is a cyclic
representation of $\rho \in \mathfrak{U}^{\ast }$. \newline
\emph{(ii)} Conversely, let $(\mathcal{H}_{\rho },\Pi _{\rho },\Omega _{\rho
})$ be any cyclic representation of $\rho \in \mathfrak{U}^{\ast }$. Then, $(%
\mathcal{H}_{\rho },\Pi _{\rho }|_{\mathcal{U}},\Omega _{\rho })$ is a
cyclic representation of $\rho \in \mathcal{U}^{\ast }$,%
\begin{equation*}
\lbrack \Pi _{\rho }\left( \mathfrak{U}\right) ]^{\prime \prime }=[\Pi
_{\rho }\left( \mathcal{U}\right) ]^{\prime \prime }\qquad \text{and}\qquad
\lbrack \Pi _{\rho }\left( \mathfrak{C}\right) ]^{\prime \prime }\subseteq
\lbrack \Pi _{\rho }\left( \mathcal{U}\right) ]^{\prime }\cap \lbrack \Pi
_{\rho }\left( \mathcal{U}\right) ]^{\prime \prime }\ .
\end{equation*}
\end{proposition}

\begin{proof}
Apply Theorem \ref{coro Extension of states} for $\mathcal{X}=\mathcal{U}$, $%
F=E$ and $\mu =\mu _{\rho }$, observing that $\mu _{\rho }$ is the unique
Choquet measure, relative to the simplex $E_{\vec{\ell}}$, representing $%
\rho $, i.e., $\rho $ is the barycenter of $\mu _{\rho }$. As a measure on $%
E $, $\mu _{\rho }$ is an orthogonal measure. See Theorem \ref{theorem
choquet}.
\end{proof}

\noindent Note that the representation $\Pi _{\rho }$ in Proposition \ref%
{coro Extension of states copy(1)} is of course $\vec{\ell}$-dependent.
Recall now the following objects introduced above:

\begin{itemize}
\item For any $\Lambda \in \mathcal{P}_{f}$ and $\mathfrak{m}\in C_{b}(%
\mathbb{R};\mathcal{M}_{\Lambda })$, $\mathbf{\varpi }^{\mathfrak{m}}\in
C\left( \mathbb{R}^{2};\mathrm{Aut}\left( E\right) \right) $ is the solution
to the self-consistency equation of Theorem \ref{theorem sdfkjsdklfjsdklfj},
proven in \cite[Theorem 6.5]{BruPedra-MFII}.

\item For any $\mathbf{\Psi }\in C(\mathbb{R};\mathfrak{W}^{\mathbb{R}})$, $(%
\mathfrak{T}_{t,s}^{\mathbf{\Psi }})_{s,t\in \mathbb{R}}$ is the strongly
continuous two-parameter family of $\ast $-auto%
\-%
morphisms of $\mathfrak{U}$ defined by Equation (\ref{definiotion tho frac}).

\item For any $\mathfrak{m}\in C\left( \mathbb{R};\mathcal{M}\right) $ and
each $\mathbf{\xi }\in C\left( \mathbb{R};\mathrm{Aut}\left( E\right)
\right) $, $\mathbf{\Phi }^{(\mathfrak{m},\mathbf{\xi })}$ is the mapping
from $\mathbb{R}$ to $\mathfrak{W}^{\mathbb{R}}$ of Definition \ref%
{definition BCS-type model approximated}. By (\ref{inequality trivial}), if $%
\mathfrak{m}\in C\left( \mathbb{R};\mathcal{M}\right) $ and $\mathbf{\xi }%
\in C\left( \mathbb{R};\mathrm{Aut}\left( E\right) \right) $ then $\mathbf{%
\Phi }^{(\mathfrak{m},\mathbf{\xi })}\in C(\mathbb{R};\mathfrak{W}^{\mathbb{R%
}})$.

\item $\mathcal{M}_{1}$ is the Banach space of all translation-invariant
long-range models defined by (\ref{translatino invariatn long range models}).
\end{itemize}

\noindent Using the orthogonality of the probability measures $\mu _{\rho }$
and the ergodicity of extreme states of $E_{\vec{\ell}}$, we obtain the
existence of the quantum part of long-range dynamics, which is the main
result of this paper:\ 

\begin{theorem}[Quantum part of long-range dynamics]
\label{theorem structure of omega}\mbox{ }\newline
Fix $\Lambda \in \mathcal{P}_{f}$, $\mathfrak{m}\in C_{b}(\mathbb{R};%
\mathcal{M}_{\Lambda }\cap \mathcal{M}_{1})$, $\vec{\ell}\in \mathbb{N}^{d}$
and $\rho \in E_{\vec{\ell}}$. Let $\left( \mathcal{H}_{\rho },\Pi _{\rho
},\Omega _{\rho }\right) $ be a cyclic representation of $\rho $, seen as a
state (\ref{ddddd}) of $\mathfrak{U}^{\ast }$. Then, for any $s,t\in \mathbb{%
R}$ and $A\in \mathcal{U}\subseteq \mathfrak{U}$, in the $\sigma $-weak
topology,%
\begin{equation}
\lim_{L\rightarrow \infty }\pi _{\rho }\left( \tau _{t,s}^{(L,\mathfrak{m}%
)}\left( A\right) \right) =\lim_{L\rightarrow \infty }\Pi _{\rho }\left(
\tau _{t,s}^{(L,\mathfrak{m})}\left( A\right) \right) =\left. \Pi _{\rho
}\left( \mathfrak{T}_{t,s}^{\mathbf{\Phi }^{(\mathfrak{m},\mathbf{\varpi }^{%
\mathfrak{m}}\left( \alpha ,\cdot \right) )}}\left( A\right) \right)
\right\vert _{\alpha =s}\in \mathcal{B}\left( \mathcal{H}_{\rho }\right) \ .
\label{tttt}
\end{equation}
\end{theorem}

\noindent If $\vec{\ell}_{1},\vec{\ell}_{2}\in \mathbb{N}^{d}$ is such that $%
\mathbb{Z}_{\vec{\ell}_{1}}^{d}\subseteq \mathbb{Z}_{\vec{\ell}_{2}}^{d}$
and $\rho \in E_{\vec{\ell}_{2}}\subseteq E_{\vec{\ell}_{1}}$ then the
extension of the state of $\mathfrak{U}^{\ast }$ taken in Theorem \ref%
{theorem structure of omega}, and thus the representation $\Pi _{\rho }$,
depends on whether one sees $\rho $ as an element of $E_{\vec{\ell}_{2}}$ or 
$E_{\vec{\ell}_{1}}$. However, the left-hand side of (\ref{tttt}) does \emph{%
not} depend on this choice, obviously. So the same is also true for the
right-hand side of (\ref{tttt}). This can directly be seen from (\ref%
{errrrrrrrtttgg})-(\ref{eeeeeeeeeeeeee}) and the direct integral
decomposition (\ref{extension1bis}) of $\Pi _{\rho }$, by proving that the
function 
\begin{equation*}
\rho \mapsto \mathfrak{T}_{t,s}^{\mathbf{\Phi }^{(\mathfrak{m},\mathbf{%
\varpi }^{\mathfrak{m}}\left( \alpha ,\rho \right) )}}\left( A\right)
\end{equation*}%
of $\mathfrak{U}$\ satisfies (\ref{errrrrrrrtttgg}). See Lemma \ref{lemma
extra} (ii). This optional proof is not done here.

Because of Theorem \ref{theorem structure of omega}, the state $\rho \circ
\tau _{t,s}^{(L,\mathfrak{m})}$ converges in the weak$^{\ast }$-topology to
the restriction to $\mathcal{U}$ of the state 
\begin{equation*}
\rho _{t,s}\doteq \rho \circ \mathfrak{T}_{t,s}^{\mathbf{\Phi }^{(\mathfrak{m%
},\mathbf{\varpi }^{\mathfrak{m}}\left( \alpha ,\cdot \right) )}}|_{\alpha
=s}\in \mathfrak{U}^{\ast }\ .
\end{equation*}%
This restriction is thus, by definition, the ($\vec{\ell}$-periodic) state
of the system at time $t\in \mathbb{R}$ when the state at initial time $s$
is $\rho \in E_{\vec{\ell}}$. Since the exact time evolution of long-range
order takes places in a $C^{\ast }$-algebra $\mathfrak{U}$ larger than $%
\mathcal{U}$, one expects that, in general, the mapping $\rho \mapsto \rho
_{t,s}|_{\mathcal{U}}$ from $E_{\vec{\ell}}$ to itself does \emph{not}
preserve the entropy density of the initial state.

Before giving the proof of Theorem \ref{theorem structure of omega}, we
first explain its heuristics: Take $\Lambda \in \mathcal{P}_{f}$, $\mathfrak{%
m}\in C_{b}(\mathbb{R};\mathcal{M}_{\Lambda }\cap \mathcal{M}_{1})$, $\vec{%
\ell}\in \mathbb{N}^{d}$ and any extreme (or ergodic) state $\hat{\rho}\in 
\mathcal{E}(E_{\vec{\ell}})\subseteq E_{\vec{\ell}}$. In this case, the
(unique) probability measure $\mu _{\hat{\rho}}$ of Theorem \ref{theorem
choquet} is the atomic one with the singleton $\{\hat{\rho}\}$ as its
support. See (\ref{decomposition choquet}). Then, 
\begin{equation*}
\Pi _{\hat{\rho}}\left( f\right) =\pi _{\hat{\rho}}\left( f\left( \hat{\rho}%
\right) \right) \ ,\qquad f\in \mathfrak{U}\ .
\end{equation*}%
Compare with Equation (\ref{extension1bis}) below. In particular, $\Pi _{%
\hat{\rho}}(\mathfrak{U})=\pi _{\hat{\rho}}(\mathcal{U})$. Similarly, for
any $s,t\in \mathbb{R}$, 
\begin{equation*}
\Pi _{\hat{\rho}}\left( \mathfrak{T}_{t,s}^{\mathbf{\Psi }}\left( f\right)
\right) =\pi _{\hat{\rho}}\left( \tau _{t,s}^{\mathbf{\Psi }\left( \hat{\rho}%
\right) }\left( f\left( \hat{\rho}\right) \right) \right) \ ,\qquad \mathbf{%
\Psi }\in C(\mathbb{R};\mathfrak{W}^{\mathbb{R}})\ ,
\end{equation*}%
using the notation (\ref{notation state interactionbis}). By Theorem \ref%
{theorem structure of omega}, for any $s,t\in \mathbb{R}$ and $A\in \mathcal{%
U}$, 
\begin{equation*}
\pi _{\hat{\rho}}\left( \tau _{t,s}^{\mathbf{\Phi }^{(\mathfrak{m},\mathbf{%
\varpi }^{\mathfrak{m}}\left( \alpha ,\cdot \right) )}(\hat{\rho})}\left(
A\right) |_{\alpha =s}-\tau _{t,s}^{(L,\mathfrak{m})}\left( A\right) \right)
\end{equation*}%
$\sigma $-weak converges to $0$, as $L\rightarrow \infty $. In fact, the
derivation of this statement for extreme (ergodic) states is the starting
point of the proof of Theorem \ref{theorem structure of omega} and
corresponds to Theorem \ref{theorem structure of omega copy(2)}.

Now, if the periodic state is non-extreme, we use its Choquet decomposition
(on extreme states), as stated in Theorem \ref{theorem choquet}. To
illustrate this, take, for instance, any state of the form 
\begin{equation*}
\rho _{\lambda }=\left( 1-\lambda \right) \hat{\rho}_{0}+\lambda \hat{\rho}%
_{1}\ ,\qquad \lambda \in \left( 0,1\right) ,\ \hat{\rho}_{0}\neq \hat{\rho}%
_{1}\in \mathcal{E}(E_{\vec{\ell}})\ ,
\end{equation*}%
i.e., $\rho _{\lambda }$ is a non-trivial convex combination of two
different extreme states of $E_{\vec{\ell}}$. Assume\footnote{%
Note that the cyclicity of the representation is unclear for general
(non-ergodic) states $\hat{\rho}_{0},\hat{\rho}_{1}$. See Section \ref%
{Direct Integrals of GNS}.} that 
\begin{equation*}
\left( \mathcal{H}_{\rho _{\lambda }},\pi _{\rho _{\lambda }},\Omega _{\rho
_{\lambda }}\right) =\left( \mathcal{H}_{\hat{\rho}_{0}}\oplus \mathcal{H}_{%
\hat{\rho}_{1}},\pi _{\hat{\rho}_{0}}\oplus \pi _{\hat{\rho}_{1}},\sqrt{%
1-\lambda }\Omega _{\hat{\rho}_{0}}\oplus \sqrt{\lambda }\Omega _{\hat{\rho}%
_{1}}\right)
\end{equation*}%
is a cyclic representation of $\rho _{\lambda }$ for $\lambda \in \left(
0,1\right) $, where, as before, $\left( \mathcal{H}_{\rho },\pi _{\rho
},\Omega _{\rho }\right) $ is any cyclic representation of $\rho \in E_{\vec{%
\ell}}$. In this case, for any $f\in \mathfrak{U}$, 
\begin{equation*}
\Pi _{\rho _{\lambda }}\left( f\right) =\pi _{\hat{\rho}_{0}}\left( f\left( 
\hat{\rho}_{0}\right) \right) \oplus \pi _{\hat{\rho}_{1}}\left( f\left( 
\hat{\rho}_{1}\right) \right) \qquad \text{and}\qquad \Pi _{\rho _{\lambda
}}(\mathfrak{U})=\pi _{\hat{\rho}_{0}}(\mathcal{U})\oplus \pi _{\hat{\rho}%
_{1}}(\mathcal{U})\ .
\end{equation*}%
Compare with Equation (\ref{extension1bis})\footnote{%
The representation $\Pi _{\rho _{\lambda }}\simeq \Pi _{\rho _{\lambda
}}^{\oplus }$ is only defined up to some unitary equivalence. This detail is
not important here.} below. It follows that, for any $s,t\in \mathbb{R}$, 
\begin{equation*}
\Pi _{\rho _{\lambda }}\left( \mathfrak{T}_{t,s}^{\mathbf{\Psi }}\left(
f\right) \right) =\pi _{\hat{\rho}_{0}}\left( \tau _{t,s}^{\mathbf{\Psi }%
\left( \hat{\rho}_{0}\right) }\left( f\left( \hat{\rho}_{0}\right) \right)
\right) \oplus \pi _{\hat{\rho}_{1}}\left( \tau _{t,s}^{\mathbf{\Psi }\left( 
\hat{\rho}_{1}\right) }\left( f\left( \hat{\rho}_{1}\right) \right) \right)
\ ,\qquad \mathbf{\Psi }\in C(\mathbb{R};\mathfrak{W}^{\mathbb{R}})\ .
\end{equation*}%
By Theorem \ref{theorem structure of omega} applied to the extreme states $%
\hat{\rho}_{0},\hat{\rho}_{1}\in \mathcal{E}(E_{\vec{\ell}})$, for any $t\in 
\mathbb{R}$ and $A\in \mathcal{U}$,%
\begin{equation*}
\pi _{\hat{\rho}_{0}}\left( \tau _{t,s}^{\mathbf{\Phi }^{(\mathfrak{m},%
\mathbf{\varpi }^{\mathfrak{m}}\left( \alpha ,\cdot \right) )}(\hat{\rho}%
_{0})}\left( A\right) \right) \oplus \pi _{\hat{\rho}_{1}}\left( \tau
_{t,s}^{\mathbf{\Phi }^{(\mathfrak{m},\mathbf{\varpi }^{\mathfrak{m}}\left(
\alpha ,\cdot \right) )}(\hat{\rho}_{1})}\left( A\right) \right) |_{\alpha
=s}-\pi _{\rho _{\lambda }}\left( \tau _{t,s}^{(L,\mathfrak{m})}\left(
A\right) \right)
\end{equation*}%
$\sigma $-weak converges to $0$, as $L\rightarrow \infty $. The proof of
Theorem \ref{theorem structure of omega} in the general case is essentially
the same, except that one has to use the direct integral decomposition
theory instead of finite direct sums. This theory for non-constant Hilbert
spaces and von Neumann algebras as well as for GNS representations of a
family of states is highly non-trivial, but it is a mature subject of
mathematics. We review it in Section \ref{app direct integrals} and use this
theory to prove below Theorem \ref{theorem structure of omega} in the
general case: \bigskip

\begin{proof}
Fix $\Lambda \in \mathcal{P}_{f}$, $\mathfrak{m}\in C_{b}(\mathbb{R};%
\mathcal{M}_{\Lambda }\cap \mathcal{M}_{1})$ and $\vec{\ell}\in \mathbb{N}%
^{d}$. For any $\rho \in E_{\vec{\ell}}$, there is a unique probability
measure $\mu _{\rho }$ on $E_{\vec{\ell}}$ satisfying (\ref{decomposition
choquet}). By Theorem \ref{theorem choquet}, this probability measure is
orthogonal and supported in the Borel set $\mathcal{E}(E_{\vec{\ell}})$
(which is not a closed set). Therefore, by orthogonality of the measure $\mu
_{\rho }$ (Theorem \ref{theorem choquet}) and Theorem \ref{coro Extension of
states} (i) with $\mathcal{X}=\mathcal{U}$ and $F=E_{\vec{\ell}}$, a cyclic
representation of $\rho \in \mathfrak{U}^{\ast }$ is given by $(\mathcal{H}%
_{\rho }^{\oplus },\Pi _{\rho }^{\oplus },\Omega _{\rho }^{\oplus })$ with 
\begin{equation}
\mathcal{H}_{\rho }^{\oplus }\doteq \int_{E_{\vec{\ell}}}\mathcal{H}_{\hat{%
\rho}}\mu _{\rho }\left( \mathrm{d}\hat{\rho}\right) ,\quad \Omega _{\rho
}^{\oplus }\doteq \int_{E_{\vec{\ell}}}\Omega _{\hat{\rho}}\mu _{\rho
}\left( \mathrm{d}\hat{\rho}\right)  \label{extension1}
\end{equation}%
and $\Pi _{\rho }^{\oplus }$ being the (direct integral) representation\ of $%
\mathfrak{U}$ on $\mathcal{H}_{\rho }^{\oplus }$ defined by%
\begin{equation}
\Pi _{\rho }^{\oplus }\left( f\right) \doteq \int_{E_{\vec{\ell}}}\pi _{\hat{%
\rho}}(f(\hat{\rho}))\mu _{\rho }(\mathrm{d}\hat{\rho})\ ,\qquad f\in 
\mathfrak{U}\ ,  \label{extension1bis}
\end{equation}%
where $\left( \mathcal{H}_{\hat{\rho}},\pi _{\hat{\rho}},\Omega _{\hat{\rho}%
}\right) $ in all integrals are always\emph{\ }the GNS representation of $%
\hat{\rho}\in E_{\vec{\ell}}$. Since $\mu _{\rho }$ is supported in the
Borel set $\mathcal{E}(E_{\vec{\ell}})$ (Theorem \ref{theorem choquet}), we
can restrict all integrals of (\ref{extension1})-(\ref{extension1bis}) to $%
\mathcal{E}(E_{\vec{\ell}})$. In particular, by (\ref{definiotion tho frac}),%
\begin{equation*}
\Pi _{\rho }^{\oplus }\left( \mathfrak{T}_{t,s}^{\mathbf{\Phi }^{(\mathfrak{m%
},\mathbf{\varpi }^{\mathfrak{m}}\left( \alpha ,\cdot \right) )}}\left(
A\right) |_{\alpha =s}-\tau _{t,s}^{(L,\mathfrak{m})}\left( A\right) \right)
=\int_{\mathcal{E}(E_{\vec{\ell}})}\pi _{\hat{\rho}}\left( \tau _{t,s}^{%
\mathbf{\Phi }^{(\mathfrak{m},\mathbf{\varpi }^{\mathfrak{m}}\left( \alpha
,\cdot \right) )}(\hat{\rho})}\left( A\right) |_{\alpha =s}-\tau _{t,s}^{(L,%
\mathfrak{m})}\left( A\right) \right) \mu \left( \mathrm{d}\hat{\rho}\right)
\ .
\end{equation*}%
Since $\tau _{t,s}^{\mathbf{\Phi }^{(\mathfrak{m},\mathbf{\varpi }^{%
\mathfrak{m}}\left( s,\cdot \right) )}(\hat{\rho})}$ and $\tau _{t,s}^{(L,%
\mathfrak{m})}$ are both $\ast $-automorphisms of $\mathcal{U}$, by
Lebesgue's dominated convergence theorem together with Theorem \ref{theorem
structure of omega copy(2)} and (\ref{extension1}), we arrive at the
assertion of Theorem \ref{theorem structure of omega} for the representation 
$\Pi _{\rho }^{\oplus }$. By \cite[Theorem 2.3.16]{BrattelliRobinsonI}, any
cyclic representation $(\mathcal{H}_{\rho },\Pi _{\rho },\Omega _{\rho })$
of $\rho \in \mathfrak{U}^{\ast }$ is unitarily equivalent to $\Pi _{\rho
}^{\oplus }$, similar to (\ref{unitary}). This concludes the proof of
Theorem \ref{theorem structure of omega}.
\end{proof}

\begin{remark}[Subcentral decompositions of periodic states]
\label{theorem structure of omega copy(3)}\mbox{ }\newline
As explained after Corollary \ref{Effros theorem}, $\pi _{\rho }^{\oplus }$
is a subcentral decomposition of the representation $\pi _{\rho _{\mu }}$,
see Definition \ref{Special direct integral representations} (ii.1). By \cite%
[Eq. (4.15)]{BruPedra2}, the GNS representations of ergodic states are, in
general, not factor representations. By Theorem \ref{central decompositions}
(ii), it follows that $\pi _{F}^{\oplus }$ is not, in general, the central
decomposition of the representation $\pi _{\rho _{\mu }}$.
\end{remark}

Note that, even if the finite-volume dynamics is autonomous, i.e., $%
\mathfrak{m}\in \mathcal{M}_{\Lambda }\cap \mathcal{M}_{1}$, the limit
long-range dynamics is generally \emph{non-autonomous}, as it can be seen
from the next corollary:

\begin{corollary}[From autonomous local dynamics to non-autonomous ones]
\label{theorem structure of omega copy(1)}\mbox{ }\newline
Fix $\Lambda \in \mathcal{P}_{f}$, $\mathfrak{m}\in \mathcal{M}_{\Lambda
}\cap \mathcal{M}_{1}$, $\vec{\ell}\in \mathbb{N}^{d}$ and $\rho \in E_{\vec{%
\ell}}$ with cyclic representation $(\mathcal{H}_{\rho },\Pi _{\rho },\Omega
_{\rho })$ seen as a state (\ref{ddddd}) of $\mathfrak{U}^{\ast }$. Then,
for any $s,t\in \mathbb{R}$ and $A\in \mathcal{U}\subseteq \mathfrak{U}$, in
the $\sigma $-weak topology,%
\begin{equation*}
\lim_{L\rightarrow \infty }\Pi _{\rho }\left( \tau _{t-s}^{(L,\mathfrak{m}%
)}\left( A\right) \right) =\left. \Pi _{\rho }\left( \mathfrak{T}_{t,s}^{%
\mathbf{\Phi }^{(\mathfrak{m},\mathbf{\varpi }^{\mathfrak{m}}\left( \alpha
,\cdot \right) )}}\left( A\right) \right) \right\vert _{\alpha =s}\in 
\mathcal{B}\left( \mathcal{H}_{\rho }\right) \ .
\end{equation*}
\end{corollary}

\section{Technical Proofs\label{dddd}}

\subsection{Cyclic Representations of Positive Functionals and Orthogonal
Measures\label{Positive Functionals}}

The dual $\mathcal{U}^{\ast }$ of the $C^{\ast }$-algebra $\mathcal{U}$ is a
locally convex space with respect to the weak$^{\ast }$-topology, which is
Hausdorff. Moreover, as $\mathcal{U}$ is separable, by \cite[Theorem 3.16]%
{Rudin}, the weak$^{\ast }$-topology is metrizable on any weak$^{\ast }$%
-compact subset of $\mathcal{U}^{\ast }$.

An important subset of $\mathcal{U}^{\ast }$ is the weak$^{\ast }$-closed
convex cone of positive functionals defined by 
\begin{equation}
\mathcal{U}_{+}^{\ast }\doteq \bigcap\limits_{A\in \mathcal{U}}\{\rho \in 
\mathcal{U}^{\ast }:\rho (A^{\ast }A)\geq 0\}\ .
\label{set of positive linear}
\end{equation}%
Equivalently, $\rho \in \mathcal{U}_{+}^{\ast }$ iff $\Vert \rho \Vert _{%
\mathcal{U}^{\ast }}=\rho (\mathfrak{1})$. Additionally, any positive
functional $\rho \in \mathcal{U}_{+}^{\ast }$ is hermitian, i.e., for all $%
A\in \mathcal{U}$, $\rho (A^{\ast })=\overline{\rho (A)}$.

By the GNS construction, any positive functional $\rho \in \mathcal{U}%
_{+}^{\ast }$ has a cyclic representation $(\mathcal{H}_{\rho },\pi _{\rho
},\Omega _{\rho })$: There exists a Hilbert space $\mathcal{H}_{\rho }$, a
representation\footnote{%
I.e., it is a $\ast $--homomorphism from $\mathcal{U}$ to $\mathcal{B}(%
\mathcal{H}_{\rho })$. See, e.g., \cite[Definition 2.3.2]{BrattelliRobinsonI}%
.} $\pi _{\rho }$ from $\mathcal{U}$ to the unital $C^{\ast }$--algebra $%
\mathcal{B}(\mathcal{H}_{\rho })$ of bounded operators on $\mathcal{H}_{\rho
}$ and a cyclic vector\footnote{%
I.e., $\mathcal{H}_{\rho }$ is the closure of (the linear span of) the set $%
\pi _{\rho }(\mathcal{U})\Omega _{\rho }\doteq \left\{ \pi _{\rho }(A)\Omega
_{\rho }:A\in \mathcal{U}\right\} $. See \cite[p. 45]{BrattelliRobinsonI}.} $%
\Omega _{\rho }\in \mathcal{H}_{\rho }$ for $\pi _{\rho }(\mathcal{U})$ such
that%
\begin{equation}
\rho (A)=\langle \Omega _{\rho },\pi _{\rho }(A)\Omega _{\rho }\rangle _{%
\mathcal{H}_{\rho }}\ ,\qquad A\in \mathcal{U}\ .  \label{equality}
\end{equation}%
The representation $\pi _{\rho }$ is faithful\footnote{%
See for instance \cite[Proposition 2.3.3]{BrattelliRobinsonI}.} if $\rho $
is faithful, that is, if $\rho (A^{\ast }A)=0$ implies $A=0$. The triple $(%
\mathcal{H}_{\rho },\pi _{\rho },\Omega _{\rho })$ is unique up to unitary
equivalence. See, e.g., \cite[Theorem 2.3.16]{BrattelliRobinsonI}, which can
trivially be extended to any positive functional. Two positive linear
functionals $\rho _{1},\rho _{2}\in \mathcal{U}_{+}^{\ast }$ are said to be 
\emph{orthogonal} whenever%
\begin{equation}
(\mathcal{H}_{\rho _{1}}\oplus \mathcal{H}_{\rho _{2}},\pi _{\rho
_{1}}\oplus \pi _{\rho _{2}},\Omega _{\rho _{1}}\oplus \Omega _{\rho _{2}})
\label{orthogonality}
\end{equation}%
is a cyclic representation for the positive functional $\rho _{1}+\rho
_{2}\in \mathcal{U}_{+}^{\ast }$. As is usual, this orthogonality property
is denoted by $\rho _{1}\perp \rho _{2}$. See, e.g., \cite[Lemma 4.1.19 and
Definition 4.1.20]{BrattelliRobinsonI}.

The set of states on $\mathcal{U}$ is the subset of $\mathcal{U}_{+}^{\ast }$
defined by (\ref{states CAR}), that is,%
\begin{equation*}
E\doteq \{\rho \in \mathcal{U}^{\ast }:\rho \geq 0,\ \rho (\mathfrak{1}%
)=1\}=\{\rho \in \mathcal{U}^{\ast }:\Vert \rho \Vert _{\mathcal{U}^{\ast
}}=\rho (\mathfrak{1})=1\}\ .
\end{equation*}%
Hence, $E$ is a weak$^{\ast }$-closed subset of the unit ball of $\mathcal{U}%
^{\ast }$ and, by the Banach-Alaoglu theorem, $E$ is weak$^{\ast }$-compact
and metrizable. See, e.g., \cite[Theorems 3.15-3.16]{Rudin}.

Let $\Sigma _{E}$ be the (Borel) $\sigma $-algebra generated by weak$^{\ast
} $-closed, or weak$^{\ast }$-open, subsets of $E$. The set of all positive
Radon measures on $(E,\Sigma _{E})$ is denoted by $\mathrm{M}(E)$. By weak$%
^{\ast }$-compactness of $E$ and the Riesz(-Markov) representation theorem,
there is a one-to-one correspondence between positive functionals of $%
\mathfrak{C}^{\ast }$ and positive Radon measures on $(E,\Sigma _{E})$ and
we write%
\begin{equation}
\mu \left( f\right) =\int_{E}f\left( \rho \right) \mu \left( \mathrm{d}\rho
\right) \ ,\qquad f\in \mathfrak{C}\doteq C\left( E;\mathbb{C}\right) \ .
\label{barycenter1bis}
\end{equation}%
By metrizability of $E$, note additionally that any positive finite Borel
measure on $(E,\Sigma _{E})$ is a positive Radon measure\footnote{%
In fact, for compact metrizable spaces, the Baire and Borel $\sigma $%
-algebras are the same.}. When $\mu (E)=1$ we say that the positive Radon
measure is normalized and $\mu $ is a probability measure. The subset of all
probability measures on $(E,\Sigma _{E})$ is denoted by $\mathrm{M}_{1}(E)$.

For each $\mu \in \mathrm{M}(E)$, we define its restriction $\mu _{\mathfrak{%
B}}\in \mathrm{M}(E)$ to any Borel set $\mathfrak{B}\in \Sigma _{E}$ by 
\begin{equation}
\mu _{\mathfrak{B}}\left( \mathfrak{B}_{0}\right) \doteq \mu \left( 
\mathfrak{B}_{0}\cap \mathfrak{B}\right) \ ,\qquad \mathfrak{B}_{0}\in
\Sigma _{E}\ .  \label{restriction radon}
\end{equation}%
Additionally, any convex and weak$^{\ast }$-closed subset $F$ of $E$\
defines a partial order $\prec _{F}$ in $\mathrm{M}(E)$: For any $\mu ,\nu
\in \mathrm{M}(E)$, $\mu \prec _{F}\nu $ if $\mu _{F}\left( f\right) \leq
\nu _{F}\left( f\right) $ for all weak$^{\ast }$-continuous convex functions 
$f:F\rightarrow \mathbb{R}$.

Each positive Radon measure $\mu \in \mathrm{M}(E)$, or equivalently a
positive finite Borel measure, represents a positive functional $\rho _{\mu
}\in \mathcal{U}_{+}^{\ast }$, which is defined by (\ref{barycenter1bis})
for $f=\hat{A}$, that is,%
\begin{equation}
\rho _{\mu }\left( A\right) \doteq \mu (\hat{A})=\int_{E}\rho \left(
A\right) \mu \left( \mathrm{d}\rho \right) \ ,\qquad A\in \mathcal{U}\ .
\label{choquet0}
\end{equation}%
Recall that $\rho \mapsto \hat{A}\left( \rho \right) \doteq \rho \left(
A\right) $, as defined by (\ref{fA}), is an affine and weak$^{\ast }$%
-continuous mapping from $E$ to $\mathbb{C}$. The positive functional $\rho
_{\mu }$ is called the \emph{barycenter}\footnote{%
Other terminology existing in the literature: \textquotedblleft $x$ is
represented by $\mu $\textquotedblright , \textquotedblleft $x$ is the
resultant of $\mu $\textquotedblright .} of $\mu \in \mathrm{M}(E)$. See,
e.g., \cite[Eq. (2.7) in Chapter I]{Alfsen}, \cite[p. 1]{Phe} or \cite[%
Definition 10.15]{BruPedra2}. By \cite[Propositions 1.1 and 1.2]{Phe},
barycenters are uniquely defined for all positive Radon measure in convex
compact subsets of locally convex spaces and the mapping $\mu \mapsto \rho
_{\mu }$ from $\mathrm{M}(E)$ to $\mathcal{U}_{+}^{\ast }$ is affine and weak%
$^{\ast }$-continuous. Clearly, if $\mu \in \mathrm{M}_{1}(E)$ then $\rho
_{\mu }\in E$. For any Borel subset $\mathfrak{B}\in \Sigma _{E}$, define by 
\begin{equation}
\mathrm{M}^{(\rho )}\left( \mathfrak{B}\right) \doteq \left\{ \mu \in 
\mathrm{M}\left( E\right) :\rho =\rho _{\mu }\quad \text{and}\quad \mu
\left( \mathfrak{B}\right) =\mu \left( E\right) \right\} \ ,
\label{choquet1}
\end{equation}%
the set of all positive Radon measures representing $\rho $ and with support
within $\mathfrak{B}\subseteq E$.

As positive functional, the barycenter $\rho _{\mu }\in \mathcal{U}%
_{+}^{\ast }$ of any positive Radon measure $\mu \in \mathrm{M}(E)$ has a
cyclic representation. A natural way to construct a triple $(\mathcal{H}%
_{\rho _{\mu }},\pi _{\rho _{\mu }},\Omega _{\rho _{\mu }})$ satisfying (\ref%
{equality}) for $\rho =\rho _{\mu }$ is to take the \emph{direct integral}
of cyclic representations $(\mathcal{H}_{\rho },\pi _{\rho },\Omega _{\rho
}) $ of $\rho \in E$ with respect to the measure $\mu \in \mathrm{M}(E)$: 
\begin{equation}
\mathcal{H}_{\rho _{\mu }}\doteq \int_{E}\mathcal{H}_{\rho }\mu \left( 
\mathrm{d}\rho \right) ,\quad \pi _{\rho _{\mu }}\doteq \int_{E}\pi _{\rho
}\mu \left( \mathrm{d}\rho \right) ,\quad \Omega _{\rho _{\mu }}\doteq
\int_{E}\Omega _{\rho }\mu \left( \mathrm{d}\rho \right) \ .
\label{direct represe}
\end{equation}%
See Section \ref{Direct Integrals of GNS}. However, $\Omega _{\rho _{\mu }}$
is, in general, \emph{not} a cyclic vector for $\pi _{\rho _{\mu }}(\mathcal{%
U})$, i.e., the subset $\pi _{\rho _{\mu }}(\mathcal{U})\Omega _{\rho _{\mu
}}$ is generally not dense in $\mathcal{H}_{\rho _{\mu }}$. In particular,
in this case, (\ref{direct represe}) is not spatially, or unitarily,
equivalent to the cyclic representation of $\rho _{\mu }$. A necessary and
sufficient condition on the measure $\mu $ to get the cyclicity of $\Omega
_{\rho _{\mu }}$ is given by the orthogonality of the measure $\mu $, in the
following sense: The measure $\mu \in \mathrm{M}(E)$ is called \emph{%
orthogonal} whenever $\rho _{\mu _{\mathfrak{B}}}\perp \rho _{\mu
_{E\backslash \mathfrak{B}}}$ for \emph{any} $\mathfrak{B}\in \Sigma _{E}$.
See Definition \ref{Orthogonal measures} as well as \cite[Definition 4.1.20]%
{BrattelliRobinsonI} for more details. The set of all orthogonal measures on 
$(E,\Sigma _{E})$ is denoted by $\mathcal{O}\left( E\right) $.

For appropriate convex and weak$^{\ast }$-closed subsets $F$ of $E$, the
orthogonality of a measure $\mu \in \mathrm{M}(E)$ can be directly related
to its maximality with respect to the partial order $\prec _{F}$ in $\mathrm{%
M}^{(\rho _{\mu })}(F)$. Examples of such a $F$ are the subsets of periodic
states discussed in the next section.

\subsection{Ergodic Orthogonal Decomposition of Periodic States}

As the state space $E$, for any $\vec{\ell}\in \mathbb{N}^{d}$, the set $E_{%
\vec{\ell}}$ of $\vec{\ell}$-periodic states defined by (\ref{periodic
invariant states}) is metrizable, weak$^{\ast }$-compact and convex, with $%
\mathcal{E}(E_{\vec{\ell}})$ denoting its (non-empty) set of extreme points.
See Equation (\ref{cov heull l perio}). By metrizability of $E_{\vec{\ell}}$%
, $\mathcal{E}(E_{\vec{\ell}})$\ is a Borel set\footnote{%
It is even a $G_{\delta }$ set. See, e.g., \cite[Proposition 1.3]{Phe}.}.
Ergo, from the Choquet theorem \cite[Theorem 10.18]{BruPedra2}, each state $%
\rho \in E_{\vec{\ell}}$ is the barycenter of a probability measure $\mu
_{\rho }$ which is supported on the set $\mathcal{E}(E_{\vec{\ell}})$ of
extreme $\vec{\ell}$-periodic states:

\begin{theorem}[Ergodic orthogonal decomposition of periodic states]
\label{theorem choquet}\mbox{ }\newline
For any $\vec{\ell}\in \mathbb{N}^{d}$ and $\rho \in E_{\vec{\ell}}$, there
is a unique probability measure $\mu _{\rho }\equiv \mu _{\rho }^{(\vec{\ell}%
)}\in \mathrm{M}^{(\rho )}\left( E_{\vec{\ell}}\right) $ with $\mu _{\rho
}\left( \mathcal{E}(E_{\vec{\ell}})\right) =1$. Moreover, $\mu _{\rho }\in 
\mathcal{O}\left( E\right) $, i.e., it is an orthogonal measure on $%
(E,\Sigma _{E})$.
\end{theorem}

\begin{proof}
The first assertion corresponds to \cite[Theorem 1.9]{BruPedra2}. In order
to prove that $\mu _{\rho }\in \mathcal{O}\left( E\right) $, observe first
that the set $\tilde{E}^{+}$ of all states on $\mathcal{U}^{+}$ (cf. (\ref%
{definition of even operators}))\ can be identified with the even-state
space $E^{+}$ defined by (\ref{gauge invariant states}): A functional $\rho
\in \mathcal{U}^{\ast }$ is even iff $\rho \circ \sigma =\rho $, with\ $%
\sigma $ being the unique $\ast $-automorphism of the $C^{\ast }$-algebra $%
\mathcal{U}$ defined by (\ref{automorphism gauge invariance}). Any even
functional $\rho \in \mathcal{U}^{\ast }$\ can be seen as a functional $%
\tilde{\rho}=\rho |_{\mathcal{U}^{+}}\in (\mathcal{U}^{+})^{\ast }$, by
restriction. Conversely, any functional $\tilde{\rho}\in (\mathcal{U}%
^{+})^{\ast }$ defines an even functional 
\begin{equation*}
\rho \doteq \tilde{\rho}\circ \left( \frac{\sigma +\mathbf{1}_{\mathcal{U}}}{%
2}\right) \in \mathcal{U}^{\ast }
\end{equation*}%
on the $C^{\ast }$-algebra $\mathcal{U}$. Both mappings $\rho \mapsto \tilde{%
\rho}$ and $\tilde{\rho}\mapsto \rho $ are linear, weak$^{\ast }$-continuous
and order-preserving. Additionally, these mappings preserve states, i.e., $%
\rho \in \mathcal{U}^{\ast }$ is a state iff $\tilde{\rho}\in (\mathcal{U}%
^{+})^{\ast }$ is a state. In particular, the mapping $\rho \mapsto \tilde{%
\rho}$ is bijective, by \cite[proof of Proposition 2.1]{BruPedra-MFII}. As a
consequence, for any positive Radon measure $\mu \in \mathrm{M}(E)$
supported on $E^{+}$ with barycenter $\rho _{\mu }\in \mathcal{U}^{\ast }$, $%
\tilde{\rho}_{\mu }\in (\mathcal{U}^{+})^{\ast }$ is the barycenter of the
pushforward of the measure $\mu $ through the mapping $\rho \mapsto \tilde{%
\rho}$. Conversely, if $\tilde{\mu}$ is a positive Radon measure on the set $%
\tilde{E}^{+}$ of all states on $\mathcal{U}^{+}$ with barycenter $\tilde{%
\rho}_{\tilde{\mu}}\in (\mathcal{U}^{+})^{\ast }$, then $\rho _{\tilde{\mu}}$
is the barycenter of the pushforward of the measure $\tilde{\mu}$ through
the mapping $\tilde{\rho}\mapsto \rho $ and the pushforward of the measure $%
\tilde{\mu}$ is supported on $E^{+}$, because it is invariant under the
pushforward through the mapping $\rho \mapsto \rho \circ \sigma $. By \cite[%
Lemma 4.1.19]{BrattelliRobinsonI}, two positive linear functionals $\rho
_{1},\rho _{2}$ on any $C^{\ast }$-algebra (like $\mathcal{U}$ or $\mathcal{U%
}^{+}$) are orthogonal iff the zero functional is the unique positive
functional on this $C^{\ast }$-algebra below $\rho _{1}$ and $\rho _{2}$. In
particular, the pushforwards of positive Radon measures, associated with the
mappings $\rho \mapsto \tilde{\rho}$ and $\tilde{\rho}\mapsto \rho $,
preserve the orthogonality of such measures: The fact that the pushforward
of positive Radon measures, associated with the mapping $\rho \mapsto \tilde{%
\rho}$, preserves orthogonality is clear, by the bijectivity, linearity and
order-preserving property of this mapping combined with \cite[Lemma 4.1.19]%
{BrattelliRobinsonI}. To show the converse assertion, take any orthogonal
positive Radon measure $\tilde{\mu}$ on $\tilde{E}^{+}$ and denote by $\mu
\in \mathrm{M}(E)$ its pushforward through the mapping $\tilde{\rho}\mapsto
\rho $. Fix any Borel set $\mathfrak{B}\in \Sigma _{E}$ and take any
positive functional $\omega $ on $\mathcal{U}$ below the positive
functionals $\rho _{\mu _{\mathfrak{B}}},\rho _{\mu _{E\backslash \mathfrak{B%
}}}\in \mathcal{U}^{\ast }$, see (\ref{restriction radon}). Then, the
restriction to $\mathcal{U}^{+}$ of $\omega $ is lower than the restrictions
to $\mathcal{U}^{+}$ of the states $\rho _{\mu _{\mathfrak{B}}},\rho _{\mu
_{E\backslash \mathfrak{B}}}$. Since $\mu $ is supported on $E^{+}$, for any 
$\mathfrak{B}\in \Sigma _{E}$, 
\begin{equation}
\mu _{\mathfrak{B}}=\mu _{\mathfrak{B}\cap E^{+}}\qquad \text{and}\qquad \mu
_{E\backslash \mathfrak{B}}=\mu _{E^{+}\backslash \mathfrak{B}}\ .
\label{orthogonal2}
\end{equation}%
Since $\mu $ is the pushforward through the mapping $\tilde{\rho}\mapsto
\rho $ of an \emph{orthogonal} positive Radon measure $\tilde{\mu}$, by \cite%
[Lemma 4.1.19]{BrattelliRobinsonI} combined with the bijectivity, linearity
and order-preserving property of the mapping, the unique \emph{even}
positive functional below $\rho _{\mu _{\mathfrak{B}}}$ and $\rho _{\mu
_{E\backslash \mathfrak{B}}}$ is the zero functional. Suppose now that $%
\omega $, not necessarily even, is a positive functional below $\rho _{\mu _{%
\mathfrak{B}}}$ and $\rho _{\mu _{E\backslash \mathfrak{B}}}$. In
particular, because $\rho _{\mu _{\mathfrak{B}}}$ and $\rho _{\mu
_{E\backslash \mathfrak{B}}}$ are even, $\left( \omega +\omega \circ \sigma
\right) /2$ is an even positive functional below $\rho _{\mu _{\mathfrak{B}%
}} $ and $\rho _{\mu _{E\backslash \mathfrak{B}}}$. In particular, 
\begin{equation*}
\omega =-\omega \circ \sigma \qquad \text{and}\qquad \left\Vert \omega
\right\Vert _{\mathcal{U}_{+}^{\ast }}=\omega \left( \mathfrak{1}\right) =0\
.
\end{equation*}%
By \cite[Lemma 4.1.19]{BrattelliRobinsonI}, the pushforward of positive
Radon measures, associated with the mapping $\tilde{\rho}\mapsto \rho $,
thus preserves orthogonality.

By \cite[Lemma 1.8, Corollary 4.3]{BruPedra2}, $E_{\vec{\ell}}\subseteq
E^{+}\equiv \tilde{E}^{+}$ for all $\vec{\ell}\in \mathbb{N}^{d}$. This
identification is pivotal because, in contrast with $\mathcal{U}$, the even $%
C^{\ast }$-subalgebra $\mathcal{U}^{+}$ is asymptotically abelian since%
\begin{equation}
\lim\limits_{|x|\rightarrow \infty }\left[ \alpha _{x}\left( A\right) ,B%
\right] =0\ ,\mathrm{\qquad }A\in \mathcal{U}^{+},B\in \mathcal{U}\ .
\label{space 00}
\end{equation}%
As is usual, $[A,B]\doteq AB-BA$. Therefore, from \cite[Propositions 4.3.3
and 4.3.7]{BrattelliRobinsonI}, it follows that $\mu _{\rho }$, which is
supported on $E_{\vec{\ell}}\subseteq E^{+}\equiv \tilde{E}^{+}$, is an
orthogonal measure, i.e., $\mu _{\rho }\in \mathcal{O}\left( E\right) $.
\end{proof}

By Theorem \ref{theorem choquet}, $E_{\vec{\ell}}$ is a Choquet simplex and
the mapping $\rho \mapsto \mu _{\rho }$ from $E_{\vec{\ell}}$ to $\mathrm{M}%
\left( E\right) $ ranges over orthogonal measures. The unique decomposition
of any $\rho \in E_{\vec{\ell}}$ in terms of extreme states $\hat{\rho}\in 
\mathcal{E}(E_{\vec{\ell}})$, given in Theorem \ref{theorem choquet}, is
also called the \emph{ergodic} decomposition of $\rho $: Fix $\vec{\ell}\in 
\mathbb{N}^{d}$ and define the space-averages of any element $A\in \mathcal{U%
}$ by%
\begin{equation}
A_{L}\equiv A_{L,\vec{\ell}}\doteq \frac{1}{|\Lambda _{L}\cap \mathbb{Z}_{%
\vec{\ell}}^{d}|}\sum\limits_{x\in \Lambda _{L}\cap \mathbb{Z}_{\vec{\ell}%
}^{d}}\alpha _{x}\left( A\right) \ ,\mathrm{\qquad }L\in \mathbb{N}\ .
\label{Limit of Space-Averages}
\end{equation}%
Then, by definition, a $\vec{\ell}$--periodic state $\hat{\rho}\in E_{\vec{%
\ell}}$ is ($\vec{\ell}$--) ergodic iff, for all $A\in \mathcal{U}$, 
\begin{equation}
\lim\limits_{L\rightarrow \infty }\hat{\rho}(A_{L}^{\ast }A_{L})=|\hat{\rho}%
(A)|^{2}\ .  \label{Ergodicity}
\end{equation}%
By \cite[Theorem 1.16]{BruPedra2}, any extreme state is ergodic and vice
versa. To be more precise,%
\begin{equation}
\mathcal{E}(E_{\vec{\ell}})=\left\{ \hat{\rho}\in E_{\vec{\ell}}:\hat{\rho}%
\text{ is }\vec{\ell}\text{-ergodic}\right\} \ ,\qquad \vec{\ell}\in \mathbb{%
N}^{d}\ .  \label{Ergodicity2}
\end{equation}%
Additionally, any extreme state $\hat{\rho}\in \mathcal{E}(E_{\vec{\ell}})$
is strongly clustering, i.e., for all $A,B\in \mathcal{U}$ and $x\in \mathbb{%
Z}_{\vec{\ell}}^{d}$,%
\begin{equation}
\lim\limits_{L\rightarrow \infty }\frac{1}{|\Lambda _{L}\cap \mathbb{Z}_{%
\vec{\ell}}^{d}|}\sum\limits_{y\in \Lambda _{L}\cap \mathbb{Z}_{\vec{\ell}%
}^{d}}\hat{\rho}\left( \alpha _{x}(A)\alpha _{y}(B)\right) =\hat{\rho}(A)%
\hat{\rho}(B)\ .  \label{strongly mixing}
\end{equation}

The ergodicity properties of extreme states has important consequences on
the structure of the sets of periodic states. For instance, the weak$^{\ast
} $-density of the set $\mathcal{E}(E_{\vec{\ell}})$ of extreme points of $%
E_{\vec{\ell}}$ is proven by using that $\vec{\ell}$-ergodic states are
extreme states of $E_{\vec{\ell}}$. See \cite[Proof of Corollary 4.6]%
{BruPedra2}. If $\vec{\ell}_{1},\vec{\ell}_{2}\in \mathbb{N}^{d}$ is such
that $\mathbb{Z}_{\vec{\ell}_{1}}^{d}\subseteq \mathbb{Z}_{\vec{\ell}%
_{2}}^{d}$ (see (\ref{group})), then $E_{\vec{\ell}_{2}}\subseteq E_{\vec{%
\ell}_{1}}$, but from (\ref{Ergodicity2}) one checks that an extreme state
of $E_{\vec{\ell}_{2}}$ is not necessarily an extreme state of $E_{\vec{\ell}%
_{1}}$, i.e., $\mathcal{E}(E_{\vec{\ell}_{2}})\nsubseteq \mathcal{E}(E_{\vec{%
\ell}_{1}})$.

For $\vec{\ell}_{j}\doteq \left( \ell _{j,1},\ldots ,\ell _{j,d}\right) \in 
\mathbb{N}^{d}$, $j\in \{1,2\}$, being such that $\mathbb{Z}_{\vec{\ell}%
_{1}}^{d}\subseteq \mathbb{Z}_{\vec{\ell}_{2}}^{d}$, we define the mapping $%
\mathfrak{x}_{\vec{\ell}_{1},\vec{\ell}_{2}}$ from $E$ to itself by 
\begin{equation}
\mathfrak{x}_{\vec{\ell}_{1},\vec{\ell}_{2}}\left( \rho \right) \doteq \frac{%
\ell _{2,1}\cdots \ell _{2,d}}{\ell _{1,1}\cdots \ell _{1,d}}%
\sum\limits_{x=(x_{1},\ldots ,x_{d}),\;x_{i}\in \{0,\ell _{2,i},2\ell
_{2,i},\ldots ,\ell _{1,i}-\ell _{2,i}\}}\rho \circ \alpha _{x}\ .
\label{space-average fcuntion}
\end{equation}%
This transformation allows us to relate the sets $\mathcal{E}(E_{\vec{\ell}%
_{1}})$ and $\mathcal{E}(E_{\vec{\ell}_{2}})$ to each other:

\begin{lemma}[From $\vec{\ell}_{1}$- to $\vec{\ell}_{2}$-periodic states]
\label{lemma extra}\mbox{ }\newline
Let $\vec{\ell}_{j}\doteq \left( \ell _{j,1},\ldots ,\ell _{j,d}\right) \in 
\mathbb{N}^{d}$, $j\in \{1,2\}$, be such that $\mathbb{Z}_{\vec{\ell}%
_{1}}^{d}\subseteq \mathbb{Z}_{\vec{\ell}_{2}}^{d}$ (see (\ref{group})).
Then, the transformation $\mathfrak{x}_{\vec{\ell}_{1},\vec{\ell}_{2}}$
defined by (\ref{space-average fcuntion}) has the following properties:\ 
\newline
\emph{(i)} It is weak$^{\ast }$-continuous, $\mathfrak{x}_{\vec{\ell}_{1},%
\vec{\ell}_{2}}\left( \rho \right) =\rho $ for all $\rho \in E_{\vec{\ell}%
_{2}}$ and 
\begin{equation*}
\mathfrak{x}_{\vec{\ell}_{1},\vec{\ell}_{2}}(E_{\vec{\ell}_{1}})=E_{\vec{\ell%
}_{2}}\ .
\end{equation*}%
\emph{(ii)} For all $\Phi \in \mathcal{W}_{1}$ and $\rho \in E_{\vec{\ell}%
_{1}}$, 
\begin{equation*}
\rho (\mathfrak{e}_{\Phi ,\vec{\ell}_{1}})=\mathfrak{x}_{\vec{\ell}_{1},\vec{%
\ell}_{2}}\left( \rho \right) (\mathfrak{e}_{\Phi ,\vec{\ell}_{2}})\ .
\end{equation*}%
\emph{(iii)} It maps extreme states of $E_{\vec{\ell}_{1}}$ to extreme
states of $E_{\vec{\ell}_{2}}$ and%
\begin{equation*}
\mathfrak{x}_{\vec{\ell}_{1},\vec{\ell}_{2}}^{-1}(\mathcal{E}(E_{\vec{\ell}%
_{2}}))\cap E_{\vec{\ell}_{1}}=\mathcal{E}(E_{\vec{\ell}_{1}})\ .
\end{equation*}
\end{lemma}

\begin{proof}
Fix all parameters of the lemma. Assertions (i)-(ii) are direct consequence
of (\ref{space-average fcuntion}). See also Equation (\ref{eq:enpersite})
defining the energy density observable $\mathfrak{e}_{\Phi ,\vec{\ell}}$ of
any $\Phi \in \mathcal{W}_{1}$. In order to prove Assertion (iii), first
fix\ $\hat{\rho}_{1}\in \mathcal{E}(E_{\vec{\ell}_{1}})$ and define $\hat{%
\rho}_{2}\doteq \mathfrak{x}_{\vec{\ell}_{1},\vec{\ell}_{2}}\left( \hat{\rho}%
_{1}\right) $. Using (\ref{Limit of Space-Averages}) for $\vec{\ell}=\vec{%
\ell}_{2}$ as well as $\mathbb{Z}_{\vec{\ell}_{1}}^{d}\subseteq \mathbb{Z}_{%
\vec{\ell}_{2}}^{d}$, we compute that, for any $A\in \mathcal{U}$,%
\begin{equation}
\hat{\rho}_{2}(|A_{L,\vec{\ell}_{2}}|^{2})=\frac{1}{|\Lambda _{L}\cap 
\mathbb{Z}_{\vec{\ell}_{2}}^{d}|^{2}}\sum\limits_{x,y\in \Lambda _{L}\cap 
\mathbb{Z}_{\vec{\ell}_{2}}^{d}}\hat{\rho}_{1}\left( \alpha _{x}\left(
A^{\ast }\right) \alpha _{y}\left( A\right) \right) +o\left( 1\right) \ ,
\label{comp1}
\end{equation}%
as $L\rightarrow \infty $, which, combined with $\mathbb{Z}_{\vec{\ell}%
_{1}}^{d}\subseteq \mathbb{Z}_{\vec{\ell}_{2}}^{d}$, implies in turn that,
as $L\rightarrow \infty $,%
\begin{equation}
\hat{\rho}_{2}(|A_{L,\vec{\ell}_{2}}|^{2})=\frac{1}{|\Lambda _{L}\cap 
\mathbb{Z}_{\vec{\ell}_{1}}^{d}|^{2}}\sum\limits_{x,y\in \Lambda _{L}\cap 
\mathbb{Z}_{\vec{\ell}_{1}}^{d}}\hat{\rho}_{1}\left( \alpha _{x}(A_{\vec{\ell%
}_{1},\vec{\ell}_{2}}^{\ast })\alpha _{y}(A_{\vec{\ell}_{1},\vec{\ell}%
_{2}})\right) +o\left( 1\right)  \label{ssss}
\end{equation}%
for any $A\in \mathcal{U}$, where 
\begin{equation}
A_{\vec{\ell}_{1},\vec{\ell}_{2}}\doteq \frac{\ell _{2,1}\cdots \ell _{2,d}}{%
\ell _{1,1}\cdots \ell _{1,d}}\sum\limits_{x=(x_{1},\ldots
,x_{d}),\;x_{i}\in \{0,\ell _{2,i},2\ell _{2,i},\ldots ,\ell _{1,i}-\ell
_{2,i}\}}\alpha _{x}(A)\ .  \label{ssssssss}
\end{equation}%
Since $\hat{\rho}_{1}\in \mathcal{E}(E_{\vec{\ell}_{1}})$, we can combine (%
\ref{ssss})-(\ref{ssssssss}) with (\ref{Limit of Space-Averages})-(\ref%
{Ergodicity2}) to arrive at the limit%
\begin{equation}
\lim_{L\rightarrow \infty }\hat{\rho}_{2}(|A_{L,\vec{\ell}_{2}}|^{2})=|\hat{%
\rho}_{1}(A_{\vec{\ell}_{1},\vec{\ell}_{2}})|^{2}\ ,\qquad A\in \mathcal{U}\
.  \label{ergodicity}
\end{equation}%
Observing from (\ref{space-average fcuntion}) that 
\begin{equation}
\rho _{1}(A_{\vec{\ell}_{1},\vec{\ell}_{2}})=\mathfrak{x}_{\vec{\ell}_{1},%
\vec{\ell}_{2}}\left( \rho _{1}\right) \left( A\right) \ ,\qquad A\in 
\mathcal{U},\ \rho _{1}\in E_{\vec{\ell}_{1}}\ ,  \label{ergodicity2}
\end{equation}%
it follows from (\ref{ergodicity}) and Assertion (i) that 
\begin{equation*}
\hat{\rho}_{2}\doteq \mathfrak{x}_{\vec{\ell}_{1},\vec{\ell}_{2}}\left( \hat{%
\rho}_{1}\right) \in E_{\vec{\ell}_{2}}
\end{equation*}%
is $\vec{\ell}_{2}$-ergodic and, thus, belongs to $\mathcal{E}(E_{\vec{\ell}%
_{2}})$, by (\ref{Ergodicity2}). In other words, $\mathfrak{x}_{\vec{\ell}%
_{1},\vec{\ell}_{2}}$ maps extreme states of $E_{\vec{\ell}_{1}}$ to extreme
states of $E_{\vec{\ell}_{2}}$. In particular, 
\begin{equation}
\mathcal{E}(E_{\vec{\ell}_{1}})\subseteq \mathfrak{x}_{\vec{\ell}_{1},\vec{%
\ell}_{2}}^{-1}(\mathcal{E}(E_{\vec{\ell}_{2}}))\cap E_{\vec{\ell}_{1}}\ .
\label{hhhhhhh1}
\end{equation}%
It remains to prove that 
\begin{equation}
\mathfrak{x}_{\vec{\ell}_{1},\vec{\ell}_{2}}^{-1}(\mathcal{E}(E_{\vec{\ell}%
_{2}}))\cap E_{\vec{\ell}_{1}}\subseteq \mathcal{E}(E_{\vec{\ell}_{1}})\ .
\label{hhhhhhh2}
\end{equation}%
We prove this by contradiction. Assume the existence of $\rho _{1}\in E_{%
\vec{\ell}_{1}}\backslash \mathcal{E}(E_{\vec{\ell}_{1}})$ such that 
\begin{equation}
\rho _{2}\doteq \mathfrak{x}_{\vec{\ell}_{1},\vec{\ell}_{2}}\left( \rho
_{1}\right) \in \mathcal{E}(E_{\vec{\ell}_{2}}).  \label{contradcition}
\end{equation}%
Then, by (\ref{Limit of Space-Averages})-(\ref{Ergodicity2}), there is $A\in 
\mathcal{U}$ such that%
\begin{equation}
\lim\limits_{L\rightarrow \infty }\rho _{1}(A_{L,\vec{\ell}_{1}}^{\ast }A_{L,%
\vec{\ell}_{1}})>|\rho _{1}(A)|^{2}\ .  \label{contradiction1}
\end{equation}%
Now, if one computes $\rho _{2}(|A_{L,\vec{\ell}_{2}}|^{2})$ as it is done
in (\ref{comp1})-(\ref{ssssssss}), then we infer from (\ref{contradiction1})
and (\ref{ergodicity2}) that 
\begin{equation*}
\lim\limits_{L\rightarrow \infty }\hat{\rho}_{2}(|A_{L,\vec{\ell}%
_{2}}|^{2})>|\hat{\rho}_{1}(A_{\vec{\ell}_{1},\vec{\ell}_{2}})|^{2}=|%
\mathfrak{x}_{\vec{\ell}_{1},\vec{\ell}_{2}}\left( \rho _{1}\right)
(A)|^{2}=|\rho _{2}(A)|^{2}\ .
\end{equation*}%
By (\ref{Limit of Space-Averages})-(\ref{Ergodicity2}), it means that $\hat{%
\rho}_{2}\notin \mathcal{E}(E_{\vec{\ell}_{2}})$, which contradicts (\ref%
{contradcition}). This yields Equation (\ref{hhhhhhh2}), which combined with
(\ref{hhhhhhh1}) in turn implies Assertion (iii).
\end{proof}

\begin{corollary}[From $\vec{\ell}_{1}$- to $\vec{\ell}_{2}$-ergodic
decompositions of periodic states]
\label{corllary-pushforward}\mbox{ }\newline
Fix $\vec{\ell}_{j}\doteq \left( \ell _{j,1},\ldots ,\ell _{j,d}\right) \in 
\mathbb{N}^{d}$, $j\in \{1,2\}$, such that $\mathbb{Z}_{\vec{\ell}%
_{1}}^{d}\subseteq \mathbb{Z}_{\vec{\ell}_{2}}^{d}$ (see (\ref{group})). Let 
$\rho \in E_{\vec{\ell}_{2}}\subseteq E_{\vec{\ell}_{1}}$ and denote by $\mu
_{\rho }^{(\vec{\ell}_{1})}\in \mathrm{M}^{(\rho )}(E_{\vec{\ell}_{1}})$ the
unique probability measure representing $\rho $ in $E_{\vec{\ell}_{1}}$ with 
$\mu _{\rho }^{(\vec{\ell}_{1})}(\mathcal{E}(E_{\vec{\ell}_{1}}))=1$. The
pushforward of $\mu _{\rho }^{(\vec{\ell}_{1})}$ through the continuous
mapping $\mathfrak{x}_{\vec{\ell}_{1},\vec{\ell}_{2}}$ is the unique
probability measure $\mu _{\rho }^{(\vec{\ell}_{2})}$\ representing $\rho $
in $E_{\vec{\ell}_{2}}$ with $\mu _{\rho }^{(\vec{\ell}_{2})}(\mathcal{E}(E_{%
\vec{\ell}_{2}}))=1$.
\end{corollary}

\begin{proof}
The proof is a direct consequence of Theorem \ref{theorem choquet} and Lemma %
\ref{lemma extra}.
\end{proof}

\subsection{Strong Limit of Space Averages}

Fix $\vec{\ell}\in \mathbb{N}^{d}$. The subject of this section is the limit 
$L\rightarrow \infty $ of the space averages (\ref{Limit of Space-Averages})
for any $A\in \mathcal{U}$, that is, 
\begin{equation}
A_{L}\doteq \frac{1}{|\Lambda _{L}\cap \mathbb{Z}_{\vec{\ell}}^{d}|}%
\sum\limits_{x\in \Lambda _{L}\cap \mathbb{Z}_{\vec{\ell}}^{d}}\alpha
_{x}\left( A\right) \ ,\mathrm{\qquad }L\in \mathbb{N}\ .
\label{Limit of Space-Averages1}
\end{equation}%
As explained in \cite[Section 4.3]{BruPedra-MFII}, this sequence does not
generally converge in the $C^{\ast }$-algebra $\mathcal{U}$. Nevertheless,
in a cyclic representation $(\mathcal{H}_{\rho },\pi _{\rho },\Omega _{\rho
})$ of any $\vec{\ell}$-periodic state $\rho \in E_{\vec{\ell}}$ this limit
exists in the sense of the strong operator topology of $\mathcal{B}(\mathcal{%
H}_{\rho })$.

We thus fix $\rho \in E_{\vec{\ell}}$ for some $\vec{\ell}\in \mathbb{N}^{d}$%
. By $\vec{\ell}$-periodicity of $\rho $ and \cite[Corollary 2.3.17]%
{BrattelliRobinsonI}, there is a unique family of unitary operators $%
\{U_{x}\}_{x\in \mathbb{Z}_{\vec{\ell}}^{d}}$ in $\mathcal{B}(\mathcal{H}%
_{\rho })$ with invariant vector $\Omega _{\rho }$, i.e., $\Omega _{\rho
}=U_{x}\Omega _{\rho }$ for any $x\in \mathbb{Z}_{\vec{\ell}}^{d}$, and 
\begin{equation}
\pi _{\rho }(\alpha _{x}(A))=U_{x}\pi _{\rho }(A)U_{x}^{\ast }\ ,\mathrm{%
\qquad }A\in \mathcal{U},\ x\in \mathbb{Z}_{\vec{\ell}}^{d}\ .
\label{spce qverqge+1}
\end{equation}%
By the von Neumann ergodic theorem\footnote{%
It is also named the Alaoglu-Birkhoff mean ergodic theorem.} \cite[Theorem
4.2]{BruPedra2} (or \cite[Proposition 4.3.4]{BrattelliRobinsonI}), the space
average 
\begin{equation}
P^{(L)}\doteq \frac{1}{|\Lambda _{L}\cap \mathbb{Z}_{\vec{\ell}}^{d}|}%
\sum\limits_{x\in \Lambda _{L}\cap \mathbb{Z}_{\vec{\ell}}^{d}}U_{x}
\label{space average PL}
\end{equation}%
strongly converges, as $L\rightarrow \infty $, to an orthogonal projection
in the commutant $[\{U_{x}\}_{x\in \mathbb{Z}_{\vec{\ell}}^{d}}]^{\prime }$,
which we denote by $P_{\rho }$. More precisely, $P_{\rho }$ is the
orthogonal projection on the closed subspace%
\begin{equation}
\bigcap_{x\in \mathbb{Z}_{\vec{\ell}}^{d}}\{\psi \in \mathcal{H}_{\rho
}\,:\,\psi =U_{x}(\psi )\}\supseteq \mathbb{C}\Omega _{\rho }
\label{space cyclic}
\end{equation}%
of all vectors of $\mathcal{H}_{\rho }$ that are invariant with respect to
the family $\{U_{x}\}_{x\in \mathbb{Z}_{\vec{\ell}}^{d}}$ of unitaries. In
particular, $\Omega _{\rho }=P_{\rho }\Omega _{\rho }$. With this
definition, we can reformulate \cite[Theorem 1.16]{BruPedra2} as follows: $%
\hat{\rho}\in \mathcal{E}(E_{\vec{\ell}})$ is an extreme state iff $P_{\hat{%
\rho}}$ is the orthogonal projection on the one-dimensional Hilbert subspace
spanned by the cyclic vector $\Omega _{\hat{\rho}}$, i.e., 
\begin{equation}
\bigcap_{x\in \mathbb{Z}_{\vec{\ell}}^{d}}\{\psi \in \mathcal{H}_{\hat{\rho}%
}\,:\,\psi =U_{x}(\psi )\}=\mathbb{C}\Omega _{\hat{\rho}}.
\label{space cyclic ergodic}
\end{equation}%
See \cite[Lemma 4.8]{BruPedra2}. Compare with (\ref{strongly mixing}) and (%
\ref{space cyclic}).

We are now in a position to study the limit of space-averages (\ref{Limit of
Space-Averages1}) seen as bounded operators acting on $\mathcal{H}_{\rho }$
via the representation $\pi _{\rho }$.

\begin{lemma}[Strong limit of space-averages]
\label{Lemma sympa copy(1)}\mbox{ }\newline
Let $\vec{\ell}\in \mathbb{N}^{d}$ and $\rho \in E_{\vec{\ell}}$ be any $%
\vec{\ell}$-periodic state with cyclic representation $\left( \mathcal{H}%
_{\rho },\pi _{\rho },\Omega _{\rho }\right) $. For every element $A\in 
\mathcal{U}^{+}$, the sequence $(\pi _{\rho }(A_{L}))_{L\in \mathbb{N}}$
defined from (\ref{Limit of Space-Averages1}) strongly converges in $%
\mathcal{B}(\mathcal{H}_{\rho })$ to the element $A_{\infty }^{\rho }\in \pi
_{\rho }(\mathcal{U})^{\prime }\cap \pi _{\rho }(\mathcal{U})^{\prime \prime
}$ uniquely defined by%
\begin{equation*}
A_{\infty }^{\rho }\pi _{\rho }\left( B\right) \Omega _{\rho }\doteq \pi
_{\rho }\left( B\right) P_{\rho }\pi _{\rho }\left( A\right) \Omega _{\rho
}\ ,\qquad B\in \mathcal{U}\text{ }.
\end{equation*}
\end{lemma}

\begin{proof}
Fix $A\in \mathcal{U}^{+}$, $\vec{\ell}\in \mathbb{N}^{d}$ and $\rho \in E_{%
\vec{\ell}}$ with $\left( \mathcal{H}_{\rho },\pi _{\rho },\Omega _{\rho
}\right) $ being its cyclic representation. Observe first that the uniformly
bounded sequence $(\pi _{\rho }(A_{L}))_{L\in \mathbb{N}}$ strongly
converges on $\mathcal{H}_{\rho }$ iff it strongly\ converges on the dense
subspace $\pi _{\rho }(\mathcal{U})\Omega _{\rho }\subseteq \mathcal{H}%
_{\rho }$, by cyclicity of $\Omega _{\rho }$. We meanwhile infer from (\ref%
{Limit of Space-Averages1})-(\ref{space cyclic}) that, for all $B\in 
\mathcal{U}$ and $L\in \mathbb{N}$,%
\begin{equation}
\pi _{\rho }\left( A_{L}\right) \pi _{\rho }\left( B\right) \Omega _{\rho
}=\pi _{\rho }\left( B\right) P^{(L)}\pi _{\rho }(A)\Omega _{\rho }+\pi
_{\rho }\left( \frac{1}{|\Lambda _{L}\cap \mathbb{Z}_{\vec{\ell}}^{d}|}%
\sum\limits_{x\in \Lambda _{L}\cap \mathbb{Z}_{\vec{\ell}}^{d}}\left[ \alpha
_{x}\left( A\right) ,B\right] \right) \Omega _{\rho }\ ,
\label{space average2}
\end{equation}%
where we recall that $[B,C]\doteq BC-CB$. By (\ref{space 00}), for all $A\in 
\mathcal{U}^{+}$ and $B\in \mathcal{U}$, 
\begin{equation*}
\lim_{L\rightarrow \infty }\frac{1}{|\Lambda _{L}\cap \mathbb{Z}_{\vec{\ell}%
}^{d}|}\sum\limits_{x\in \Lambda _{L}\cap \mathbb{Z}_{\vec{\ell}%
}^{d}}\left\Vert \left[ \alpha _{x}\left( A\right) ,B\right] \right\Vert _{%
\mathcal{U}}=0
\end{equation*}%
and $P^{(L)}$ strongly converges, as $L\rightarrow \infty $, to $P_{\rho }$.
Consequently, we deduce from (\ref{space average2}) the existence of a
bounded operator $A_{\infty }^{\rho }\in \mathcal{B}(\mathcal{H}_{\rho })$
such that%
\begin{equation}
A_{\infty }^{\rho }\pi _{\rho }(B)\Omega _{\rho }=\pi _{\rho }\left(
B\right) P_{\rho }\pi _{\rho }(A)\Omega _{\rho }=\lim_{L\rightarrow \infty
}\pi _{\rho }\left( A_{L}\right) \pi _{\rho }\left( B\right) \Omega _{\rho
}\ ,\qquad B\in \mathcal{U}\text{ }.  \label{ass strong conv Ainfty}
\end{equation}%
This equation uniquely defines $A_{\infty }^{\rho }\in \mathcal{B}(\mathcal{H%
}_{\rho })$, by cyclicity of $\Omega _{\rho }$. The limit operator $%
A_{\infty }^{\rho }$ belongs to the commutant $\pi _{\rho }(\mathcal{U}%
)^{\prime }$ because $\Omega _{\rho }$ is a cyclic vector and Equation (\ref%
{ass strong conv Ainfty}) implies that%
\begin{equation*}
\left[ \pi _{\rho }\left( C\right) ,A_{\infty }^{\rho }\right] \pi _{\rho
}(B)\Omega _{\rho }=0\ ,\qquad B,C\in \mathcal{U}\text{ }.
\end{equation*}%
Further, $A_{\infty }^{\rho }$ has to be an element of the bicommutant $\pi
_{\rho }(\mathcal{U})^{\prime \prime }$, because it is the strong limit of
elements of $\pi _{\rho }(\mathcal{U})$.
\end{proof}

The limit operator $A_{\infty }^{\rho }$ of Lemma \ref{Lemma sympa copy(1)}
takes a very simple form when $\rho $ is an extreme $\vec{\ell}$-periodic
state, i.e., $\rho \in \mathcal{E}(E_{\vec{\ell}})$.

\begin{corollary}[Strong limit of space-averages for ergodic states]
\label{Lemma sympa}\mbox{ }\newline
Let $\vec{\ell}\in \mathbb{N}^{d}$ and $\hat{\rho}\in \mathcal{E}(E_{\vec{%
\ell}})$ with cyclic representation $\left( \mathcal{H}_{\hat{\rho}},\pi _{%
\hat{\rho}},\Omega _{\hat{\rho}}\right) $. For every element $A\in \mathcal{U%
}$, the sequence $(\pi _{\rho }(A_{L}))_{L\in \mathbb{N}}$ defined from (\ref%
{Limit of Space-Averages1}) strongly converges in $\mathcal{B}(\mathcal{H}_{%
\hat{\rho}})$ to 
\begin{equation*}
A_{\infty }^{\hat{\rho}}=\hat{\rho}\left( A\right) \mathbf{1}_{\mathcal{H}_{%
\hat{\rho}}}\text{ }.
\end{equation*}
\end{corollary}

\begin{proof}
Since $\hat{\rho}\in \mathcal{E}(E_{\vec{\ell}})$ is an extreme $\vec{\ell}$%
-periodic state, by \cite[Lemma 4.8]{BruPedra2}, $P_{\hat{\rho}}$ is the
orthogonal projection on the one-dimensional Hilbert subspace spanned by $%
\Omega _{\hat{\rho}}$. See Equation (\ref{space cyclic ergodic}). In fact,
this is directly related to the ergodicity of extreme states, see (\ref%
{Limit of Space-Averages})-(\ref{Ergodicity2}). By Lemma \ref{Lemma sympa
copy(1)} and Equation (\ref{equality}), the assertion then follows.
\end{proof}

\begin{corollary}[Strong limit of local energies for ergodic states]
\label{Lemma sympa copy(2)}\mbox{ }\newline
Let $\vec{\ell}\in \mathbb{N}^{d}$ and $\hat{\rho}\in \mathcal{E}(E_{\vec{%
\ell}})$ with cyclic representation $\left( \mathcal{H}_{\hat{\rho}},\pi _{%
\hat{\rho}},\Omega _{\hat{\rho}}\right) $. For every translation-invariant
interaction $\Phi \in \mathcal{W}_{1}$, the sequence 
\begin{equation*}
\left( \pi _{\hat{\rho}}\left( \left\vert \Lambda _{L}\right\vert
^{-1}U_{L}^{\Phi }\right) \right) _{L\in \mathbb{N}}
\end{equation*}%
defined by (\ref{eq:def lambda n}) and (\ref{equation fininte vol dynam0})
strongly converges in $\mathcal{B}(\mathcal{H}_{\hat{\rho}})$ to $\hat{\rho}(%
\mathfrak{e}_{\Phi ,\vec{\ell}})\mathbf{1}_{\mathcal{H}_{\hat{\rho}}}$,
where $\mathfrak{e}_{\Phi ,\vec{\ell}}$ is the even observable defined by (%
\ref{eq:enpersite}).
\end{corollary}

\begin{proof}
Fix $\vec{\ell}\in \mathbb{N}^{d}$ and any translation-invariant interaction 
$\Phi \in \mathcal{W}_{1}$. By Equations (\ref{ti interaction})-(\ref%
{eq:enpersite}) and (\ref{equation fininte vol dynam0}),%
\begin{multline}
\left\Vert \frac{U_{L}^{\Phi }}{\left\vert \Lambda _{L}\right\vert }-\frac{1%
}{|\Lambda _{L}\cap \mathbb{Z}_{\vec{\ell}}^{d}|}\sum\limits_{y\in \Lambda
_{L}\cap \mathbb{Z}_{\vec{\ell}}^{d}}\alpha _{y}(\mathfrak{e}_{\Phi ,\vec{%
\ell}})\right\Vert _{\mathcal{U}}\leq \sum\limits_{x=(x_{1},\cdots
,x_{d}),\;x_{i}\in \{0,\cdots ,\ell _{i}-1\}}\sum\limits_{\Lambda \in 
\mathcal{P}_{f},\Lambda \ni 0}\frac{\left\Vert \Phi _{\Lambda }\right\Vert _{%
\mathcal{U}}}{|\Lambda |}  \label{inequality a la con space average} \\
\times \frac{1}{|\Lambda _{L}|}\sum\limits_{y\in \Lambda _{L}\cap \mathbb{Z}%
_{\vec{\ell}}^{d}}\left( 1-\mathbf{1}\left[ x+y\in \Lambda _{L}\right] 
\mathbf{1}\left[ \Lambda \subseteq (\Lambda _{L}-x-y)\right] \right) \ , 
\notag
\end{multline}%
where, by definition, $\mathbf{1}\left[ p\right] =1$ whenever the
proposition $p$ is true, otherwise $\mathbf{1}\left[ p\right] =0$. Since $%
\Phi \in \mathcal{W}_{1}$, by (\ref{iteration0}) and Lebesgue's dominated
convergence theorem, we deduce from last inequality that 
\begin{equation*}
\lim_{L\rightarrow \infty }\left\Vert \frac{U_{L}^{\Phi }}{\left\vert
\Lambda _{L}\right\vert }-\frac{1}{|\Lambda _{L}\cap \mathbb{Z}_{\vec{\ell}%
}^{d}|}\sum\limits_{y\in \Lambda _{L}\cap \mathbb{Z}_{\vec{\ell}}^{d}}\alpha
_{y}(\mathfrak{e}_{\Phi ,\vec{\ell}})\right\Vert _{\mathcal{U}}=0\ ,
\end{equation*}%
using that%
\begin{equation*}
\lim_{L\rightarrow \infty }\frac{1}{|\Lambda _{L}|}\sum\limits_{y\in \Lambda
_{L}\cap \mathbb{Z}_{\vec{\ell}}^{d}}\left( 1-\mathbf{1}\left[ x+y\in
\Lambda _{L}\right] \mathbf{1}\left[ \Lambda \subseteq (\Lambda _{L}-x-y)%
\right] \right) =0
\end{equation*}%
for all $x\in \mathfrak{L}$ and $\Lambda \in \mathcal{P}_{f}$. Therefore,
the assertion is a direct consequence of Corollary \ref{Lemma sympa} applied
to $A=\mathfrak{e}_{\Phi ,\vec{\ell}}\in \mathcal{U}$.
\end{proof}

\subsection{Commutator Estimates from Lieb-Robinson Bounds}

In Section \ref{Long-Range Dynamics on Ergodic States}, we prove Theorem \ref%
{theorem structure of omega} for ergodic states. To this end, we employ the
following uniform bound, which is a direct consequence of Lieb-Robinson
bounds \cite[Theorem 4.3]{brupedraLR}:

\begin{lemma}[Commutator estimates from Lieb-Robinson bounds]
\label{Lemma sympa copy(3)}\mbox{ }\newline
For any time-dependent model $\mathfrak{m}\in C_{b}\left( \mathbb{R};%
\mathcal{M}\right) $, $\Phi \in \mathcal{W}$, times $s,t\in \mathbb{R}$,
length $L\in \mathbb{N}$, subset $\Lambda \subseteq \Lambda _{L}$ and $A\in 
\mathcal{U}_{\Lambda }$, 
\begin{equation*}
\left\Vert \lbrack \tau _{t,s}^{(L,\mathfrak{m})}\left( A\right)
,U_{L}^{\Phi }]\right\Vert _{\mathcal{U}}\leq 2\left\vert \Lambda
\right\vert \left\Vert A\right\Vert _{\mathcal{U}}\left\Vert \Phi
\right\Vert _{\mathcal{W}}\mathrm{e}^{16\left( \mathbf{D}+2\left\Vert 
\mathbf{F}\right\Vert _{1,\mathfrak{L}}+1\right) \int_{s}^{t}\left\Vert 
\mathfrak{m}\left( \varsigma \right) \right\Vert _{\mathcal{M}}\mathrm{d}%
\varsigma }
\end{equation*}%
with $(\tau _{t,s}^{(L,\mathfrak{m})})_{_{s,t\in \mathbb{R}}}$ being the
unique (fundamental) solution to (\ref{cauchy1})-(\ref{cauchy2}), while $%
U_{L}^{\Phi }$ is the energy element defined by (\ref{equation fininte vol
dynam0}).
\end{lemma}

\begin{proof}
For any model $\mathfrak{m}\doteq \left( \Phi ,\mathfrak{a}\right) \in 
\mathcal{M}$ and $L\in \mathbb{N}$, there is an interaction $\Phi ^{(L,%
\mathfrak{m})}$ such that 
\begin{equation}
U_{L}^{\mathfrak{m}}\doteq U_{L}^{\Phi ^{(L,\mathfrak{m})}}\doteq
\sum\limits_{\Lambda \subseteq \Lambda _{L}}\Phi _{\Lambda }^{(L,\mathfrak{m}%
)}  \label{equation long range energy2}
\end{equation}%
with $\Lambda _{L}$ being the cubic box defined by (\ref{eq:def lambda n})
for $L\in \mathbb{N}$. Recall that $U_{L}^{\mathfrak{m}}$ is the local
Hamiltonian of the model $\mathfrak{m}\in \mathcal{M}$, defined by (\ref%
{equation long range energy}) for $L\in \mathbb{N}$. Since, by (\ref%
{equation fininte vol dynam0}), 
\begin{equation*}
U_{L}^{\Psi ^{(1)}}\cdots U_{L}^{\Psi ^{(n)}}=\sum\limits_{\Lambda \in 
\mathcal{P}_{f}}\mathbf{1}\left[ \Lambda \subseteq \Lambda _{L}\right]
\sum_{\Lambda _{1},\ldots ,\Lambda _{n}\in \mathcal{P}_{f}\mathfrak{:}%
\Lambda =\Lambda _{1}\cup \Lambda _{2}\cup \cdots \cup \Lambda _{n}}\Psi
_{\Lambda _{1}}^{(1)}\cdots \Psi _{\Lambda _{n}}^{(n)}\ ,
\end{equation*}%
for any $n\in \mathbb{N}$ and $\Psi ^{(1)},\ldots ,\Psi ^{(n)}\in \mathcal{W}
$, such an interaction $\Phi ^{(L,\mathfrak{m})}$ is explicitly given by%
\begin{equation}
\Phi _{\Lambda }^{(L,\mathfrak{m})}\doteq \Phi _{\Lambda }+\sum_{n\in 
\mathbb{N}}\frac{\mathbf{1}\left[ \Lambda \subseteq \Lambda _{L}\right] }{%
\left\vert \Lambda _{L}\right\vert ^{n-1}}\int_{\mathbb{S}^{n}}\mathfrak{a}%
_{n}\left( \mathrm{d}\Psi ^{(1)},\ldots ,\mathrm{d}\Psi ^{(n)}\right)
\sum_{\Lambda _{1},\ldots ,\Lambda _{n}\in \mathcal{P}_{f}\mathfrak{:}%
\Lambda =\Lambda _{1}\cup \Lambda _{2}\cup \cdots \cup \Lambda _{n}}\Psi
_{\Lambda _{1}}^{(1)}\cdots \Psi _{\Lambda _{n}}^{(n)}  \label{phi L}
\end{equation}%
for all $\Lambda \in \mathcal{P}_{f}$. See again (\ref{equation long range
energy}) and (\ref{equation long range energy2}).

Fix $L\in \mathbb{N}$. We next define the positive-valued symmetric function 
$\mathbf{F}_{L}:\Lambda _{L}^{2}\rightarrow (0,1]$ by 
\begin{equation*}
\mathbf{F}_{L}\left( x,y\right) \doteq \frac{|\Lambda _{L}|}{|\Lambda _{L}|+1%
}\left( \mathbf{F}\left( x,y\right) +\frac{1}{|\Lambda _{L}|}\right) \
,\qquad x,y\in \Lambda _{L}\ ,
\end{equation*}%
where $\mathbf{F}:\mathfrak{L}^{2}\rightarrow (0,1]$ is some positive-valued
symmetric function with maximum value $\mathbf{F}\left( x,x\right) =1$ for
all $x\in \mathfrak{L}$, satisfying (\ref{(3.1) NS})-(\ref{(3.2) NS}). Note
that $\mathbf{F}_{L}$ has also maximum value $\mathbf{F}_{L}\left(
x,x\right) =1$ for all $x\in \mathfrak{L}$. Then, a $L$-dependent seminorm
for interactions $\Phi $ is defined by 
\begin{equation}
\left\Vert \Phi \right\Vert _{\mathcal{W}_{L}}\doteq \underset{x,y\in
\Lambda _{L}}{\sup }\sum\limits_{\Lambda \subseteq \Lambda _{L},\;\Lambda
\supseteq \{x,y\}}\frac{\Vert \Phi _{\Lambda }\Vert _{\mathcal{U}}}{\mathbf{F%
}_{L}\left( x,y\right) }\leq \left( 1+|\Lambda _{L}|^{-1}\right) \left\Vert
\Phi \right\Vert _{\mathcal{W}}\ .  \label{estimate2}
\end{equation}%
Compare with Equation (\ref{iteration0}). This seminorm is in fact a norm in
the finite-dimensional space $\mathcal{W}_{L}$ of interactions $\Phi $
supported on the cubic box $\Lambda _{L}$, i.e., $\Phi _{\Lambda }=0$
whenever $\Lambda \nsubseteq \Lambda _{L}$. Compare with the Banach space $%
\mathcal{W}$ defined by (\ref{banach space short range})-(\ref{iteration0}).
The (real) subspace of all self-adjoint interactions of $\mathcal{W}_{L}$ is
denoted by $\mathcal{W}_{L}^{\mathbb{R}}\varsubsetneq \mathcal{W}_{L}$,
similar to $\mathcal{W}^{\mathbb{R}}\varsubsetneq \mathcal{W}$. For any $%
\mathfrak{m}\doteq \left( \Phi ,\mathfrak{a}\right) \in \mathcal{M}$, we
compute from (\ref{phi L}) and (\ref{estimate2}) that%
\begin{eqnarray}
\left\Vert \Phi ^{(L,\mathfrak{m})}\right\Vert _{\mathcal{W}_{L}} &\leq
&\left( 1+|\Lambda _{L}|^{-1}\right) \left\Vert \Phi \right\Vert _{\mathcal{W%
}}+\sum_{n\in \mathbb{N}}\frac{1}{\left\vert \Lambda _{L}\right\vert ^{n-1}}%
\int_{\mathbb{S}^{n}}\left\vert \mathfrak{a}_{n}\right\vert \left( \mathrm{d}%
\Psi ^{(1)},\ldots ,\mathrm{d}\Psi ^{(n)}\right)  \label{upper bound WL} \\
&&\underset{x,y\in \Lambda _{L}}{\sup }\sum\limits_{\Lambda \subseteq
\Lambda _{L},\;\Lambda \supseteq \{x,y\}}\sum_{\Lambda _{1},\ldots ,\Lambda
_{n}\in \mathcal{P}_{f}\mathfrak{:}\Lambda =\Lambda _{1}\cup \Lambda
_{2}\cup \cdots \cup \Lambda _{n}}\frac{\Vert \Psi _{\Lambda
_{1}}^{(1)}\Vert _{\mathcal{U}}\cdots \Vert \Psi _{\Lambda _{n}}^{(n)}\Vert
_{\mathcal{U}}}{\mathbf{F}_{L}\left( x,y\right) }\ .  \notag
\end{eqnarray}%
Now, observe that, for any $\Psi ^{(1)},\Psi ^{(2)}\in \mathbb{S}$ (the unit
sphere in $\mathcal{W}$), 
\begin{eqnarray}
&&\frac{1}{|\Lambda _{L}|}\underset{x,y\in \Lambda _{L}}{\sup }%
\sum\limits_{\Lambda \subseteq \Lambda _{L},\;\Lambda \supseteq
\{x,y\}}\sum_{\Lambda _{1},\Lambda _{2}\in \mathcal{P}_{f}\mathfrak{:}%
\Lambda =\Lambda _{1}\cup \Lambda _{2}}\frac{\Vert \Psi _{\Lambda
_{1}}^{(1)}\Vert _{\mathcal{U}}\Vert \Psi _{\Lambda _{2}}^{(2)}\Vert _{%
\mathcal{U}}}{\mathbf{F}_{L}\left( x,y\right) }  \notag \\
&\leq &\frac{2}{|\Lambda _{L}|}\underset{x,y\in \Lambda _{L}}{\sup }%
\sum\limits_{\Lambda _{1},\Lambda _{2}\subseteq \Lambda _{L}:x\in \Lambda
_{1},y\in \Lambda _{2}}\frac{\Vert \Psi _{\Lambda _{1}}^{(1)}\Vert _{%
\mathcal{U}}\Vert \Psi _{\Lambda _{2}}^{(2)}\Vert _{\mathcal{U}}}{\mathbf{F}%
_{L}\left( x,y\right) }  \label{estimate chiante1} \\
&&+\frac{\left\Vert \Psi ^{(1)}\right\Vert _{\mathcal{W}_{L}}}{|\Lambda _{L}|%
}\sum\limits_{\Lambda \subseteq \Lambda _{L}}\Vert \Psi _{\Lambda
}^{(2)}\Vert _{\mathcal{U}}+\frac{\left\Vert \Psi ^{(2)}\right\Vert _{%
\mathcal{W}_{L}}}{|\Lambda _{L}|}\sum\limits_{\Lambda \subseteq \Lambda
_{L}}\Vert \Psi _{\Lambda }^{(1)}\Vert _{\mathcal{U}}\ .  \notag
\end{eqnarray}%
Note additionally that 
\begin{equation}
\frac{1}{\left\vert \Lambda _{L}\right\vert }\sum\limits_{\Lambda \subseteq
\Lambda _{L}}\Vert \Psi _{\Lambda }\Vert _{\mathcal{U}}\leq \left\Vert 
\mathbf{F}\right\Vert _{1,\mathfrak{L}}\ ,\qquad \Psi \in \mathbb{S}\ ,
\label{estimate chiante2}
\end{equation}%
(compare with (\ref{norm Uphi})) and, using 
\begin{equation*}
|\Lambda _{L}|\mathbf{F}_{L}\left( x,y\right) \geq \frac{|\Lambda _{L}|}{%
|\Lambda _{L}|+1}\ ,\qquad x,y\in \Lambda _{L}\ ,
\end{equation*}%
one also verifies in a similar way that, for any $\Psi ^{(1)},\Psi ^{(2)}\in 
\mathbb{S}$, 
\begin{equation}
\frac{1}{|\Lambda _{L}|}\underset{x,y\in \Lambda _{L}}{\sup }%
\sum\limits_{\Lambda _{1},\Lambda _{2}\subseteq \Lambda _{L}:x\in \Lambda
_{1},y\in \Lambda _{2}}\frac{\Vert \Psi _{\Lambda _{1}}^{(1)}\Vert _{%
\mathcal{U}}\Vert \Psi _{\Lambda _{2}}^{(2)}\Vert _{\mathcal{U}}}{\mathbf{F}%
_{L}\left( x,y\right) }\leq 1+|\Lambda _{L}|^{-1}\ .
\label{estimate chiante3}
\end{equation}%
By combining (\ref{estimate chiante1}) with (\ref{estimate2}), (\ref%
{estimate chiante2})-(\ref{estimate chiante3}) and $\left\Vert \mathbf{F}%
\right\Vert _{1,\mathfrak{L}}\geq 1$, we deduce that%
\begin{equation*}
\underset{x,y\in \Lambda _{L}}{\sup }\sum\limits_{\Lambda \subseteq \Lambda
_{L},\;\Lambda \supseteq \{x,y\}}\sum_{\Lambda _{1},\Lambda _{2}\in \mathcal{%
P}_{f}\mathfrak{:}\Lambda =\Lambda _{1}\cup \Lambda _{2}}\frac{\Vert \Psi
_{\Lambda _{1}}^{(1)}\Vert _{\mathcal{U}}\Vert \Psi _{\Lambda
_{2}}^{(2)}\Vert _{\mathcal{U}}}{\mathbf{F}_{L}\left( x,y\right) }\leq
4\left( 1+|\Lambda _{L}|^{-1}\right) \left\Vert \mathbf{F}\right\Vert _{1,%
\mathfrak{L}}
\end{equation*}%
for any $\Psi ^{(1)},\Psi ^{(2)}\in \mathbb{S}$. Using this estimate, as
well as $n-2$ times Inequality (\ref{estimate chiante2}), in order to bound
the right-hand side of (\ref{upper bound WL}) from above, we obtain that%
\begin{equation}
\left\Vert \Phi ^{(L,\mathfrak{m})}\right\Vert _{\mathcal{W}_{L}}\leq \left(
1+|\Lambda _{L}|^{-1}\right) \left( \left\Vert \Phi \right\Vert _{\mathcal{W}%
}+4\sum_{n\in \mathbb{N}}n^{2}\left\Vert \mathbf{F}\right\Vert _{1,\mathfrak{%
L}}^{n-1}\Vert \mathfrak{a}\Vert _{\mathcal{S}(\mathbb{S}^{n}\mathbb{)}%
}\right) \ ,  \notag
\end{equation}%
(see (\ref{definition 0})), which in turn implies that, for any $\mathfrak{m}%
\in \mathcal{M}$ and $L\in \mathbb{N}$,%
\begin{equation}
\left\Vert \Phi ^{(L,\mathfrak{m})}\right\Vert _{\mathcal{W}_{L}}\leq
4\left( 1+|\Lambda _{L}|^{-1}\right) \left\Vert \mathfrak{m}\right\Vert _{%
\mathcal{M}}\leq 8\left\Vert \mathfrak{m}\right\Vert _{\mathcal{M}}\ ,
\label{rewritten0}
\end{equation}%
according to (\ref{definition 0bis}) and (\ref{def long range2}). Since%
\begin{equation*}
\Phi ^{(L,\mathfrak{m}_{1})}-\Phi ^{(L,\mathfrak{m}_{2})}=\Phi ^{(L,%
\mathfrak{m}_{1}-\mathfrak{m}_{2})}\ ,\qquad \mathfrak{m}_{1},\mathfrak{m}%
_{2}\in \mathcal{M},\ L\in \mathbb{N}\ ,
\end{equation*}%
we thus deduce from (\ref{equation long range energy2}), (\ref{phi L}) and (%
\ref{rewritten0}) that, for any $\mathfrak{m}\in C_{b}\left( \mathbb{R};%
\mathcal{M}\right) $, there is $\Psi ^{(L,\mathfrak{m})}\in C(\mathbb{R};%
\mathcal{W}_{L}^{\mathbb{R}})$ such that 
\begin{equation}
\tau _{t,s}^{(L,\Psi ^{(L,\mathfrak{m})})}\equiv \tau _{t,s}^{(L,(\Psi ^{(L,%
\mathfrak{m})},0))}=\tau _{t,s}^{(L,\mathfrak{m})}\ ,\qquad L\in \mathbb{N}\
,  \label{rewritten}
\end{equation}%
where $(\tau _{t,s}^{(L,\mathfrak{m})})_{_{s,t\in \mathbb{R}}}$ is the
unique (fundamental) solution to (\ref{cauchy1})-(\ref{cauchy2}). By (\ref%
{(3.1) NS})-(\ref{(3.2) NS}), note that, for all $L\in \mathbb{N}$,%
\begin{equation}
\left\Vert \mathbf{F}_{L}\right\Vert _{1,\Lambda _{L}}\doteq \underset{y\in
\Lambda _{L}}{\sup }\sum_{x\in \Lambda _{L}}\mathbf{F}_{L}\left( x,y\right)
\leq 1+\left\Vert \mathbf{F}\right\Vert _{1,\mathfrak{L}}
\label{(3.1) NSbis}
\end{equation}%
and 
\begin{equation}
\mathbf{D}_{L}\doteq \underset{x,y\in \Lambda _{L}}{\sup }\sum_{z\in \Lambda
_{L}}\frac{\mathbf{F}_{L}\left( x,z\right) \mathbf{F}_{L}\left( z,y\right) }{%
\mathbf{F}_{L}\left( x,y\right) }\leq \mathbf{D}+2\left\Vert \mathbf{F}%
\right\Vert _{1,\mathfrak{L}}+1\ .  \label{(3.2) NSbis}
\end{equation}%
As a consequence, since $\Psi ^{(L,\mathfrak{m})}\in C(\mathbb{R};\mathcal{W}%
_{L}^{\mathbb{R}})$, by (\ref{rewritten0}) and (\ref{rewritten}), we can
apply the Lieb-Robinson bounds \cite[Proposition 3.8 with $\mathfrak{L}%
=\Lambda _{L}$]{BruPedra-MFII} on the right-hand side of the inequality 
\begin{equation*}
\left\Vert \lbrack \tau _{t,s}^{(L,\mathfrak{m})}\left( A\right)
,U_{L}^{\Phi }]\right\Vert _{\mathcal{U}}\leq \sum\limits_{\mathcal{Z}%
\subseteq \Lambda _{L}}\left\Vert [\tau _{t,s}^{(L,\mathfrak{m})}\left(
A\right) ,\Phi _{\mathcal{Z}}]\right\Vert _{\mathcal{U}}
\end{equation*}%
for any $\Lambda \subseteq \Lambda _{L}$ and element $A\in \mathcal{U}%
_{\Lambda }$ to get that 
\begin{equation*}
\left\Vert \lbrack \tau _{t,s}^{(L,\mathfrak{m})}\left( A\right)
,U_{L}^{\Phi }]\right\Vert _{\mathcal{U}}\leq 2\mathbf{D}_{L}^{-1}\left\Vert
A\right\Vert _{\mathcal{U}}\left( \mathrm{e}^{16\mathbf{D}%
_{L}\int_{s}^{t}\left\Vert \mathfrak{m}\left( \varsigma \right) \right\Vert
_{\mathcal{M}}\mathrm{d}\varsigma }-1\right) \sum_{x\in \Lambda }\sum_{y\in
\Lambda _{L}}\sum\limits_{\mathcal{Z}\subseteq \Lambda _{L},\ \mathcal{Z}%
\supseteq \{y\}}\left\Vert \Phi _{\mathcal{Z}}\right\Vert _{\mathcal{U}}%
\mathbf{F}_{L}\left( x,y\right) \ .
\end{equation*}%
By (\ref{iteration0}) and (\ref{(3.1) NSbis}), we infer from the last
inequality that%
\begin{equation*}
\left\Vert \lbrack \tau _{t,s}^{(L,\mathfrak{m})}\left( A\right)
,U_{L}^{\Phi }]\right\Vert _{\mathcal{U}}\leq 2\mathbf{D}_{L}^{-1}\left\vert
\Lambda \right\vert \left\Vert A\right\Vert _{\mathcal{U}}\left\Vert \Phi
\right\Vert _{\mathcal{W}}\left( \mathrm{e}^{16\mathbf{D}_{L}\int_{s}^{t}%
\left\Vert \mathfrak{m}\left( \varsigma \right) \right\Vert _{\mathcal{M}}%
\mathrm{d}\varsigma }-1\right) \left\Vert \mathbf{F}_{L}\right\Vert _{1,%
\mathfrak{L}}
\end{equation*}%
for any $\mathfrak{m}\in C_{b}\left( \mathbb{R};\mathcal{M}\right) $, $\Phi
\in \mathcal{W}$, $s,t\in \mathbb{R}$, $L\in \mathbb{N}$, $\Lambda \in 
\mathcal{P}_{f}$ and element $A\in \mathcal{U}_{\Lambda }$. This yields the
lemma, by (\ref{(3.1) NSbis})-(\ref{(3.2) NSbis}) combined with rough
estimates on $\mathbf{D}_{L}$.
\end{proof}

\subsection{Long-Range Dynamics on Ergodic States\label{Long-Range Dynamics
on Ergodic States}}

In this section, we prove Theorem \ref{theorem structure of omega} for any
ergodic state. To this end, we recall again some important objects and
notation:

\begin{itemize}
\item For any $\Lambda \in \mathcal{P}_{f}$ and $\mathfrak{m}\in C_{b}(%
\mathbb{R};\mathcal{M}_{\Lambda })$, $\mathbf{\varpi }^{\mathfrak{m}}\in
C\left( \mathbb{R}^{2};\mathrm{Aut}\left( E\right) \right) $ is the solution
to the self-consistency equation of Theorem \ref{theorem sdfkjsdklfjsdklfj}.

\item For any $\Psi \in C(\mathbb{R};\mathcal{W}^{\mathbb{R}})$, $(\tau
_{t,s}^{\Psi })_{s,t\in {\mathbb{R}}}$ is the strongly continuous two-para%
\-%
meter family defined as the strong limit, for fixed $s,t$, of the local
dynamics $(\tau _{t,s}^{(L,\Psi )})_{s,t\in {\mathbb{R}}}$ defined by (\ref%
{cauchy1})-(\ref{cauchy2}). See \cite[Proposition 3.7]{BruPedra-MFII}.

\item For any $\mathbf{\Psi }\in C(\mathbb{R};\mathfrak{W})$ and $\rho \in E$%
, $\mathbf{\Psi }\left( \rho \right) \in C(\mathbb{R};\mathcal{W})$ stands
for the time-dependent interaction defined by 
\begin{equation*}
\mathbf{\Psi }\left( \rho \right) \left( t\right) \doteq \mathbf{\Psi }%
\left( t;\rho \right) \ ,\qquad \rho \in E,\ t\in \mathbb{R}\ .
\end{equation*}%
This refers to Equation (\ref{notation state interactionbis}).

\item For any $\mathfrak{m}\in C\left( \mathbb{R};\mathcal{M}\right) $ and
each $\mathbf{\xi }\in C\left( \mathbb{R};\mathrm{Aut}\left( E\right)
\right) $ (typically, $\mathbf{\xi =\varpi }^{\mathfrak{m}}\left( \alpha
,\cdot \right) $ at fixed $\alpha \in \mathbb{R}$), the approximating
interaction $\mathbf{\Phi }^{(\mathfrak{m},\mathbf{\xi })}$ is the mapping
from $\mathbb{R}$ to $\mathfrak{W}^{\mathbb{R}}$ of Definition \ref%
{definition BCS-type model approximated}. By (\ref{inequality trivial}), if $%
\mathfrak{m}\in C\left( \mathbb{R};\mathcal{M}\right) $ and $\mathbf{\xi }%
\in C\left( \mathbb{R};\mathrm{Aut}\left( E\right) \right) $ then $\mathbf{%
\Phi }^{(\mathfrak{m},\mathbf{\xi })}\left( \rho \right) \in C\left( \mathbb{%
R};\mathcal{W}^{\mathbb{R}}\right) $.

\item $\mathcal{M}_{1}$ is the Banach space of all translation-invariant
long-range models defined by (\ref{translatino invariatn long range models}).
\end{itemize}

\noindent Now, we are in a position to state the main theorem of this
subsection:

\begin{theorem}[Long-range dynamics on ergodic states]
\label{theorem structure of omega copy(2)}\mbox{ }\newline
Fix $\Lambda \in \mathcal{P}_{f}$, $\mathfrak{m}\in C_{b}(\mathbb{R};%
\mathcal{M}_{\Lambda }\cap \mathcal{M}_{1})$, $\vec{\ell}\in \mathbb{N}^{d}$
and $\hat{\rho}\in \mathcal{E}(E_{\vec{\ell}})$ with cyclic representation $%
\left( \mathcal{H}_{\hat{\rho}},\pi _{\hat{\rho}},\Omega _{\hat{\rho}%
}\right) $. Then, for any $s,t\in \mathbb{R}$ and $A\in \mathcal{U}\subseteq 
\mathfrak{U}$, in the $\sigma $-weak topology,%
\begin{equation*}
\lim_{L\rightarrow \infty }\pi _{\hat{\rho}}\left( \tau _{t,s}^{(L,\mathfrak{%
m})}\left( A\right) \right) =\left. \pi _{\hat{\rho}}\left( \tau _{t,s}^{%
\mathbf{\Phi }^{(\mathfrak{m},\mathbf{\varpi }^{\mathfrak{m}}\left( \alpha
,\cdot \right) )}(\hat{\rho})}\left( A\right) \right) \right\vert _{\alpha
=s}\in \mathcal{B}\left( \mathcal{H}_{\hat{\rho}}\right) \ .
\end{equation*}
\end{theorem}

\begin{proof}
Fix once and for all $\Lambda \in \mathcal{P}_{f}$, $\mathfrak{m}\in C_{b}(%
\mathbb{R};\mathcal{M}_{\Lambda }\cap \mathcal{M}_{1})$, $\vec{\ell}\in 
\mathbb{N}^{d}$ and $\hat{\rho}\in \mathcal{E}(E_{\vec{\ell}})$ with a
cyclic representation denoted by $\left( \mathcal{H}_{\hat{\rho}},\pi _{\hat{%
\rho}},\Omega _{\hat{\rho}}\right) $. By (\ref{eq:enpersitebisbis})-(\ref%
{eq:enpersitebisbisbis}), we can assume without loss of generality that $%
\mathfrak{e}_{\Phi ,\vec{\ell}}\in \mathcal{U}_{\Lambda }$. In order to
simplify the notation, we denote 
\begin{equation}
\gimel _{t,s}\doteq \tau _{t,s}^{\mathbf{\Phi }^{(\mathfrak{m},\mathbf{%
\varpi }^{\mathfrak{m}}\left( \alpha ,\cdot \right) )}(\hat{\rho})}|_{\alpha
=s}\ ,\qquad \alpha ,s,t\in \mathbb{R}\ .  \label{notation interaction}
\end{equation}%
The proof is done in several steps: \bigskip

\noindent \underline{Step 1:} For any $s,t\in \mathbb{R}$, the sequence 
\begin{equation*}
\left\{ \pi _{\hat{\rho}}\left( \gimel _{t,s}\left( A\right) -\tau
_{t,s}^{(L,\mathfrak{m})}\left( A\right) \right) \right\} _{L\in \mathbb{N}%
}\subseteq \mathcal{B}\left( \mathcal{H}_{\hat{\rho}}\right)
\end{equation*}%
is norm equicontinuous with respect to $A\in \mathcal{U}$ and we can
consider, without loss of generality, only elements $A$ within some dense
set of $\mathcal{U}$, like the dense $\ast $-algebra $\mathcal{U}_{0}$ of
local elements defined by (\ref{simple}). This sequence is also uniformly
bounded by $2\Vert A\Vert _{\mathcal{U}}$ in the operator norm of $\mathcal{B%
}(\mathcal{H}_{\hat{\rho}})$\ and hence, using \cite[Proposition 2.4.2]%
{BrattelliRobinsonI}, we only need to prove the weak-operator convergence on
any dense set of $\mathcal{H}_{\hat{\rho}}$, like 
\begin{equation*}
\left\{ \pi _{\hat{\rho}}\left( B\right) \Omega _{\hat{\rho}}:B\in \mathcal{U%
}\right\} \subseteq \mathcal{H}_{\hat{\rho}}\ ,
\end{equation*}%
in order to get the desired $\sigma $-weak convergence. Moreover, by \cite[%
Proposition 3.7]{BruPedra-MFII}, we can replace $\gimel _{t,s}$ with the
local dynamics 
\begin{equation*}
\gimel _{t,s}^{(L)}\doteq \tau _{t,s}^{(L,\mathbf{\Phi }^{(\mathfrak{m},%
\mathbf{\varpi }^{\mathfrak{m}}\left( \alpha ,\cdot \right) )}(\hat{\rho}%
))}|_{\alpha =s}\ ,\qquad \alpha ,s,t\in \mathbb{R},\ L\in \mathbb{N}\ .
\end{equation*}%
To summarize, at fixed $s,t\in \mathbb{R}$, it suffices to prove that%
\begin{eqnarray}
&&\lim_{L\rightarrow \infty }\left\langle \pi _{\hat{\rho}}\left( B\right)
\Omega _{\hat{\rho}},\pi _{\hat{\rho}}\left( \gimel _{t,s}^{(L)}\left(
A\right) -\tau _{t,s}^{(L,\mathfrak{m})}\left( A\right) \right) \pi _{\hat{%
\rho}}\left( B\right) \Omega _{\hat{\rho}}\right\rangle _{\mathcal{H}_{\hat{%
\rho}}}  \notag \\
&=&\lim_{L\rightarrow \infty }\hat{\rho}\left( B^{\ast }\left( \gimel
_{t,s}^{(L)}\left( A\right) -\tau _{t,s}^{(L,\mathfrak{m})}\left( A\right)
\right) B\right) =0  \label{limit to prove}
\end{eqnarray}%
for all elements $A\in \mathcal{U}_{0}$ and $B\in \mathcal{U}$ in order to
prove the theorem.\bigskip

\noindent \underline{Step 2:} By Duhamel's formula and Equations (\ref%
{equation derivation long range})-(\ref{cauchy2}), for any $L\in \mathbb{N}$
and $s,t\in \mathbb{R}$, 
\begin{equation}
\gimel _{t,s}^{(L)}-\tau _{t,s}^{(L,\mathfrak{m})}=\int_{s}^{t}\gimel
_{u,s}^{(L)}\circ \left( \delta _{L}^{\mathbf{\Phi }^{(\mathfrak{m},\mathbf{%
\varpi }^{\mathfrak{m}}\left( s,u\right) )}(\hat{\rho})}-\delta _{L}^{%
\mathfrak{m}\left( u\right) }\right) \circ \tau _{t,u}^{(L,\mathfrak{m})}%
\mathrm{d}u\ ,  \label{Duhamel's formula}
\end{equation}%
where $\delta _{L}^{\Phi }$ and $\delta _{L}^{\mathfrak{\tilde{m}}}$ are the
bounded symmetric derivations defined by (\ref{equation derivation long
range}) for any model $\mathfrak{\tilde{m}}\in \mathcal{M}$ and short-range
interaction $\Phi \equiv \left( \Phi ,0\right) \in \mathcal{W}\subseteq 
\mathcal{M}$. Observe from Equations (\ref{equation fininte vol dynam0}), (%
\ref{equation long range energy}), (\ref{equation derivation long range})
and Definition \ref{definition BCS-type model approximated} together with
explicit computations that, for any $L\in \mathbb{N}$, $u\in \mathbb{R}$ and 
$A\in \mathcal{U}$, 
\begin{eqnarray}
&&\left( \delta _{L}^{\mathbf{\Phi }^{(\mathfrak{m},\mathbf{\varpi }^{%
\mathfrak{m}}\left( s,u\right) )}(\hat{\rho})}-\delta _{L}^{\mathfrak{m}%
\left( u\right) }\right) \left( A\right)  \notag \\
&=&\sum_{n=2}^{\infty }\sum_{m=1}^{n}\int_{\mathbb{S}^{n}}\mathfrak{a}%
_{n}\left( u\right) \left( \mathrm{d}\Psi ^{(1)},\ldots ,\mathrm{d}\Psi
^{(n)}\right)  \notag \\
&&\qquad \left( \left( \prod\limits_{j=1}^{m-1}\mathbf{\varpi }^{\mathfrak{m}%
}\left( s,u;\hat{\rho}\right) (\mathfrak{e}_{\Psi ^{(j)},\vec{\ell}})\right)
\delta _{L}^{\Psi ^{(m)}}\left( A\right) \left( \prod\limits_{j=m+1}^{n}%
\mathbf{\varpi }^{\mathfrak{m}}\left( s,u;\hat{\rho}\right) (\mathfrak{e}%
_{\Psi ^{(j)},\vec{\ell}})\right) \right.  \notag \\
&&\qquad \qquad \qquad \qquad \qquad \qquad \left. -\left(
\prod\limits_{j=1}^{m-1}\frac{U_{L}^{\Psi ^{(j)}}}{\left\vert \Lambda
_{L}\right\vert }\right) \delta _{L}^{\Psi ^{(m)}}\left( A\right) \left(
\prod\limits_{j=m+1}^{n}\frac{U_{L}^{\Psi ^{(j)}}}{\left\vert \Lambda
_{L}\right\vert }\right) \right) \ ,  \label{delta difgferente}
\end{eqnarray}%
where 
\begin{equation*}
\prod\limits_{j=1}^{m-1<j}\left( \cdot \right) \doteq 1\doteq
\prod\limits_{j=m+1>n}^{n}\left( \cdot \right) .
\end{equation*}%
Combining (\ref{delta difgferente}) with (\ref{Duhamel's formula}), we
deduce that, for any $L\in \mathbb{N}$, $s,t\in \mathbb{R}$, $A\in \mathcal{U%
}_{0}$ and $B\in \mathcal{U}$,%
\begin{eqnarray}
&&\left\vert \hat{\rho}\left( B^{\ast }\left( \gimel _{t,s}^{(L)}\left(
A\right) -\tau _{t,s}^{(L,\mathfrak{m})}\left( A\right) \right) B\right)
\right\vert  \label{upper boundddvdv} \\
&\leq &\sum_{n=2}^{\infty }\sum_{m=1}^{n}\int_{s}^{t}\mathrm{d}u\int_{%
\mathbb{S}^{n}}\mathfrak{a}\left( u\right) _{n}\left( \mathrm{d}\Psi
^{(1)},\ldots ,\mathrm{d}\Psi ^{(n)}\right) \left\vert \mathbf{Y}%
_{L}^{(n,m)}\left( u;\Psi ^{(1)},\ldots ,\Psi ^{(n)}\right) \right\vert \ , 
\notag
\end{eqnarray}%
where, for any integer $n\geq 2$, $m\in \left\{ 1,\ldots ,n\right\} $, $L\in 
\mathbb{N}$, $u\in \mathbb{R}$ and $\Psi ^{(1)},\ldots ,\Psi ^{(n)}\in 
\mathbb{S}$, 
\begin{eqnarray}
&&\mathbf{Y}_{L}^{(n,m)}\left( u;\Psi ^{(1)},\ldots ,\Psi ^{(n)}\right)
\label{YL} \\
&\doteq &\hat{\rho}\left( B^{\ast }\gimel _{u,s}^{(L)}\left( \left(
\prod\limits_{j=1}^{m-1}\mathbf{\varpi }^{\mathfrak{m}}\left( s,u;\hat{\rho}%
\right) (\mathfrak{e}_{\Psi ^{(j)},\vec{\ell}})\right) \delta _{L}^{\Psi
^{(m)}}\circ \tau _{t,u}^{(L,\mathfrak{m})}\left( A\right) \left(
\prod\limits_{j=m+1}^{n}\mathbf{\varpi }^{\mathfrak{m}}\left( s,u;\hat{\rho}%
\right) (\mathfrak{e}_{\Psi ^{(j)},\vec{\ell}})\right) \right) B\right) 
\notag \\
&&-\hat{\rho}\left( B^{\ast }\gimel _{u,s}^{(L)}\left( \left(
\prod\limits_{j=1}^{m-1}\frac{U_{L}^{\Psi ^{(j)}}}{\left\vert \Lambda
_{L}\right\vert }\right) \delta _{L}^{\Psi ^{(m)}}\circ \tau _{t,u}^{(L,%
\mathfrak{m})}\left( A\right) \left( \prod\limits_{j=m+1}^{n}\frac{%
U_{L}^{\Psi ^{(j)}}}{\left\vert \Lambda _{L}\right\vert }\right) \right)
B\right) \ .  \notag
\end{eqnarray}%
By (\ref{e phi}) and (\ref{norm Uphi}),%
\begin{equation}
\left\vert \mathbf{\varpi }^{\mathfrak{m}}\left( s,u;\hat{\rho}\right) (%
\mathfrak{e}_{\Psi ,\vec{\ell}})\right\vert \leq \left\Vert \mathbf{F}%
\right\Vert _{1,\mathfrak{L}}\quad \text{and}\quad \left\Vert \left\vert
\Lambda _{L}\right\vert ^{-1}U_{L}^{\Psi }\right\Vert _{\mathcal{U}}\leq
\left\Vert \mathbf{F}\right\Vert _{1,\mathfrak{L}}
\label{uniformemement borne0}
\end{equation}%
for any $\Psi \in \mathbb{S}$, while, by Lemma \ref{Lemma sympa copy(3)}, 
\begin{equation}
\left\Vert \delta _{L}^{\Psi }\circ \tau _{t,u}^{(L,\mathfrak{m})}\left(
A\right) \right\Vert _{\mathcal{U}}\leq 2\left\vert \mathcal{Z}\right\vert
\left\Vert A\right\Vert _{\mathcal{U}}\mathrm{e}^{16\left( \mathbf{D}%
+2\left\Vert \mathbf{F}\right\Vert _{1,\mathfrak{L}}+1\right)
\int_{s}^{t}\left\Vert \mathfrak{m}\left( \varsigma \right) \right\Vert _{%
\mathcal{M}}\mathrm{d}\varsigma }  \label{uniformemement borne}
\end{equation}%
for any $\Psi \in \mathbb{S}$, $L\in \mathbb{N}$, $u\in \lbrack s,t]$, $%
\mathcal{Z}\subseteq \Lambda _{L}$ and $A\in \mathcal{U}_{\mathcal{Z}}$.
Therefore, since $\gimel _{u,s}^{(L)}$ is a $\ast $-automorphism of $%
\mathcal{U}$, for any integer $n\geq 2$, $m\in \left\{ 1,\ldots ,n\right\} $%
, $L\in \mathbb{N}$, $u\in \lbrack s,t]$, $\Psi ^{(1)},\ldots ,\Psi
^{(n)}\in \mathbb{S}$, $\mathcal{Z}\subseteq \Lambda _{L}$ and $A\in 
\mathcal{U}_{\mathcal{Z}}$, 
\begin{equation*}
\left\vert \mathbf{Y}_{L}^{(n,m)}\left( u;\Psi ^{(1)},\ldots ,\Psi
^{(n)}\right) \right\vert \leq 4\left\vert \mathcal{Z}\right\vert \left\Vert
A\right\Vert _{\mathcal{U}}\left\Vert B\right\Vert _{\mathcal{U}%
}^{2}\left\Vert \mathbf{F}\right\Vert _{1,\mathfrak{L}}^{n-1}\mathrm{e}%
^{16\left( \mathbf{D}+2\left\Vert \mathbf{F}\right\Vert _{1,\mathfrak{L}%
}+1\right) \int_{s}^{t}\left\Vert \mathfrak{m}\left( \varsigma \right)
\right\Vert _{\mathcal{M}}\mathrm{d}\varsigma }\ .
\end{equation*}%
Since $\mathfrak{m}\in C_{b}(\mathbb{R};\mathcal{M})$, by Equations (\ref%
{definition 0})-(\ref{definition 0bis}), (\ref{def long range2}) and (\ref%
{upper boundddvdv}), we deduce from the last estimate and Lebesgue's
dominated convergence theorem that (\ref{limit to prove}) follows from%
\begin{equation}
\lim_{L\rightarrow \infty }\left\vert \mathbf{Y}_{L}^{(n,m)}\left( u;\Psi
^{(1)},\ldots ,\Psi ^{(n)}\right) \right\vert =0  \label{limit}
\end{equation}%
for $n\geq 2$, $m\in \left\{ 1,\ldots ,n\right\} $, $u\in \mathbb{R}$, $\Psi
^{(1)},\ldots ,\Psi ^{(n)}\in \mathbb{S}$.\bigskip

\noindent \underline{Step 3:} For any integer $n\geq 2$, $k\in \left\{
1,\ldots ,n+1\right\} $, $l\in \left\{ 0,\ldots ,n\right\} $, $k\leq l$ and $%
\Psi ^{(k)},\ldots ,\Psi ^{(l)}\in \mathbb{S}$, define%
\begin{equation}
\mathbf{\Theta }_{L}^{(k,l)}\left( \Psi ^{(k)},\ldots ,\Psi ^{(l)}\right)
\doteq \prod\limits_{j=k}^{l}\mathbf{\varpi }^{\mathfrak{m}}\left( s,u;\hat{%
\rho}\right) (\mathfrak{e}_{\Psi ^{(j)},\vec{\ell}})-\prod\limits_{j=k}^{l}%
\frac{U_{L}^{\Psi ^{(j)}}}{\left\vert \Lambda _{L}\right\vert }.
\label{delta difgferente2}
\end{equation}%
By Equations (\ref{YL})-(\ref{uniformemement borne}) for $A\in \mathcal{U}%
_{0}$, and since $\gimel _{u,s}^{(L)}$ is a $\ast $-automorphism of $%
\mathcal{U}$, the limit we want to prove, i.e., (\ref{limit}), follows if we
are able to show that 
\begin{equation}
\lim_{L\rightarrow \infty }\left\Vert \pi _{\hat{\rho}}\circ \gimel
_{u,s}^{(L)}\left( \mathbf{\Theta }_{L}^{(k,l)}\left( \Psi ^{(k)},\ldots
,\Psi ^{(l)}\right) \right) \pi _{\hat{\rho}}\left( B\right) \Omega _{\hat{%
\rho}}\right\Vert _{\mathcal{H}_{\hat{\rho}}}=0  \label{limit final}
\end{equation}%
for any integer $n\geq 2$, $k,l\in \left\{ 1,\ldots ,n\right\} $ with $k\leq
l$, $B\in \mathcal{U}$ and $\Psi ^{(k)},\ldots ,\Psi ^{(l)}\in \mathbb{S}$.
Note that 
\begin{equation*}
\gimel _{u,s}\circ \alpha _{x}=\alpha _{x}\circ \gimel _{u,s}\ ,\qquad x\in 
\mathbb{Z}^{d},\ u,s\in \mathbb{R}\ ,
\end{equation*}%
because $\mathfrak{m}\in C_{b}(\mathbb{R};\mathcal{M}_{1})$. So, as $\hat{%
\rho}\in \mathcal{E}(E_{\vec{\ell}})$, we infer from \cite[Proposition 3.7]%
{BruPedra-MFII}, Lemma \ref{Lemma sympa copy(1)}, Corollary \ref{Lemma sympa
copy(2)} and Equation (\ref{delta difgferente2}) that 
\begin{eqnarray*}
&&\lim_{L\rightarrow \infty }\pi _{\hat{\rho}}\circ \gimel
_{u,s}^{(L)}\left( \mathbf{\Theta }_{L}^{(k,l)}\left( \Psi ^{(k)},\ldots
,\Psi ^{(l)}\right) \right) \pi _{\hat{\rho}}\left( B\right) \Omega _{\hat{%
\rho}} \\
&=&\left( \prod\limits_{j=k}^{l}\mathbf{\varpi }^{\mathfrak{m}}\left( s,u;%
\hat{\rho}\right) (\mathfrak{e}_{\Psi ^{(j)},\vec{\ell}})-\prod%
\limits_{j=k}^{l}\hat{\rho}\circ \gimel _{u,s}(\mathfrak{e}_{\Psi ^{(j)},%
\vec{\ell}})\right) \pi _{\hat{\rho}}\left( B\right) \Omega _{\hat{\rho}}
\end{eqnarray*}%
in the Hilbert space $\mathcal{H}_{\hat{\rho}}$, for any integer $n\geq 2$, $%
k,l\in \left\{ 1,\ldots ,n\right\} $ with $k\leq l$, $B\in \mathcal{U}$ and $%
\Psi ^{(k)},\ldots ,\Psi ^{(l)}\in \mathbb{S}$. We finally invoke the
self-consistency equations, that is, Theorem \ref{theorem sdfkjsdklfjsdklfj}
(cf. (\ref{notation interaction})), to arrive from the last equality at
Equation (\ref{limit final}), which, by going backwards, in turn implies the
theorem.
\end{proof}

\subsection{Direct Integrals of GNS Representations of Families of States 
\label{Direct Integrals of GNS}}

Theorem \ref{theorem structure of omega} is already proven for any ergodic
state, by Theorem \ref{theorem structure of omega copy(2)}. In order to
extend this result to all periodic states we need to decompose periodic
states into ergodic states, as stated in Theorem \ref{theorem choquet}. This
leads to a technically convenient cyclic representation of each periodic
state by using the direct integral of the GNS representation of ergodic
states. In this subsection (and in the next one), we explain the direct
integrals of GNS spaces in a general framework, since the particularities of
the CAR algebra $\mathcal{U}$ are never used, apart from the general fact
that it is a separable unital $C^{\ast }$-algebra.

Note additionally that this subsection is a collection of results that are
rather standard. Nevertheless, it is important to present them in a coherent
and self-contained manner because (i) we do not know any simple reference
allowing the reader to get the relevant information in a concise way and
(ii) the content of this subsection, in terms of results, notation, etc., is
essential for Section \ref{Representation cool}. Proofs are also included
here, making this subsection also useful for a full understanding by a
non-expert reader.

Let $\mathcal{X}$ be any separable unital $C^{\ast }$-algebra and denote by $%
E\subseteq \mathcal{X}^{\ast }$ its set of states\footnote{%
I.e., continuous linear functionals $\rho \in \mathcal{X}^{\ast }$ which are
positive, i.e., $\rho (A^{\ast }A)\geq 0$ for all $A\in \mathcal{X}$, and
normalized, i.e., $\rho (\mathfrak{1})=1$.}. For any state $\rho \in E$, its
GNS representation is denoted by the triplet $(\mathcal{H}_{\rho },\pi
_{\rho },\Omega _{\rho })$ with the following definitions (see, e.g., \cite[%
Section 2.3.3]{BrattelliRobinsonI}):\bigskip

\noindent \underline{($\mathcal{H}$):} $\mathcal{L}_{\rho }\doteq \{X\in 
\mathcal{X}:\rho (X^{\ast }X)=0\}$ is a closed left--ideal of $\mathcal{X}$
and $\mathcal{H}_{\rho }\doteq \overline{\mathcal{X}/\mathcal{L}_{\rho }}$
is the separable GNS Hilbert space with scalar product satisfying 
\begin{equation}
\langle \lbrack X]_{\rho },[Y]_{\rho }\rangle _{\mathcal{H}_{\rho }}=\rho
\left( X^{\ast }Y\right) \text{ },\text{\qquad }[X]_{\rho },[Y]_{\rho }\in 
\mathcal{X}/\mathcal{L}_{\rho }\subseteq \mathcal{H}_{\rho }\ .
\label{GNS innerprod}
\end{equation}

\noindent \underline{($\pi $):} $\pi _{\rho }$ is a representation of $%
\mathcal{X}$ on $\mathcal{B}(\mathcal{H}_{\rho })$ uniquely defined by 
\begin{equation}
\pi _{\rho }\left( A\right) [X]_{\rho }=[AX]_{\rho }\in \mathcal{X}/\mathcal{%
L}_{\rho }\text{ },\text{\qquad }A\in \mathcal{X},\ [X]_{\rho }\in \mathcal{X%
}/\mathcal{L}_{\rho }\subseteq \mathcal{H}_{\rho }\ .  \label{GNS innerprod2}
\end{equation}

\noindent \underline{($\Omega $):} $\Omega _{\rho }\doteq \lbrack \mathfrak{1%
}]_{\rho }\in \mathcal{X}/\mathcal{L}_{\rho }\subseteq \mathcal{H}_{\rho }$
is a cyclic vector for $\pi _{\rho }(\mathcal{X})$, i.e., the set $\pi
_{\rho }(\mathcal{X})\Omega _{\rho }$ is dense in $\mathcal{H}_{\rho }$.
\bigskip

We apply now the general theory discussed in Section \ref{app direct
integrals} to the GNS objects (space, representation and cyclic vectors) of
a separable $C^{\ast }$-algebra: Recall that $E$ is compact and metrizable
with respect to the weak$^{\ast }$ topology. Let $\Sigma _{E}$ be the
(Borel) $\sigma $-algebra generated by the weak$^{\ast }$ topology of $E$,
like in Section \ref{Positive Functionals}. Pick any fixed (weak$^{\ast }$)
Borel subset $F\in \Sigma _{E}$ and denote by $(F,\Sigma _{F})$ the
measurable space associated with the $\sigma $-algebra $\Sigma _{F}$
generated by the weak$^{\ast }$ topology of $F$. Note that%
\begin{equation*}
\Sigma _{F}=\left\{ F\cap \mathfrak{B}:\mathfrak{B}\in \Sigma _{E}\right\} \
.
\end{equation*}%
(In Section \ref{app direct integrals}, $F$ refers to the set denoted by $%
\mathcal{Z}$.) Since compact metric spaces are complete and separable, the
measurable space $(F,\Sigma _{F})$ is standard (Definition \ref{Standard
measure spaces}), whenever $F$\ is closed.

For any $F\in \Sigma _{E}$, $\mathcal{H}_{F}\doteq (\mathcal{H}_{\rho
})_{\rho \in F}$ is a family\ of separable GNS Hilbert spaces and, by the
GNS\ construction, any element $X\in \mathcal{X}$ defines a vector field 
\begin{equation}
v\doteq (v_{\rho })_{\rho \in F}\doteq ([X]_{\rho })_{\rho \in F}
\label{def a la con}
\end{equation}%
over the family $\mathcal{H}_{F}$. Also, $\pi _{F}\doteq (\pi _{\rho
})_{\rho \in F}$ is a field of representations of $\mathcal{X}$ (Definition %
\ref{Field of representation} (i)) on $\mathcal{H}_{\rho }$ for $\rho \in F$%
. Now, it suffices to use Theorems \ref{Measurable families of Hilbert
spaces - Equivalent formulation} and \ref{coherence from fields} to get the
measurability of all these objects:

\begin{lemma}[Measurability of\ GNS\ Hilbert spaces and representations]
\label{lemma direct integral von neumann2 copy(1)}\mbox{ }\newline
Let $\mathcal{X}$ be a separable unital $C^{\ast }$-algebra and $F\in \Sigma
_{E}$ any (weak$^{\ast }$) Borel subset of states. Then, $\mathcal{H}_{F}$
is measurable and there is a unique (equivalence class of) coherence $\alpha
_{F}$ making, via (\ref{def a la con}), any countable (norm) dense set of $%
\mathcal{X}$ a sequence of $\alpha _{F}$-measurable fields . Moreover, $\pi
_{F}$\ is $\alpha _{F}$-measurable.
\end{lemma}

\begin{proof}
Let $\{X^{(n)}\}_{n\in \mathbb{N}}\subseteq \mathcal{X}$ be any countable
(norm) dense set. By the GNS construction, this set defines via (\ref{def a
la con}) a countable dense subset of the GNS Hilbert space $\mathcal{H}%
_{\rho }$ for all $\rho \in E$. Hence, it defines a sequence of vector
fields over $\mathcal{H}_{F}$, denoted by $(v^{(n)})_{n\in \mathbb{N}}$,
where, for all $\rho \in E$, the set $\{v_{\rho }^{(n)}\}_{n\in \mathbb{N}}$
is dense in $\mathcal{H}_{\rho }$ and, in particular, total in this space.
Since, for all $B\in \mathcal{X}$, the mapping $\rho \mapsto \rho (B)$ from $%
F$ to $\mathbb{C}$ is weak$^{\ast }$-continuous, we deduce from (\ref{GNS
innerprod}) that the mapping $\rho \mapsto \langle v_{\rho }^{(m)},v_{\rho
}^{(n)}\rangle _{\mathcal{H}_{\rho }}$ from $F$ to $\mathbb{C}$ is\ $\Sigma
_{F}$-measurable for all $m,n\in \mathbb{N}$. Thus, the sequence $%
(v^{(n)})_{n\in \mathbb{N}}$ of vector fields fulfills Conditions (a)-(b) of
Theorem \ref{Measurable families of Hilbert spaces - Equivalent formulation}
and $\mathcal{H}_{F}$ is thus measurable.

We take the unique (up to an equivalence of coherences) coherence $\alpha
_{F}$ for $\mathcal{H}_{F}$\ such that $(v^{(n)})_{n\in \mathbb{N}}$ is a
sequence of $\alpha _{F}$-measurable fields, see Theorem \ref{coherence from
fields} (i). Observing that the point-wise limit of a sequence of measurable
functions is measurable, and using again (\ref{GNS innerprod}) and Theorem %
\ref{coherence from fields} (i)-(ii), one checks that (the equivalence class
of) the coherence $\alpha _{F}$ does not depend on the particular choice of
the dense countable set $\{X^{(n)}\}_{n\in \mathbb{N}}\subseteq \mathcal{X}$
originally taken.

Finally, since $\{X^{(n)}\}_{n\in \mathbb{N}}\subseteq \mathcal{X}$ is any
dense set, by (\ref{GNS innerprod})-(\ref{def a la con}), we have 
\begin{equation}
\langle v_{\rho }^{(n)},\pi _{\rho }(A)v_{\rho }^{(m)}\rangle _{\mathcal{H}%
_{\rho }}=\rho ((X^{(n)})^{\ast }AX^{(m)})\text{ },\text{\qquad }n,m\in 
\mathbb{N},\ A\in \mathcal{X}\text{\ },  \label{measruability}
\end{equation}%
and the sequence $(v^{(n)})_{n\in \mathbb{N}}$ of vector fields satisfies
Conditions (a)-(b) of Theorem \ref{Measurable families of Hilbert spaces -
Equivalent formulation}. It follows from Theorem \ref{coherence from fields}
and Definition \ref{Field of representation} that $\pi _{\rho }$\ is $\alpha
_{F}$-measurable.
\end{proof}

In view of Section \ref{Positive Functionals}, take now any (weak$^{\ast }$)
closed set $F\in \Sigma _{E}$. The set of all positive Radon measures on $%
(F,\Sigma _{F})$ is denoted by $\mathrm{M}(F)$, each element of which
corresponds (one-to-one) to a positive regular Borel measure. In fact, by
separability of the $C^{\ast }$-algebra $\mathcal{X}$, $E$ is metrizable and
thus, any positive finite Borel measure on $(E,\Sigma _{E})$ is regular, as
already explained for $\mathcal{X}=\mathcal{U}$.

On the one hand, by Lemma \ref{lemma direct integral von neumann2 copy(1)},
one can construct a direct integral triplet $(\mathcal{H}_{F}^{\oplus },\pi
_{F}^{\oplus },\Omega _{F}^{\oplus })$ associated with each closed set $F\in
\Sigma _{E}$ and any positive Radon measure $\mu \in \mathrm{M}(F)$: \bigskip

\noindent \underline{($\mathcal{H}^{\oplus }$):} $\mathcal{H}_{F}^{\oplus }$
is the direct integral Hilbert space associated with the measurable family $%
\mathcal{H}_{F}$. See Definition \ref{Direct integrals of Hilbert spaces}.
Since $\mathcal{H}_{F}$ has a canonical (equivalence class of) coherence $%
\alpha _{F}$, we use the notation%
\begin{equation}
\mathcal{H}_{F}^{\oplus }\equiv \int_{F}^{\alpha _{F}}\mathcal{H}_{\rho }\mu
(\mathrm{d}\rho )\equiv \int_{F}\mathcal{H}_{\rho }\mu (\mathrm{d}\rho )\ .
\label{direct integral GNS}
\end{equation}%
Similarly, we say that any vector field (respectively operator field) $%
v\doteq (v_{\rho })_{\rho \in F}$ (respectively $A\doteq (A_{\rho })_{\rho
\in F})$) over $\mathcal{H}_{F}$ is measurable\ whenever it is $\alpha _{F}$%
-measurable. When $F$\ is a closed set, $(F,\Sigma _{F},\mu )$ is standard
and hence, $\mathcal{H}_{F}^{\oplus }$ is in this case separable, by Theorem %
\ref{Nielsen Th 7.1 copy(1)}. \bigskip

\noindent \underline{($\pi ^{\oplus }$):} $\pi _{F}^{\oplus }$ is the direct
integral (representation)\ of the ($\alpha _{F}$-) measurable representation
field $\pi _{F}$, see (\ref{direct integral representation1})-(\ref{direct
integral representation2}). Similar to (\ref{direct integral GNS}), we use
the notation%
\begin{equation}
\pi _{F}^{\oplus }\equiv \int_{F}^{\alpha _{F}}\pi _{\rho }\mu (\mathrm{d}%
\rho )\equiv \int_{F}\pi _{\rho }\mu (\mathrm{d}\rho )\ .
\label{direct integral representation states}
\end{equation}%
It is a representation of the separable unital $C^{\ast }$-algebra $\mathcal{%
X}$\ on the direct integral Hilbert space $\mathcal{H}_{F}^{\oplus }$.
\bigskip

\noindent \underline{($\Omega ^{\oplus }$):} The vector $\Omega _{F}^{\oplus
}$ is the element of the direct integral Hilbert space $\mathcal{H}%
_{F}^{\oplus }$ defined by%
\begin{equation}
\Omega _{F}^{\oplus }\doteq \int_{F}^{\alpha _{F}}\Omega _{\rho }\mu (%
\mathrm{d}\rho )\equiv \int_{F}\Omega _{\rho }\mu (\mathrm{d}\rho )\in 
\mathcal{H}_{F}^{\oplus }\ .  \label{posible cyclic vector}
\end{equation}%
This vector is well-defined because positive Radon measures on compact
spaces are always finite. In contrast with the usual GNS\ representation,
note that $\Omega _{F}^{\oplus }$ is generally not a cyclic vector for $\pi
_{F}^{\oplus }(\mathcal{X})$.\bigskip

On the other hand, as explained around (\ref{choquet0}), a positive Radon
measure $\mu \in \mathrm{M}(F)$ represents a unique positive functional $%
\rho _{\mu }\in \mathcal{X}^{\ast }$, which is defined by 
\begin{equation*}
\rho _{\mu }\left( A\right) \doteq \int_{F}\rho \left( A\right) \mu \left( 
\mathrm{d}\rho \right) \ ,\qquad A\in \mathcal{X}\ .
\end{equation*}%
$\rho _{\mu }$ is called the barycenter of $\mu \in \mathrm{M}(F)$ and, as
any positive functional of $\mathcal{X}^{\ast }$, it has a GNS\
representation. Clearly, the triplet $(\mathcal{H}_{F}^{\oplus },\pi
_{F}^{\oplus },\Omega _{F}^{\oplus })$ can be used to represent the positive
functional $\rho _{\mu }$, in the sense that%
\begin{equation*}
\rho _{\mu }(A)=\langle \Omega _{F}^{\oplus },\pi _{F}^{\oplus }(A)\Omega
_{F}^{\oplus }\rangle _{\mathcal{H}_{F}^{\oplus }}\ ,\qquad A\in \mathcal{X}%
\ .
\end{equation*}%
However, as already mentioned, $\Omega _{F}^{\oplus }$ is not necessarily a
cyclic vector for $\pi _{F}^{\oplus }(\mathcal{X})$ and, in general, $(%
\mathcal{H}_{F}^{\oplus },\pi _{F}^{\oplus },\Omega _{F}^{\oplus })$ is only
quasi-equivalent (see \cite[Theorem 2.4.26]{BrattelliRobinsonI}) to any
cyclic representation of $\rho _{\mu }$:

\begin{lemma}[Direct integral and GNS representations]
\label{Direct integral and GNS lemma}\mbox{ }\newline
Let $\mathcal{X}$ be a separable unital $C^{\ast }$-algebra, $F\in \Sigma
_{E}$ any (weak$^{\ast })$ Borel subset of states and $\mu \in \mathrm{M}(F)$
a positive Radon measure with barycenter $\rho _{\mu }\in \mathcal{X}^{\ast
} $ and associated cyclic representation$\ (\mathcal{H}_{\rho _{\mu }},\pi
_{\rho _{\mu }},\Omega _{\rho _{\mu }})$. Then, $\pi _{F}^{\oplus }$ is
quasi-equivalent to $\pi _{\rho _{\mu }}$. It is equivalent to $\pi _{\rho
_{\mu }}$ iff $\Omega _{F}^{\oplus }$ is a cyclic vector for $\pi
_{F}^{\oplus }(\mathcal{X})$.
\end{lemma}

\begin{proof}
Fix all parameters of the lemma. Let $P_{F}$ be the orthogonal projection on 
$\mathcal{H}_{F}^{\oplus }$ whose range $\mathrm{ran}P_{F}$ is the closure
of the subspace $\pi _{F}^{\oplus }(\mathcal{X})\Omega _{F}^{\oplus }$: 
\begin{equation*}
\mathrm{ran}P_{F}\doteq \overline{\left\{ \pi _{F}^{\oplus }(A)\Omega
_{F}^{\oplus }:A\in \mathcal{X}\right\} }\subseteq \mathcal{H}_{F}^{\oplus
}\ .
\end{equation*}%
Clearly, $P_{F}\in \lbrack \pi _{F}^{\oplus }(\mathcal{X})]^{\prime }$ and $%
P_{F}\Omega _{F}^{\oplus }=\Omega _{F}^{\oplus }$. Therefore, the mapping $%
A\mapsto \pi _{F}^{\oplus }(A)|_{\mathrm{ran}P_{F}}$ from $\mathcal{X}$ to $%
\mathcal{B}(P_{F}\mathcal{H}_{F}^{\oplus })$ defines a representation $%
\tilde{\pi}_{F}^{\oplus }$ of $\mathcal{X}$ on the Hilbert space $P_{F}%
\mathcal{H}_{F}^{\oplus }$. Additionally, there is a $\ast $-isomorphism $%
\gimel $ from the von Neumann algebra $[\pi _{F}^{\oplus }(\mathcal{X}%
)]^{\prime \prime }$ to $[\tilde{\pi}_{F}^{\oplus }(\mathcal{X})]^{\prime
\prime }$ such that 
\begin{equation*}
\gimel \left( \pi _{F}^{\oplus }(A)\right) \doteq \pi _{F}^{\oplus }(A)|_{%
\mathrm{ran}P_{F}}=\tilde{\pi}_{F}^{\oplus }(A)\ ,\qquad A\in \mathcal{X}\ .
\end{equation*}%
This follows from the weak-operator continuity of $P_{F}(\cdot )P_{F}$. By 
\cite[Theorem 2.4.26]{BrattelliRobinsonI}, the representations $\pi
_{F}^{\oplus }$ and $\tilde{\pi}_{F}^{\oplus }$ are thus quasi-equivalent.
On the other hand, by construction, $(P_{F}\mathcal{H}_{F}^{\oplus },\tilde{%
\pi}_{F}^{\oplus },\Omega _{F}^{\oplus })$ is a cyclic representation for
the positive functional $\rho _{\mu }$ and is thus spatially, or unitarily,
equivalent to $\pi _{\rho _{\mu }}$, by \cite[Theorem 2.3.16]%
{BrattelliRobinsonI} (which can trivially be extended to any positive
functional of $\mathcal{X}^{\ast }$).
\end{proof}

\noindent Unless $\Omega _{F}^{\oplus }$ is a cyclic vector for $\pi
_{F}^{\oplus }(\mathcal{X})$, i.e., the orthogonal projection $P_{F}$ of the
last proof is the identity operator on $\mathcal{H}_{F}^{\oplus }$, the
triplet $(\mathcal{H}_{F}^{\oplus },\pi _{F}^{\oplus },\Omega _{F}^{\oplus
}) $ is not spatially equivalent to any cyclic representation of $\rho _{\mu
}$. In the following we discuss necessary and sufficient conditions on the
positive Radon measure $\mu $ for $\Omega _{F}^{\oplus }$ to be cyclic, in
order to have a direct integral decomposition of the cyclic\ representation
of $\rho _{\mu }$ as $(\mathcal{H}_{F}^{\oplus },\pi _{F}^{\oplus },\Omega
_{F}^{\oplus })$.

For each $\mu \in \mathrm{M}(F)$, we define its restriction $\mu _{\mathfrak{%
B}}\in \mathrm{M}(F)$ to any Borel set $\mathfrak{B}\in \Sigma _{F}$ by 
\begin{equation*}
\mu _{\mathfrak{B}}\left( \mathfrak{B}_{0}\right) \doteq \mu \left( 
\mathfrak{B}_{0}\cap \mathfrak{B}\right) \ ,\qquad \mathfrak{B}_{0}\in
\Sigma _{F}\ .
\end{equation*}%
See (\ref{restriction radon}). If $\Omega _{F}^{\oplus }$ is a cyclic vector
for $\pi _{F}^{\oplus }(\mathcal{X})$ then one easily checks that $\Omega _{%
\mathfrak{B}}^{\oplus }$ is also cyclic for $\pi _{\mathfrak{B}}^{\oplus }(%
\mathcal{X})$\ and so, $(\mathcal{H}_{\mathfrak{B}}^{\oplus },\pi _{%
\mathfrak{B}}^{\oplus },\Omega _{\mathfrak{B}}^{\oplus })$ is a cyclic
representation of the barycenter $\rho _{\mu _{\mathfrak{B}}}$ of the
restricted positive Radon measure $\mu _{\mathfrak{B}}\in \mathrm{M}(F)$. In
particular, for all Borel sets $\mathfrak{B}\in \Sigma _{F}$, 
\begin{equation}
\left( \mathcal{H}_{\rho _{\mu _{\mathfrak{B}}}}\oplus \mathcal{H}_{\rho
_{\mu _{F\backslash \mathfrak{B}}}},\ \pi _{\rho _{\mu _{\mathfrak{B}%
}}}\oplus \pi _{\rho _{\mu _{F\backslash \mathfrak{B}}}}\ ,\ \Omega _{\rho
_{\mu _{\mathfrak{B}}}}\oplus \Omega _{\rho _{\mu _{F\backslash \mathfrak{B}%
}}}\right)  \label{decomposition}
\end{equation}%
is a cyclic representation of $\rho _{\mu }$. This motivates the following
definition:

\begin{definition}[Orthogonal measures]
\label{Orthogonal measures}\mbox{ }\newline
Let $\mathcal{X}$ be a separable unital $C^{\ast }$-algebra and $F\subseteq 
\mathcal{X}^{\ast }$ any weak$^{\ast }$-closed subset of states. A positive
Radon measure $\mu \in \mathrm{M}(F)$ is orthogonal whenever, for all Borel
sets $\mathfrak{B}\in \Sigma _{F}$, (\ref{decomposition}) is a cyclic
representation of its barycenter $\rho _{\mu }\in \mathcal{X}^{\ast }$,
i.e., $\rho _{\mu _{\mathfrak{B}}}\perp \rho _{\mu _{F\backslash \mathfrak{B}%
}}$ (see around (\ref{orthogonality})).
\end{definition}

\noindent As already explained, if $(\mathcal{H}_{F}^{\oplus },\pi
_{F}^{\oplus },\Omega _{F}^{\oplus })$ is a cyclic representation of $\rho
_{\mu }$\ then $\mu $ is an orthogonal measure. We prove below that this
orthogonality property is also a sufficient condition for the cyclicity of $%
\Omega _{F}^{\oplus }$.

The positive Radon measures we are interested in concern those coming from
the Choquet theorem \cite[Theorem 10.18]{BruPedra2}, which allow to
decompose states of a compact convex set into extreme ones. Such measures
are always probability measures, i.e., normalized positive Radon measures.
The subset of all probability measures on $(F,\Sigma _{F})$ is denoted by $%
\mathrm{M}_{1}(F)$. So, for simplicity, we consider, from now on, only
probability measures $\mu \in \mathrm{M}_{1}(F)$ on $F$.

Recall that $L^{\infty }(F,\mu )$ is the space of all (equivalence classes
of) essentially bounded measurable complex-valued functions on $F$
associated with the measure space $(F,\Sigma _{F},\mu )$. The $\mathrm{ess}%
\sup $ norm of $L^{\infty }(F,\mu )$ is denoted by $\Vert \cdot \Vert
_{\infty }$. We give a first, very useful, lemma (cf. \cite[Lemma 4.1.21]%
{BrattelliRobinsonI}), similar to Theorem \ref{Nielsen Th 7.1 copy(2)},
which links $L^{\infty }(F,\mu )$ with the commutant $[\pi _{\rho _{\mu }}(%
\mathcal{X})]^{\prime }$:

\begin{lemma}[Bounded measurable functions and the GNS\ representation]
\label{lemma chi mu}\mbox{ }\newline
Let $\mathcal{X}$ be a separable unital $C^{\ast }$-algebra and $F\subseteq 
\mathcal{X}^{\ast }$ any weak$^{\ast }$-closed subset of states. For any
probability measure $\mu \in \mathrm{M}_{1}(F)$, there is a unique linear
map $\varkappa _{\mu }:L^{\infty }(F,\mu )\rightarrow \lbrack \pi _{\rho
_{\mu }}(\mathcal{X})]^{\prime }$ such that 
\begin{equation}
\int_{F}f(\rho )\rho (A)\mu (\mathrm{d}\rho )=\left\langle \Omega _{\rho
_{\mu }},\pi _{\rho _{\mu }}(A)\varkappa _{\mu }(f)\Omega _{\rho _{\mu
}}\right\rangle _{\mathcal{H}_{\rho _{\mu }}}\ ,\qquad A\in \mathcal{X}\ ,
\label{sfsdfsfdsdfs}
\end{equation}%
and $\left\Vert \varkappa _{\mu }(f)\right\Vert _{\mathcal{B}(\mathcal{H}%
_{\rho _{\mu }})}\leq \left\Vert f\right\Vert _{\infty }$ for all $f\in
L^{\infty }(F,\mu )$, where $(\mathcal{H}_{\rho _{\mu }},\pi _{\rho _{\mu
}},\Omega _{\rho _{\mu }})$ is any cyclic representation of the barycenter $%
\rho _{\mu }$ of $\mu $. Additionally, $\varkappa _{\mu }$ is unital,
positivity-preserving and, for all $f\in L^{\infty }(F,\mu )$, there is a
unique $\varkappa _{\mu }(f)\in \mathcal{B}(\mathcal{H}_{\rho _{\mu }})$
satisfying (\ref{sfsdfsfdsdfs}).
\end{lemma}

\begin{proof}
Let $\mu \in \mathrm{M}_{1}(F)$ with $(\mathcal{H}_{\rho _{\mu }},\pi _{\rho
_{\mu }},\Omega _{\rho _{\mu }})$ being a cyclic representation of its
barycenter $\rho _{\mu }$. Choose any (essentially) bounded function $f\in
L^{\infty }(F,\mu )$. Observe that 
\begin{equation*}
(\pi _{\rho _{\mu }}(A)\Omega _{\rho _{\mu }},\pi _{\rho _{\mu }}(B)\Omega
_{\rho _{\mu }})\mapsto \int_{F}f(\rho )\rho (A^{\ast }B)\mu (\mathrm{d}\rho
)\ ,\qquad A,B\in \mathcal{X}\ ,
\end{equation*}%
defines a unique sesquilinear form on $\mathcal{H}_{\rho _{\mu }}$ bounded
by $\left\Vert f\right\Vert _{\infty }$. Therefore, there is a unique $%
\varkappa _{\mu }(f)\in \mathcal{B}(\mathcal{H}_{\rho _{\mu }})$ such that $%
\left\Vert \varkappa _{\mu }(f)\right\Vert _{\mathcal{B}(\mathcal{H}_{\rho
_{\mu }})}\leq \left\Vert f\right\Vert _{\infty }$ and satisfying (\ref%
{sfsdfsfdsdfs}). Clearly, if $f\geq 0$ then $\varkappa _{\mu }(f)\geq 0$ and 
$\varkappa _{\mu }(f)=1$ when $f=1$. Now, elementary computations show that $%
\varkappa _{\mu }(f)\in \lbrack \pi _{\rho _{\mu }}(\mathcal{X})]^{\prime }$%
. See \cite[Proof of Theorem 2.3.19]{BrattelliRobinsonI} for more details.
\end{proof}

If $(\mathcal{H}_{F}^{\oplus },\pi _{F}^{\oplus },\Omega _{F}^{\oplus })$ is
equivalent to any cyclic representation of $\rho _{\mu }$ then the mapping $%
\varkappa _{\mu }$ of Lemma \ref{lemma chi mu} is nothing else as (up to
unitary equivalence) the $\ast $-isomorphism from $L^{\infty }(F,\mu )$ to
the abelian von Neumann algebra $N_{F}$ of diagonalizable operators on $%
\mathcal{H}_{F}^{\oplus }$, as given in Theorem \ref{Nielsen Th 7.1 copy(2)}%
. In particular, $\varkappa _{\mu }$ has to be a $\ast $-homomorphism (and
not only a linear map) whenever $(\mathcal{H}_{F}^{\oplus },\pi _{F}^{\oplus
},\Omega _{F}^{\oplus })$ is a cyclic representation of the barycenter $\rho
_{\mu }$ of $\mu \in \mathrm{M}_{1}(F)$. We exploit now this observation to
show that $(\mathcal{H}_{F}^{\oplus },\pi _{F}^{\oplus },\Omega _{F}^{\oplus
})$ is a cyclic representation of $\rho _{\mu }$ iff $\mu $ is an orthogonal
measure. This is a consequence of the Tomita theorem:

\begin{proposition}[\textquotedblleft half\textquotedblright\ Tomita's
theorem]
\label{half Tomita}\mbox{ }\newline
Let $\mathcal{X}$ be a separable unital $C^{\ast }$-algebra, $F\subseteq 
\mathcal{X}^{\ast }$ any weak$^{\ast }$-closed subset of states and $\mu \in 
\mathrm{M}_{1}(F)$ any probability measure with barycenter $\rho _{\mu }$
and associated cyclic representation$\ (\mathcal{H}_{\rho _{\mu }},\pi
_{\rho _{\mu }},\Omega _{\rho _{\mu }})$. If $\mu $ is an orthogonal
measure, then the mapping $\varkappa _{\mu }$ of Lemma \ref{lemma chi mu} is
a $\ast $-isomorphism from $L^{\infty }(F,\mu )$ to $[\pi _{\rho _{\mu }}(%
\mathcal{X})]^{\prime }$. In particular $\varkappa _{\mu }(L^{\infty }(F,\mu
))$ is an abelian von Neumann subalgebra of the commutant $[\pi _{\rho _{\mu
}}(\mathcal{X})]^{\prime }$.
\end{proposition}

\begin{proof}
For completeness, we reproduce here the proof of \cite[Proposition 4.1.22; $%
(1)\Rightarrow (2)$]{BrattelliRobinsonI}. Let $\mu \in \mathrm{M}_{1}(F)$ be
an orthogonal probability measure with barycenter $\rho _{\mu }$ and
corresponding cyclic representation$\ (\mathcal{H}_{\rho _{\mu }},\pi _{\rho
_{\mu }},\Omega _{\rho _{\mu }})$. Then, by Definition \ref{Orthogonal
measures} and \cite[Theorem 2.3.16]{BrattelliRobinsonI}, for any Borel set $%
\mathfrak{B}\in \Sigma _{F}$, there is an orthogonal projection $P_{%
\mathfrak{B}}\in \lbrack \pi _{\rho _{\mu }}(\mathcal{X})]^{\prime }$ acting
on $\mathcal{H}_{\rho _{\mu }}$ such that, for all $A\in \mathcal{X}$, 
\begin{equation}
\rho _{\mu _{\mathfrak{B}}}(A)=\left\langle \Omega _{\rho _{\mu }},\pi
_{\rho _{\mu }}(A)P_{\mathfrak{B}}\Omega _{\rho _{\mu }}\right\rangle _{%
\mathcal{H}_{\rho _{\mu }}}\qquad \text{and}\qquad \rho _{\mu _{F\backslash 
\mathfrak{B}}}=\left\langle \Omega _{\rho _{\mu }},\pi _{\rho _{\mu }}(A)(%
\mathbf{1}_{\mathcal{H}_{\rho _{\mu }}}-P_{\mathfrak{B}})\Omega _{\rho _{\mu
}}\right\rangle _{\mathcal{H}_{\rho _{\mu }}}\ .  \label{definition bis}
\end{equation}%
Therefore, if $\zeta _{\mathfrak{B}}$ denotes the characteristic function of
any Borel set $\mathfrak{B}\in \Sigma _{F}$, then, by Lemma \ref{lemma chi
mu}, $\varkappa _{\mu }(\zeta _{\mathfrak{B}})=P_{\mathfrak{B}}$ is always
an orthogonal projection and, for all $\mathfrak{B}_{1},\mathfrak{B}_{2}\in
\Sigma _{F}$, $\mathfrak{B}_{1}\mathfrak{\cap B}_{2}=\emptyset $, 
\begin{equation}
\varkappa _{\mu }\left( \zeta _{\mathfrak{B}_{1}}\right) \varkappa _{\mu
}\left( \zeta _{\mathfrak{B}_{2}}\right) =0\ .  \label{disjoint}
\end{equation}%
This last equality comes from the fact that, whenever $\mathfrak{B}_{1}%
\mathfrak{\cap B}_{2}=\emptyset $, $\zeta _{\mathfrak{B}_{1}}\leq 1-\zeta _{%
\mathfrak{B}_{2}}$, leading by $\varkappa _{\mu }(1)=\mathbf{1}_{\mathcal{H}%
_{\rho _{\mu }}}$, linearity and positivity of $\varkappa _{\mu }$, to $%
\varkappa _{\mu }(\zeta _{\mathfrak{B}_{1}})\leq \mathbf{1}_{\mathcal{H}%
_{\rho _{\mu }}}-\varkappa _{\mu }(\zeta _{\mathfrak{B}_{2}})$, in turn
implying (\ref{disjoint}). Now, for any $\mathfrak{B}_{1},\mathfrak{B}%
_{2}\in \Sigma _{F}$, we can rewrite the characteristic functions $\zeta _{%
\mathfrak{B}_{1}},\zeta _{\mathfrak{B}_{2}}$ as%
\begin{equation*}
\zeta _{\mathfrak{B}_{1}}=\zeta _{\mathfrak{B}_{1}}\zeta _{\mathfrak{B}%
_{2}}+\zeta _{\mathfrak{B}_{1}}(\zeta _{F}-\zeta _{\mathfrak{B}_{2}})\qquad 
\text{and}\qquad \zeta _{\mathfrak{B}_{2}}=\zeta _{\mathfrak{B}_{2}}\zeta _{%
\mathfrak{B}_{1}}+\zeta _{\mathfrak{B}_{2}}(\zeta _{F}-\zeta _{\mathfrak{B}%
_{1}})
\end{equation*}%
and use (\ref{disjoint}) for disjoint Borel subsets to deduce the equality%
\begin{equation}
\varkappa _{\mu }(\zeta _{\mathfrak{B}_{1}})\varkappa _{\mu }(\zeta _{%
\mathfrak{B}_{2}})=\varkappa _{\mu }(\zeta _{\mathfrak{B}_{1}}\zeta _{%
\mathfrak{B}_{2}})\ ,\qquad \mathfrak{B}_{1},\mathfrak{B}_{2}\in \Sigma
_{F}\ .  \label{disjoint2}
\end{equation}%
All functions of $L^{\infty }(F,\mu )$ can be approximated in this Banach
space by linear combination of characteristic functions and since $\varkappa
_{\mu }$ is linear and a contractive mapping (Lemma \ref{lemma chi mu}), we
deduce from (\ref{disjoint2}) that $\varkappa _{\mu }$ is a $\ast $%
-homomorphism. Using now this property, one easily checks that $\varkappa
_{\mu }(f)\Omega _{\rho _{\mu }}=0$ iff $f=0$. Since $\Omega _{\rho _{\mu }}$
is a cyclic vector and $\varkappa _{\mu }(f)\in \lbrack \pi _{\rho _{\mu }}(%
\mathcal{X})]^{\prime }$, it follows that $\varkappa _{\mu }$ is in fact a $%
\ast $-isomorphism.
\end{proof}

Proposition \ref{half Tomita} is part of the Tomita theorem, which says that 
$\mu $ is an orthogonal measure iff the mapping $\varkappa _{\mu }$ of Lemma %
\ref{lemma chi mu} is a $\ast $-homomorphism from $L^{\infty }(F,\mu )$ to $%
[\pi _{\rho _{\mu }}(\mathcal{X})]^{\prime }$. See \cite[Proposition 4.1.22]%
{BrattelliRobinsonI}. In the case that $\varkappa _{\mu }(L^{\infty }(F,\mu
))$ is an abelian von Neumann subalgebra of $[\pi _{\rho _{\mu }}(\mathcal{X}%
)]^{\prime }$, the corresponding direct integral decomposition of the
representation $\pi _{\rho _{\mu }}$ is, as expected, $\pi _{F}^{\oplus }$,
see (\ref{direct integral representation states}). In fact, we have the
following result, which refers to the Effros theorem \cite[Theorem 4.4.9]%
{BrattelliRobinsonI}:

\begin{corollary}[Effros]
\label{Effros theorem}\mbox{ }\newline
Let $\mathcal{X}$ be a separable unital $C^{\ast }$-algebra, $F\subseteq 
\mathcal{X}^{\ast }$ any weak$^{\ast }$-closed subset of states and $\mu \in 
\mathrm{M}_{1}(F)$ any probability measure. $\mu $ is an orthogonal measure
iff $\Omega _{F}^{\oplus }$ is a cyclic vector for $\pi _{F}^{\oplus }(%
\mathcal{X})$. In particular, if $\mu $ is orthogonal then $(\mathcal{H}%
_{F}^{\oplus },\pi _{F}^{\oplus },\Omega _{F}^{\oplus })$ is a cyclic
representation of its barycenter.
\end{corollary}

\begin{proof}
If $\Omega _{F}^{\oplus }$ is a cyclic vector for $\pi _{F}^{\oplus }(%
\mathcal{X})$ then, as already explained before Definition \ref{Orthogonal
measures}, $\mu $ has to be orthogonal. Now, suppose that the mapping $%
\varkappa _{\mu }$ of Lemma \ref{lemma chi mu} is a $\ast $-homomorphism. In
particular, 
\begin{equation}
\varkappa _{\mu }(|f|^{2})=\varkappa _{\mu }(f)^{\ast }\varkappa _{\mu }(f)\
,\qquad f\in L^{\infty }\left( F,\mu \right) \ .  \label{dsdsdfjh}
\end{equation}%
Then, a simple computation using (\ref{dsdsdfjh}) and $\varkappa _{\mu
}(f)\in \lbrack \pi _{\rho _{\mu }}(\mathcal{X})]^{\prime }$ shows that, for
all $A,B\in \mathcal{X}$ and $f\in L^{\infty }(F,\mu )$,%
\begin{equation}
\left\Vert \pi _{F}^{\oplus }(B)\Omega _{F}^{\oplus }-\phi (f)\pi
_{F}^{\oplus }(A)\Omega _{F}^{\oplus }\right\Vert _{\mathcal{H}_{F}^{\oplus
}}=\left\Vert \pi _{\rho _{\mu }}(B)\Omega _{\rho _{\mu }}-\varkappa _{\mu
}(f)\pi _{\rho _{\mu }}(A)\Omega _{\rho _{\mu }}\right\Vert _{\mathcal{H}%
_{\rho _{\mu }}}\text{ },  \label{sfgasfgfgsdg}
\end{equation}%
provided $\varkappa _{\mu }$ is a $\ast $-homomorphism, where%
\begin{equation*}
\phi (f)\doteq \int_{F}f(\rho )\mathbf{1}_{\mathcal{H}_{\rho }}\mu (\mathrm{d%
}\rho )\in \mathcal{B}\left( \mathcal{H}_{F}^{\oplus }\right) \text{ }%
,\qquad f\in L^{\infty }(F,\mu )\ .
\end{equation*}%
Now, by applying Lemma \ref{construction of coherences I copy(1)}, observe
that the set 
\begin{equation*}
\left\{ \phi \left( f\right) \pi _{F}^{\oplus }\left( A\right) \Omega
_{F}^{\oplus }:f\in L^{\infty }\left( F,\mu \right) ,\ A\in \mathcal{X}%
\right\}
\end{equation*}%
is dense in $\mathcal{H}_{F}^{\oplus }$. Hence, if $\mu $ is orthogonal
then, as $\Omega _{\rho _{\mu }}$ is cyclic for $\pi _{\rho _{\mu }}(%
\mathcal{X})$, we deduce from Proposition \ref{half Tomita} and Equality (%
\ref{sfgasfgfgsdg}) that $\Omega _{F}^{\oplus }$ is a cyclic vector for $\pi
_{F}^{\oplus }(\mathcal{X})$. By Lemma \ref{Direct integral and GNS lemma},
this implies in turn that $(\mathcal{H}_{F}^{\oplus },\pi _{F}^{\oplus
},\Omega _{F}^{\oplus })$ is a cyclic representation of the barycenter $\rho
_{\mu }$ of $\mu $.
\end{proof}

Therefore, Proposition \ref{half Tomita}, Corollary \ref{Effros theorem} and 
\cite[Theorem 2.3.16]{BrattelliRobinsonI} show, in a constructive way, that,
for any cyclic representation $(\mathcal{H}_{\rho _{\mu }},\pi _{\rho _{\mu
}},\Omega _{\rho _{\mu }})$ of an orthogonal probability measure $\mu \in 
\mathrm{M}_{1}(F)$, the von Neumann algebra $[\pi _{\rho _{\mu }}(\mathcal{X}%
)]^{\prime \prime }$ is decomposable with respect to the abelian von Neumann
subalgebra $\mathfrak{N}_{\rho _{\mu }}\doteq \varkappa _{\mu }(L^{\infty
}\left( F,\mu \right) )$ of the commutant $[\pi _{\rho _{\mu }}(\mathcal{X}%
)]^{\prime }$. See Definition \ref{def direct integral represe copy(1)}. By
Theorem \ref{lemma direct integral von neumann4 copy(1)}, it means that 
\begin{equation}
\mathfrak{N}_{\rho _{\mu }}\subseteq \lbrack \pi _{\rho _{\mu }}(\mathcal{X}%
)]^{\prime }\cap \lbrack \pi _{\rho _{\mu }}(\mathcal{X})]^{\prime \prime }
\label{fgkjghdfkjghdfkj}
\end{equation}%
while, by Corollary \ref{lemma direct integral von neumann6}, 
\begin{equation*}
\left[ \pi _{F}^{\oplus }(\mathcal{X})\right] ^{\prime \prime }=\int_{F}%
\left[ \pi _{\rho }(\mathcal{X})\right] ^{\prime \prime }\mu (\mathrm{d}\rho
)\ .
\end{equation*}%
In this case, $\pi _{F}^{\oplus }$ is a so-called subcentral decomposition
(of the representation $\pi _{\rho _{\mu }}$), see Definition \ref{Special
direct integral representations} (ii.1).

\subsection{$C^{\ast }$-Algebra of $\mathcal{X}$-valued Continuous Functions
on States\label{Representation cool}}

We conclude with some properties of the $C^{\ast }$-algebra $C(F;\mathcal{X}%
) $ of $\mathcal{X}$-valued weak$^{\ast }$-continuous functions on any weak$%
^{\ast }$-closed subset $F\subseteq E$ of states, where $\mathcal{X}$ is a
separable unital $C^{\ast }$-algebra. Such a space is crucial to study the
dynamics of long-range models at infinite volume. Similar to (\ref%
{metaciagre set})-(\ref{metaciagre set bis}), $C(F;\mathcal{X})$ is endowed
with the (point-wise) algebra operations inherited from $\mathcal{X}$. The
unique $C^{\ast }$- norm of $C(F;\mathcal{X})$ is 
\begin{equation*}
\left\Vert f\right\Vert _{C\left( F;\mathcal{X}\right) }\doteq \max_{\rho
\in F}\left\Vert f\left( \rho \right) \right\Vert _{\mathcal{X}}\ ,\qquad
f\in C\left( F;\mathcal{X}\right) \ .
\end{equation*}%
We identify $\mathcal{X}$ with the subalgebra of constant functions of $C(F;%
\mathcal{X})$ and $C(F)\doteq C(F;\mathbb{C})$ with the subalgebra of
functions whose values are multiples of the unit $\mathfrak{1}\in \mathcal{X}
$. In other words, 
\begin{equation*}
\mathcal{X}\subseteq C\left( F;\mathcal{X}\right) \qquad \text{and}\qquad
C\left( F\right) \doteq C\left( F;\mathbb{C}\right) \subseteq C\left( F;%
\mathcal{X}\right) \ .
\end{equation*}%
In fact, $\mathcal{X}$ and $C(F)$ are $C^{\ast }$-subalgebras of $C(F;%
\mathcal{X})$ and the set $C(F)\cup \mathcal{X}$ generates this $C^{\ast }$%
-algebra:

\begin{lemma}[Generation of $C(F;\mathcal{X})$ by elementary functions]
\label{lemma dense func alg}\mbox{ }\newline
Let $\mathcal{X}$ be a separable unital $C^{\ast }$-algebra and $F\subseteq 
\mathcal{X}^{\ast }$ any weak$^{\ast }$-closed subset of states. If $%
\mathcal{X}_{0}$ is a dense set of $\mathcal{X}$, then%
\begin{equation*}
C(F)\mathcal{X}_{0}\doteq \left\{ fA:f\in C(F),\text{ }A\in \mathcal{X}%
_{0}\right\}
\end{equation*}%
is total in $C(F;\mathcal{X})$. If $\mathcal{X}_{0}$ is a $\ast $-subalgebra
of $\mathcal{X}$\ then $C(F)\mathcal{X}_{0}$ is a $\ast $-subalgebra of $C(F;%
\mathcal{X})$.
\end{lemma}

\begin{proof}
Use the density of $\mathcal{X}_{0}\subseteq \mathcal{X}$ as well as the weak%
$^{\ast }$-compactness of $F$ together with the existence of partitions of
unity subordinated to any open cover of the metrizable (weak$^{\ast }$%
-compact) space $F$.
\end{proof}

Recall that $\Sigma _{E}$ is the (Borel) $\sigma $-algebra generated by the
weak$^{\ast }$ topology of the weak$^{\ast }$-compact and metrizable space $%
E $ of states. For any weak$^{\ast }$-closed subset $F\in \Sigma _{E}$, $%
\Sigma _{F}$ is the $\sigma $-algebra generated by the weak$^{\ast }$
topology of $F$. The GNS representation of any state $\rho \in E$ is denoted
by $(\mathcal{H}_{\rho },\pi _{\rho },\Omega _{\rho })$. For any $F\in
\Sigma _{E}$, $\mathcal{H}_{F}\doteq (\mathcal{H}_{\rho })_{\rho \in F}$ is
a measurable family\ of separable GNS\ Hilbert spaces and $\pi _{F}\doteq
(\pi _{\rho })_{\rho \in F}$ is a measurable field of GNS representations of 
$\mathcal{X}$ on $\mathcal{H}_{\rho }$ for $\rho \in F$, see Lemma \ref%
{lemma direct integral von neumann2 copy(1)}. Similarly, for any weak$^{\ast
}$-closed subset $F\in \Sigma _{E}$ and all $f\in C(F;\mathcal{X})$, the
bounded field $(\pi _{\rho }(f(\rho )))_{\rho \in F}$ of operators over $(%
\mathcal{H}_{\rho })_{\rho \in F}$ is also measurable:

\begin{lemma}[Measurability of\ GNS\ representations applied to $\mathcal{X}$%
-valued functions]
\label{representations applied to}\mbox{ }\newline
Let $\mathcal{X}$ be a separable unital $C^{\ast }$-algebra and $F\subseteq 
\mathcal{X}^{\ast }$ any weak$^{\ast }$-closed subset of states. Then, for
all $f\in C(F;\mathcal{X})$, $(\pi _{\rho }(f(\rho )))_{\rho \in F}$ is a
bounded ($\alpha _{F}$-) measurable field of operators over $\mathcal{H}%
_{F}\doteq (\mathcal{H}_{\rho })_{\rho \in F}$.
\end{lemma}

\begin{proof}
The proof is the same as the one proving that $\pi _{F}$ is measurable. In
particular, for any $f\in C(F;\mathcal{X})$, one uses (\ref{measruability})
for $A=f(\rho )$ together with the weak$^{\ast }$-continuity of functions of 
$C(F;\mathcal{X})$ and Theorem \ref{coherence from fields} (iii) to deduce
the assertion.
\end{proof}

For any $F\in \Sigma _{E}$ and any positive Radon measure $\mu \in \mathrm{M}%
(F)$, $\mathcal{H}_{F}^{\oplus }$ is the direct integral Hilbert space (\ref%
{direct integral GNS}) associated with $\mathcal{H}_{F}$. The direct
integral $\pi _{F}^{\oplus }$ of $\pi _{F}$, defined by (\ref{direct
integral representation states}), is a representation of $\mathcal{X}$\ on
the direct integral Hilbert space $\mathcal{H}_{F}^{\oplus }$. In the same
way, for any weak$^{\ast }$-closed subset $F$ of states, we obtain from
Lemma \ref{representations applied to} a representation of $C(F;\mathcal{X})$%
\ on $\mathcal{H}_{F}^{\oplus }$, that is, a $\ast $-homomorphism from $C(F;%
\mathcal{X})$ to $\mathcal{B}(\mathcal{H}_{F}^{\oplus })$. This
representation is a natural extension of $\pi _{F}^{\oplus }$ from $\mathcal{%
X}$\ to the $C^{\ast }$-algebra $C(F;\mathcal{X})$:\bigskip

\noindent \underline{($\Pi ^{\oplus }$):} $\Pi _{F}^{\oplus }$ is the
(direct integral) representation\ of $C(F;\mathcal{X})$ on $\mathcal{H}%
_{F}^{\oplus }$ defined by%
\begin{equation}
\Pi _{F}^{\oplus }\left( f\right) \doteq \int_{F}^{\alpha _{F}}\pi _{\rho
}(f(\rho ))\mu (\mathrm{d}\rho )\equiv \int_{F}\pi _{\rho }(f(\rho ))\mu (%
\mathrm{d}\rho )\ ,\qquad f\in C(F;\mathcal{X})\ .
\label{direct intergral representation of C(F,A)}
\end{equation}%
Compare with Equation (\ref{direct integral representation states}) defining 
$\pi _{F}^{\oplus }$.\bigskip

Recall that $\Omega _{F}^{\oplus }\in \mathcal{H}_{F}^{\oplus }$ is defined
from the measurable family $\Omega _{F}\doteq (\Omega _{\rho })_{\rho \in F}$
by the direct integral (\ref{posible cyclic vector}). By the Effros theorem
(Corollary \ref{Effros theorem}), if $\mu \in \mathrm{M}_{1}(F)$ is an
orthogonal probability measure then $(\mathcal{H}_{F}^{\oplus },\pi
_{F}^{\oplus },\Omega _{F}^{\oplus })$ is a cyclic representation of its
barycenter $\rho _{\mu }$ and, by \cite[Theorem 2.3.16]{BrattelliRobinsonI}, 
$(\mathcal{H}_{F}^{\oplus },\pi _{F}^{\oplus },\Omega _{F}^{\oplus })$ is
equivalent to any cyclic\ representation of $\rho _{\mu }$. In this case, $%
C(F;\mathcal{X})$ can be represented in the Hilbert space $\mathcal{H}_{\rho
_{\mu }}$ of any cyclic representation of $\rho _{\mu }$. More precisely,
from the Tomita theorem (Proposition \ref{half Tomita}), we obtain the
following assertion:

\begin{proposition}[Orthogonal measures and representations of $C(F;\mathcal{%
X})$]
\label{proposition repr C orthogonal case}\mbox{ }\newline
Let $\mathcal{X}$ be a separable unital $C^{\ast }$-algebra, $F\subseteq 
\mathcal{X}^{\ast }$ any weak$^{\ast }$-closed subset of states and $\mu \in 
\mathrm{M}_{1}(F)$ any probability measure with barycenter $\rho _{\mu }$
and associated cyclic representation$\ (\mathcal{H}_{\rho _{\mu }},\pi
_{\rho _{\mu }},\Omega _{\rho _{\mu }})$. If $\mu $ is an orthogonal
measure, then there exists a unique representation $\Pi _{\rho _{\mu }}$ of $%
C(F;\mathcal{X})$ on $\mathcal{H}_{\rho _{\mu }}$ such that 
\begin{equation*}
\Pi _{\rho _{\mu }}|_{\mathcal{X}}=\pi _{\rho _{\mu }}\qquad \text{and}%
\qquad \Pi _{\rho _{\mu }}|_{C\left( F\right) }=\varkappa _{\mu }|_{C\left(
F\right) }
\end{equation*}%
with $\varkappa _{\mu }$ being the $\ast $-isomorphism from $L^{\infty
}(F,\mu )$ to $[\pi _{\rho _{\mu }}(\mathcal{X})]^{\prime }$ originally
defined in Lemma \ref{lemma chi mu}. Additionally, 
\begin{equation}
\left\langle \Omega _{\rho _{\mu }},\Pi _{\rho _{\mu }}\left( f\right)
\Omega _{\rho _{\mu }}\right\rangle _{\mathcal{H}_{\rho _{\mu
}}}=\int_{F}\rho (f(\rho ))\mu (\mathrm{d}\rho )\ ,\quad f\in C\left( F;%
\mathcal{X}\right) \ .  \label{sdfsdfsdsdfsdfsdf0}
\end{equation}
\end{proposition}

\begin{proof}
Fix all parameters of the proposition. The representation $\Pi _{F}^{\oplus
} $ of $C(F;\mathcal{X})$ on $\mathcal{H}_{F}^{\oplus }$ defined by (\ref%
{direct intergral representation of C(F,A)}) obviously satisfies $\Pi
_{F}^{\oplus }|_{\mathcal{X}}=\pi _{F}^{\oplus }$, see (\ref{direct integral
representation states}). By (\ref{posible cyclic vector}) and (\ref{direct
intergral representation of C(F,A)}),%
\begin{equation}
\left\langle \Omega _{F}^{\oplus },\Pi _{F}^{\oplus }(f)\Omega _{F}^{\oplus
}\right\rangle _{\mathcal{H}_{F}^{\oplus }}=\int_{F}f(\rho )\rho (A)\mu (%
\mathrm{d}\rho )\ ,\quad f\in C\left( F;\mathcal{X}\right) \ .
\label{esdsfklgjsds}
\end{equation}%
By Corollary \ref{Effros theorem} and \cite[Theorem 2.3.16]%
{BrattelliRobinsonI}, if $\mu $ is an orthogonal measure then there is a
unitary operator $\mathrm{U}:\mathcal{H}_{F}^{\oplus }\rightarrow \mathcal{H}%
_{\rho _{\mu }}$ such that $\mathrm{U}\Omega _{F}^{\oplus }=\Omega _{\rho
_{\mu }}$ and $\pi _{\rho _{\mu }}(A)=\mathrm{U}\pi _{F}^{\oplus }(A)\mathrm{%
U}^{\ast }$ for any $A\in \mathcal{X}$. Let $\Pi _{\rho _{\mu }}$ be the
representation of $C(F;\mathcal{X})$ on $\mathcal{H}_{\rho _{\mu }}$ defined
by 
\begin{equation}
\Pi _{\rho _{\mu }}\left( f\right) \doteq \mathrm{U}\Pi _{F}^{\oplus }\left(
f\right) \mathrm{U}^{\ast }\ ,\qquad f\in C\left( F;\mathcal{X}\right) \ .
\label{unitary}
\end{equation}%
Clearly, $\Pi _{F}^{\oplus }|_{\mathcal{X}}=\pi _{F}^{\oplus }$ yields $\Pi
_{\rho _{\mu }}|_{\mathcal{X}}=\pi _{\rho _{\mu }}$ and, by (\ref%
{esdsfklgjsds}), we also deduce Equation (\ref{sdfsdfsdsdfsdfsdf0}). In
particular, for any $f\in C(F)$ and $A\in \mathcal{X}$, 
\begin{equation}
\left\langle \Omega _{\rho _{\mu }},\pi _{\rho _{\mu }}\left( A\right) \Pi
_{\rho _{\mu }}\left( f\right) \Omega _{\rho _{\mu }}\right\rangle _{%
\mathcal{H}_{\rho _{\mu }}}=\int_{F}f(\rho )\rho (A)\mu (\mathrm{d}\rho )\ .
\label{sdfsdfsdsdfsdfsdf}
\end{equation}%
By Lemma \ref{lemma chi mu}, $\Pi _{\rho _{\mu }}\left( f\right) =\varkappa
_{\mu }\left( f\right) $\ for any $f\in C(F)$. By Lemma \ref{lemma dense
func alg}, a representation of $C(F;\mathcal{X})$ on any\ Hilbert space is
uniquely defined by its values on $\mathcal{X}$ and $C(F)$. Therefore, $\Pi
_{\rho _{\mu }}$ is the unique representation which equals $\pi _{\rho _{\mu
}}$ and $\varkappa _{\mu }$ on $\mathcal{X}$ and $C(F)$, respectively.
\end{proof}

The support $\mathrm{supp}\mu $ of a positive Radon measure $\mu \in \mathrm{%
M}(F)$ is, by definition, the set%
\begin{equation*}
\left\{ \rho \in F:\mu (V_{\rho })>0\text{ for any weak}^{\ast }\text{-open
neighborhood }V_{\rho }\text{\ of }\rho \right\} \subseteq F\text{ }.
\end{equation*}%
In particular, it is weak$^{\ast }$-compact. As Radon measures are (inner
and outer) regular, $\mathrm{supp}$ $\mu $ has full measure, i.e., $\mu (%
\mathrm{supp}\mu )=1$. Therefore, one can assume, without loss of
generality, that $F=\mathrm{supp}\mu $. In this case, if the state $\rho \in
F$ is $\mu $-almost everywhere faithful then the representations $\pi
_{F}^{\oplus }$ and $\Pi _{F}^{\oplus }$ are faithful. If additionally $\mu
\in \mathrm{M}_{1}(F)$ is an orthogonal probability measure on $F$, then,
from Lemma \ref{lemma dense func alg}, Proposition \ref{proposition repr C
orthogonal case} and Equation (\ref{unitary}), we deduce the existence of a
unique $\ast $-isomorphism from the\ $C^{\ast }$-subalgebra of $\mathcal{B}(%
\mathcal{H}_{\rho _{\mu }})$ generated by $\pi _{\rho _{\mu }}(\mathcal{X}%
)\cup \varkappa _{\mu }(C(F))$ onto the $C^{\ast }$-algebra $C(F;\mathcal{X}%
) $, satisfying 
\begin{equation*}
\varkappa _{\mu }(f)\pi _{\rho _{\mu }}(A)\mapsto fA
\end{equation*}%
for all $f\in C(F)$\ and $A\in \mathcal{X}$.

We conclude the section by observing that the barycenter of any probability
measure on weak$^{\ast }$-closed subset of states can naturally be extended
to a state on the $C^{\ast }$-algebra $C(F;\mathcal{X})$:

\begin{definition}[Extension of states on $\mathcal{X}$ to the whole $%
C^{\ast }$-algebra $C(F;\mathcal{X})$]
\label{def Extension of states}\mbox{ }\newline
Let $\mathcal{X}$ be a separable unital $C^{\ast }$-algebra and $F\subseteq 
\mathcal{X}^{\ast }$ any weak$^{\ast }$-closed subset of states. For any
probability measure $\mu \in \mathrm{M}_{1}(F)$, its barycenter $\rho _{\mu
}\in F$ can naturally be extended to a state on $C(F;\mathcal{X})$, again
denoted by $\rho _{\mu }$, via the definition%
\begin{equation*}
\rho _{\mu }\left( f\right) \doteq \int_{F}\rho (f(\rho ))\mu (\mathrm{d}%
\rho )\text{ },\text{\qquad }f\in C(F;\mathcal{X})\text{\ }.
\end{equation*}
\end{definition}

\noindent Proposition \ref{proposition repr C orthogonal case} directly
yields a natural characterization of cyclic representations of the extension
to $C(F;\mathcal{X})$ of barycenters of orthogonal probability measures:

\begin{theorem}[Cyclic representations of barycenters]
\label{coro Extension of states}\mbox{ }\newline
Let $\mathcal{X}$ be a separable unital $C^{\ast }$-algebra, $F\subseteq 
\mathcal{X}^{\ast }$ any weak$^{\ast }$-closed subset of states and $\mu \in 
\mathrm{M}_{1}(F)$ any orthogonal probability measure with barycenter $\rho
_{\mu }$ seen as a state of either $\mathcal{X}^{\ast }$ or $C(F;\mathcal{X}%
)^{\ast }$\emph{.}\newline
\emph{(i)} $(\mathcal{H}_{F}^{\oplus },\Pi _{F}^{\oplus },\Omega
_{F}^{\oplus })$ is a cyclic representation of $\rho _{\mu }\in C(F;\mathcal{%
X})^{\ast }$. \newline
\emph{(ii)} Let $(\mathcal{H}_{\rho _{\mu }},\pi _{\rho _{\mu }},\Omega
_{\rho _{\mu }})$ be any cyclic representation of $\rho _{\mu }\in \mathcal{X%
}^{\ast }$. Then, there exists a unique representation $\Pi _{\rho _{\mu }}$
of $C(F;\mathcal{X})$ on $\mathcal{H}_{\rho _{\mu }}$ such that $\Pi _{\rho
_{\mu }}|_{\mathcal{X}}=\pi _{\rho _{\mu }}$ and $(\Omega _{\rho _{\mu
}},\Pi _{\rho _{\mu }},\mathcal{H}_{\rho _{\mu }})$ is a cyclic
representation of $\rho _{\mu }\in C(F;\mathcal{X})^{\ast }$. \newline
\emph{(iii)} Conversely, let $(\mathcal{H}_{\rho _{\mu }},\Pi _{\rho _{\mu
}},\Omega _{\rho _{\mu }})$ be any cyclic representation of $\rho _{\mu }\in
C(F;\mathcal{X})^{\ast }$. Then, $(\mathcal{H}_{\rho _{\mu }},\Pi _{\rho
_{\mu }}|_{\mathcal{X}},\Omega _{\rho _{\mu }})$ is a cyclic representation
of $\rho _{\mu }\in \mathcal{X}^{\ast }$,%
\begin{equation*}
\lbrack \Pi _{\rho _{\mu }}\left( C\left( F;\mathcal{X}\right) \right)
]^{\prime \prime }=[\Pi _{\rho _{\mu }}\left( \mathcal{X}\right) ]^{\prime
\prime }\qquad \text{and}\qquad \lbrack \Pi _{\rho _{\mu }}\left( C\left(
F\right) \right) ]^{\prime \prime }\subseteq \lbrack \Pi _{\rho _{\mu
}}\left( \mathcal{X}\right) ]^{\prime }\cap \lbrack \Pi _{\rho _{\mu
}}\left( \mathcal{X}\right) ]^{\prime \prime }\ .
\end{equation*}
\end{theorem}

\begin{proof}
Fix all parameters of the theorem.

\noindent \underline{(i):} By Corollary \ref{Effros theorem}, $\Omega
_{F}^{\oplus }$ is cyclic for the $\ast $-subalgebra 
\begin{equation*}
\pi _{F}^{\oplus }\left( \mathcal{X}\right) =\Pi _{F}^{\oplus }\left( 
\mathcal{X}\right) \subseteq \Pi _{F}^{\oplus }\left( C\left( F;\mathcal{X}%
\right) \right) \text{ }.
\end{equation*}%
Therefore, it is also cyclic for $\Pi _{F}^{\oplus }(C(F;\mathcal{X}))$.
Since $\Pi _{F}^{\oplus }$ is a representation of $C(F;\mathcal{X})$ on $%
\mathcal{H}_{F}^{\oplus }$, by (\ref{esdsfklgjsds}), the first assertion
follows.

\noindent \underline{(ii):} Let $(\mathcal{H}_{\rho _{\mu }},\pi _{\rho
_{\mu }},\Omega _{\rho _{\mu }})$ be any cyclic representation of $\rho
_{\mu }\in \mathcal{X}^{\ast }$. By Proposition \ref{proposition repr C
orthogonal case}, there exists a unique representation $\Pi _{\rho _{\mu }}$
of $C(F;\mathcal{X})$ on $\mathcal{H}_{\rho _{\mu }}$ such that $\Pi _{\rho
_{\mu }}|_{\mathcal{X}}=\pi _{\rho _{\mu }}$, $\Pi _{\rho _{\mu }}|_{C\left(
F\right) }=\varkappa _{\mu }|_{C\left( F\right) }$ and, since $\Omega _{\rho
_{\mu }}$ is cyclic for the $\ast $- subalgebra $\pi _{\rho _{\mu }}\left( 
\mathcal{X}\right) $, $(\Omega _{\rho _{\mu }},\Pi _{\rho _{\mu }},\mathcal{H%
}_{\rho _{\mu }})$ is a cyclic representation of $\rho _{\mu }\in C(F;%
\mathcal{X})^{\ast }$. See (\ref{sdfsdfsdsdfsdfsdf0}) and Definition \ref%
{def Extension of states}. Assume now the existence of another
representation $\tilde{\Pi}_{\rho _{\mu }}$ such that $\tilde{\Pi}_{\rho
_{\mu }}|_{\mathcal{X}}=\pi _{\rho _{\mu }}$ and $(\Omega _{\rho _{\mu }},%
\tilde{\Pi}_{\rho _{\mu }},\mathcal{H}_{\rho _{\mu }})$ is a cyclic
representation of $\rho _{\mu }\in C(F;\mathcal{X})^{\ast }$. Since 
\begin{equation*}
\tilde{\Pi}_{\rho _{\mu }}|_{\mathcal{X}}=\Pi _{\rho _{\mu }}|_{\mathcal{X}%
}=\pi _{\rho _{\mu }}\ ,
\end{equation*}%
we deduce that 
\begin{equation*}
\left\langle \pi _{\rho _{\mu }}\left( A\right) \Omega _{\rho _{\mu
}},\left( \tilde{\Pi}_{\rho _{\mu }}\left( f\right) -\Pi _{\rho _{\mu
}}\left( f\right) \right) \Omega _{\rho _{\mu }}\right\rangle _{\mathcal{H}%
_{\rho _{\mu }}}=0\text{ },\qquad f\in C(F),\ A\in \mathcal{X}\ .
\end{equation*}%
By cyclicity of $\Omega _{\rho _{\mu }}$ for $\pi _{\rho _{\mu }}\left( 
\mathcal{X}\right) $, it follows that 
\begin{equation*}
\tilde{\Pi}_{\rho _{\mu }}\left( f\right) =\Pi _{\rho _{\mu }}\left(
f\right) =\varkappa _{\mu }|_{C\left( F\right) }\text{ },\qquad f\in C(F)\ .
\end{equation*}%
By Proposition \ref{proposition repr C orthogonal case}, we conclude that $%
\tilde{\Pi}_{\rho _{\mu }}=\Pi _{\rho _{\mu }}$.

\noindent \underline{(iii):} By Corollary \ref{Effros theorem}, $(\mathcal{H}%
_{F}^{\oplus },\pi _{F}^{\oplus },\Omega _{F}^{\oplus })$ is a cyclic
representation of $\rho _{\mu }\in \mathcal{X}^{\ast }$ and $(\mathcal{H}%
_{F}^{\oplus },\Pi _{F}^{\oplus },\Omega _{F}^{\oplus })$ is thus a cyclic
representation of $\rho _{\mu }\in C(F;\mathcal{X})^{\ast }$, where $\Pi
_{F}^{\oplus }$ is the representation of $C(F;\mathcal{X})$ on $\mathcal{H}%
_{F}^{\oplus }$ defined by (\ref{direct intergral representation of C(F,A)}%
). By \cite[Theorem 2.3.16]{BrattelliRobinsonI}, any cyclic representation $(%
\mathcal{H}_{\rho _{\mu }},\Pi _{\rho _{\mu }},\Omega _{\rho _{\mu }})$ of $%
\rho _{\mu }\in C(F;\mathcal{X})^{\ast }$ is equivalent to $(\mathcal{H}%
_{F}^{\oplus },\pi _{F}^{\oplus },\Omega _{F}^{\oplus })$ and, since $\Omega
_{F}^{\oplus }$ is cyclic for $\Pi _{F}^{\oplus }(\mathcal{X})=\pi
_{F}^{\oplus }(\mathcal{X})$, the unit vector $\Omega _{\rho _{\mu }}$ is
also cyclic for $\Pi _{\rho _{\mu }}(\mathcal{X})$. In particular, $(%
\mathcal{H}_{\rho _{\mu }},\Pi _{\rho _{\mu }}|_{\mathcal{X}},\Omega _{\rho
_{\mu }})$ is a cyclic representation of $\rho _{\mu }\in \mathcal{X}^{\ast
} $. Using the definition $\pi _{\rho _{\mu }}\doteq \Pi _{\rho _{\mu }}|_{%
\mathcal{X}}$, Equation (\ref{sdfsdfsdsdfsdfsdf}) holds true, again by (\ref%
{esdsfklgjsds}) and \cite[Theorem 2.3.16]{BrattelliRobinsonI}. Therefore, $%
\Pi _{\rho _{\mu }}$ must be the unique representation of Proposition \ref%
{proposition repr C orthogonal case}. By Equation (\ref{fgkjghdfkjghdfkj}), $%
[\Pi _{\rho _{\mu }}\left( C\left( F\right) \right) ]^{\prime \prime }$ is
an abelian von Neumann subalgebra of the center of $[\Pi _{\rho _{\mu
}}\left( \mathcal{X}\right) ]^{\prime \prime }$. From Lemma \ref{lemma dense
func alg} it follows that 
\begin{equation*}
\lbrack \Pi _{\rho _{\mu }}\left( C\left( F;\mathcal{X}\right) \right)
]^{\prime \prime }=[\Pi _{\rho _{\mu }}\left( \mathcal{X}\right) \cup \Pi
_{\rho _{\mu }}\left( C\left( F\right) \right) ]^{\prime \prime }=[\Pi
_{\rho _{\mu }}\left( \mathcal{X}\right) ]^{\prime \prime }\ .
\end{equation*}
\end{proof}

\section{Appendix: Direct Integrals and Spatial Decompositions\label{app
direct integrals}}

In this section we review important aspects of the theory of direct
integrals of measurable families of Hilbert spaces, operators, von Neumann
algebras, and $C^{\ast }$-algebra representations, which are useful in the
scope of the present work. Mathematical foundations of the theory go back to
von Neumann in the pivotal paper \cite{von neuman}\footnote{%
Note that this paper was already written in 1937-1938, but only published
more than ten years later.}, aiming to obtain factor decompositions of
strongly closed operator algebras, i.e., von Neumann algebras.

Nowadays, constant-fiber spaces or algebras are much more popular\footnote{%
For instance, the theory for constant-fiber Hilbert spaces is a standard
tool to study Schr\"{o}dinger operators with periodic potentials, as
explained in \cite[Section XIII.16]{ReedSimonIV}.} than the more general
situation needed here, i.e., the non-constant fiber case, which were already
introduced by von Neumann in \cite{von neuman}. Thus, for self-containedness
of the paper and the reader's convenience, we concisely explain the general
theory of direct integrals. For a more thorough exposition on the subject,
as well as complete proofs, we refer to the monograph \cite{Niesen-direct
integrals}. Indeed, in our opinion, the approach of \cite{Niesen-direct
integrals}, based on the notion of \textquotedblleft \emph{coherence}%
\textquotedblright , is more intuitive, being more explicit, than other
well-known mathematical expositions of direct integrals of separable Hilbert
spaces. For another presentation of direct integrals, see, e.g., \cite[%
Section 4.4, in particular 4.4.1]{BrattelliRobinsonI}.

\begin{notation}
\label{remark constant copy(1)}\mbox{
}\newline
For any set $\mathcal{Z}$, $\mathcal{T}_{\mathcal{Z}}$ always denotes a
family $\mathcal{T}_{\mathcal{Z}}\doteq (\mathcal{T}_{z})_{z\in \mathcal{Z}}$%
. If such a family is an operator or vector field, as defined above, we even
omit the subscript $\mathcal{Z}$ to simplify expressions.
\end{notation}

\subsection{Measurable Families of Separable Hilbert Spaces}

To start we fix some conventions. Through out Section \ref{app direct
integrals}, $\mathcal{H}$ (with decoration or not) always stands for either
a separable (complex) Hilbert space or a family of such separable spaces.
Recall that the $\sigma $-algebra generated by the norm topology coincides
with the one generated by the weak topology of the Hilbert space $\mathcal{H}
$. It is denoted here by $\mathfrak{F}_{\mathcal{H}}$. We say that a mapping
from a measurable space $(\mathcal{Z},\mathfrak{F})$ to $\mathcal{H}$ is 
\emph{measurable}\ if it is $(\mathfrak{F}_{\mathcal{H}},\mathfrak{F})$%
-measurable. Similarly, the weak-operator, $\sigma $-weak, strong (operator)
and $\sigma $-strong topologies of the space $\mathcal{B}(\mathcal{H})$\ of
bounded (linear) operators acting on $\mathcal{H}$ all generate the same
(Borel) $\sigma $-algebra, which we denote by $\mathfrak{F}_{\mathcal{B}(%
\mathcal{H})}$. Note, however, that the (Borel) $\sigma $-algebra generated
by the norm topology of $\mathcal{B}(\mathcal{H})$ is generally strictly
bigger than $\mathfrak{F}_{\mathcal{B}(\mathcal{H})}$. Again, we say that a
mapping from $(\mathcal{Z},\mathfrak{F})$ to $\mathcal{B}(\mathcal{H})$\ is 
\emph{measurable}\ if it is $(\mathfrak{F}_{\mathcal{B}(\mathcal{H})},%
\mathfrak{F})$-measurable. $\mathbb{M}(\mathcal{Z};\mathcal{H})$ and $%
\mathbb{M}(\mathcal{Z};\mathcal{B}(\mathcal{H}))$ denote the spaces of
measurable mappings from $(\mathcal{Z},\mathfrak{F})$ to $\mathcal{H}$ and $%
\mathcal{B}(\mathcal{H})$, respectively. The following two lemmata are
useful characterizations of $(\mathfrak{F}_{\mathcal{B}(\mathcal{H})},%
\mathfrak{F})$- and $(\mathfrak{F}_{\mathcal{H}},\mathfrak{F})$-mappings:

\begin{lemma}[Characterization of measurable mappings]
\label{lema equiv meas total family}\mbox{ }\newline
Let $(\mathcal{Z},\mathfrak{F})$ be a measurable space and $\mathfrak{T}%
\subseteq \mathcal{H}$ any total\footnote{$\mathcal{H}$ is the norm closure
of the linear hull of $\mathfrak{T}$.} family in a separable Hilbert space $%
\mathcal{H}$.\newline
\emph{(i)} Any mapping $\varkappa :\mathcal{Z}\rightarrow \mathcal{H}$ is
measurable iff the mapping $z\mapsto \left\langle v,\varkappa
(z)\right\rangle _{\mathcal{H}}$ from $\mathcal{Z}$ to $\mathbb{C}$ is
measurable for any $v\in \mathfrak{T}$. In this case, the mapping $z\mapsto
\left\Vert \varkappa (z)\right\Vert _{\mathcal{H}}$ from $\mathcal{Z}$ to $%
\mathbb{R}$ is also measurable.\newline
\emph{(ii)} Any mapping $\varkappa :\mathcal{Z}\rightarrow \mathcal{B}(%
\mathcal{H})$ is measurable iff the mapping $z\mapsto \left\langle
v,\varkappa (z)w\right\rangle _{\mathcal{H}}$ from $\mathcal{Z}$ to $\mathbb{%
C}$ is measurable for any $v,w\in \mathfrak{T}$. In this case, the mapping $%
z\mapsto \left\Vert \varkappa (z)\right\Vert _{\mathcal{B}(\mathcal{H})}$
from $\mathcal{Z}$ to $\mathbb{R}$ is also measurable.
\end{lemma}

\begin{proof}
(i) refers to the fact that the measurability corresponds to the weak
topology of the Hilbert space $\mathcal{H}$. Note that the mapping 
\begin{equation*}
z\mapsto \left\Vert \varkappa (z)\right\Vert _{\mathcal{H}}=\sup_{v\in 
\mathfrak{T}:\left\Vert v\right\Vert _{\mathcal{H}}=1}\left\langle
v,\varkappa (z)\right\rangle _{\mathcal{H}}
\end{equation*}%
from $\mathcal{Z}$ to $\mathbb{R}$ is also measurable, since $\mathcal{H}$
is separable and the supremum of a sequence of measurable functions is
measurable. Similar arguments imply Assertion (ii). We omit the details.
\end{proof}

\begin{lemma}[$\mathbb{M}(\mathcal{Z};\mathcal{H})$ as a $\mathbb{M}(%
\mathcal{Z};\mathcal{B}(\mathcal{H}))$-module]
\label{lemma left module}\mbox{ }\newline
$\mathbb{M}(\mathcal{Z};\mathcal{B}(\mathcal{H}))$ is a $\ast $-algebra and $%
\mathbb{M}(\mathcal{Z};\mathcal{H})$ is a left $\mathbb{M}(\mathcal{Z};%
\mathcal{B}(\mathcal{H}))$-module with respect to point-wise operations. In
particular, for all $A,B\in \mathbb{M}(\mathcal{Z};\mathcal{B}(\mathcal{H}))$
and all $\varphi \in \mathbb{M}(\mathcal{Z};\mathcal{H})$, $A\cdot B,A^{\ast
}\in \mathbb{M}(\mathcal{Z};\mathcal{B}(\mathcal{H}))$\ and $A\varphi \in 
\mathbb{M}(\mathcal{Z};\mathcal{H})$, where $A\cdot B(z)\doteq A(z)\cdot
B(z) $, $A^{\ast }(z)\doteq A(z)^{\ast }$ and $A\varphi (z)\doteq
A(z)(\varphi (z))$ for $z\in \mathcal{Z}$.
\end{lemma}

\begin{proof}
As explained in \cite[Chap. 2]{Niesen-direct integrals}, the fact that $%
\mathbb{M}(\mathcal{Z};\mathcal{B}(\mathcal{H}))$ is a $\ast $-algebra with
respect to point-wise operations directly follows from Lemma \ref{lema equiv
meas total family} (ii), since the involution on $\mathcal{B}(\mathcal{H})$
is continuous in the weak-operator topology and the multiplication on $%
\mathcal{B}(\mathcal{H})\times \mathcal{B}(\mathcal{H})$ is jointly strongly
continuous on bounded sets. The fact that $\mathbb{M}(\mathcal{Z};\mathcal{H}%
)$ is a left $\mathbb{M}(\mathcal{Z};\mathcal{B}(\mathcal{H}))$-module is
also clear.
\end{proof}

It is important at this point to introduce the notion of measurable family\
of Hilbert spaces as defined in \cite[Chap. 2]{Niesen-direct integrals}.

\begin{definition}[Measurable families of separable Hilbert spaces]
\label{def measurable family hilbert space}\mbox{ }\newline
Let $(\mathcal{Z},\mathfrak{F})$ be a measurable space. Then, $\mathcal{H}_{%
\mathcal{Z}}\doteq (\mathcal{H}_{z})_{z\in \mathcal{Z}}$ is said to be a
measurable family\ of Hilbert spaces if each subset 
\begin{equation}
\mathcal{Z}_{n}\doteq \{z\in \mathcal{Z}:\dim \mathcal{H}_{z}=n\}\subseteq 
\mathcal{Z}\text{ },\text{\qquad }n\in \mathbb{N}_{0}\text{ },  \label{Zn}
\end{equation}%
is measurable, i.e., $\mathcal{Z}_{n}\in \mathfrak{F}$ for all $n\in \mathbb{%
N}_{0}$.
\end{definition}

\noindent Notice that, if $\mathcal{Z}_{n}$ is measurable for all $n\in 
\mathbb{N}_{0}$, then $\mathcal{Z}_{\infty }$ is also measurable\footnote{%
This results from the elementary facts that the complement of any measurable
set is measurable while countable unions of measurable sets are measurable.}.

Observe that, with point-wise operations, $\prod_{z\in \mathcal{Z}}{\mathcal{%
H}_{z}}$ and $\prod_{z\in \mathcal{Z}}{\mathcal{B}(\mathcal{H}_{z})}$ are a
vector space and a $\ast $-algebra, respectively. Elements of the vector
space are named \emph{vector\ fields} over the family $\mathcal{H}_{\mathcal{%
Z}}$ of separable Hilbert spaces. Similarly, elements of the above $\ast $%
-algebra are \emph{operator\ fields} over $\mathcal{H}_{\mathcal{Z}}$\emph{.}

In the mathematical literature, measurable families\ of separable Hilbert
spaces are often defined by the existence of a sequence of vector fields
which is fiberwise total (or even dense) and whose fiberwise scalar products
are measurable. See, e.g. \cite[Definition 4.4.1B]{BrattelliRobinsonI}. By 
\cite[Proposition 8.1]{Niesen-direct integrals}, this alternative definition
is equivalent to Definition \ref{def measurable family hilbert space}:

\begin{theorem}[Measurable families of separable Hilbert spaces - Equivalent
formulation]
\label{Measurable families of Hilbert spaces - Equivalent formulation}%
\mbox{
}\newline
Let $(\mathcal{Z},\mathfrak{F})$ be a measurable space and $\mathcal{H}_{%
\mathcal{Z}}$ a family of separable Hilbert spaces. Then, $\mathcal{H}_{%
\mathcal{Z}}$ is measurable iff there is a sequence $(v^{(n)})_{n\in \mathbb{%
N}}$ of vector fields over $\mathcal{H}_{\mathcal{Z}}$ such that:\newline
\emph{(a)} For all $m,n\in \mathbb{N}$, the mapping $z\mapsto \langle
v_{z}^{(n)},v_{z}^{(m)}\rangle _{{\mathcal{H}_{z}}}$ from $\mathcal{Z}$ to $%
\mathbb{C}$ is measurable.\newline
\emph{(b)} For each $z\in \mathcal{Z}$, the subset $\{v_{z}^{(n)}\}_{n\in 
\mathbb{N}}$ is total in $\mathcal{H}_{z}$.
\end{theorem}

\begin{proof}
Let $(\mathcal{Z},\mathfrak{F})$ be a measurable space and $\mathcal{H}_{%
\mathcal{Z}}$ a family of separable Hilbert spaces. In order to simplify the
discussion of the proof, we assume, without loss of generality, that $%
\mathcal{Z}_{0}=\emptyset $.

Let $\mathcal{H}_{\mathcal{Z}}$ be measurable and, for all $z\in \mathcal{Z}$%
, $(e_{n}^{z})_{n=1}^{\dim \mathcal{H}_{z}}$ any orthonormal basis of the
separable Hilbert space $\mathcal{H}_{z}$. Then define, for all $n\in 
\mathbb{N}$, $x_{z}^{(n)}\doteq e_{n}^{z}$ if $n\leq \dim \mathcal{H}_{z}$,
and $x_{z}^{(n)}\doteq 0$ otherwise. With this definition, (b) with $%
v_{z}^{(n)}=x_{z}^{(n)}$ holds true, by construction. In order to prove (a),
observe that, if $m\neq n$ then the mapping $z\mapsto \langle
x_{z}^{(n)},x_{z}^{(m)}\rangle _{{\mathcal{H}_{z}}}$ from $\mathcal{Z}$ to $%
\mathbb{C}$ is trivially measurable, for it is the zero function.
Additionally, for all $n\in \mathbb{N}$, the mapping $z\mapsto \langle
x_{z}^{(n)},x_{z}^{(n)}\rangle _{{\mathcal{H}_{z}}}$ from $\mathcal{Z}$ to $%
\mathbb{C}$ is the characteristic function of the set $\{z\in \mathcal{Z}$ $%
: $ $\dim \mathcal{H}_{z}\geq n\}$ and is, hence, measurable, as $\mathcal{H}%
_{\mathcal{Z}}$ is a measurable family of Hilbert spaces. See Definition \ref%
{def measurable family hilbert space}. Thus, there exists a sequence $%
(x^{(n)})_{n\in \mathbb{N}}$ of vector fields over $\mathcal{H}_{\mathcal{Z}%
} $ satisfying the following properties: \bigskip

\noindent (\~{a}) For all $m,n\in \mathbb{N}$, the mapping $z\mapsto \langle
x_{z}^{(n)},x_{z}^{(m)}\rangle _{{\mathcal{H}_{z}}}$ from $\mathcal{Z}$ to $%
\mathbb{C}$ is measurable.

\noindent (\~{b}) For each $z\in \mathcal{Z}$, $(x_{z}^{(n)})_{n=1}^{\dim 
\mathcal{H}_{z}}$ is an orthonormal basis of $\mathcal{H}_{z}$ and $%
x_{z}^{(n)}=0$ whenever $n>\dim \mathcal{H}_{z}$.\bigskip

Conversely, assume the existence of a sequence $(v^{(n)})_{n\in \mathbb{N}}$
of vector fields over $\mathcal{H}_{\mathcal{Z}}$ satisfying Conditions
(a)-(b). Apply next the Gram-Schmidt orthonormalization process to the total
sequence $(v_{z}^{(n)})_{n\in \mathbb{N}}$ for each $z\in \mathcal{Z}$: $%
x_{z}^{(1)}\doteq \lbrack v_{z}^{(1)}]$ and, for all $n\in \mathbb{N}$, $n>1$%
, 
\begin{equation*}
x_{z}^{(n)}\doteq \left\lceil v_{z}^{(n)}-\left\langle
x_{z}^{(n-1)},v_{z}^{(n)}\right\rangle _{\mathcal{H}_{z}}x_{z}^{(n-1)}-%
\cdots -\left\langle x_{z}^{(1)},v_{z}^{(n)}\right\rangle _{\mathcal{H}%
_{z}}x_{z}^{(1)}\right\rceil \text{ },
\end{equation*}%
where, for all $w\in \mathcal{H}_{z}$, $\lceil w\rceil \doteq 0$ whenever $%
w=0$ and $\lceil w\rceil \doteq \Vert w\Vert _{\mathcal{H}_{z}}^{-1}w$
otherwise. The new sequence $(x^{(n)})_{n\in \mathbb{N}}$ of vector fields
over $\mathcal{H}_{\mathcal{Z}}$ satisfies, by construction, Conditions (%
\~{a})-(\~{b}). In this case, by Lemma \ref{lema equiv meas total family}
(i), the mapping $z\mapsto \Vert x_{z}^{(n)}\Vert _{\mathcal{H}_{z}}$ from $%
\mathcal{Z}$ to $\mathbb{R}$ is measurable, while 
\begin{equation*}
\mathcal{Z}_{n}\doteq \left\{ z\in \mathcal{Z}:\dim \mathcal{H}%
_{z}=n\right\} =\left\{ z\in \mathcal{Z}:\sum_{m\in \mathbb{N}}\Vert
x_{z}^{(m)}\Vert _{\mathcal{H}_{z}}=n\right\} \text{ },\text{\qquad }n\in 
\mathbb{N}\text{ }.
\end{equation*}%
It follows that $\mathcal{H}_{\mathcal{Z}}$ is measurable, by Definition \ref%
{def measurable family hilbert space}.
\end{proof}

\subsection{Coherences and Measurable Fields}

We now introduce the notion of coherences and measurable fields. To this
end, we denote by 
\begin{equation*}
\ell _{\infty }^{2}\doteq \left\{ (x_{k})_{k\in \mathbb{N}}\subseteq \mathbb{%
C}:\sum_{k\in \mathbb{N}}\left\vert x_{k}\right\vert ^{2}<\infty \right\}
\end{equation*}
the Hilbert space of all square-summable sequences of complex numbers. For
each integer $n\in \mathbb{N}$, let $\ell _{n}^{2}\varsubsetneq \ell
_{\infty }^{2}$ be the subspace of sequences $(x_{k})_{k\in \mathbb{N}}\in
\ell _{\infty }^{2}$ such that $x_{k}=0$ for all $k>n$.

\begin{definition}[Coherences for families of separable Hilbert spaces]
\label{Coherences}\mbox{ }\newline
Let $\mathcal{H}_{\mathcal{Z}}\doteq (\mathcal{H}_{z})_{z\in \mathcal{Z}}$
be a family of separable Hilbert spaces. A family $\alpha _{\mathcal{Z}%
}=(\alpha _{z})_{z\in \mathcal{Z}}$ is a coherence\ for $\mathcal{H}_{%
\mathcal{Z}}$ if, for each $z\in \mathcal{Z}$, $\alpha _{z}$ is a linear
isometry\footnote{$\alpha _{z}$ is not defined as a linear isometry from $%
\mathcal{H}_{z}$ onto $\ell _{\dim \mathcal{H}_{z}}^{2}$ to avoid a $z$%
-dependent domain of its adjoint $\alpha _{z}^{\ast }$.} from $\mathcal{H}%
_{z}$ into $\ell _{\infty }^{2}$ with range $\ell _{\dim \mathcal{H}%
_{z}}^{2} $.
\end{definition}

\noindent Note that the separability of each $\mathcal{H}_{z}$ is a
necessary and sufficient condition for the existence of coherences. The
concept of coherence is often omitted in the literature on direct integrals,
but it is useful for it converts the case of fiber-dependent Hilbert spaces
into the constant-fiber case.

We next introduce measurable fields with respect to some fixed coherence. To
this end, recall that elements of the vector space $\prod_{z\in \mathcal{Z}}{%
\mathcal{H}_{z}}$ are named vector\ fields over the family $\mathcal{H}_{%
\mathcal{Z}}$ of separable Hilbert spaces. Similarly, elements of the $\ast $%
-algebra $\prod_{z\in \mathcal{Z}}{\mathcal{B}(\mathcal{H}_{z})}$ are
operator\ fields over $\mathcal{H}_{\mathcal{Z}}$. See also Notation \ref%
{remark constant copy(1)}.

\begin{definition}[Measurability of fields and equivalence of coherences]
\label{def measurable fields}\mbox{ }\newline
Let $(\mathcal{Z},\mathfrak{F})$ be a measurable space and $\mathcal{H}_{%
\mathcal{Z}}$ a family of separable Hilbert spaces. \newline
\emph{(i)} A vector field $v$ (respectively operator field $A$) is called $%
\alpha _{\mathcal{Z}}$-measurable\ if the mapping $z\mapsto \alpha
_{z}\left( v_{z}\right) $ from $\mathcal{Z}$ to $\ell _{\infty }^{2}$
(respectively $z\mapsto \alpha _{z}A_{z}\alpha _{z}^{\ast }$ from $\mathcal{Z%
}$ to ${\mathcal{B}(\ell _{\infty }^{2})}$) is measurable. \newline
\emph{(ii)} Two coherences $\alpha _{\mathcal{Z}}$ and $\beta _{\mathcal{Z}}$
for $\mathcal{H}_{\mathcal{Z}}$\ are equivalent if, for all vector fields $v$%
\ over $\mathcal{H}_{\mathcal{Z}}$, $v$ is $\alpha _{\mathcal{Z}}$%
-measurable iff it is $\beta _{\mathcal{Z}}$-measurable.
\end{definition}

\noindent In this definition, $\ell _{\infty }^{2}$ and ${\mathcal{B}(\ell
_{\infty }^{2})}$ are seen as measurable spaces with respect to the $\sigma $%
-algebras $\mathfrak{F}_{\ell _{\infty }^{2}}$\ and $\mathfrak{F}_{{\mathcal{%
B}(\ell _{\infty }^{2})}}$, as defined above for any separable Hilbert space 
$\mathcal{H}$.

The $\alpha _{\mathcal{Z}}$-measurable vector (respectively operator) fields
over $\mathcal{H}_{\mathcal{Z}}$ form a subspace of $\prod_{z\in \mathcal{Z}}%
{\mathcal{H}_{z}}$ (respectively $\prod_{z\in \mathcal{Z}}{\mathcal{B}(%
\mathcal{H}_{z})}$). By Lemma \ref{lemma left module}, the $\alpha _{%
\mathcal{Z}}$-measurable operator fields over $\mathcal{H}_{\mathcal{Z}}$
even form a $\ast $-algebra. Moreover, by the same lemma, if $v$ is a $%
\alpha _{\mathcal{Z}}$-measurable vector field and $A$ is a $\alpha _{%
\mathcal{Z}}$-measurable operator field, then\ the mapping $z\mapsto
A_{z}v_{z}$ from $\mathcal{Z}$\ to $\prod_{z\in \mathcal{Z}}{\mathcal{H}}$
is again a $\alpha _{\mathcal{Z}}$-measurable vector field. By Lemma \ref%
{lema equiv meas total family}, if $v,w$ are $\alpha _{\mathcal{Z}}$%
-measurable vector fields and $A$ is a $\alpha _{\mathcal{Z}}$-measurable
operator field then $z\mapsto \langle v_{z},w_{z}\rangle _{{\mathcal{H}_{z}}%
} $ and $z\mapsto \lVert A_{z}\rVert _{{\mathcal{B}(\mathcal{H}_{z})}}$ from 
$\mathcal{Z}$ to $\mathbb{C}$ are measurable functions\footnote{%
In order to prove the measurability of $z\mapsto \langle v_{z},w_{z}\rangle
_{{\mathcal{H}_{z}}}$, one also uses an orthonormal basis of ${\ell _{\infty
}^{2}}$ together with the elementary facts that sums and products of
measurable functions are measurable and the point-wise limit of a sequence
of measurable complex-valued functions is also measurable.}.

Given a measurable space $(\mathcal{Z},\mathfrak{F})$ and a measurable
family $\mathcal{H}_{\mathcal{Z}}$ of separable Hilbert spaces, Theorem \ref%
{Measurable families of Hilbert spaces - Equivalent formulation} yields a
canonical procedure to construct coherences for $\mathcal{H}_{\mathcal{Z}}$:

\begin{theorem}[Coherences associated with sequences of fields]
\label{coherence from fields}\mbox{ }\newline
Let $(\mathcal{Z},\mathfrak{F})$ be a measurable space and $\mathcal{H}_{%
\mathcal{Z}}$ a family of separable Hilbert spaces. Take any sequence $%
(v^{(n)})_{n\in \mathbb{N}}$ of vector fields over $\mathcal{H}_{\mathcal{Z}%
} $ satisfying Conditions (a)-(b) of Theorem \ref{Measurable families of
Hilbert spaces - Equivalent formulation}. Then, one has: \newline
\emph{(i)} Up to an equivalence of coherences, there is a unique coherence $%
\alpha _{\mathcal{Z}}$ for $\mathcal{H}_{\mathcal{Z}}$\ such that $%
(v^{(n)})_{n\in \mathbb{N}}$ is a sequence of $\alpha _{\mathcal{Z}}$%
-measurable fields. \newline
\emph{(ii)} An arbitrary vector field $w$ over $\mathcal{H}_{\mathcal{Z}}$\
is $\alpha _{\mathcal{Z}}$-measurable iff, for all $n\in \mathbb{N}$, the
mapping $z\mapsto \langle w_{z},v_{z}^{(n)}\rangle _{\mathcal{H}_{z}}$ from $%
\mathcal{Z}$ to $\mathbb{C}$ is measurable. \newline
\emph{(iii)} Similar to (ii), if $A$ is an operator field over $\mathcal{H}_{%
\mathcal{Z}}$\ then it is $\alpha _{\mathcal{Z}}$-measurable iff, for all $%
n,m\in \mathbb{N}$, the mapping $z\mapsto \langle
v_{z}^{(n)},A_{z}v_{z}^{(m)}\rangle _{\mathcal{H}_{z}}$ from $\mathcal{Z}$
to $\mathbb{C}$ is measurable.
\end{theorem}

\begin{proof}
For any sequence $(v^{(n)})_{n\in \mathbb{N}}$ of vector fields over $%
\mathcal{H}_{\mathcal{Z}}$ satisfying Conditions (a)-(b) of Theorem \ref%
{Measurable families of Hilbert spaces - Equivalent formulation}, one uses
the Gram-Schmidt orthonormalization process to construct a sequence $%
(x^{(n)})_{n\in \mathbb{N}}$ of vector fields over $\mathcal{H}_{\mathcal{Z}%
} $ satisfying Conditions (\~{a})-(\~{b}), as explained in the proof of
Theorem \ref{Measurable families of Hilbert spaces - Equivalent formulation}%
. Let $\alpha _{\mathcal{Z}}$ be the unique coherence for $\mathcal{H}_{%
\mathcal{Z}}$ such that, for each $z\in \mathcal{Z}$ and $k\in \{1,2,\ldots
,\dim \mathcal{H}_{z}\}$, $\alpha _{z}x_{z}^{(n_{k})}$ $=e_{k}$ with $n_{k}$
being the $k$-th natural number satisfying $x_{z}^{(n_{k})}\neq 0$. Then, as
one can easily check from Lemma \ref{lema equiv meas total family}, $%
(v^{(n)})_{n\in \mathbb{N}}$ is a sequence of $\alpha _{\mathcal{Z}}$%
-measurable fields and the coherence $\alpha _{z}$ we build satisfies
(i)-(iii).
\end{proof}

Theorem \ref{coherence from fields} gives a useful characterization for the
set of $\alpha _{\mathcal{Z}}$-measurable fields of an implicitly defined
coherence $\alpha _{\mathcal{Z}}$ for a measurable family of separable
Hilbert spaces. Additionally, the proofs of Theorems \ref{Measurable
families of Hilbert spaces - Equivalent formulation} and \ref{coherence from
fields} give an explicit, very natural, construction of coherences.

\subsection{Direct Integrals of Measurable Families of Hilbert Spaces}

Let $(\mathcal{Z},\mathfrak{F},\mu )$ be a $\sigma $-finite measure space
and $\alpha _{\mathcal{Z}}$ a coherence for a measurable family $\mathcal{H}%
_{\mathcal{Z}}$\ of separable Hilbert spaces. Recall also Notation \ref%
{remark constant copy(1)}. Denote by 
\begin{equation*}
\mathcal{\tilde{H}}_{\mathcal{Z}}^{\oplus }\subseteq \prod_{z\in \mathcal{Z}}%
{\mathcal{H}_{z}}
\end{equation*}%
the subspace of $\alpha _{\mathcal{Z}}$-measurable vector fields $v$ over $%
\mathcal{H}_{\mathcal{Z}}$ for which the mapping $z\mapsto \lVert
v_{z}\rVert _{\mathcal{H}_{z}}$ from $\mathcal{Z}$ to $\mathbb{C}$ belongs
to the Hilbert space $L^{2}(\mathcal{Z},\mu )$ of complex-valued functions
that are square-integrable with respect to the $\sigma $-finite measure $\mu 
$. A semi-inner-product on this space is naturally defined by%
\begin{equation}
\langle v,w\rangle _{\mathcal{\tilde{H}}_{\mathcal{Z}}^{\oplus }}\doteq
\int_{\mathcal{Z}}{\langle v_{z},w_{z}\rangle }_{\mathcal{H}_{z}}\mu (%
\mathrm{d}z)\text{ },\qquad v,w\in \mathcal{\tilde{H}}_{\mathcal{Z}}^{\oplus
}\ .  \label{inner prof direct int}
\end{equation}%
Then, as is usual, we define the seminorm 
\begin{equation*}
\left\Vert v\right\Vert _{\mathcal{\tilde{H}}_{\mathcal{Z}}^{\oplus }}\doteq 
\sqrt{\langle v,v\rangle _{\mathcal{\tilde{H}}_{\mathcal{Z}}^{\oplus }}}=%
\sqrt{\int_{\mathcal{Z}}\lVert v_{z}\rVert _{\mathcal{H}_{z}}^{2}\mu (%
\mathrm{d}z)\text{ }},\qquad v\in \mathcal{\tilde{H}}_{\mathcal{Z}}^{\oplus }%
\text{ },
\end{equation*}%
and identify $v$ and $w$ whenever $\lVert v-w\rVert _{\mathcal{\tilde{H}}_{%
\mathcal{Z}}^{\oplus }}=0$ to get a Hilbert space:\ 

\begin{definition}[Direct integrals of Separable Hilbert spaces]
\label{Direct integrals of Hilbert spaces}\mbox{ }\newline
Let $(\mathcal{Z},\mathfrak{F},\mu )$ be a $\sigma $-finite measure space
and $\alpha _{\mathcal{Z}}$ a coherence for a measurable family $\mathcal{H}%
_{\mathcal{Z}}$\ of separable Hilbert spaces. The direct integral\ of $%
\mathcal{H}_{\mathcal{Z}}$\ with respect to $\mu $ and $\alpha _{\mathcal{Z}%
} $, denoted by 
\begin{equation*}
\mathcal{H}_{\mathcal{Z}}^{\oplus }\equiv \int_{\mathcal{Z}}^{\alpha _{%
\mathcal{Z}}}\mathcal{H}_{z}\mu (\mathrm{d}z)\ ,
\end{equation*}%
is the Hilbert space of equivalence classes of elements of $\mathcal{\tilde{H%
}}_{\mathcal{Z}}^{\oplus }$, with inner product defined from (\ref{inner
prof direct int}). $\mathcal{H}_{z}$, $z\in \mathcal{Z}$, are named fiber\
Hilbert spaces.
\end{definition}

\noindent There is a canonical mapping from $\mathcal{\tilde{H}}_{\mathcal{Z}%
}^{\oplus }$ to $\mathcal{H}_{\mathcal{Z}}^{\oplus }$ defined by 
\begin{equation}
v=(v_{z})_{z\in \mathcal{Z}}\mapsto \left[ v\right] \equiv \int_{\mathcal{Z}%
}^{\alpha _{\mathcal{Z}}}v_{z}\mu (\mathrm{d}z)\ .
\label{canonical isomorphism}
\end{equation}%
To simplify notation, we often implicitly omit the distinction between $v\in 
\mathcal{\tilde{H}}_{\mathcal{Z}}^{\oplus }$ and the equivalence class $%
[v]\in \mathcal{H}_{\mathcal{Z}}^{\oplus }$.

The existence of coherences to define the direct integrals is very useful
because it converts the study of non-constant fiber Hilbert spaces into the
analysis of constant ones, in a natural way:

\begin{lemma}[Conversion into direct integrals of constant fiber Hilbert
spaces]
\label{lemma-useful-coherences}\mbox{ }\newline
Let $(\mathcal{Z},\mathfrak{F},\mu )$ be a $\sigma $-finite measure space
and $\alpha _{\mathcal{Z}}$ a coherence for a measurable family $\mathcal{H}%
_{\mathcal{Z}}$\ of separable Hilbert spaces. For any $n\in \mathbb{N}%
_{0}\cup \{\infty \}$, let $\mu _{n}$ be the restriction to the measurable
set $\mathcal{Z}_{n}$ of $\mu $, see (\ref{Zn}) by including the case $%
n=\infty $. Then, the mapping 
\begin{equation*}
\Upsilon _{\alpha _{\mathcal{Z}}}:\mathcal{H}_{\mathcal{Z}}^{\oplus }\equiv
\int_{\mathcal{Z}}^{\alpha _{\mathcal{Z}}}\mathcal{H}_{z}\mu (\mathrm{d}%
z)\rightarrow \bigoplus\limits_{n\in \mathbb{N}_{0}\cup \left\{ \infty
\right\} }\int_{\mathcal{Z}_{n}}{\ell _{n}^{2}\ }\mu _{n}(\mathrm{d}z)\ ,
\end{equation*}%
defined, for all $v\in \mathcal{H}_{\mathcal{Z}}^{\oplus }$, by%
\begin{equation}
\Upsilon _{\alpha _{\mathcal{Z}}}\int_{\mathcal{Z}}^{\alpha _{\mathcal{Z}%
}}v_{z}\mu (\mathrm{d}z)=\sum\limits_{n\in \mathbb{N}_{0}\cup \left\{ \infty
\right\} }\int_{\mathcal{Z}_{n}}\alpha _{z}\left( v_{z}\right) {\ }\mu _{n}(%
\mathrm{d}z)  \label{defined toto}
\end{equation}%
is a unitary mapping.
\end{lemma}

\begin{proof}
Fix all parameters of the lemma. For all $v\in \mathcal{H}_{\mathcal{Z}%
}^{\oplus }$, observe that 
\begin{equation*}
\left\Vert v\right\Vert _{\mathcal{H}_{\mathcal{Z}}^{\oplus }}^{2}\doteq
\int_{\mathcal{Z}}^{\alpha _{\mathcal{Z}}}\left\Vert v_{z}\right\Vert _{%
\mathcal{H}_{z}}^{2}\mu (\mathrm{d}z)=\sum\limits_{n\in \mathbb{N}_{0}\cup
\left\{ \infty \right\} }\int_{\mathcal{Z}_{n}}\left\Vert \alpha _{z}\left(
v_{z}\right) \right\Vert _{{\ell _{n}^{2}}}^{2}{\ }\mu _{n}(\mathrm{d}z)\ ,
\end{equation*}%
using Definition \ref{Coherences} and Lebesgue's monotone convergence
theorem. We thus deduce that $\Upsilon _{\alpha _{\mathcal{Z}}}$, as defined
by (\ref{defined toto}), is a linear isometry. Additionally, any element 
\begin{equation*}
\left[ w\right] \in \bigoplus\limits_{n\in \mathbb{N}_{0}\cup \left\{ \infty
\right\} }\int_{\mathcal{Z}_{n}}{\ell _{n}^{2}\ }\mu _{n}(\mathrm{d}z)
\end{equation*}%
is, by definition, a sequence $(\left[ w_{n}\right] )_{n\in \mathbb{N}%
_{0}\cup \left\{ \infty \right\} }$ with $\left[ w_{n}\right] \in \int_{%
\mathcal{Z}_{n}}{\ell _{n}^{2}\ }\mu _{n}(\mathrm{d}z)$ for any $n\in 
\mathbb{N}_{0}\cup \left\{ \infty \right\} $. Then, for any $n\in \mathbb{N}%
_{0}\cup \left\{ \infty \right\} $ and any representative $\left(
w_{n,z}\right) _{z\in \mathcal{Z}_{n}}$ of the equivalence class $\left[
w_{n}\right] $, define 
\begin{equation*}
v_{z}\doteq \alpha _{z}^{\ast }w_{n,z}\ ,\qquad z\in \mathcal{Z}_{n},\ n\in 
\mathbb{N}_{0}\cup \left\{ \infty \right\} ,
\end{equation*}%
and observe that the constructed element $\left( v_{z}\right) _{z\in 
\mathcal{Z}}\in \mathcal{\tilde{H}}_{\mathcal{Z}}^{\oplus }$ leads to a
unique equivalence class 
\begin{equation*}
\left[ v\right] \equiv \int_{\mathcal{Z}}^{\alpha _{\mathcal{Z}}}v_{z}\mu (%
\mathrm{d}z)\in \mathcal{H}_{\mathcal{Z}}^{\oplus }\ ,
\end{equation*}%
which satisfies $\Upsilon _{\alpha _{\mathcal{Z}}}\left[ v\right] =\left[ w%
\right] $. Therefore, $\Upsilon _{\alpha _{\mathcal{Z}}}$ is a surjective
linear isometry between two Hilbert spaces, and thus a unitary mapping.
\end{proof}

All the study of the general theory of direct integrals can be based on the
well-known theory of constant fiber direct integrals. Additionally, by Lemma %
\ref{lemma-useful-coherences}, two equivalent coherences $\alpha _{\mathcal{Z%
}}$ and $\beta _{\mathcal{Z}}$ for $\mathcal{H}_{\mathcal{Z}}$\ (Definition %
\ref{def measurable fields} (ii)) clearly imply the same direct integral:\ 
\begin{equation*}
\mathcal{H}_{\mathcal{Z}}^{\oplus }\equiv \int_{\mathcal{Z}}^{\alpha _{%
\mathcal{Z}}}\mathcal{H}_{z}\mu (\mathrm{d}z)=\int_{\mathcal{Z}}^{\beta _{%
\mathcal{Z}}}\mathcal{H}_{z}\mu (\mathrm{d}z)\text{ }.
\end{equation*}%
Note also that, similar to Theorem \ref{Measurable families of Hilbert
spaces - Equivalent formulation}, there is at least one sequence of vector
fields over $\mathcal{H}_{\mathcal{Z}}^{\oplus }$ which is a total family in
each fiber.

As is usual, 
\begin{equation*}
L^{\infty }\left( \mathcal{Z},\mu \right) \equiv L^{\infty }\left( \mathcal{Z%
},\mu ;\mathbb{C}\right)
\end{equation*}%
is the $C^{\ast }$-algebra of (equivalence classes of almost everywhere
equal) measurable complex-valued functions on $\mathcal{Z}$ with 
\begin{equation*}
\left\Vert f\right\Vert _{L^{\infty }\left( \mathcal{Z},\mu \right) }\equiv
\left\Vert f\right\Vert _{\infty }\doteq \underset{z\in \mathcal{Z}}{\mathrm{%
ess}\sup }\left\vert f\left( z\right) \right\vert <\infty
\end{equation*}%
being the essential supremum of $f$ associated with the ($\sigma $-finite)
measure space $(\mathcal{Z},\mathfrak{F},\mu )$. As a Banach space, it is
the (topological) dual space $L^{1}\left( \mathcal{Z},\mu \right) ^{\ast }$
of the Banach space 
\begin{equation*}
L^{1}\left( \mathcal{Z},\mu \right) \equiv L^{1}\left( \mathcal{Z},\mu ;%
\mathbb{C}\right)
\end{equation*}%
of (equivalence classes of) complex-valued functions on $\mathcal{Z}$ that
are absolutely integrable with respect to the ($\sigma $-finite) measure $%
\mu $. Any element of $L^{\infty }(\mathcal{Z},\mu )$ can also be seen%
\footnote{%
In fact, $L^{\infty }(\mathcal{Z},\mu )$ is $\ast $-isomorphic to an abelian
von Neumann algebra.} as a bounded operator acting, by the point-wise
multiplication, on the Hilbert space%
\begin{equation*}
L^{2}\left( \mathcal{Z},\mu \right) \equiv L^{2}\left( \mathcal{Z},\mu ;%
\mathbb{C}\right)
\end{equation*}%
of (equivalence class of almost everywhere equal) complex-valued functions
on $\mathcal{Z}$ that are square-integrable with respect to the ($\sigma $%
-finite) measure $\mu $.

\begin{lemma}[Sequence of vector fields as fiberwise-total families]
\label{construction of coherences I copy(1)}\mbox{ }\newline
Let $(\mathcal{Z},\mathfrak{F},\mu )$ be a $\sigma $-finite measure space
and $\alpha _{\mathcal{Z}}$ a coherence for a measurable family $\mathcal{H}%
_{\mathcal{Z}}$\ of separable Hilbert spaces. \newline
\emph{(i)} There exists a sequence $(v^{(n)})_{n\in \mathbb{N}}$ in $%
\mathcal{H}_{\mathcal{Z}}^{\oplus }$ such that $\{v_{z}^{(n)}\}_{n\in 
\mathbb{N}}$ is total in $\mathcal{H}_{z}$ for each $z\in \mathcal{Z}$.%
\newline
\emph{(ii)} Let $\{\varphi ^{(i)}\}_{i\in I}\subseteq L^{\infty }(\mathcal{Z}%
,\mu )\equiv L^{1}\left( \mathcal{Z},\mu \right) ^{\ast }$ be a weak$^{\ast
} $-total family and $(v^{(n)})_{n\in \mathbb{N}}$ a sequence like in (i).
Then the family of $\alpha _{\mathcal{Z}}$-measurable fields $\{\varphi
^{(i)}v^{(n)}\}_{(i,n)\in I\times \mathbb{N}}$ is total in $\mathcal{H}_{%
\mathcal{Z}}^{\oplus }$. If, for all $z\in \mathcal{Z}$, $%
\{v_{z}^{(n)}\}_{n\in \mathbb{N}}$ is dense then $\{\varphi
^{(i)}v^{(n)}\}_{(i,n)\in I\times \mathbb{N}}$ is dense in $\mathcal{H}_{%
\mathcal{Z}}^{\oplus }$.
\end{lemma}

\begin{proof}
Fix all the assumptions of the lemma. Recall that $(e_{n})_{n\in \mathbb{N}}$
is the canonical orthonormal basis of $\ell _{\infty }^{2}$. To prove (i),
let $(x^{(n)})_{n\in \mathbb{N}}$ be the $\alpha _{\mathcal{Z}}$-measurable
vector fields over $\mathcal{H}_{\mathcal{Z}}$ defined by $x_{z}^{(n)}\doteq
\alpha _{z}^{\ast }e_{n}$ for each $z\in \mathcal{Z}$ and $n\in \mathbb{N}$.
Since coherences are linear isometries from $\mathcal{H}_{z}$ into $\ell
_{\infty }^{2}$ with range $\ell _{\dim \mathcal{H}_{z}}^{2}$ (Definition %
\ref{Coherences}), the family $\{x_{z}^{(n)}\}_{n\in \mathbb{N}}$ is total
in $\mathcal{H}_{z}$ for each $z\in \mathcal{Z}$, but $x^{(n)}\doteq
(x_{z}^{(n)})_{z\in \mathcal{Z}}$, $n\in \mathbb{N}$, are not necessarily
elements of $\mathcal{\tilde{H}}_{\mathcal{Z}}^{\oplus }$, for they are
possibly non-square-integrable. Since $\mu $ is, by assumption, a $\sigma $%
-finite measure, there is a strictly positive measurable function $f$ such
that 
\begin{equation*}
\int_{\mathcal{Z}}{f}\left( {z}\right) \mu (\mathrm{d}z)<\infty \ .
\end{equation*}%
Then, define $v_{z}^{(n)}=\sqrt{{f}\left( {z}\right) }x_{z}^{(n)}$ for each $%
z\in \mathcal{Z}$ and $n\in \mathbb{N}$ to arrive at Assertion (i).

In order to get Assertion (ii), it suffices to prove that any element $w\in 
\mathcal{H}_{\mathcal{Z}}^{\oplus }$ that is orthogonal to $\varphi
^{(i)}v^{(n)}$ for all $(i,n)\in I\times \mathbb{N}$ must be zero. This is
straightforward. See \cite[Lemma 7.3]{Niesen-direct integrals} for more
details.
\end{proof}

The relation between different direct integrals relative to absolutely
continuous measures is also very natural:

\begin{lemma}[Sequence of vector fields as a total family in each fiber]
\label{construction of coherences I copy(2)}\mbox{ }\newline
Let $(\mathcal{Z},\mathfrak{F},\mu )$ be a $\sigma $-finite measure space
and $\alpha _{\mathcal{Z}}$ a coherence for a measurable family $\mathcal{H}%
_{\mathcal{Z}}$\ of separable Hilbert spaces. If $\tilde{\mu}$ is a $\sigma $%
-finite measure that is absolutely continuous with respect to $\mu $, then 
\begin{equation*}
{v_{z}\mapsto }\left( \frac{\mathrm{d}\tilde{\mu}}{\mathrm{d}\mu }(z)\right)
^{1/2}{v_{z}}
\end{equation*}%
is a linear isometry from $\int_{\mathcal{Z}}^{\alpha _{\mathcal{Z}}}%
\mathcal{H}_{z}\tilde{\mu}(\mathrm{d}z)$ to $\int_{\mathcal{Z}}^{\alpha _{%
\mathcal{Z}}}\mathcal{H}_{z}\mu (\mathrm{d}z)$.
\end{lemma}

\begin{proof}
The proof follows from direct computations. Note that the existence of the
measurable function $\mathrm{d}\tilde{\mu}/\mathrm{d}\mu $ in $L^{1}(%
\mathcal{Z},\mu )$, which defines the linear isometry, is a direct
consequence of the Radon-Nikodym theorem.
\end{proof}

Note finally that the direct integral $\mathcal{H}_{\mathcal{Z}}^{\oplus }$
of Hilbert spaces, as defined above, is not necessarily separable. This very
important property of a Hilbert space holds true when the measure space $(%
\mathcal{Z},\mathfrak{F},\mu )$ is \emph{standard}, in the following sense:

\begin{definition}[Standard measure spaces]
\label{Standard measure spaces}\mbox{ }\newline
The measurable space $(\mathcal{Z},\mathfrak{F})$ is standard\ if $\mathfrak{%
F}$ is the Borel $\sigma $-algebra of a polish space\footnote{%
There is a metric $\mathfrak{d}$ on $\mathcal{Z}$ such that $(\mathcal{Z},%
\mathfrak{d})$ is a separable and complete metric space and $\mathfrak{F}$\
is the Borel $\sigma $--algebra associated with $\mathfrak{d}$.}. The
measure space $(\mathcal{Z},\mathfrak{F},\mu )$ is standard\ if it is $%
\sigma $-finite and $(\mathcal{Z},\mathfrak{F})$ is standard as a measurable
space.
\end{definition}

\noindent Standard measure spaces lead to the separability of direct
integrals of families of separable Hilbert spaces:

\begin{theorem}[Separability of direct integrals]
\label{Nielsen Th 7.1 copy(1)}\mbox{ }\newline
Let $(\mathcal{Z},\mathfrak{F},\mu )$ be a standard measure space and $%
\alpha _{\mathcal{Z}}$ a coherence for a measurable family $\mathcal{H}_{%
\mathcal{Z}}$\ of separable Hilbert spaces. Then, $\mathcal{H}_{\mathcal{Z}%
}^{\oplus }\equiv \int_{\mathcal{Z}}^{\alpha _{\mathcal{Z}}}\mathcal{H}%
_{z}\mu (\mathrm{d}z)$ is separable.
\end{theorem}

\begin{proof}
By Lemma \ref{lemma-useful-coherences}, we can assume without loss of
generality that $\mathcal{H}_{\mathcal{Z}}$ is a family of constant
separable Hilbert spaces, i.e., $\mathcal{H}_{z}=\mathcal{H}$ is fixed for
all $z\in \mathcal{Z}$. In this special case, we can apply \cite[Proposition
5.2]{Niesen-direct integrals} which states the existence of a unitary
operator mapping $\mathcal{H}_{\mathcal{Z}}^{\oplus }$ onto $L^{2}(\mathcal{Z%
},\mu )\otimes \mathcal{H}$. It is well-known that $L^{2}(\mathcal{Z},\mu )$
is separable when $\mu $ is a standard measure, see, e.g., \cite[Corollary
5.3]{Niesen-direct integrals}. Since $\mathcal{H}$ is, by assumption, also a
separable Hilbert space, the assertion follows.
\end{proof}

\subsection{Decomposable Operators}

Let $(\mathcal{Z},\mathfrak{F},\mu )$ still denote a $\sigma $-finite
measure space and $\alpha _{\mathcal{Z}}$ a coherence for a measurable
family $\mathcal{H}_{\mathcal{Z}}$\ of separable Hilbert spaces. See also
Notation \ref{remark constant copy(1)}. Let $A=(A_{z})_{z\in \mathcal{Z}}$
be an $\alpha _{\mathcal{Z}}$-measurable operator field over $\mathcal{H}_{%
\mathcal{Z}}$. If $A$ is $\mu $-essentially bounded, i.e., the mapping $%
z\mapsto \lVert A_{z}\rVert _{\mathcal{B}(\mathcal{H}_{z})}$ from $\mathcal{Z%
}$ to $\mathbb{C}$\ belongs to $L^{\infty }(\mathcal{Z},\mu )$, then there
is a unique bounded operator acting on $\mathcal{H}_{\mathcal{Z}}^{\oplus }$%
, denoted by 
\begin{equation*}
\int_{\mathcal{Z}}^{\alpha _{\mathcal{Z}}}{A_{z}}\mu (\mathrm{d}z)\in 
\mathcal{B}\left( \mathcal{H}_{\mathcal{Z}}^{\oplus }\right) \ ,
\end{equation*}%
satisfying 
\begin{equation*}
\left( \int_{\mathcal{Z}}^{\alpha _{\mathcal{Z}}}A_{z}\mu (\mathrm{d}%
z)\right) v=\int_{\mathcal{Z}}^{\alpha _{\mathcal{Z}}}A_{z}v_{z}\text{ }\mu (%
\mathrm{d}z)\text{ },\qquad v\in \mathcal{H}_{\mathcal{Z}}^{\oplus }\ .
\end{equation*}%
(See also (\ref{canonical isomorphism})). Operators of this type refer to 
\emph{decomposable} operators:

\begin{definition}[Decomposable operators]
\label{Decomposable operators}\mbox{ }\newline
Let $(\mathcal{Z},\mathfrak{F},\mu )$ be a $\sigma $-finite measure space
and $\alpha _{\mathcal{Z}}$ a coherence for a measurable family $\mathcal{H}%
_{\mathcal{Z}}$\ of separable Hilbert spaces. $A\in \mathcal{B}(\mathcal{H}_{%
\mathcal{Z}}^{\oplus })$ is decomposable whenever there is a $\mu $%
-essentially bounded, $\alpha _{\mathcal{Z}}$-measurable operator field $%
(A_{z})_{z\in \mathcal{Z}}$ such that $A=\int_{\mathcal{Z}}^{\alpha _{%
\mathcal{Z}}}{A_{z}}\mu (\mathrm{d}z)$. We denote by $M_{\mathcal{Z}%
}\subseteq \mathcal{B}(\mathcal{H}_{\mathcal{Z}}^{\oplus })$ the subspace of
decomposable operators.
\end{definition}

\noindent Important decomposable operators are the so-called \emph{%
diagonalizable} ones:

\begin{definition}[Diagonalizable operators]
\label{Diagonalizable operators}\mbox{ }\newline
Let $(\mathcal{Z},\mathfrak{F},\mu )$ be a $\sigma $-finite measure space
and $\alpha _{\mathcal{Z}}$ a coherence for a measurable family $\mathcal{H}%
_{\mathcal{Z}}$\ of separable Hilbert spaces. $A\in \mathcal{B}(\mathcal{H}_{%
\mathcal{Z}}^{\oplus })$ is diagonalizable whenever there is $\varphi \in
L^{\infty }(\mathcal{Z},\mu )$ such that $A=\int_{\mathcal{Z}}^{\alpha _{%
\mathcal{Z}}}\varphi (z)\mathbf{1}_{\mathcal{H}_{z}}\mu (\mathrm{d}z)$. We
denote by $N_{\mathcal{Z}}\subseteq M_{\mathcal{Z}}$ the subspace of
diagonalizable operators.
\end{definition}

In order to explicitly characterize the subspaces $N_{\mathcal{Z}}$ and $M_{%
\mathcal{Z}}$ of operators, the existence of coherences to define the direct
integral is very useful, by Lemma \ref{lemma-useful-coherences}. First note
the following fact:

\begin{lemma}[Reduction to constant fiber Hilbert spaces]
\label{lemma-useful-coherences copy(1)}\mbox{ }\newline
Let $(\mathcal{Z},\mathfrak{F},\mu )$ be a $\sigma $-finite measure space
and $\alpha _{\mathcal{Z}}$ a coherence for a measurable family $\mathcal{H}%
_{\mathcal{Z}}$\ of separable Hilbert spaces. Then, 
\begin{equation*}
M_{\mathcal{Z}}=\Upsilon _{\alpha _{\mathcal{Z}}}^{\ast
}\bigoplus\limits_{n\in \mathbb{N}_{0}\cup \left\{ \infty \right\} }M_{%
\mathcal{Z}}^{(n)}\Upsilon _{\alpha _{\mathcal{Z}}}\qquad \text{and}\qquad
N_{\mathcal{Z}}=\Upsilon _{\alpha _{\mathcal{Z}}}^{\ast
}\bigoplus\limits_{n\in \mathbb{N}_{0}\cup \left\{ \infty \right\} }N_{%
\mathcal{Z}}^{(n)}\Upsilon _{\alpha _{\mathcal{Z}}}
\end{equation*}%
with $\Upsilon _{\alpha _{\mathcal{Z}}}$ being the unitary mapping of Lemma %
\ref{lemma-useful-coherences}, and where $M_{\mathcal{Z}}^{(n)}$ and $N_{%
\mathcal{Z}}^{(n)}$ are the subspaces of respectively decomposable and
diagonalizable operators acting on\ the constant fiber direct integral $%
\int_{\mathcal{Z}_{n}}{\ell _{n}^{2}\ }\mu (\mathrm{d}z)$ for each $n\in 
\mathbb{N}_{0}\cup \{\infty \}$.
\end{lemma}

\begin{proof}
The proof is very similar to the one of Lemma \ref{lemma-useful-coherences}
and we thus omit the details. See \cite[p. 25]{Niesen-direct integrals}.
\end{proof}

Therefore, all the study of decomposable and diagonalizable operators can be
based on the well-known theory of constant fiber direct integrals. We thus
obtain the following result:

\begin{theorem}[Structure of the subspace of diagonalizable operators]
\label{Nielsen Th 7.1 copy(2)}\mbox{ }\newline
Let $(\mathcal{Z},\mathfrak{F},\mu )$ be a $\sigma $-finite measure space
and $\alpha _{\mathcal{Z}}$ a coherence for a measurable family $\mathcal{H}%
_{\mathcal{Z}}$\ of separable Hilbert spaces. $N_{\mathcal{Z}}$ is a von
Neumann algebra on the Hilbert space $\mathcal{H}_{\mathcal{Z}}^{\oplus }$
and the mapping 
\begin{equation*}
\varphi \mapsto \int_{\mathcal{Z}}^{\alpha _{\mathcal{Z}}}\varphi (z)\mathbf{%
1}_{\mathcal{H}_{z}}\mu (\mathrm{d}z)
\end{equation*}%
defines a $\ast $-isomorphism from the abelian von Neumann algebra $%
L^{\infty }(\mathcal{Z},\mu )\subseteq \mathcal{B}(L^{2}(\mathcal{Z},\mu ))$
to $N_{\mathcal{Z}}$.
\end{theorem}

\begin{proof}
By Lemmata \ref{lemma-useful-coherences} and \ref{lemma-useful-coherences
copy(1)}, we can assume without loss of generality that $\mathcal{H}_{%
\mathcal{Z}}$ is a family of constant separable Hilbert spaces, i.e., $%
\mathcal{H}_{z}=\mathcal{H}$ is fixed for all $z\in \mathcal{Z}$. Similar to
the proof of Theorem \ref{Nielsen Th 7.1 copy(1)}, we apply again \cite[%
Proposition 5.2]{Niesen-direct integrals} which directly implies that the
set of diagonalizable operators acting on $\mathcal{H}_{\mathcal{Z}}^{\oplus
}\equiv $ $L^{2}(\mathcal{Z},\mu )\otimes \mathcal{H}$ is a von Neumann
algebra which is $\ast $-isomorphic to the abelian von Neumann algebra $%
L^{\infty }(\mathcal{Z},\mu )\subseteq \mathcal{B}(L^{2}(\mathcal{Z},\mu ))$.
\end{proof}

The previous statement on diagonalizable operators has the following
implication for the subspace $M_{\mathcal{Z}}$ of decomposable operators
(cf. \cite[Theorem 7.1 (iii)-(vii)]{Niesen-direct integrals}):

\begin{theorem}[Structure of the subspace of decomposable operators]
\label{Nielsen Th 7.1}\mbox{ }\newline
Let $(\mathcal{Z},\mathfrak{F},\mu )$ be a $\sigma $-finite measure space
and $\alpha _{\mathcal{Z}}$ a coherence for a measurable family $\mathcal{H}%
_{\mathcal{Z}}$\ of separable Hilbert spaces. \newline
\emph{(i)} $M_{\mathcal{Z}}$ is the commutant\footnote{%
\bigskip The commutant $\mathfrak{S}^{\prime }$ of a set $\mathfrak{S}%
\subseteq \mathcal{B}(\mathcal{H})$ ($\mathcal{H}$ being some Hilbert space)
is, by definition, the set of all elements of $\mathcal{B}(\mathcal{H})$
that commute with all $A\in \mathfrak{S}$.} of the abelian von Neumann
algebra $N_{\mathcal{Z}}$, i.e., $M_{\mathcal{Z}}=N_{\mathcal{Z}}^{\prime }$%
. In particular, $M_{\mathcal{Z}}$ is also a von Neumann algebra on the
Hilbert space $\mathcal{H}_{\mathcal{Z}}^{\oplus }$ and $M_{\mathcal{Z}%
}^{\prime }=N_{\mathcal{Z}}$. \emph{\newline
(ii)} $A\doteq (A_{z})_{z\in \mathcal{Z}}\mapsto \int_{\mathcal{Z}}^{\alpha
_{\mathcal{Z}}}A_{z}\mu (\mathrm{d}z)$ defines a $\ast $-homomorphism from
the $\ast $-algebra of the $\mu $-essentially bounded and $\alpha _{\mathcal{%
Z}}$-measurable operator fields over $\mathcal{H}_{\mathcal{Z}}$ to $M_{%
\mathcal{Z}}\subseteq \mathcal{B}(\mathcal{H}_{\mathcal{Z}}^{\oplus })$ and 
\begin{equation*}
\left\Vert \int_{\mathcal{Z}}^{\alpha _{\mathcal{Z}}}A_{z}\mu (\mathrm{d}%
z)\right\Vert _{\mathcal{B}(\mathcal{H}_{\mathcal{Z}}^{\oplus })}=\underset{%
z\in \mathcal{Z}}{\mathrm{ess}\sup }\left\Vert A_{z}\right\Vert _{\mathcal{B}%
(\mathcal{H}_{z})}\ .
\end{equation*}%
\emph{(iii)} Let $A,A_{n}$, $n\in \mathbb{N}$, be essentially bounded $%
\alpha _{\mathcal{Z}}$-measurable fields of operators over $\mathcal{H}_{%
\mathcal{Z}}$ such that 
\begin{equation*}
\underset{z\in \mathcal{Z}}{\mathrm{ess}\sup }\sup_{n\in \mathbb{N}%
}\left\Vert A_{n,z}\right\Vert _{\mathcal{B}(\mathcal{H}_{z})}<\infty \ .
\end{equation*}%
If $A_{n,z}$ converges in the strong operator topology of $\mathcal{B}(%
\mathcal{H}_{z})$ to $A_{z}$ $\mu $-almost everywhere in $\mathcal{Z}$ when $%
n\rightarrow \infty $, then $\int_{\mathcal{Z}}^{\alpha _{\mathcal{Z}%
}}A_{n,z}\mu (\mathrm{d}z)$ tends to $\int_{\mathcal{Z}}^{\alpha _{\mathcal{Z%
}}}A_{z}\mu (\mathrm{d}z)$ in the strong operator topology of $\mathcal{B}%
\left( \mathcal{H}_{\mathcal{Z}}^{\oplus }\right) $, as $n\rightarrow \infty 
$. \emph{\newline
(iv)} Conversely, if $\int_{\mathcal{Z}}^{\alpha _{\mathcal{Z}}}A_{n,z}\mu (%
\mathrm{d}z)$ tends to $\int_{\mathcal{Z}}^{\alpha _{\mathcal{Z}}}A_{z}\mu (%
\mathrm{d}z)$ in the strong operator topology when $n\rightarrow \infty $
then there is a subsequence $\{n_{k}\}_{k\in \mathbb{N}}$ such that $%
A_{n_{k},z}$ converges in the strong operator topology to $A_{z}$ $\mu $%
-almost everywhere in $\mathcal{Z}$, as $k\rightarrow \infty $.
\end{theorem}

\begin{proof}
(i) $M_{\mathcal{Z}}=N_{\mathcal{Z}}^{\prime }$ is proven in \cite[Theorem
6.2]{Niesen-direct integrals} for constant fiber direct integrals.
Therefore, by Lemma \ref{lemma-useful-coherences copy(1)} together with \cite%
[Theorem 6.2]{Niesen-direct integrals}, 
\begin{equation*}
N_{\mathcal{Z}}^{\prime }=\Upsilon _{\alpha _{\mathcal{Z}}}^{\ast
}\bigoplus\limits_{n\in \mathbb{N}_{0}\cup \left\{ \infty \right\} }\left(
N_{\mathcal{Z}}^{(n)}\right) ^{\prime }\Upsilon _{\alpha _{\mathcal{Z}}}=M_{%
\mathcal{Z}}
\end{equation*}%
which, combined with Theorem \ref{Nielsen Th 7.1 copy(2)}, implies Assertion
(i). Note that $M_{\mathcal{Z}}^{\prime }=N_{\mathcal{Z}}^{\prime \prime
}=N_{\mathcal{Z}}$ is a direct consequence of the celebrated bicommutant
theorem \cite[Theorem 2.4.11]{BrattelliRobinsonI}.

To prove Assertions (ii)-(iv) we can assume without loss of generality that $%
\mathcal{H}_{\mathcal{Z}}$ is a family of constant separable Hilbert spaces,
by Lemmata \ref{lemma-useful-coherences} and \ref{lemma-useful-coherences
copy(1)}. (ii) refers to \cite[Proposition 6.1 (b)]{Niesen-direct integrals}%
, which is straightforward to prove. (iii)-(iv) in the constant-fiber case
is \cite[Proposition 6.3]{Niesen-direct integrals}. In this case, one can
again use \cite[Proposition 5.2]{Niesen-direct integrals}, i.e., $\mathcal{H}%
_{\mathcal{Z}}^{\oplus }\equiv $ $L^{2}(\mathcal{Z},\mu )\otimes \mathcal{H}$%
. We omit the details.
\end{proof}

By Theorems \ref{Nielsen Th 7.1 copy(2)}-\ref{Nielsen Th 7.1}, if $(\mathcal{%
Z},\mathfrak{F},\mu )$ is standard\ then the space $N_{\mathcal{Z}}$ of
diagonalizable operators is an abelian von Neumann over a separable Hilbert
space. The following theorem says that this situation is universal, up to
spatial isomorphisms\footnote{%
Von Neumann algebras $\mathfrak{M}_{1},\mathfrak{M}_{2}$ over, respectively,
the Hilbert spaces $\mathcal{H}_{1}$,$\mathcal{H}_{2}$ are spatially
isomorphic iff there is a unitary map $\mathrm{U}$ from $\mathcal{H}_{1}$ to 
$\mathcal{H}_{2}$ such that $\mathfrak{M}_{2}=\mathrm{U}\mathfrak{M}_{1}%
\mathrm{U}^{\ast }$.}:

\begin{theorem}[Abelian von Neumann algebras as spaces of diagonalizable
operators]
\label{universality}\mbox{ }\newline
Assume that $\mathfrak{N}$ is an abelian von Neumann algebra on a separable
Hilbert space $\mathcal{H}$. Then, there is a standard measure space $(%
\mathcal{Z},\mathfrak{F},\mu )$, a $\ast $-isomorphism $\phi :L^{\infty }(%
\mathcal{Z},\mu )\rightarrow \mathfrak{N}$, a measurable family $\mathcal{H}%
_{\mathcal{Z}}$ of Hilbert spaces, a coherence $\alpha _{\mathcal{Z}}$ for $%
\mathcal{H}_{\mathcal{Z}}$ and a unitary mapping $\mathrm{U}$ from $\mathcal{%
H}$ to $\mathcal{H}_{\mathcal{Z}}^{\oplus }$ such that 
\begin{equation*}
N_{\mathcal{Z}}=\mathrm{U}\mathfrak{N}\mathrm{U}^{\ast }\qquad \text{and}%
\qquad \mathrm{U}\phi (f)\mathrm{U}^{\ast }=\int_{\mathcal{Z}}^{\alpha _{%
\mathcal{Z}}}f(z)\mathbf{1}_{\mathcal{H}_{z}}\mu (\mathrm{d}z)\text{ }%
,\qquad f\in L^{\infty }(\mathcal{Z},\mu )\ .
\end{equation*}
\end{theorem}

\begin{proof}
Any abelian von Neumann algebra on a separable Hilbert space $\mathcal{H}$
is $\ast $-isomorphic to $L^{\infty }(\mathcal{Z},\mu )$ for some standard
measure space $(\mathcal{Z},\mathfrak{F},\mu )$. This is a known result of
the theory of abelian von Neumann algebras. See, e.g., \cite[Proposition A.4]%
{Niesen-direct integrals}. The remaining part of the proof is not trivial
and requires some rather long arguments. We thus refer to \cite[Theorem 9.1]%
{Niesen-direct integrals} for a detailed proof.
\end{proof}

\noindent This result motivates the following definition:

\begin{definition}[Decomposition of Hilbert spaces via abelian von Neumann
algebras]
\label{Hilbert decomposition}\mbox{ }\newline
Let $\mathfrak{N}$ be an abelian von Neumann algebra on a separable Hilbert
space $\mathcal{H}$. We say that the direct integral Hilbert space $\mathcal{%
H}_{\mathcal{Z}}^{\oplus }\equiv \int_{\mathcal{Z}}^{\alpha _{\mathcal{Z}}}%
\mathcal{H}_{z}\mu (\mathrm{d}z)$ is a decomposition of $\mathcal{H}$\ with
respect to the abelian von Neumann algebra $\mathfrak{N}$ whenever $(%
\mathcal{Z},\mathfrak{F},\mu )$\ is a standard measure space, $\alpha _{%
\mathcal{Z}}$ is a coherence for a measurable family $\mathcal{H}_{\mathcal{Z%
}}$\ of separable Hilbert spaces, and $N_{\mathcal{Z}}$ is spatially
isomorphic to $\mathfrak{N}$.
\end{definition}

By Theorem \ref{universality}, a separable Hilbert space $\mathcal{H}$
admits a decomposition with respect to any abelian von Neumann algebra $%
\mathfrak{N}$ on $\mathcal{H}$. Moreover, the standard space $(\mathcal{Z},%
\mathfrak{F},\mu )$ and the fiber\ Hilbert spaces $\mathcal{H}_{z}$, $z\in 
\mathcal{Z}$, are unique up to certain natural equivalences. So, one can
speak about \emph{the}\ decomposition of $\mathcal{H}$ with respect to $%
\mathfrak{N}$. For more details, see \cite[Theorem 9.2]{Niesen-direct
integrals}. The decomposability of separable Hilbert spaces is pivotal in
the sequel and an analog property holds true for von Neumann algebras and
representations of $C^{\ast }$-algebras.

\subsection{Direct Integrals of Representations of Separable Unital Banach $%
\ast $-Algebras}

In this section, $\mathcal{X}$ denotes an arbitrary, separable, unital
Banach $\ast $-algebra, also named (separable, unital) involutive Banach
algebra. This means that $\mathcal{X}$ is a (complex) Banach algebra with a
unit $\mathfrak{1}$ and is endowed with an antilinear involution $A\mapsto
A^{\ast }$ from $\mathcal{X}$ to itself satisfying $(AB)^{\ast }=B^{\ast
}A^{\ast }$ for all $A,B\in \mathcal{X}$. Like Hilbert spaces under
consideration, here $\mathcal{X}$ is always assumed to be separable. The
main example we have in mind is $\mathcal{X}$ being a separable unital $%
C^{\ast }$-algebra, i.e., a (separable) Banach $\ast $-algebra such that $%
\Vert A^{\ast }A\Vert _{\mathcal{X}}=\Vert A\Vert _{\mathcal{X}}^{2}$ for $%
A\in \mathcal{X}$. Recall that a representation\footnote{%
A representation of a $\ast $-algebra $\mathcal{X}$ is also defined to be a
pair, in this case $(\mathcal{H},\pi )$. See \cite[Definition 2.3.2]%
{BrattelliRobinsonI}.} $\pi $ of $\mathcal{X}$ on a Hilbert space $\mathcal{H%
}$ is a $\ast $-homomorphism of $\mathcal{X}$ to $\mathcal{B}(\mathcal{H})$.

\begin{definition}[Field of representations of separable unital Banach $\ast 
$-algebras]
\label{Field of representation}\mbox{ }\newline
\emph{(i)} For any set $\mathcal{Z}$, a field $\pi _{\mathcal{Z}}$\ of
representations of a separable unital Banach $\ast $-algebra $\mathcal{X}$
is a family $\pi _{\mathcal{Z}}\doteq (\pi _{z})_{z\in \mathcal{Z}}$ of
representations $\pi _{z}$ of $\mathcal{X}$ on a separable (complex) Hilbert
space $\mathcal{H}_{z}$ for all $z\in \mathcal{Z}$. \newline
\emph{(ii)} If $(\mathcal{Z},\mathfrak{F})$ is a measurable space and $%
\alpha _{\mathcal{Z}}$ is a coherence for $\mathcal{H}_{\mathcal{Z}}\doteq (%
\mathcal{H}_{z})_{z\in \mathcal{Z}}$, as defined in (i), we say that $\pi _{%
\mathcal{Z}}$\ is $\alpha _{\mathcal{Z}}$-measurable iff $\pi _{\mathcal{Z}%
}(A)\doteq (\pi _{z}(A))_{z\in \mathcal{Z}}$ is a $\alpha _{\mathcal{Z}}$%
-measurable\ field of operators for all $A\in \mathcal{X}$. See Definition %
\ref{def measurable fields} (i).
\end{definition}

Let $(\mathcal{Z},\mathfrak{F},\mu )$\ be a $\sigma $-finite measure space
and $\mathcal{H}_{\mathcal{Z}}\doteq (\mathcal{H}_{z})_{z\in \mathcal{Z}}$
be the family of Definition \ref{Field of representation}. As
representations are norm-contractive \cite[Proposition 2.3.1]%
{BrattelliRobinsonI}, if $\pi _{\mathcal{Z}}$ is a $\alpha _{\mathcal{Z}}$%
-measurable field of representations\ of $\mathcal{X}$\ then, for all $A\in 
\mathcal{X}$, $\pi _{\mathcal{Z}}(A)$ defines a decomposable bounded
operator on the direct integral $\mathcal{H}_{\mathcal{Z}}^{\oplus }\equiv
\int_{\mathcal{Z}}^{\alpha _{\mathcal{Z}}}\mathcal{H}_{z}\mu (\mathrm{d}z)$
(Definitions \ref{Direct integrals of Hilbert spaces} and \ref{Decomposable
operators}). It is easy to check that the mapping 
\begin{equation}
\pi _{\mathcal{Z}}^{\oplus }\equiv \int_{\mathcal{Z}}^{\alpha _{\mathcal{Z}%
}}\pi _{z}\mu (\mathrm{d}z):\mathcal{X}\rightarrow M_{\mathcal{Z}}\subseteq 
\mathcal{B(H}_{\mathcal{Z}}^{\oplus })
\label{direct integral representation1}
\end{equation}%
defined by 
\begin{equation}
\pi _{\mathcal{Z}}^{\oplus }(A)\doteq \int_{\mathcal{Z}}^{\alpha _{\mathcal{Z%
}}}\pi _{z}(A)\mu (\mathrm{d}z)\ ,\qquad A\in \mathcal{X}\ ,
\label{direct integral representation2}
\end{equation}%
is a representation of $\mathcal{X}$. We term it the \emph{direct integral
(representation)}\ of the representation field $\pi _{\mathcal{Z}}$. Note,
additionally, that $N_{\mathcal{Z}}\subseteq \lbrack \pi _{\mathcal{Z}%
}^{\oplus }(\mathcal{X})]^{\prime }$, i.e., the subspace $N_{\mathcal{Z}}$
of diagonalizable operators (Definition \ref{Diagonalizable operators})
belongs to the commutant of the space $\pi _{\mathcal{Z}}^{\oplus }(\mathcal{%
X})$. See, e.g., Theorem \ref{Nielsen Th 7.1} (i).

Similar to\ Definition \ref{Hilbert decomposition}, one can decompose a
representation via an abelian von Neumann algebra as follows:

\begin{definition}[Decomposition of representations via abelian von Neumann
algebras]
\label{def direct integral represe}\mbox{ }\newline
Let $\Pi $ be a representation of a separable unital Banach $\ast $-algebra $%
\mathcal{X}$\ on the separable Hilbert space $\mathcal{\mathcal{H}}$, $%
\mathfrak{N}$ an abelian von Neumann subalgebra of $[\Pi (\mathcal{X}%
)]^{\prime }$, $(\mathcal{Z},\mathfrak{F},\mu )$ a standard measure space
and $\pi _{\mathcal{Z}}$ a $\alpha _{\mathcal{Z}}$-measurable\ field of
representations of $\mathcal{X}$, as in\ Definition \ref{Field of
representation}. $\pi _{\mathcal{Z}}^{\oplus }\equiv \int_{\mathcal{Z}%
}^{\alpha _{\mathcal{Z}}}\pi _{z}\mu (\mathrm{d}z)$ is a direct integral
decomposition of $\Pi $ with respect to $\mathfrak{N}$\ if there is a
unitary mapping $\mathrm{U}:\mathcal{\mathcal{H}}\rightarrow \mathcal{H}_{%
\mathcal{Z}}^{\oplus }$ such that, for all $A\in \mathcal{X}$, 
\begin{equation*}
\pi _{\mathcal{Z}}^{\oplus }(A)=\mathrm{U}\Pi (A)\mathrm{U}^{\ast }\qquad 
\text{and}\qquad N_{\mathcal{Z}}=\mathrm{U}\mathfrak{N}\mathrm{U}^{\ast }\ ,
\end{equation*}%
i.e., $\mathrm{U}\mathfrak{N}\mathrm{U}^{\ast }$ is precisely the algebra of
diagonal operators on $\mathcal{H}_{\mathcal{Z}}^{\oplus }$, see\ Definition %
\ref{Diagonalizable operators}. In this case, we say that $\Pi $ is
decomposable with respect to $\mathfrak{N}$.
\end{definition}

\noindent The following result ensures the existence of decompositions of
representations:

\begin{theorem}[Decompositions of representations -- I]
\label{thm decompo repre}\mbox{ }\newline
Let $(\mathcal{Z},\mathfrak{F},\mu )$\ be a $\sigma $-finite measure space
and $\Pi $ a representation of a separable unital Banach $\ast $-algebra $%
\mathcal{X}$ on a direct integral $\mathcal{H}_{\mathcal{Z}}^{\oplus }$
(Definition \ref{Direct integrals of Hilbert spaces}) such that $N_{\mathcal{%
Z}}\subseteq \lbrack \Pi (\mathcal{X})]^{\prime }$ (Definition \ref%
{Diagonalizable operators}). Then, there is a $\alpha _{\mathcal{Z}}$%
-measurable\ field $\pi _{\mathcal{Z}}$\ of representations of $\mathcal{X}$
such that $\Pi =\pi _{\mathcal{Z}}^{\oplus }$.
\end{theorem}

\begin{proof}
By Theorem \ref{Nielsen Th 7.1} (i), the assumption $N_{\mathcal{Z}%
}\subseteq \lbrack \Pi (\mathcal{X})]^{\prime }$ implies that $\Pi (\mathcal{%
X})\subseteq \lbrack \Pi (\mathcal{X})]^{\prime \prime }\subseteq M_{%
\mathcal{Z}}$, i.e., for any $A\in \mathcal{X}$, there is a $\mu $%
-essentially bounded, $\alpha _{\mathcal{Z}}$-measurable operator field $%
(A_{z})_{z\in \mathcal{Z}}$ such that 
\begin{equation*}
\Pi \left( A\right) =\int_{\mathcal{Z}}^{\alpha _{\mathcal{Z}}}A_{z}\mu (%
\mathrm{d}z)\ .
\end{equation*}%
See Definition \ref{Decomposable operators}. It means that the field $\pi _{%
\mathcal{Z}}$\ of mappings from $\mathcal{X}$ to $\mathcal{B}(\mathcal{H}%
_{z})$ defined by $\pi _{z}(A)\doteq {A_{z}}$ may be a good candidate for a $%
\alpha _{\mathcal{Z}}$-measurable\ field\ of representations of $\mathcal{X}$%
. In fact, all $\ast $-algebraic operations in $M_{\mathcal{Z}}$ (Theorem %
\ref{Nielsen Th 7.1} (i)) refers, $\mu $-almost everywhere, to the
corresponding operations in each fiber. There is, however, a complication
here in converting $\ast $-algebraic operations in $M_{\mathcal{Z}}$ into
fiberwise ones, because these only hold true $\mu $-almost everywhere, on
subsets of $\mathcal{Z}$ that possibly depend on the elements of $\mathcal{X}
$ that are taken. Therefore, as $\mathcal{X}$ is a \emph{separable} (unital)
Banach $\ast $-algebra, one takes a \emph{countable} dense subset $\mathcal{X%
}_{0}$ of $\mathcal{X}$ which is a $\ast $-algebra over $\mathbb{Q}+i\mathbb{%
Q}$. Since the countable union of measurable sets of zero measure has zero
measure, there is a fixed subset $\mathcal{Z}_{0}\subseteq \mathcal{Z}$
satisfying $\mu \left( \mathcal{Z}\backslash \mathcal{Z}_{0}\right) =0$ such
that the definitions $\pi _{z}(A)\doteq {A_{z}}$ for $z\in \mathcal{Z}_{0}$
and $\pi _{z}(A)\doteq 0$ for $z\in \mathcal{Z}\backslash \mathcal{Z}_{0}$
lead to a $\alpha _{\mathcal{Z}}$-measurable\ field $\pi _{\mathcal{Z}}$\ of
representations of $\mathcal{X}$ satisfying $\Pi =\pi _{\mathcal{Z}}^{\oplus
}$. The desired properties of this family $\pi _{\mathcal{Z}}$ of mappings
is a consequence of the density of $\mathcal{X}_{0}\subseteq \mathcal{X}$
and Theorem \ref{Nielsen Th 7.1} (iii), together with \cite[Proposition 2.3.1%
]{BrattelliRobinsonI}. For more details, see \cite[Theorem 12.3]%
{Niesen-direct integrals}.
\end{proof}

\begin{corollary}[Decompositions of representations -- II]
\label{thm decompo repre2}\mbox{ }\newline
Let $\Pi $ be any representation of a separable unital Banach $\ast $%
-algebra $\mathcal{X}$\ on a (separable) Hilbert space $\mathcal{H}$. If $%
\mathfrak{N}$ is an abelian von Neumann subalgebra of $[\Pi (\mathcal{X}%
)]^{\prime }$, then it is decomposable with respect to $\mathfrak{N}$.
\end{corollary}

\begin{proof}
The assertion is a consequence of Theorems \ref{universality} and \ref{thm
decompo repre}.
\end{proof}

\noindent Observe that decompositions of a given representation $\Pi $ with
respect to $\mathfrak{N}$ are unique, up to natural equivalences. In
particular, one can speak in this case about \emph{the} direct integral
decompositions of $\Pi $. For more details, see \cite[Theorems 12.1 and 12.4]%
{Niesen-direct integrals}.

Different types of direct integral decompositions can be defined:

\begin{definition}[Special direct integral representations]
\label{Special direct integral representations}\mbox{ }\newline
Under Conditions of Definition \ref{def direct integral represe} we define
the following terminology for the direct integral decomposition $\pi _{%
\mathcal{Z}}^{\oplus }\equiv \int_{\mathcal{Z}}^{\alpha _{\mathcal{Z}}}\pi
_{z}\mu (\mathrm{d}z)$\ of $\Pi $ with respect to $\mathfrak{N}$:\newline
\emph{(i.1)} $\pi _{\mathcal{Z}}^{\oplus }$ is a maximal decomposition if $%
\mathfrak{N}$ is a maximal abelian von Neumann subalgebra of $[\Pi (\mathcal{%
X})]^{\prime }$, i.e., 
\begin{equation*}
\mathfrak{N}^{\prime }\cap \lbrack \Pi (\mathcal{X})]^{\prime }\subseteq 
\mathfrak{N}\ .
\end{equation*}%
\emph{(i.2)} $\pi _{\mathcal{Z}}^{\oplus }$ is an irreducible decomposition\
whenever $\pi _{z}$ is irreducible for $\mu $-almost everywhere $z\in 
\mathcal{Z}$, i.e., $\{0\}$ and $\mathcal{H}_{z}$ are the only closed
subspaces of $\mathcal{H}_{z}$ that are invariant under the action of $\pi
_{z}(\mathcal{X})$. \newline
\emph{(ii.1)} $\pi _{\mathcal{Z}}^{\oplus }$ is a subcentral\ decomposition
if $\mathfrak{N}$ is contained in the center of the von Neumann algebra $%
[\Pi (\mathcal{X})]^{\prime }$, i.e.,%
\begin{equation*}
\mathfrak{N}\subseteq \lbrack \Pi (\mathcal{X})]^{\prime }\cap \lbrack \Pi (%
\mathcal{X})]^{\prime \prime }\ .
\end{equation*}%
If the equality holds true, then we speak about the central decomposition. 
\newline
\emph{(ii.2)} $\pi _{\mathcal{Z}}^{\oplus }$ is a factor decomposition\
whenever $\pi _{z}$ is a factor representation of $\mathcal{X}$\ for $\mu $%
-almost everywhere $z\in \mathcal{Z}$, i.e., 
\begin{equation*}
\lbrack \pi _{z}(\mathcal{X})]^{\prime }\cap \lbrack \pi _{z}(\mathcal{X}%
)]^{\prime \prime }=\mathbb{C}\mathbf{1}_{\mathcal{H}_{z}}\qquad \text{(}\mu 
\text{-almost everywhere).}
\end{equation*}
\end{definition}

The maximality of the abelian von Neumann algebra is equivalent to the
irreducibility of the representations appearing in direct integral
decompositions, i.e., Definitions \ref{Special direct integral
representations} (i.1) and (i.2) are equivalent:

\begin{theorem}[Maximal versus irreducible decompositions]
\label{irreducible fiber representations}\mbox{ }\newline
Let $\Pi $ be a representation of a separable unital Banach $\ast $-algebra $%
\mathcal{X}$\ on a (separable) Hilbert space $\mathcal{H}$ and $\mathfrak{N}$
an abelian von Neumann subalgebra of $[\Pi (\mathcal{X})]^{\prime }$. Let $%
\pi _{\mathcal{Z}}^{\oplus }$ be the direct integral decomposition of $\Pi $
with respect to $\mathfrak{N}$ (Corollary \ref{thm decompo repre2}). Then, $%
\mathfrak{N}$ is a maximal abelian von Neumann subalgebra of $[\Pi (\mathcal{%
X})]^{\prime }$ iff $\pi _{z}$ is irreducible for $\mu $-almost everywhere $%
z\in \mathcal{Z}$.
\end{theorem}

\begin{proof}
Without loss of generality we assume that $\Pi $ is a representation of $%
\mathcal{X}$ on a direct integral $\mathcal{H}_{\mathcal{Z}}^{\oplus }$
(Definition \ref{Direct integrals of Hilbert spaces}) such that $N_{\mathcal{%
Z}}=\mathfrak{N}\subseteq \lbrack \Pi (\mathcal{X})]^{\prime }$ (Definition %
\ref{Diagonalizable operators}), where $(\mathcal{Z},\mathfrak{F},\mu )$\ is
a standard measure space.

Assume that $\pi _{z}$ is irreducible for $\mu $-almost everywhere $z\in 
\mathcal{Z}$. $N_{\mathcal{Z}}$ is maximal abelian in $[\pi _{\mathcal{Z}%
}^{\oplus }(\mathcal{X})]^{\prime }$ iff $N_{\mathcal{Z}}^{\prime }\cap
\lbrack \pi _{\mathcal{Z}}^{\oplus }(\mathcal{X})]^{\prime }\subseteq N_{%
\mathcal{Z}}$. Therefore, take $A\in N_{\mathcal{Z}}^{\prime }\cap \lbrack
\pi _{\mathcal{Z}}^{\oplus }(\mathcal{X})]^{\prime }$. By Theorem \ref%
{Nielsen Th 7.1} (i), there is a $\mu $-essentially bounded, $\alpha _{%
\mathcal{Z}}$-measurable operator field $(A_{z})_{z\in \mathcal{Z}}$ such
that $A=\int_{\mathcal{Z}}^{\alpha _{\mathcal{Z}}}A_{z}\mu (\mathrm{d}z)$,
see Definition \ref{Decomposable operators}. Using similar arguments as in
the proof of Theorem \ref{thm decompo repre}, one uses a countable dense
subset of $\mathcal{X}$ to prove that $A\in \lbrack \pi _{\mathcal{Z}%
}^{\oplus }(\mathcal{X})]^{\prime }$ yields ${A_{z}}\in \lbrack \pi _{z}(%
\mathcal{X})]^{\prime }$ for $\mu $-almost everywhere $z\in \mathcal{Z}$.
Since $\pi _{z}$ is $\mu $-almost everywhere irreducible, $[\pi _{z}(%
\mathcal{X})]^{\prime }=\mathbb{C}\mathbf{1}_{\mathcal{H}_{z}}$ for $\mu $%
-almost everywhere $z\in \mathcal{Z}$, by \cite[Proposition 2.3.8]%
{BrattelliRobinsonI}. As a consequence, ${A_{z}}\in \lbrack \pi _{z}(%
\mathcal{X})]^{\prime }$ for $\mu $-almost everywhere $z\in \mathcal{Z}$
implies that $A\in N_{\mathcal{Z}}$.

The converse statement, namely the fact that $N_{\mathcal{Z}}^{\prime }\cap
\lbrack \pi _{\mathcal{Z}}^{\oplus }(\mathcal{X})]^{\prime }\subseteq N_{%
\mathcal{Z}}$ yields the irreducibility of $\pi _{z}$ for $\mu $-almost
everywhere $z\in \mathcal{Z}$, requires more arguments. First, an assertion
which is similar to Lemmata \ref{lemma-useful-coherences} and \ref%
{lemma-useful-coherences copy(1)} holds true for the direct integral
representation $\pi _{\mathcal{Z}}^{\oplus }$, see \cite[Eq. (*) in Section
11, p. 46]{Niesen-direct integrals}. So, one can assume without loss of
generality that $\mathcal{H}_{\mathcal{Z}}$ is a family of constant
separable Hilbert spaces. In this case, one then uses the properties of
standard measure spaces (see \cite[Theorems 4.1 and 4.3]{Niesen-direct
integrals}) to show the existence of two Borel sets $\mathcal{Z}_{0},%
\mathcal{Z}_{1}\subseteq \mathcal{Z}$ and a $\alpha _{\mathcal{Z}}$%
-measurable operator field $(A_{z})_{z\in \mathcal{Z}_{0}}$ over $\mathcal{H}%
_{\mathcal{Z}_{0}}\doteq (\mathcal{H}_{z})_{z\in \mathcal{Z}_{0}}$ such that%
\begin{equation*}
\mu \left( \mathcal{Z}_{1}\backslash \mathcal{Z}_{0}\right) =0\ ,\qquad
\int_{\mathcal{Z}_{0}}^{\alpha _{\mathcal{Z}}}A_{z}\mu (\mathrm{d}z)\in
\left( M_{\mathcal{Z}}\backslash N_{\mathcal{Z}}\right) \cap \lbrack \pi _{%
\mathcal{Z}}^{\oplus }(\mathcal{X})]^{\prime }
\end{equation*}%
and 
\begin{equation*}
\mathcal{Z}_{0}\subseteq S\doteq \{z\in \mathcal{Z}:\pi _{z}\text{ is not
irreducible}\}\subseteq \mathcal{Z}_{1}\ .
\end{equation*}%
For the precise arguments leading to these facts, see \cite[Theorem 13.1]%
{Niesen-direct integrals}. By Theorem \ref{Nielsen Th 7.1} (i) recall that $%
M_{\mathcal{Z}}=N_{\mathcal{Z}}^{\prime }$ and since, by assumption, $N_{%
\mathcal{Z}}^{\prime }\cap \lbrack \pi _{\mathcal{Z}}^{\oplus }(\mathcal{X}%
)]^{\prime }\subseteq N_{\mathcal{Z}}$, it follows that $\mu \left( \mathcal{%
Z}_{0}\right) =0$, which in turn implies that $\mu \left( \mathcal{Z}%
_{1}\right) =0$.
\end{proof}

By\ Theorem \ref{irreducible fiber representations}, if $\Pi $ is a
representation of $\mathcal{X}$\ on $\mathcal{H}$ then \emph{maximal}
abelian von Neumann subalgebras of $[\Pi (\mathcal{X})]^{\prime }$ determine
decompositions of $\Pi $\ with respect to \emph{irreducible}
representations. Observe that the representation $\Pi $ is irreducible iff $%
[\Pi (\mathcal{X})]^{\prime }=\mathbb{C}\mathbf{1}_{\mathcal{H}}$, by \cite[%
Proposition 2.3.8]{BrattelliRobinsonI}. Such a representation is a
particular example of a factor\ representation, for one trivially has in
this case that 
\begin{equation*}
\lbrack \Pi (\mathcal{X})]^{\prime }\cap \lbrack \Pi (\mathcal{X})]^{\prime
\prime }=\mathbb{C}\mathbf{1}_{\mathcal{H}}\ .
\end{equation*}

The next result establishes a relation between the central decomposition,
i.e., the decompositions with respect to the center $[\Pi (\mathcal{X}%
)]^{\prime }\cap \lbrack \Pi (\mathcal{X})]^{\prime \prime }$ of the von
Neumann algebra $[\Pi (\mathcal{X})]^{\prime }$, and factor decompositions:

\begin{theorem}[Central versus factor decompositions]
\label{central decompositions}\mbox{ }\newline
Let $\Pi $ be a representation of a separable unital Banach $\ast $-algebra $%
\mathcal{X}$\ on a (separable) Hilbert space $\mathcal{H}$ and $\mathfrak{N}$
an abelian von Neumann subalgebra of $[\Pi (\mathcal{X})]^{\prime }$. Let $%
\pi _{\mathcal{Z}}^{\oplus }$ be the direct integral decomposition of $\Pi $
with respect to $\mathfrak{N}$ (Corollary \ref{thm decompo repre2}).\newline
\emph{(i)} If $\pi _{z}$ is a factor representation of $\mathcal{X}$\ for $%
\mu $-almost everywhere $z\in \mathcal{Z}$\ then $[\Pi (\mathcal{X}%
)]^{\prime }\cap \lbrack \Pi (\mathcal{X})]^{\prime \prime }\subseteq 
\mathfrak{N}$. \newline
\emph{(ii)} If $\mathfrak{N}=[\Pi (\mathcal{X})]^{\prime }\cap \lbrack \Pi (%
\mathcal{X})]^{\prime \prime }$ then $\pi _{z}$ is a factor representation
of $\mathcal{X}$\ for $\mu $-almost everywhere $z\in \mathcal{Z}$.
\end{theorem}

\begin{proof}
Without loss of generality we assume that $\Pi =\pi _{\mathcal{Z}}^{\oplus }$
is a representation of $\mathcal{X}$ on a direct integral $\mathcal{H}_{%
\mathcal{Z}}^{\oplus }$ (Definition \ref{Direct integrals of Hilbert spaces}%
) such that $\mathfrak{N}=N_{\mathcal{Z}}\subseteq \lbrack \pi _{\mathcal{Z}%
}^{\oplus }(\mathcal{X})]^{\prime }$ (Definition \ref{Diagonalizable
operators}), where $(\mathcal{Z},\mathfrak{F},\mu )$\ is a standard measure
space.

(i) On the one hand, by Theorem \ref{Nielsen Th 7.1} (i), $N_{\mathcal{Z}%
}\subseteq \lbrack \pi _{\mathcal{Z}}^{\oplus }(\mathcal{X})]^{\prime }$
yields $[\pi _{\mathcal{Z}}^{\oplus }(\mathcal{X})]^{\prime \prime
}\subseteq M_{\mathcal{Z}}$. Therefore, for any $A\in \lbrack \pi _{\mathcal{%
Z}}^{\oplus }(\mathcal{X})]^{\prime }\cap \lbrack \pi _{\mathcal{Z}}^{\oplus
}(\mathcal{X})]^{\prime \prime }$, there is a $\mu $-essentially bounded, $%
\alpha _{\mathcal{Z}}$-measurable operator field $(A_{z})_{z\in \mathcal{Z}}$
such that $A=\int_{\mathcal{Z}}^{\alpha _{\mathcal{Z}}}A_{z}\mu (\mathrm{d}%
z) $ and ${A_{z}}\in \lbrack \pi _{z}(\mathcal{X})]^{\prime }$ for $\mu $%
-almost everywhere $z\in \mathcal{Z}$, using similar arguments as in the
proof of Theorem \ref{irreducible fiber representations}. On the other hand,
since $\mathcal{H}_{\mathcal{Z}}^{\oplus }$ is separable (Theorem \ref%
{Nielsen Th 7.1 copy(1)}), one can apply the Kaplanski density theorem \cite[%
Theorem A.2]{Niesen-direct integrals} to approximate, in the strong
topology, any decomposable operator 
\begin{equation*}
A=\int_{\mathcal{Z}}^{\alpha _{\mathcal{Z}}}A_{z}\mu (\mathrm{d}z)\in
\lbrack \pi _{\mathcal{Z}}^{\oplus }(\mathcal{X})]^{\prime \prime }\cap M_{%
\mathcal{Z}}
\end{equation*}%
by a sequence\footnote{%
Recall that$\mathcal{X}$ is assumed to be unital.} $(\pi _{\mathcal{Z}%
}^{\oplus }(A_{n}))_{n\in \mathbb{N}}$ with $A_{n}\in \mathcal{X}$. One then
uses Theorem \ref{Nielsen Th 7.1} (iv) to show that ${A_{z}}\in \lbrack \pi
_{z}(\mathcal{X})]^{\prime \prime }$ for $\mu $-almost everywhere $z\in 
\mathcal{Z}$. As a consequence, if $A\in \lbrack \pi _{\mathcal{Z}}^{\oplus
}(\mathcal{X})]^{\prime }\cap \lbrack \pi _{\mathcal{Z}}^{\oplus }(\mathcal{X%
})]^{\prime \prime }$ then, for $\mu $-almost everywhere $z\in \mathcal{Z}$, 
${A_{z}}$ belongs to the center of $[\pi _{z}(\mathcal{X})]^{\prime }$,
which consists of multiples of the identity, by assumption. Therefore, 
\begin{equation*}
\lbrack \pi _{\mathcal{Z}}^{\oplus }(\mathcal{X})]^{\prime }\cap \lbrack \pi
_{\mathcal{Z}}^{\oplus }(\mathcal{X})]^{\prime \prime }\subseteq N_{\mathcal{%
Z}}\ .
\end{equation*}

(ii) Since $\mathcal{H}_{\mathcal{Z}}^{\oplus }$ is separable (Theorem \ref%
{Nielsen Th 7.1 copy(1)}), by the Kaplanski density theorem \cite[Theorem A.2%
]{Niesen-direct integrals}, there is a sequence $(C_{n})_{n\in \mathbb{N}}$
of operators that is dense in the weak-operator topology within the unit
ball of the commutant $[\pi _{\mathcal{Z}}^{\oplus }(\mathcal{X})]^{\prime }$%
. As $\mathcal{X}$ is separable and since representations are
norm-contractive \cite[Proposition 2.3.1]{BrattelliRobinsonI}, we construct
the separable unital $C^{\ast }$-algebra $\mathcal{C}$ of operators on the
direct integral $\mathcal{H}_{\mathcal{Z}}^{\oplus }$ that is generated by $%
\pi _{\mathcal{Z}}^{\oplus }(\mathcal{X})$ and $(C_{n})_{n\in \mathbb{N}%
}\subseteq \lbrack \pi _{\mathcal{Z}}^{\oplus }(\mathcal{X})]^{\prime }$.
This $C^{\ast }$-algebra satisfies, by construction, the equality 
\begin{equation*}
\mathcal{C}^{\prime }=[\pi _{\mathcal{Z}}^{\oplus }(\mathcal{X})]^{\prime
}\cap \lbrack \pi _{\mathcal{Z}}^{\oplus }(\mathcal{X})]^{\prime \prime }\ ,
\end{equation*}%
which equals $N_{\mathcal{Z}}$, by assumption. Using now \cite[Proposition
2.3.1]{BrattelliRobinsonI} and Theorem \ref{Nielsen Th 7.1} (iii) together
with Corollary \ref{thm decompo repre2} and Theorem \ref{irreducible fiber
representations}, where the separable Banach $\ast $-algebra, the
representation and the abelian sublagebra, are respectively $\mathcal{C}$,
the identity mapping and $N_{\mathcal{Z}}=\mathcal{C}^{\prime }$, we deduce
the existence of a $\alpha _{\mathcal{Z}}$-measurable field $\varkappa _{%
\mathcal{Z}}$ of $\mu $-almost everywhere \emph{irreducible} representations
of $\mathcal{C}$ on $\mathcal{H}_{z}$ such that 
\begin{equation*}
A=\int_{\mathcal{Z}}^{\alpha _{\mathcal{Z}}}\varkappa _{z}\left( A\right)
\mu (\mathrm{d}z)\ ,\qquad A\in \mathcal{C}\ ,
\end{equation*}%
with $\pi _{z}=\varkappa _{z}\circ \pi _{\mathcal{Z}}^{\oplus }$ and $%
\varkappa _{z}(C_{n})\in \lbrack \pi _{z}(\mathcal{X})]^{\prime }$ for $\mu $%
-almost everywhere $z\in \mathcal{Z}$ and $n\in \mathbb{N}$. By the
weak-operator density of $(C_{n})_{n\in \mathbb{N}}$ in $[\pi _{\mathcal{Z}%
}^{\oplus }(\mathcal{X})]^{\prime }$ and \cite[Proposition 2.3.8]%
{BrattelliRobinsonI}, it follows that%
\begin{equation*}
\lbrack \pi _{z}(\mathcal{X})]^{\prime }\cap \lbrack \pi _{z}(\mathcal{X}%
)]^{\prime \prime }=[\varkappa _{z}\circ \pi _{\mathcal{Z}}^{\oplus }(%
\mathcal{X})]^{\prime }\cap \lbrack (\varkappa _{z}(C_{n}))_{n\in \mathbb{N}%
}]^{\prime }=[\varkappa _{z}(\mathcal{C})]^{\prime }=\mathbb{C}\mathbf{1}_{%
\mathcal{H}}
\end{equation*}%
for $\mu $-almost everywhere $z\in \mathcal{Z}$.
\end{proof}

\noindent In general, $[\Pi (\mathcal{X})]^{\prime }\cap \lbrack \Pi (%
\mathcal{X})]^{\prime \prime }\subseteq \mathfrak{N}$ does not necessarily
imply that $\pi _{z}$ is a factor representation of $\mathcal{X}$\ for $\mu $%
-almost everywhere $z\in \mathcal{Z}$. See \cite[Example 15.3]{Niesen-direct
integrals}. In particular, by Theorem \ref{central decompositions} (ii),
central decompositions (Definition \ref{Special direct integral
representations} (ii.1)) are always factor decompositions (Definition \ref%
{Special direct integral representations} (ii.2)), but the converse is
false, in general. In other words, the two definitions are not equivalent.

By Theorem \ref{central decompositions}, observe that any representation of $%
\mathcal{X}$ admits a central decomposition. One remarkable fact about
central, or subcentral, decompositions is that the corresponding factor
representations are non-redundant in the following sense:

\begin{theorem}[Non-redundancy of subcentral decompositions]
\label{central decompositions copy(1)}\mbox{ }\newline
Let $\Pi $ be a representation of a separable unital Banach $\ast $-algebra $%
\mathcal{X}$\ on a (separable) Hilbert space $\mathcal{H}$ and $\pi _{%
\mathcal{Z}}^{\oplus }$ be the direct integral decomposition of $\Pi $ with
respect to $\mathfrak{N}\subseteq \lbrack \Pi (\mathcal{X})]^{\prime }\cap
\lbrack \Pi (\mathcal{X})]^{\prime \prime }$ (Corollary \ref{thm decompo
repre2}). Then, for some measurable subset $\mathcal{Z}_{0}\subseteq 
\mathcal{Z}$, $\mu (\mathcal{Z}_{0})=0$, and all $z_{1},z_{2}\in \mathcal{Z}%
\backslash \mathcal{Z}_{0}$, $z_{1}\neq z_{2}$, there exists no $\ast $%
-isomorphism $\kappa :[\pi _{z_{1}}(\mathcal{X})]^{\prime \prime }\mapsto
\lbrack \pi _{z_{2}}(\mathcal{X})]^{\prime \prime }$ such that $\kappa \circ
\pi _{z_{1}}=\pi _{z_{2}}$.
\end{theorem}

\begin{proof}
As in previous proofs, let $\Pi =\pi _{\mathcal{Z}}^{\oplus }$ be a
representation of $\mathcal{X}$ on a direct integral $\mathcal{H}_{\mathcal{Z%
}}^{\oplus }$ such that $\mathfrak{N}=N_{\mathcal{Z}}\subseteq \lbrack \pi _{%
\mathcal{Z}}^{\oplus }(\mathcal{X})]^{\prime }\cap \lbrack \pi _{\mathcal{Z}%
}^{\oplus }(\mathcal{X})]^{\prime \prime }$, where $(\mathcal{Z},\mathfrak{F}%
,\mu )$\ is a standard measure space. By \cite[Proposition B.5]%
{Niesen-direct integrals}, it suffices to prove that $\pi _{z_{1}}$ and $\pi
_{z_{2}}$ are disjoint for all $z_{1},z_{2}\in \mathcal{Z}\backslash 
\mathcal{Z}_{0}$ with $z_{1}\neq z_{2}$. This means that there is no
non-zero operator $T$ from $\mathcal{H}_{z_{1}}$ to $\mathcal{H}_{z_{2}}$
such that 
\begin{equation}
T\pi _{z_{1}}\left( A\right) =\pi _{z_{2}}\left( A\right) T\ ,\qquad A\in 
\mathcal{X}\ .  \label{interwine}
\end{equation}%
To this end, one takes a sequence $(\mathcal{Z}_{n})_{n\in \mathbb{N}}$ of
Borel subsets of $\mathcal{Z}$ separating the points of $\mathcal{Z}$. By
assumption, $N_{\mathcal{Z}}\subseteq \lbrack \pi _{\mathcal{Z}}^{\oplus }(%
\mathcal{X})]^{\prime \prime }$ and we deduce that 
\begin{equation}
\chi _{n}\doteq \int_{\mathcal{Z}}^{\alpha _{\mathcal{Z}}}\mathbf{1}\left[
z\in \mathcal{Z}_{n}\right] \mathbf{1}_{\mathcal{H}_{z}}\mu (\mathrm{d}z)\in
\lbrack \pi _{\mathcal{Z}}^{\oplus }(\mathcal{X})]^{\prime \prime }\ ,\qquad
n\in \mathbb{N}\ .  \label{interwine2}
\end{equation}%
Assume now that $T$ is an operator from $\mathcal{H}_{z_{1}}$ to $\mathcal{H}%
_{z_{2}}$ satisfying (\ref{interwine}) with $z_{1}\in \mathcal{Z}_{n}$ and $%
z_{2}\notin \mathcal{Z}_{n}$ for some $n\in \mathbb{N}$. Then, by (\ref%
{interwine})-(\ref{interwine2}), we \emph{formally} expect that, for $\mu $%
-almost everywhere $z_{1},z_{2}\in \mathcal{Z}$, $z_{1}\neq z_{2}$, 
\begin{equation}
T\mathbf{1}\left[ z_{1}\in \mathcal{Z}_{n}\right] =T\pi _{z_{1}}\left( \chi
_{n}\right) =\pi _{z_{2}}\left( \chi _{n}\right) T=0\ ,
\label{equation conne}
\end{equation}%
implying $T=0$. To complete the arguments, by making sense of the formal
elements $\pi _{z}\left( \chi _{n}\right) $, one uses a convenient
approximation of $(\chi _{n})_{n\in \mathbb{N}}\subseteq \lbrack \pi _{%
\mathcal{Z}}^{\oplus }(\mathcal{X})]^{\prime \prime }$, in the strong
topology, via the Kaplanski density theorem \cite[Theorem A.2]{Niesen-direct
integrals}, as in Theorem \ref{central decompositions} (i). Then, we use
Theorem \ref{Nielsen Th 7.1} (iv) and the fact that a countable union of
measurable sets of zero measure has zero measure to obtain a set $\mathcal{Z}%
_{0}$ of zero measure such that Equation (\ref{equation conne}) holds true
for any $z_{1}\in \mathcal{Z}_{n}\cap \mathcal{Z}\backslash \mathcal{Z}_{0}$
and $z_{2}\notin \mathcal{Z}_{n}\cap \mathcal{Z}\backslash \mathcal{Z}_{0}$.
For more details, see \cite[Theorem 13.3]{Niesen-direct integrals}.
\end{proof}

Theorem \ref{central decompositions copy(1)} means that $\pi _{z_{1}}$\ and $%
\pi _{z_{2}}$ are not \emph{quasi-equivalent}\footnote{%
The definition of quasi-equivalent representations in \cite[Definition 2.4.25%
]{BrattelliRobinsonI} differs from the one of \cite[Appendix B, p. 146]%
{Niesen-direct integrals}, but they are equivalent, by \cite[Theorem 2.4.26]%
{BrattelliRobinsonI}.} representations for any $z_{1},z_{2}\in \mathcal{Z}%
\backslash \mathcal{Z}_{0}$, $z_{1}\neq z_{2}$. When we only have a factor
decomposition of a representation $\Pi $ with respect to $\mathfrak{N}$
(Corollary \ref{thm decompo repre2}) such that 
\begin{equation}
\lbrack \Pi (\mathcal{X})]^{\prime }\cap \lbrack \Pi (\mathcal{X})]^{\prime
\prime }\varsubsetneq \mathfrak{N}\subseteq \lbrack \Pi (\mathcal{X}%
)]^{\prime },  \label{redondant}
\end{equation}%
we have a redundancy in the following sense:\ Since two factor
representations $\pi _{1},\pi _{2}$ are either quasi-equivalent or disjoint 
\cite[Proposition B.5]{Niesen-direct integrals}, when (\ref{redondant})
holds true, the direct integral representation $\pi _{\mathcal{Z}}^{\oplus }$
of $\Pi $ is constructed from a $\alpha _{\mathcal{Z}}$-measurable\ field $%
\pi _{\mathcal{Z}}$\ of representations of $\mathcal{X}$ having
quasi-equivalent representations within a set of non-zero measure. This
redundancy disappears when $\mathfrak{N}$ is exactly the center of $[\Pi (%
\mathcal{X})]^{\prime }$, or $[\Pi (\mathcal{X})]^{\prime \prime }$, by
Theorem \ref{central decompositions copy(1)}. If 
\begin{equation}
\mathfrak{N}\varsubsetneq \lbrack \Pi (\mathcal{X})]^{\prime }\cap \lbrack
\Pi (\mathcal{X})]^{\prime \prime }  \label{ineauqlity sup}
\end{equation}%
then there is also no redundancy, in the same way. However, by Theorem \ref%
{irreducible fiber representations}, when (\ref{ineauqlity sup}) holds true,
the irreducibility of fiber representations cannot be true $\mu $-almost
everywhere since $\mathfrak{N}$ is clearly not a maximal abelian von Neumann
subalgebra of $[\Pi (\mathcal{X})]^{\prime }$. Note that\ (\ref{ineauqlity
sup}) implies that $\Pi $ \emph{cannot} be an irreducible representation of
a separable unital Banach $\ast $-algebra $\mathcal{X}$\ on a (separable)
Hilbert space $\mathcal{H}$.

Note finally that, even if there exists a $\ast $-isomorphism $\kappa :[\pi
_{1}(\mathcal{X})]^{\prime \prime }\mapsto \lbrack \pi _{2}(\mathcal{X}%
)]^{\prime \prime }$ such that $\kappa \circ \pi _{1}=\pi _{2}$, two
representations $\pi _{1},\pi _{2}$ of the same $C^{\ast }$-algebra $%
\mathcal{X}$ on two Hilbert spaces $\mathcal{H}_{1},\mathcal{H}_{2}$,
respectively, are not necessarily (unitarily) equivalent, which means the
existence of a unitary operator $\mathrm{U}$ from $\mathcal{H}_{1}$ to $%
\mathcal{H}_{2}$ such that $\pi _{1}(A)=\mathrm{U}^{\ast }\pi _{2}(A)\mathrm{%
U}$ for any $A\in \mathcal{X}$. This is related with the question whether
isomorphisms between von Neumann algebras can be unitarily implemented. By 
\cite[Theorem 2.4.26]{BrattelliRobinsonI}, if $\pi _{1},\pi _{2}$ \ are two
quasi-equivalent representations then $\pi _{1},\pi _{2}$ are (unitarily)
equivalent up to multiplicity.

\subsection{Direct Integrals of von Neumann Algebras}

In this section, we study the direct integrals of von Neumann algebras $%
\mathfrak{M}$ on separable Hilbert spaces $\mathcal{H}$, i.e., a $\ast $%
-subalgebra of $\mathcal{B}(\mathcal{H})$ so that $\mathfrak{M}^{\prime
\prime }=\mathfrak{M}$. See \cite[Definition 2.4.8. and Theorem 2.4.11]%
{BrattelliRobinsonI}.

\begin{definition}[Fields of von Neumann algebras]
\label{lemma direct integral von neumann0}\mbox{ }\newline
\emph{(i)} For any set $\mathcal{Z}$, a field of von Neumann algebras\ over
a measurable family $\mathcal{H}_{\mathcal{Z}}$ of separable Hilbert spaces
is a family $\mathfrak{M}_{\mathcal{Z}}\doteq (\mathfrak{M}_{z})_{z\in 
\mathcal{Z}}$ of von Neumann algebras $\mathfrak{M}_{z}$ acting on $\mathcal{%
H}_{z}$ for all $z\in \mathcal{Z}$. \newline
\emph{(ii)} If $(\mathcal{Z},\mathfrak{F})$ is a measurable space and $%
\alpha _{\mathcal{Z}}$ is a coherence for the family $\mathcal{H}_{\mathcal{Z%
}}$ of (i), we say that $\mathfrak{M}_{\mathcal{Z}}$\ is $\alpha _{\mathcal{Z%
}}$-measurable\ iff there exists a sequence $(A^{(n)})_{n\in \mathbb{N}}$ of 
$\alpha _{\mathcal{Z}}$-measurable fields of operators such that $%
\{A_{z}^{(n)}:n\in \mathbb{N}\}\subseteq \mathcal{B}(\mathcal{H}_{z})$
generates\footnote{%
I.e., $\mathfrak{M}_{z}$ is the closure, in the strong or weak operator
topology, of the $\ast $-algebra generated by this set.} $\mathfrak{M}_{z}$
for all $z\in \mathcal{Z}$. Such a sequence $(A^{(n)})_{n\in \mathbb{N}}$ is
named a $\alpha _{\mathcal{Z}}$-measurable generating sequence for $%
\mathfrak{M}_{\mathcal{Z}}$.
\end{definition}

If $\mathcal{H}_{\mathcal{Z}}$\ is a measurable family of separable Hilbert
spaces (see Definition \ref{def measurable family hilbert space}), then $(%
\mathbb{C}\mathbf{1}_{\mathcal{H}_{z}})_{z\in \mathcal{Z}}$ and $(\mathcal{B}%
(\mathcal{H}_{z}))_{z\in \mathcal{Z}}$ are trivial examples of $\alpha _{%
\mathcal{Z}}$-measurable fields\ of von Neumann algebras on $\mathcal{H}_{%
\mathcal{Z}}$. Another less trivial example is given by the following lemma:

\begin{lemma}[Fields of bicommutants of representations]
\label{lemma direct integral von neumann1}\mbox{ }\newline
Let $(\mathcal{Z},\mathfrak{F})$ be a measurable space and $\pi _{\mathcal{Z}%
}$ a $\alpha _{\mathcal{Z}}$-measurable field of\ representations of a
separable unital Banach $\ast $-algebra $\mathcal{X}$\ over a measurable
family $\mathcal{H}_{\mathcal{Z}}$ of separable Hilbert spaces. Then, the
field $([\pi _{z}(\mathcal{X})]^{\prime \prime })_{z\in \mathcal{Z}}$\ of
von Neumann algebras is $\alpha _{\mathcal{Z}}$-measurable.
\end{lemma}

\begin{proof}
For any $z\in \mathcal{Z}$, $[\pi _{z}(\mathcal{X})]^{\prime \prime }$ is
generated by the set $\pi _{z}(\mathcal{X})$, while $\mathcal{X}$ is
separable and a representation is norm-contractive \cite[Proposition 2.3.1]%
{BrattelliRobinsonI}. Therefore, the lemma is a direct consequence of the $%
\alpha _{\mathcal{Z}}$-measurability of $\pi _{\mathcal{Z}}$. See
Definitions \ref{Field of representation} (ii) and \ref{lemma direct
integral von neumann0} (ii).
\end{proof}

\noindent In fact, all $\alpha _{\mathcal{Z}}$-measurable fields of von
Neumann algebras are of this form, by \cite[Lemma 18.1]{Niesen-direct
integrals}.

It turns out that measurable fields of von Neumann algebras are stable with
respect to simple point-wise operations on fields of von Neumann algebras:

\begin{lemma}[Stability of $\protect\alpha _{\mathcal{Z}}$-measurable fields
of von\ Neumann algebras]
\label{lemma direct integral von neumann2}\mbox{ }\newline
Let $(\mathcal{Z},\mathfrak{F})$ be a measurable space and $\mathfrak{M}_{%
\mathcal{Z}}$, $\mathfrak{M}_{\mathcal{Z}}^{(n)}$, $n\in \mathbb{N}$, $%
\alpha _{\mathcal{Z}}$-measurable fields $\mathfrak{M}_{\mathcal{Z}}$\ of
von Neumann algebras over a measurable family $\mathcal{H}_{\mathcal{Z}}$ of
separable Hilbert spaces. Then, the fields $(\mathfrak{M}_{z}^{\prime
})_{z\in \mathcal{Z}}$, $(\cap _{n\in \mathbb{N}}\mathfrak{M}%
_{z}^{(n)})_{z\in \mathcal{Z}}$, $([\cup _{n\in \mathbb{N}}\mathfrak{M}%
_{z}^{(n)}]^{\prime \prime })_{z\in \mathcal{Z}}$ of von Neumann algebras
over $\mathcal{H}_{\mathcal{Z}}$ are also $\alpha _{\mathcal{Z}}$-measurable.
\end{lemma}

\begin{proof}
This is a consequence of the properties of the so-called Effros-Borel
structure. See \cite[Theorem 17.1]{Niesen-direct integrals}. We omit the
details.
\end{proof}

We now define direct integrals of von Neumann algebras via the direct
integral of Hilbert spaces (Definition \ref{Direct integrals of Hilbert
spaces}) and the subalgebra of decomposable operators (Definition \ref%
{Decomposable operators}):

\begin{definition}[Direct integrals of von Neumann algebras]
\label{lemma direct integral von neumann3}\mbox{ }\newline
Let $(\mathcal{Z},\mathfrak{F},\mu )$ be a $\sigma $-finite measure space, $%
\mathfrak{M}_{\mathcal{Z}}$ a $\alpha _{\mathcal{Z}}$-measurable field of
von Neumann algebras over a measurable family $\mathcal{H}_{\mathcal{Z}}$\
of separable Hilbert spaces. The direct integral 
\begin{equation*}
\mathfrak{M}_{\mathcal{Z}}^{\oplus }\equiv \int_{\mathcal{Z}}^{\alpha _{%
\mathcal{Z}}}\mathfrak{M}_{z}\mu (\mathrm{d}z)\subseteq M_{\mathcal{Z}%
}\subseteq \mathcal{B}(\mathcal{H}_{\mathcal{Z}}^{\oplus })
\end{equation*}%
of $\mathfrak{M}_{\mathcal{Z}}$ with respect to $\mu $ and $\alpha _{%
\mathcal{Z}}$ is the $\ast $-subalgebra of decomposable operators $A=\int_{%
\mathcal{Z}}^{\alpha _{\mathcal{Z}}}{A_{z}}\mu (\mathrm{d}z)$ for which $%
A_{z}\in \mathfrak{M}_{z}\subseteq \mathcal{B}(\mathcal{H}_{z})$ for all $%
\mu $-almost everywhere $z\in \mathcal{Z}$.
\end{definition}

Using this definition%
\begin{equation*}
\int_{\mathcal{Z}}^{\alpha _{\mathcal{Z}}}[\mathbb{C}\mathbf{1}_{\mathcal{H}%
_{z}}]\mu (\mathrm{d}z)=N_{\mathcal{Z}}\qquad \text{and}\qquad \int_{%
\mathcal{Z}}^{\alpha _{\mathcal{Z}}}\mathcal{B}(\mathcal{H}_{z})\mu (\mathrm{%
d}z)=M_{\mathcal{Z}}
\end{equation*}%
are nothing else as the von Neumann algebras $N_{\mathcal{Z}}$ and $M_{%
\mathcal{Z}}$ of diagonalizable and decomposable operators on $\mathcal{H}_{%
\mathcal{Z}}^{\oplus }$, respectively. See Definitions \ref{Decomposable
operators} and \ref{Diagonalizable operators}. Observe further that, for any 
$\alpha _{\mathcal{Z}}$-measurable field $\mathfrak{M}_{\mathcal{Z}}$ of von
Neumann algebras over $\mathcal{H}_{\mathcal{Z}}$,%
\begin{equation*}
\int_{\mathcal{Z}}^{\alpha _{\mathcal{Z}}}[\mathbb{C}\mathbf{1}_{\mathcal{H}%
_{z}}]\mu (\mathrm{d}z)\subseteq \mathfrak{M}_{\mathcal{Z}}^{\oplus }\cap %
\left[ \mathfrak{M}_{\mathcal{Z}}^{\oplus }\right] ^{\prime }\text{ },
\end{equation*}%
i.e., the algebra $N_{\mathcal{Z}}$ of diagonalizable operator is always a
subalgebra of the center of $\mathfrak{M}_{\mathcal{Z}}^{\oplus }$.

The next theorem gives a sufficient condition on direct integrals\ of von
Neumann algebras to be themselves von Neumann algebras:

\begin{theorem}[Direct integrals\ of von Neumann algebras as von Neumann
algebras]
\label{lemma direct integral von neumann4}\mbox{ }\newline
Let $(\mathcal{Z},\mathfrak{F},\mu )$ be a standard measure space, $%
\mathfrak{M}_{\mathcal{Z}}$ a $\alpha _{\mathcal{Z}}$-measurable field of
von Neumann algebras over a measurable family $\mathcal{H}_{\mathcal{Z}}$\
of separable Hilbert spaces. Then, $\mathfrak{M}_{\mathcal{Z}}^{\oplus }$ is
the von Neumann subalgebra of $\mathcal{B}(\mathcal{H}_{\mathcal{Z}}^{\oplus
})$ generated by $N_{\mathcal{Z}}$ and $\int_{\mathcal{Z}}^{\alpha _{%
\mathcal{Z}}}A_{z}^{(n)}\mu (\mathrm{d}z)$, $n\in \mathbb{N}$, where $%
(A^{(n)})_{n\in \mathbb{N}}$ is any $\alpha _{\mathcal{Z}}$-measurable
generating sequence for $\mathfrak{M}_{\mathcal{Z}}$.
\end{theorem}

\begin{proof}
Fix all parameters of the theorem. Pick, in particular, an arbitrary $\alpha
_{\mathcal{Z}}$-measurable generating sequence $(A^{(n)})_{n\in \mathbb{N}}$
for $\mathfrak{M}_{\mathcal{Z}}$. Denote by $\mathfrak{M}$ the von Neumann
subalgebra of $\mathcal{B}(\mathcal{H}_{\mathcal{Z}}^{\oplus })$ generated
by $N_{\mathcal{Z}}$ and $\int_{\mathcal{Z}}^{\alpha _{\mathcal{Z}%
}}A_{z}^{(n)}\mu (\mathrm{d}z)$, $n\in \mathbb{N}$. If $B\in \mathfrak{M}%
^{\prime }$ then there is a $\mu $-essentially bounded, $\alpha _{\mathcal{Z}%
}$-measurable operator field $(B_{z})_{z\in \mathcal{Z}}$ such that $B=\int_{%
\mathcal{Z}}^{\alpha _{\mathcal{Z}}}B{_{z}}\mu (\mathrm{d}z)\in M_{\mathcal{Z%
}}$, by Theorem \ref{Nielsen Th 7.1} (i), and, for $\mu $-almost everywhere $%
z\in \mathcal{Z}$, $B_{z}\in \mathfrak{M}_{z}^{\prime }$ by similar
arguments as in the proof of Theorem \ref{thm decompo repre}. Therefore, $%
\mathfrak{M}^{\prime }\subseteq \lbrack \mathfrak{M}_{\mathcal{Z}}^{\oplus
}]^{\prime }$, which in turn implies that $\mathfrak{M}_{\mathcal{Z}%
}^{\oplus }\subseteq \mathfrak{M}$.

To show the reverse inclusion, use the Kaplanski density theorem \cite[%
Theorem A.2]{Niesen-direct integrals} to approximate, in the strong
topology, any $B\in \mathfrak{M}$\ by a sequence $(B^{(n)})_{n\in \mathbb{N}%
}\subseteq \mathfrak{M}_{\mathcal{Z}}^{\oplus }$ of elements of the $\ast $%
-algebra generated by $N_{\mathcal{Z}}$ and $\int_{\mathcal{Z}}^{\alpha _{%
\mathcal{Z}}}A_{z}^{(n)}\mu (\mathrm{d}z)$, $n\in \mathbb{N}$. Observe that $%
\mathfrak{M}\subseteq M_{\mathcal{Z}}$, because $N_{\mathcal{Z}}\subseteq 
\mathfrak{M}^{\prime }$. See also Theorem \ref{Nielsen Th 7.1} (i). In other
words, $B,(B^{(n)})_{n\in \mathbb{N}}\in \mathfrak{M}$ are all decomposable
operators and, using $(B^{(n)})_{n\in \mathbb{N}}\subseteq \mathfrak{M}_{%
\mathcal{Z}}^{\oplus }$ and Theorem \ref{Nielsen Th 7.1} (iv), we arrive at $%
B\in \mathfrak{M}_{\mathcal{Z}}^{\oplus }$, i.e., $\mathfrak{M}\subseteq 
\mathfrak{M}_{\mathcal{Z}}^{\oplus }$.
\end{proof}

\begin{corollary}[Direct integrals of von Neumann algebras and fiber
inclusions]
\label{lemma direct integral von neumann5}\mbox{ }\newline
Let $(\mathcal{Z},\mathfrak{F},\mu )$ be a standard measure space, $%
\mathfrak{M}_{\mathcal{Z}}$, $\widetilde{\mathfrak{M}}_{\mathcal{Z}}$ two $%
\alpha _{\mathcal{Z}}$-measurable field of von Neumann algebras over the
same measurable family $\mathcal{H}_{\mathcal{Z}}$\ of separable Hilbert
spaces. Then, $\mathfrak{M}_{\mathcal{Z}}\subseteq \widetilde{\mathfrak{M}}_{%
\mathcal{Z}}$ iff $\mathfrak{M}_{z}\subseteq \widetilde{\mathfrak{M}}_{z}$
for all $\mu $-almost everywhere $z\in \mathcal{Z}$.
\end{corollary}

\begin{proof}
It is an obvious consequence of Theorem \ref{lemma direct integral von
neumann4}.
\end{proof}

Another consequence of Theorem \ref{lemma direct integral von neumann4}
concerns the difference between the von Neumann algebra constructed from a
direct integral representation and the direct integral of the von Neumann
algebras constructed from the fields of von Neumann algebras generated by
the corresponding fiber representation, as stated in Lemma \ref{lemma direct
integral von neumann1}. These von Neumann algebras are in general different.
Necessary and sufficient conditions to have equality are given in the
following corollary:\ 

\begin{corollary}[Direct integrals of von Neumann algebras and
representations]
\label{lemma direct integral von neumann6}\mbox{ }\newline
Let $\Pi $ be a representation of a separable unital Banach $\ast $-algebra $%
\mathcal{X}$\ on a (separable) Hilbert space $\mathcal{H}$ and $\mathfrak{N}$
an abelian von Neumann subalgebra of $[\Pi (\mathcal{X})]^{\prime }$. Let $%
\pi _{\mathcal{Z}}^{\oplus }$ be the direct integral decomposition of $\Pi $
with respect to $\mathfrak{N}$ (Corollary \ref{thm decompo repre2}). Then, 
\begin{equation*}
\left[ \pi _{\mathcal{Z}}^{\oplus }(\mathcal{X})\right] ^{\prime \prime
}=\int_{\mathcal{Z}}^{\alpha _{\mathcal{Z}}}\left[ \pi _{z}(\mathcal{X})%
\right] ^{\prime \prime }\mu (\mathrm{d}z)\qquad \text{iff}\qquad \mathfrak{N%
}\subseteq \left[ \Pi (\mathcal{X})\right] ^{\prime }\cap \left[ \Pi (%
\mathcal{X})\right] ^{\prime \prime }\ .
\end{equation*}
\end{corollary}

\begin{proof}
Let $\pi _{\mathcal{Z}}^{\oplus }$ be the direct integral decomposition of a
representation $\Pi $ of a separable unital Banach $\ast $-algebra $\mathcal{%
X}$ with respect to an abelian von Neumann subalgebra of $[\Pi (\mathcal{X}%
)]^{\prime }$, as stated in Definition \ref{def direct integral represe}.
See also Corollary \ref{thm decompo repre2}. Then, by Lemma \ref{lemma
direct integral von neumann1} and Theorem \ref{lemma direct integral von
neumann4}, the field $(\mathfrak{M}_{z}\doteq \lbrack \pi _{z}(\mathcal{X}%
)]^{\prime \prime })_{z\in \mathcal{Z}}$\ of von Neumann algebras is $\alpha
_{\mathcal{Z}}$-measurable and its direct integral $\mathfrak{M}_{\mathcal{Z}%
}^{\oplus }$ is a von Neumann algebra. As $\pi _{\mathcal{Z}}^{\oplus }(%
\mathcal{X})\subseteq \mathfrak{M}_{\mathcal{Z}}^{\oplus }$, we \emph{always}
have the natural inclusion%
\begin{equation}
\lbrack \pi _{\mathcal{Z}}^{\oplus }(\mathcal{X})]^{\prime \prime }\subseteq 
\mathfrak{M}_{\mathcal{Z}}^{\oplus }\equiv \int_{\mathcal{Z}}^{\alpha _{%
\mathcal{Z}}}\left[ \pi _{z}(\mathcal{X})\right] ^{\prime \prime }\mu (%
\mathrm{d}z)\ ,  \label{toto}
\end{equation}%
keeping in mind that $[\pi _{\mathcal{Z}}^{\oplus }(\mathcal{X})]^{\prime
\prime }$ is the smallest von Neumann algebra containing $\pi _{\mathcal{Z}%
}^{\oplus }(\mathcal{X})$. In fact, by Theorem \ref{lemma direct integral
von neumann4}, $\mathfrak{M}_{\mathcal{Z}}^{\oplus }$ is the von Neumann
subalgebra of $\mathcal{B}(\mathcal{H}_{\mathcal{Z}}^{\oplus })$ generated
by $N_{\mathcal{Z}}$ and the bicommutant $[\pi _{\mathcal{Z}}^{\oplus }(%
\mathcal{X})]^{\prime \prime }$. Since obviously $N_{\mathcal{Z}}\subseteq
\lbrack \pi _{\mathcal{Z}}^{\oplus }(\mathcal{X})]^{\prime }$, (\ref{toto})
holds true with equality iff $N_{\mathcal{Z}}$ is contained in the center of 
$[\pi _{\mathcal{Z}}^{\oplus }(\mathcal{X})]^{\prime }$, i.e., $\mathfrak{N}%
\subseteq \lbrack \Pi (\mathcal{X})]^{\prime }\cap \lbrack \Pi (\mathcal{X}%
)]^{\prime \prime }$, by Definition \ref{def direct integral represe}.
\end{proof}

Note that the bicommutant $[\pi _{\mathcal{Z}}^{\oplus }(\mathcal{X}%
)]^{\prime \prime }$ of Corollary \ref{lemma direct integral von neumann6}
does not necessarily include all diagonal operators and in this case, 
\begin{equation*}
\left[ \pi _{\mathcal{Z}}^{\oplus }(\mathcal{X})\right] ^{\prime \prime
}\varsubsetneq \int_{\mathcal{Z}}^{\alpha _{\mathcal{Z}}}\left[ \pi _{z}(%
\mathcal{X})\right] ^{\prime \prime }\mu (\mathrm{d}z)\ .
\end{equation*}%
See the argument just below (\ref{toto}) that uses Theorem \ref{lemma direct
integral von neumann4}. In any case, (\ref{toto}) always holds true and, for
instance, any separating vector $\Psi =\left( \Psi _{z}\right) _{z\in 
\mathcal{Z}}\in \mathcal{\tilde{H}}_{\mathcal{Z}}^{\oplus }\equiv \mathcal{H}%
_{\mathcal{Z}}^{\oplus }$ such that $\Psi _{z}$ is separating for $[\pi _{z}(%
\mathcal{X})]^{\prime \prime }$ for all $\mu $-almost everywhere $z\in 
\mathcal{Z}$ yields a separating vector $\Psi \equiv \lbrack \Psi ]$ for $%
[\pi _{\mathcal{Z}}^{\oplus }(\mathcal{X})]^{\prime \prime }$, i.e., $A\Psi
=0$ implies $A=0$ for all $A\in \lbrack \pi _{\mathcal{Z}}^{\oplus }(%
\mathcal{X})]^{\prime \prime }$.

By \cite[Theorem 13.3]{Niesen-direct integrals}, or the proof of Theorem \ref%
{central decompositions copy(1)}, the condition of Corollary \ref{lemma
direct integral von neumann6}, that is, 
\begin{equation*}
\mathfrak{N}\subseteq \lbrack \Pi (\mathcal{X})]^{\prime }\cap \lbrack \Pi (%
\mathcal{X})]^{\prime \prime },
\end{equation*}%
is directly related with the fact that the fiber representations are
mutually disjoint, or equivalently, for some measurable subset $\mathcal{Z}%
_{0}\subseteq \mathcal{Z}$ with $\mu (\mathcal{Z}_{0})=0$\ and all $%
z_{1},z_{2}\in \mathcal{Z}\backslash \mathcal{Z}_{0}$, there exists no $\ast 
$-isomorphism $\kappa :[\pi _{z_{1}}(\mathcal{X})]^{\prime \prime }\mapsto
\lbrack \pi _{z_{2}}(\mathcal{X})]^{\prime \prime }$ such that $\kappa \circ
\pi _{z_{1}}=\pi _{z_{2}}$ whenever $z_{1}\neq z_{2}$. In other words, to
have the equality%
\begin{equation*}
\left[ \pi _{\mathcal{Z}}^{\oplus }(\mathcal{X})\right] ^{\prime \prime
}=\int_{\mathcal{Z}}^{\alpha _{\mathcal{Z}}}\left[ \pi _{z}(\mathcal{X})%
\right] ^{\prime \prime }\mu (\mathrm{d}z)\ ,
\end{equation*}%
it is necessary to have \emph{no} redundancy in the fiber representations,
as expected. From Theorem \ref{central decompositions copy(1)} and
discussions after Equation (\ref{ineauqlity sup}), we see that the natural
candidate for an abelian von Neumann subalgebra $\mathfrak{N}\subseteq
\lbrack \Pi (\mathcal{X})]^{\prime }$, with respect to which $[\Pi (\mathcal{%
X})]^{\prime \prime }$ is decomposed, is precisely the center $\left[ \Pi (%
\mathcal{X})\right] ^{\prime }\cap \left[ \Pi (\mathcal{X})\right] ^{\prime
\prime }$.

By Theorem \ref{Nielsen Th 7.1} (i), note that $\int_{\mathcal{Z}}^{\alpha _{%
\mathcal{Z}}}\mathcal{B}(\mathcal{H}_{z})\mu (\mathrm{d}z)$ is the commutant
of $\int_{\mathcal{Z}}^{\alpha _{\mathcal{Z}}}[\mathbb{C}\mathbf{1}_{%
\mathcal{H}_{z}}]\mu (\mathrm{d}z)$. Thus, if a von Neumann algebra $%
\mathfrak{M}$ over $\mathcal{H}_{\mathcal{Z}}^{\oplus }$ is a direct
integral of a $\alpha _{\mathcal{Z}}$-measurable field $\mathfrak{M}_{%
\mathcal{Z}}$ of von Neumann algebras over $\mathcal{H}_{\mathcal{Z}}$, then 
$N_{\mathcal{Z}}\subseteq \mathfrak{M}\subseteq N_{\mathcal{Z}}^{\prime }$
and hence, the center of $\mathfrak{M}$ contains all diagonalizable
operators. By Corollary \ref{lemma direct integral von neumann6}, the
converse should clearly be true and this refers to \cite[Theorem 19.4]%
{Niesen-direct integrals} and the existence of a direct integral
decomposition of a von Neumann algebra via an abelian von Neumann algebra,
similar to Definitions \ref{Hilbert decomposition} and \ref{def direct
integral represe}:

\begin{definition}[Direct integral decomposition of von Neumann algebras]
\label{def direct integral represe copy(1)}\mbox{ }\newline
Let $\mathfrak{M}$ be a von Neumann algebra on a separable Hilbert space $%
\mathcal{H}$, $\mathfrak{N}$ an abelian von Neumann subalgebra of $\mathfrak{%
M}^{\prime }$, $(\mathcal{Z},\mathfrak{F},\mu )$ a standard measure space
and $\mathfrak{M}_{\mathcal{Z}}$ a $\alpha _{\mathcal{Z}}$-measurable\ field
of von Neumann algebras over a measurable family $\mathcal{H}_{\mathcal{Z}}$%
\ of separable Hilbert spaces. $\mathfrak{M}_{\mathcal{Z}}^{\oplus }$ is a
direct integral decomposition of $\mathfrak{M}$ with respect to $\mathfrak{N}
$\ if there is a $\ast $-isomorphism $\phi :L^{\infty }(\mathcal{Z},\mu
)\rightarrow \mathfrak{N}$ and a unitary mapping $\mathrm{U}:\mathcal{%
\mathcal{H}}\rightarrow \mathcal{H}_{\mathcal{Z}}^{\oplus }$ such that, 
\begin{equation*}
\mathfrak{M}_{\mathcal{Z}}^{\oplus }=\mathrm{U}\mathfrak{M}\mathrm{U}^{\ast
}\ ,\quad N_{\mathcal{Z}}=\mathrm{U}\mathfrak{N}\mathrm{U}^{\ast }\ ,\quad 
\mathrm{U}\phi (f)\mathrm{U}^{\ast }=\int_{\mathcal{Z}_{\mathfrak{N}%
}}^{\alpha _{\mathcal{Z}_{\mathfrak{N}}}}f(z)\mathbf{1}_{\mathcal{H}_{z}}\mu
(\mathrm{d}z)\ ,\quad f\in L^{\infty }(\mathcal{Z},\mu )\text{ }.
\end{equation*}%
In this case, we say that $\mathfrak{M}$ is decomposable with respect to $%
\mathfrak{N}$.
\end{definition}

\noindent Recall that (i) any separable Hilbert has a decomposition with
respect to any abelian von Neumann algebra on it and (ii) such a
decomposition is related to a standard measure space, by Theorem \ref%
{universality}. The following result, which is similar to Theorem \ref%
{universality} and Corollary \ref{thm decompo repre2}, ensures the existence
of direct integral decompositions of von Neumann algebras:

\begin{theorem}[Direct integral decomposition of von Neumann algebras]
\label{lemma direct integral von neumann4 copy(1)}\mbox{ }\newline
Let $\mathfrak{M}$ be a von Neumann algebra acting on a separable Hilbert
space $\mathcal{H}$ and $\mathfrak{N}$ an abelian von Neumann subalgebra of $%
\mathfrak{M}^{\prime }$.Then, $\mathfrak{M}$ is decomposable with respect
to\ $\mathfrak{N}$ iff $\mathfrak{N}\subseteq \mathfrak{M}^{\prime }\cap 
\mathfrak{M}^{\prime \prime }$.
\end{theorem}

\begin{proof}
If $\mathfrak{M}$ is decomposable with respect to\ $\mathfrak{N}$ then we
must have $\mathfrak{N}\subseteq \mathfrak{M}^{\prime }\cap \mathfrak{M}%
^{\prime \prime }$. See discussion before Definition \ref{def direct
integral represe copy(1)}. Conversely, any von Neumann algebra $\mathfrak{M}$
on a separable Hilbert space $\mathcal{H}$ is the strong closure of a
norm-separable $C^{\ast }$-algebra $\mathcal{X}$: To see this, note that the
unit closed ball of any von Neumann algebra $\mathfrak{M}$ on $\mathcal{H}$
is compact with respect to the weak-operator topology. Therefore, if the
Hilbert space $\mathcal{H}$ is separable, the weak-operator topology is
metrizable on any ball of $\mathfrak{M}$, which is thus separable in this
topology. In particular, a von Neumann algebra $\mathfrak{M}$ on a separable
Hilbert space is separable with respect to the weak-operator topology. Thus
take any weak-operator-dense countable subset $\mathcal{X}_{0}\subseteq 
\mathfrak{M}$ and let $\mathcal{X}$ be the separable $C^{\ast }$-algebra
generated by $\mathcal{X}_{0}$. Clearly, $\mathcal{X}^{\prime \prime }=%
\mathfrak{M}$. Define the representation $\Pi $ on $\mathcal{X}$ to be the
identity mapping. Let $\pi _{\mathcal{Z}}^{\oplus }$ be the direct integral
decomposition of $\Pi $ with respect to $\mathfrak{N}$ (Corollary \ref{thm
decompo repre2}). The assertion then follows from Corollary \ref{lemma
direct integral von neumann6} as $[\Pi (\mathcal{X})]^{\prime \prime }=%
\mathfrak{M}$. See also Theorem \ref{universality} for the existence of the $%
\ast $-isomorphism $\phi :L^{\infty }(\mathcal{Z},\mu )\rightarrow \mathfrak{%
N}$.
\end{proof}

Hence, in the context of standard spaces, the theory of direct integrals of
fields of von Neumann algebras corresponds to the study of von Neumann
subalgebras of the algebra of decomposable operators whose center contains
the diagonalizable ones.

When $\mathfrak{N}=\mathfrak{M}^{\prime }\cap \mathfrak{M}^{\prime \prime }$%
, one talks about the \emph{central} decomposition of $\mathfrak{M}$,
similar to Definition \ref{Special direct integral representations} (ii.1).
One can also talk about \emph{factor} decompositions of von Neumann
algebras, similar to Definition \ref{Special direct integral representations}
(ii.2).\bigskip

\noindent \textit{Acknowledgments:} This work is supported by CNPq
(308337/2017-4), FAPESP (2017/22340-9), as well as by the Basque Government
through the grant IT641-13 and the BERC 2018-2021 program, and by the
Spanish Ministry of Science, Innovation and Universities: BCAM Severo Ochoa
accreditation SEV-2017-0718, MTM2017-82160-C2-2-P.

\end{document}